\documentclass[graybox,envcountchap,envcountsame,sectrefs]{svmono}
\usepackage[T1]{fontenc}
\usepackage[utf8]{inputenc}
\usepackage{longtable}
\usepackage{type1cm}
\usepackage{makeidx}         
\usepackage{graphicx}
\usepackage{multicol}        
\usepackage[bottom]{footmisc}
\usepackage{etoolbox}
\undef{\warning}
\usepackage{fourier}
\usepackage[american,provide=*]{babel}
\usepackage{theorem}
\usepackage{index}
\usepackage{color,xcolor}
\usepackage[colorlinks=true,hyperindex=true]{hyperref}
\colorlet{darkblue}{blue!50!black}
\hypersetup{
    colorlinks,%
    citecolor=darkblue,%
    filecolor=red,%
    linkcolor=darkblue,%
    urlcolor=darkblue,%
    pdfnewwindow=true,%
    pdfstartview={FitH}
}
\usepackage{fancyhdr}
\usepackage{amsmath}
\usepackage{bbm} 
\usepackage{bm} 
\usepackage{amssymb}
\usepackage{amsfonts}
\usepackage{mathtools}
\usepackage{enumerate}
\usepackage{enumitem}
\usepackage{xspace}
\usepackage{xparse}
\usepackage{xargs}
\usepackage{needspace}
\usepackage[prependcaption]{todonotes}
\usepackage[style=cap-alpha,texencoding=UTF-8,refsection=chapter]{biblatex}
\def\bep{\begin{proposition}}
\def\eep{\end{proposition}}
\def\bel{\begin{lemma}}
\def\eel{\end{lemma}}
\def\bet{\begin{theorem}}
\def\eet{\end{theorem}}
\def\bed{\begin{definition}}
\def\eed{\end{definition}}
\def\bec{\begin{corollary}}
\def\eec{\end{corollary}}
\def\bex{\begin{example}}
\def\eex{\end{example}}
\newcommand{\beq}{\begin{equation}}
\newcommand{\eeq}{\end{equation}}
\newcommand{\bear}[1]{\begin{array}{#1}}
\newcommand{\ear}{\end{array}}
\newcommand{\ben}{\begin{enumerate}[label=\upshape{\textbf{(\arabic*)}}]}
\newcommand{\een}{\end{enumerate}}
\renewcommand{\part}[1]{\Needspace{5\baselineskip}\medskip\noindent\textbf{#1}\;}
\newcommand{\mypar}[1]{\medskip\noindent{\bf #1.\;}}
\RenewDocumentCommand{\remark}{o}{%
\IfNoValueTF{#1}%
{\medskip\noindent\textbf{Remark. }}%
{\medskip\noindent\textbf{Remark #1.}}}
\RenewDocumentCommand{\proof}{o}{%
\IfNoValueTF{#1}%
{\medskip\noindent\textbf{Proof. }}%
{\medskip\noindent\textbf{Proof #1.}}}
\newcommand{\ds}{\displaystyle}

\let\oldcaption\caption
\renewcommand{\caption}[1]{\oldcaption{\normalsize #1}}
\AtBeginEnvironment{theorem}{\Needspace{5\baselineskip}}
\AtBeginEnvironment{proposition}{\Needspace{3\baselineskip}}
\AtBeginEnvironment{definition}{\Needspace{5\baselineskip}}
\AtBeginEnvironment{corollary}{\Needspace{5\baselineskip}}
\AtBeginEnvironment{lemma}{\Needspace{5\baselineskip}}
\AtBeginEnvironment{quote}{\Needspace{5\baselineskip}}
\renewcommand{\d}{\mathrm{d}}  \newcommand{\e}{\mathrm{e}}  \renewcommand{\i}{\mathrm{i}}
\renewcommand{\qed}{\hfill$\Box$}
\newcommand{\bra}{\langle}
\newcommand{\ket}{\rangle}
\newcommand{\Id}{\operatorname{Id}}
\renewcommand{\Re}{\operatorname{Re}}

\newcommand{\Dom}{\operatorname{Dom}}
\newcommand{\Fix}{\operatorname{Fix}}

\newcommand{\dist}{\operatorname{dist}}
\newcommand{\supp}{\operatorname{supp}}
\newcommand{\tr}{\operatorname{tr}}

\newcommand{\wlim}{\operatorname*{w-lim}}
    
  \newcommand{\bK}{\mathbf{K}}

\newcommand{\bF}{\mathbf{F}}

\newcommand{\bI}{\mathbf{I}}  

\newcommand{\bsI}{{\boldsymbol{I}}}  
\newcommand{\bsa}{{\boldsymbol{a}}}    \newcommand{\bss}{{\boldsymbol{s}}}
\newcommand{\bsb}{{\boldsymbol{b}}}    \newcommand{\bst}{{\boldsymbol{t}}}
    \newcommand{\bsu}{{\boldsymbol{u}}}
    
  \newcommand{\bsn}{{\boldsymbol{n}}}  
    \newcommand{\bsx}{{\boldsymbol{x}}}

  \newcommand{\bsr}{{\boldsymbol{r}}}
\newcommand{\cA}{\mathcal{A}}  \newcommand{\cJ}{\mathcal{J}}  \newcommand{\cS}{\mathcal{S}}
\newcommand{\cB}{\mathcal{B}}    \newcommand{\cT}{\mathcal{T}}
\newcommand{\cC}{\mathcal{C}}  \newcommand{\cL}{\mathcal{L}}  
\newcommand{\cD}{\mathcal{D}}    \newcommand{\cV}{\mathcal{V}}
\newcommand{\cE}{\mathcal{E}}    
\newcommand{\cF}{\mathcal{F}}  \newcommand{\cO}{\mathcal{O}}  \newcommand{\cX}{\mathcal{X}}
\newcommand{\cG}{\mathcal{G}}  \newcommand{\cP}{\mathcal{P}}  
    \newcommand{\cZ}{\mathcal{Z}}
  \newcommand{\cR}{\mathcal{R}}
  \newcommand{\fJ}{\mathfrak{J}}  \newcommand{\fS}{\mathfrak{S}}

	\newcommand{\fM}{\mathfrak{M}}	
		
\newcommand{\fF}{\mathfrak{F}}		\newcommand{\fX}{\mathfrak{X}}
	\newcommand{\fP}{\mathfrak{P}}	\newcommand{\fY}{\mathfrak{Y}}
		
	\newcommand{\fR}{\mathfrak{R}}
    \newcommand{\fs}{\mathfrak{s}}

                                \newcommand{\fp}{\mathfrak{p}}	
\newcommand{\fh}{\mathfrak{h}}		
                          	    
\newcommand{\BB}{\mathbb{B}}     
\newcommand{\CC}{\mathbb{C}}     
   \newcommand{\MM}{\mathbb{M}}  \newcommand{\VV}{\mathbb{V}}
\newcommand{\EE}{\mathbb{E}}   \newcommand{\NN}{\mathbb{N}}  
     
   \newcommand{\PP}{\mathbb{P}}  
   \newcommand{\QQ}{\mathbb{Q}}  \newcommand{\ZZ}{\mathbb{Z}}
\newcommand{\II}{\mathbb{I}}   \newcommand{\RR}{\mathbb{R}}
\newcommand{\Ent}{\mathop{\mathrm{Ent}}\nolimits}

\newcommand{\ep}{\mathop{\mathrm{ep}}\nolimits}
\newcommand{\epr}{\mathrm{epr}}
\newcommand{\wP}{\widehat{\mathbb{P}}}
\newcommand{\wE}{\widehat{\mathbb{E}}}
\newcommand{\st}{\mathrm{st}}
\newcommand{\Orb}{\mathop{\mathrm{Orb}}\nolimits}
\newcommand{\argdot}{{\bm\cdot}}

\newcommand{\ndex}[1]{{\sl #1}\index{#1}}
\newcommand{\hC}{\hyperlink{hyp.C}{\textbf{(C)}}\xspace}
\newcommand{\hH}{\hyperlink{hyp.H}{\textbf{(H)}}\xspace}
\newcommand{\hWPS}{\hyperlink{hyp.WPS}{\textbf{(WPS)}}\xspace}
\newcommand{\hS}{\hyperlink{hyp.S}{\textbf{(S)}}\xspace}
\newcommand{\hUSCE}{\hyperlink{hyp.USCE}{\textbf{(USCE)}}\xspace}
\newcommand{\hED}{\hyperlink{hyp.ED}{\textbf{(ED)}}\xspace}
\newcommand{\hPAP}{\hyperlink{hyp.PAP}{\textbf{(PAP)}}\xspace}
\makeindex

\begin{document}
\pagenumbering{Alph}
\author{No\'e Cuneo\\
Vojkan Jak\v si\'c\\
Claude-Alain Pillet\\
Armen Shirikyan}
\title{What is a Fluctuation Theorem?}
\def\today{}
\maketitle
\frontmatter
\mathtoolsset{mathic}

\begin{dedication}
Dedicated to David Ruelle, on the occasion of his 90th birthday
\end{dedication}


\preface

Fluctuation Relations (FRs for short) describe statistical features of the
fluctuations of some physical quantities in systems subject to a discrete
symmetry. They provide quantitative information on the spontaneous breakdown of
this symmetry. FRs were first discovered in numerical experiments on shear
flows, in the early 1990s, by D.J.~Evans, E.G.D.~Cohen and
G.P.~Morris.\footnote{Phys. Rev. Lett. 71, 2401--4 (1993).} There, the quantity
of interest was entropy production and the symmetry was time-reversal. Its
breakdown led to the emergence of the arrow of time, and the FRs gave
quantitative estimates on the probability of violation of the second law of
thermodynamics. D.J.~Evans and D.J.~Searles\footnote{Phys. Rev. E 50, 1645--1648
(1994).} subsequently showed that, on finite time intervals, the rate of entropy
production of time-reversal invariant systems relaxing to a nonequilibrium
steady state satisfies a FR (such FRs are now called transient). The last and
crucial pioneering step was achieved by G.~Gallavotti and
E.G.D.~Cohen\footnote{Phys. Rev. Lett. 74, 2694--7 (1995).} who proved that,
under an appropriate chaotic hypothesis, FRs also hold in the steady states of
time-reversal invariant systems. These early discoveries and developments
generated an enormous body of numerical, theoretical and experimental works in
the physics literature which have fundamentally modified our understanding of
non-equilibrium statistical physics, with applications extending to chemistry
and biology.

This monograph, addressed to mathematical physicists,  is an introduction to FRs
and the associated Fluctuation Theorems in the context of statistical mechanics,
stochastic dynamics and dynamical systems.

\vspace{\baselineskip}
\begin{flushright}\noindent
Paris, Milan, Toulon, Cergy\hfill {\it Noé Cuneo}\\
April 2025\hfill {\it Vojkan Jak\v si\'c}\\
\hfill {\it Claude-Alain Pillet}\\
\hfill {\it Armen Shirikyan}
\end{flushright}

\newpage
\thispagestyle{empty}
\noindent
Noé Cuneo\\
Université Paris Cité and Sorbonne Université, CNRS,\\
Laboratoire de Probabilités, Statistique et Modélisation,\\
75013 Paris, France

\medskip\noindent
Vojkan Jak\v si\'c\\
Dipartimento di Matematica\\
Politecnico di Milano\\
Piazza Leonardo da Vinci, 32\\
20133 Milan, Italy

\medskip\noindent
Claude-Alain Pillet\\
Centre de Physique Théorique\\
Université de Toulon --- CNRS --- Aix-Marseille Université\\
83957 La Garde, France

\medskip\noindent
Armen Shirikyan\\
Département de Mathématiques\\
CY Cergy Paris Université --- CNRS\\
2, avenue Adolphe Chauvin\\
95302 Cergy-Pontoise, France

\extrachap{Acknowledgements}

We are grateful to R.~Chetrite and H.~Comman for providing us with useful
references, and to L.~Bruneau and C.-E.~Pfister for helpful comments.

\smallskip 
The work of  AS,  CAP, and VJ was partly funded by the CY Initiative grant {\sl
Investissements d'Avenir}, grant number ANR-16-IDEX-0008. We also acknowledge
the support of the ANR project {\sl DYNACQUS}, grant number ANR-24-CE40-5714. VJ
acknowledges the support of NSERC, Canada and the support of the MUR grant {\em Dipartimento di Eccellenza 2023-2027} of Dipartimento di Matematica, Politecnico di Milano.

\smallskip 

This document is the preprint version of Volume 54 of the {\em SpringerBriefs in Mathematical Physics} collection. The final authenticated version is available online at: \url{https://doi.org/10.1007/978-3-032-02095-6}
\tableofcontents
\mainmatter
\pagenumbering{arabic}
\chapter{Introduction}
\label{chap:Introduction}

\section{Historical Overview}

Despite being random, fluctuations in physical systems do carry useful
information. One can hardly overestimate the impact of the earliest and probably
best known application of this fact: the determination of Avogadro's constant
$N_A$ by Perrin~\cite{Perrin1908,Perrin1909b} on the basis of a statistical
analysis of \ndex{Brownian motion}.\footnote{The interested reader should
consult~\cite[Chapter~8]{Gallavotti1999a} for a comprehensive account on the
theory of Brownian motion and its history.}\;The theoretical foundation of
Perrin's experiment was the Einstein--Smoluchowski
relation~\cite{Einstein1905a,Smoluchowski1906}, linking the mobility $\mu$ of a
tracer particle immersed in a fluid to the diffusion constant $D$,
\begin{equation}
D=\mu k_BT,\qquad k_B=\frac{R}{N_A},
\label{eq:Einstein}
\end{equation}
where $T$ is the equilibrium temperature of the fluid, $k_B$ is the  Boltzmann
constant, and $R$ is the ideal gas constant. In the modern language of
statistical mechanics, this is a \ndex{fluctuation--dissipation relation} which,
as the name indicates, relates the fluctuations of some observable, typically
encoded in a correlation function, to some quantity linked to the dissipative
properties of the system. In the case of Brownian motion, the left-hand side
of~\eqref{eq:Einstein} is related to the fluctuations of the velocity $\dot X$
of the tracer by\footnote{Here, $\langle\,\argdot\,\rangle$
denotes the ensemble
average. The factor of $3$ is the dimension of the configuration space $\RR^3$,
and "$\,\cdot\,$" denotes the Euclidean inner product.}
$$
3D=\int_0^\infty\langle\dot X_0\cdot \dot X_t\rangle\d t.
$$
The fluid's viscosity is responsible for the dissipative drag force,
$F=-\mu^{-1}\dot X_t$, opposed to the motion of the tracer. For a spherical
tracer of radius $r$, the mobility appearing on the right-hand side
of~\eqref{eq:Einstein} is related to the dynamic viscosity $\eta$ of the fluid
by Stokes' formula
$$
\mu=\frac{1}{6\pi\eta r}.
$$
We also mention that the velocity power spectrum $S$,
$$
S(\omega)\equiv\lim_{T\to\infty}\frac1T\langle|\hat V_T(\omega)|^2\rangle,\qquad
\hat V_T(\omega)=\int_{0}^T\dot X_t\e^{-\i\omega t}\d t,
$$
is linked to the  correlation function by the Wiener--Khinchin relation~\cite{Khintchine1934},
$$
\langle\dot X_0\cdot \dot X_t\rangle=\int_{-\infty}^\infty S(\omega)\e^{\i\omega t}\frac{\d\omega}{2\pi}.
$$
Since the instantaneous power dissipated by the Stokes force is given by the
quadratic expression $P_t=\mu^{-1}|\dot X_t|^2$, its steady averaged value is
related to the power spectrum as
$$
\langle P_0\rangle=\frac1\mu\int_{-\infty}^\infty S(\omega)\frac{\d\omega}{2\pi}.
$$

\medskip
The \ndex{Johnson–Nyquist noise}~\cite{Johnson1928,Nyquist1928} --- the electrical
counterpart of Brownian motion --- was discovered a century after Brown's
observations~\cite{Brown1828}. The electromotive force $U$ induced by the
thermally generated motion of charge carriers in a resistor of resistance $R$
kept at a constant temperature $T$ has a power spectrum given by the celebrated
Nyquist formula~\cite[Section~XI.2]{Kampen1981}
$$
S(\omega)=2k_BTR.
$$
Invoking the Wiener--Khinchin relation yields the fluctuation--dissipation relation
$$
\int_0^\infty\langle U_0U_t\rangle\d t=2Rk_BT.
$$

In nonequilibrium thermodynamics, \ndex{linear response theory} describes the reaction
of a system driven out of equilibrium by weak external forces~\cite[Chapters~III-IV]{Groot1962}.
Distinctive features of this reaction are the induced currents (particles species,
heat, charges) $J_i$ and the associated \ndex{entropy production rate}
$$
\sigma=\sum_iJ_iX_i,
$$
where $X_i$ is the thermodynamic force, or affinity, conjugated to the current
$J_i$. In the linear response regime, the currents are linear functions of the
forces
$$
J_j=\sum_iL_{ji}X_i.
$$
The $L_{ji}$ are called transport coefficients, and the second law of thermodynamics
requires
$$
\sigma=\sum_{i,j}L_{ij}X_iX_j\ge0.
$$
In other words,  the symmetric part of the matrix $L$ is non-negative: $L+L^T\ge0$.

A few years after the aforementioned works of Johnson and Nyquist, Onsager
developed an approach to the transport coefficients based on the general theory
of fluctuations laid down by Einstein~\cite{Einstein1910} on the basis of
Boltzmann's entropy formula. The main result of these
investigations~\cite{Onsager1931a,Onsager1931b} is a general formulation of the
so-called \ndex{reciprocity relations}, the first instances of which were discovered by
Thomson (Lord Kelvin), Helmholtz and others: if the system satisfies microscopic
reversibility, then the matrix $L$ is symmetric, $L=L^T\ge0$.

Two decades later, invoking quantum mechanics, Callen and Welton derived a
general form of the Nyquist formula~\cite{Callen1951}, including a low
temperature quantum correction, while Green and Kubo~\cite{Green1952,Kubo1957}
derived a general relation between transport coefficients and current-current
correlation functions, since then called the {\sl Green--Kubo formula},\index{Green--Kubo formula}
\begin{equation}
L_{ij}=\frac12\int_{-\infty}^\infty\langle\Phi_i(0)\Phi_j(t)\rangle_{\mathrm{eq}}\d t,
\label{equ:GKFDR}
\end{equation}
where $\Phi_j(t)$ denotes the instantaneous current observable at time $t$ and
$\langle\,\argdot\,\rangle_\mathrm{eq}$ the thermal equilibrium state.

To the best of our knowledge, Bochkov and Kuzovlev were the first to investigate
the formal consequences of microscopic reversibility on the statistical
properties of fluctuations {\sl away from thermal
equilibrium}~\cite{Bochkov1977}. They consider a classical system governed by
the Hamiltonian $H_0$, assumed to be invariant under a time-reversal map, {\sl
i.e.,} $H_0=H_0\circ\theta$ for some anti-symplectic involution $\theta$ of the
phase space. The system can be acted upon through external ``forces'' $X_i$,
with the associated Hamiltonian
$$
H(X)=H_0-\sum_iX_iQ_i,
$$
where  $Q_i$'s are  phase functions that are either even or odd under time-reversal,
$$
Q_i\circ\theta=\theta_iQ_i,\qquad \theta_i=\pm1.
$$
Given a steady forcing $X$, the state of thermal equilibrium at temperature $T$
is described by  the Gibbs probability measure
$$
\d\mu_{X}=\e^{-\beta(H(X)-F(X))}\d\gamma,\qquad F(X)=-\frac1\beta\log\int\e^{-\beta H(X)}\d\gamma,
$$
where $\beta=1/k_BT$, $\gamma$ denotes Liouville's measure on phase space, and $F(X)$
the Helmholtz free energy.

The system, initially in the state $\mu_0$, is driven out of thermal equilibrium
by a time-dependent forcing $X:[0,\tau]\ni t\mapsto X(t)$. We denote by
$[0,\tau]\ni t\mapsto\Phi^t(\,\argdot\,|X)$ the flow generated by the
time-dependent Hamiltonian $[0,\tau]\ni t\mapsto H(X(t))$, and by
$$
\Phi_X:x\mapsto\left(\Phi^t(x|X)\right)_{t\in[0,\tau]}
$$
the induced map between initial conditions and solutions of the Hamilton equation
of motion.

Let $\PP_X=\mu_0\circ\Phi_X^{-1}$ be the probability measure induced by $\mu_0$
on the set of phase space paths $(x_t)_{t\in[0,\tau]}$, and denote by
$\Theta:(x_t)_{t\in[0,\tau]}\mapsto(\theta x_{\tau-t})_{t\in[0,\tau]}$ the
time-reversal map. Bochkov and Kuzovlev derived the relation\footnote{Strictly
speaking, their main result, \cite[Equ.~(7)]{Bochkov1977}, is slightly less
general. See the appendix to this introduction for a derivation
of~\eqref{equ:BKFR} and comparison with other fluctuation relations.}\index{FR!Bochkov--Kuzovlev}
\begin{equation}
\frac{\d\PP_X}{\d\PP_X\circ\Theta}=\e^{S_\tau},
\label{equ:BKFR}
\end{equation}
where
$$
S_\tau=\sum_i\int_0^\tau J_i(t)X_i(t)\d t=\beta W_\tau^\mathrm{ex},
$$
$J_i(t)=\beta\partial_tQ_i\circ\Phi^t$ is the current conjugated to $X_i$, and
$W_\tau^\mathrm{ex}=H_0\circ\Phi^\tau-H_0$ is the increase of the internal
energy $H_0$ during the period $[0,\tau]$. For later reference, we  mention
that $W_\tau^\mathrm{ex}$ is sometimes called the {\sl exclusive work}\index{work!exclusive} performed
on the system by the external forces. The term ``exclusive'' is motivated by the
fact that the free Hamiltonian $H_0$, and not the full one $H(X(t))$ is used in
the energy bookkeeping. The {\sl inclusive work}\index{work!inclusive}
$$
W_\tau^\mathrm{in}=H(X(\tau))\circ\Phi^\tau-H(X(0))
=-\sum_i\int_0^\tau Q_i(t)\dot X_i(t)\d t
$$
will enter the scene below. Rewriting~\eqref{equ:BKFR} as
$\e^{-S_\tau}\d\PP_X=\d\PP_X\circ\Theta$ leads to
\begin{equation}
\int\e^{-\beta W_\tau^\mathrm{ex}}\d\PP_X=1.
\label{equ:BKIFR}
\end{equation}
For an obvious reason, such an identity is often referred to as an {\sl integral
fluctuation relation,} to distinguish it from a more {\sl detailed fluctuation
relation} of the form~\eqref{equ:BKFR}.\index{FR!integral/detailed}

From~\eqref{equ:BKFR}, Bochkov and Kuzovlev further deduced the existence of an
infinite hierarchy of fluctuation--dissipation relations, the Green-Kubo
formula~\eqref{equ:GKFDR} being just a limiting form of one of them. A simple
extension of~\eqref{equ:BKFR} to more general initial
states~\cite{Bochkov1979,Bochkov1981a} allowed them to deal also with thermally
forced systems and to derive fluctuation relations for nonequilibrium steady
states of open systems.

A decade later, while investigating microscopic violations of the $2^\text{nd}$
Law in numerical simulations of sheared viscous fluids, Evans, Cohen and
Morriss~\cite{Evans1993} discovered a fluctuation relation formally similar
to~\eqref{equ:BKFR}. Their model consists of a many particles system driven by
non-conservative forces and coupled to a Gaussian thermostat to prevent heating
up and ensure relaxation toward a nonequilibrium steady state. As pointed out by
Gallavotti and Cohen~\cite{Cohen1999}, a precise statement of this new relation
writes
\begin{equation}
\lim_{\tau\to\infty}\frac1\tau\log\frac{\d P_\tau}{\d P_\tau\circ\vartheta}(s)=s,
\label{eq:GCFR}
\end{equation}
where $P_\tau$ denotes the law of the entropy production rate\footnote{In the
model considered, entropy production rate is simply the phase-space contraction
rate.} over the time period $[0,\tau]$ {\sl induced by the steady state} of the
system, and the map $\vartheta$ time-reversal at the entropic level, {\sl i.e.,}
$\vartheta(s)=-s$. Evans and Searles~\cite{Evans1994} derived a relation for
transient states of thermostated systems that  is  formally similar
to~\eqref{equ:BKFR}. In their version of this relation, the measure $\PP_\tau$
is induced by a time-reversal invariant initial state (e.g., Liouville's measure
or the Gibbs state $\mu_0$) and pertains to a trajectory segment of duration
$\tau$. The time-reversal $\Theta$ maps a trajectory with entropy production
$S_\tau$ to a trajectory with the reversed entropy production $-S_\tau$, and in
this case
\[
\frac1\tau\log\frac{\d P_\tau}{\d P_\tau\circ\vartheta}(s)=s
\]
holds for all $\tau$. Relations of this kind are often referred to as {\sl
Evans--Searles,} {\sl transient} or {\sl finite-time fluctuation relations,} due\index{FR!transient/finite-time}
to the fact that they hold on any time intervals $[0,\tau]$ of finite durations,
with a time-reversal invariant initial state, as opposed to the {\sl
Gallavotti--Cohen} or {\sl steady state fluctuation relations}~\eqref{eq:GCFR},\index{FR!steady-state}
which only hold in the limit $\tau\to\infty$, with a nonequilibrium steady state
typically breaking time-reversal symmetry. While transient fluctuation relations
usually follow from simple calculations, their steady state counterparts are
much deeper and rely on a precise analysis of the relaxation process leading to
the steady state or on  detailed properties of the latter. Gallavotti and
Cohen~\cite{Gallavotti1995b,Gallavotti1995c} proved that the steady state
fluctuation relation holds for systems with sufficiently chaotic dynamics, more
precisely for the SRB measure of Anosov dynamical systems.
Gallavotti~\cite{Gallavotti1995d} further showed how, in such cases, large
deviation estimates can be related to the fluctuation relation, giving rise to
what we shall call a {\sl fluctuation theorem.}

Unaware of the works of Bochkov and Kuzovlev, but working in the same
Hamiltonian setup, Jarzynski~\cite{Jarzynski1997a,Jarzynski2004,Jarzynski2007}
derived the integral fluctuation relation resulting from the replacement of the
exclusive work $W_\tau^\mathrm{ex}$ by the inclusive one $W_\tau^\mathrm{in}$
in~\eqref{equ:BKIFR}, namely\index{FR!Jarzynski--Crooks}\index{work!inclusive}
\begin{equation}
\int\e^{-\beta W_\tau^\mathrm{in}}\d\PP_X=\e^{-\beta(F(X(\tau))-F(X(0)))},
\label{equ:JarIntFR}
\end{equation}
where $\PP_X=\mu_{X(0)}\circ\Phi_X^{-1}$. What makes Jarzynski's fluctuation
relation particularly interesting is that its two sides have a radically
distinct nature. On the left-hand side, one has the ensemble average of a
quantity depending on the work done by a non-reversible process transforming the
equilibrium state $\mu_{X(0)}$ into the nonequilibrium state
$\mu_{X(0)}\circ(\Phi^{\tau})^{-1}$. On the right-hand side, the free energy
difference $\Delta F=F(X(\tau))-F(X(0))$ is the reversible work, {\sl i.e.,} the
work that has to be done on the system to bring it from the initial equilibrium
state $\mu_{X(0)}$ to the final one $\mu_{X(\tau)}$ by an arbitrary reversible
process. This connection between a nonequilibrium process and equilibrium
quantities provides an opportunity to measure the latter,
see~\cite{Jarzynski2008} and references therein for examples of such
applications.

Jarzynski, Crooks and Kurchan~\cite{Jarzynski1997b,Crooks1998,Kurchan1998}
extended the derivation of the above transient and integral fluctuation
relations from the deterministic Hamiltonian setting to stochastic dynamics,
while Lebowitz and Spohn developed  the Gallavotti--Cohen steady-state
fluctuation theorem~\cite{Lebowitz1999} in the Markovian setting.
Maes~\cite{Maes1999} interpreted fluctuation theorem within Gibbs formalism. A
detailed version of Jarzynski's integral relation~\eqref{equ:JarIntFR} was
obtained by Crooks~\cite{Crooks1999} in the stochastic case, and by
Jarzynski~\cite{Jarzynski2000} in the deterministic case. The results  of
Gallavotti--Cohen~\cite{Gallavotti1995b,Gallavotti1995c} were extended to
general chaotic dynamical systems on compact metric spaces in~\cite{Maes2003b}.

The difficulties involved with experimental studies of the fluctuation relations
explain, at least partly, why the works of Bochkov and Kuzovlev remained largely
unnoticed~\cite{Pitaevskii2011}. In this respect, two decades later, new
techniques became available, and the subject generated an enormous body of
theoretical, numerical and experimental works which have fundamentally altered
our understanding of non-equilibrium physics, with applications extending to
chemistry and biology. For reviews of these various developments we refer the
reader to~\cite{Ruelle1999,Evans2002,Jiang2004,Gaspard2005,Rondoni2007,Marconi2008,Jaksic2011}

\section{Appendix: Hamiltonian Fluctuation Relations}

We finish this introduction with a formal derivation of various forms of
detailed fluctuation relations for Hamiltonian dynamics. The setup is similar to
that of~\cite{Bochkov1977} and~\cite{Jarzynski2000}, and any necessary
regularity is assumed without notice.

We consider a system with phase space $\Gamma$, equipped with its symplectic
form, Liouville measure $\gamma$, and anti-symplectic time-reversal map
$\theta$. We assume that the system Hamiltonian $\Gamma\ni x\mapsto H(x|X)$
depends on some control parameters $X=(X_1,\ldots,X_n)\in\RR^n$, describing
adjustable conditions of the model, e.g., some applied voltage bias, the
coupling strengths between some of its subsytems, {\sl etc.} A protocol
is a scenario specifying how these controls vary in the course of time.

For a fixed time $\tau$ and given the protocol
\begin{equation}
X:[0,\tau]\to\RR^n,
\label{eq:Xdef}
\end{equation}
we denote by $[0,\tau]\ni t\mapsto\Phi^t(\,\argdot\,|X)$ the phase space flow
associated with the time-dependent Hamiltonian $[0,\tau]\ni t\mapsto
H(\,\argdot\,|X(t))$. We also assume the time-reversal covariance of the
Hamiltonian,
\begin{equation}
H(\theta(x)|X)=H(x|\vartheta(X)),
\label{eq:HthetaInv}
\end{equation}
where $\vartheta$ is a linear involution of $\RR^n$, and define the
time-reversal of the protocol~\eqref{eq:Xdef} by
$$
\widehat{X}:[0,\tau]\mapsto\vartheta(X(\tau-t)).
$$

For a given protocol $X$, $\Phi_X:x\mapsto(\Phi^t(x|X))_{t\in[0,\tau]}$ maps an
initial condition to its phase space paths, and
$\Theta:(x_t)_{t\in[0,\tau]}\mapsto(\theta(x_{\tau-t}))_{t\in[0,\tau]}$ maps
such a path to its time-reversal.

The following is the mother of all fluctuation relations concerning
deterministic Hamiltonian systems.

\begin{proposition}\label{prop:genHamFR}
Let
$$
\nu_X(\d x)\coloneq\e^{-S(x|X(0))}\gamma(\d x)
$$
be a probability measure on $\Gamma$ and\, $\PP_X=\nu_X\circ\Phi_X^{-1}$ the
induced path-space measure. Then, the time-reversed measure\,
$\PP_{\widehat{X}}\circ\Theta$ is absolutely continuous w.r.t.\;$\PP_X$, with
the Radon--Nikodym derivative
\begin{equation}
\bsx=(x_t)_{t\in[0,\tau]}\mapsto
\frac{\d\PP_{\widehat{X}}\circ\Theta}{\d\PP_X}(\bsx)
=\e^{-S(\theta(x_\tau)|\vartheta(X(\tau)))+S(x_0|X(0))},
\label{eq:genHamFR}
\end{equation}
{\sl i.e.,} for any path function $F$,
$$
\int F(\Theta(\bsx))\PP_{\widehat{X}}(\d\bsx)
=\int\e^{-S(\theta(x_\tau)|\vartheta(X(\tau)))+S(x_0|X(0))}F(\bsx)\PP_X(\d\bsx).
$$
\end{proposition}

\proof Defining the time-reversed Hamiltonian flow by
\begin{equation}
\widehat\Phi^t(x|X)=\theta\circ\Phi^{\tau-t}(\theta(x)|\widehat X),
\label{eq:hatPhiDef}
\end{equation}
one easily checks that
\begin{equation}
\Phi^t(x|X)=\widehat\Phi^t(\Phi^\tau(x|X)|X)
\label{eq:PhiPhi}
\end{equation}
holds for all $x\in\Gamma$ and $t\in[0,\tau]$ (both sides of this identity
satisfy Hamilton's equation of motion with the same final condition). Denoting\,
$\EE_X$ the expectation functional associated to\, $\PP_X$, for
$F(\bsx)=f(x_{t_1},\ldots,x_{t_n})$ one has
$$
\EE_{\widehat{X}}[F\circ\Theta]=\int_\Gamma f\left(\theta\circ\Phi^{\tau-t_1}(x|\widehat{X}),\ldots,\theta\circ\Phi^{\tau-t_n}(x|\widehat{X})\right)
\e^{-S(x|\widehat{X}(0))}\gamma(\d x),
$$
and invoking~\eqref{eq:hatPhiDef},
$$
\EE_{\widehat{X}}[F\circ\Theta]=\int_\Gamma
f\left(\widehat{\Phi}^{t_1}(\theta(x)|X),\ldots,\widehat{\Phi}^{t_n}(\theta(x)|X)\right)
\e^{-S(x|\widehat{X}(0))}\gamma(\d x).
$$
Using the fact that Liouville's measure is $\theta$-invariant, we can write
$$
\EE_{\widehat{X}}[F\circ\Theta]=\int_\Gamma
f\left(\widehat{\Phi}^{t_1}(x|X),\ldots,\widehat{\Phi}^{t_n}(x|X)\right)
\e^{-S(\theta(x)|\widehat{X}(0))}\gamma(\d x),
$$
and its invariance under the Hamiltonian flow further yields
$$
\EE_{\widehat{X}}[F\circ\Theta]=\int_\Gamma
f\left(\widehat{\Phi}^{t_1}(\Phi^\tau(x|X)|X),\ldots,\widehat{\Phi}^{t_n}(\Phi^\tau(x|X)|X)\right)
\e^{-S(\theta(\Phi^\tau(x|X))|\widehat{X}(0))}\gamma(\d x).
$$
In view of~\eqref{eq:PhiPhi}, we finally get
\begin{align*}
\EE_{\widehat{X}}[F\circ\Theta]&=\int_\Gamma
f\left(\Phi^{t_1}(x|X),\ldots,\Phi^{t_n}(x|X)\right)
\e^{S(x|X(0))-S(\theta(\Phi^\tau(x|X))|\widehat{X}(0))}\nu_X(\d x)\\
&=\EE_X\left[F\e^{S(x_0|X(0))-S(\theta(x_\tau)|\vartheta(X(\tau)))}\right].
\end{align*}
\qed

\medskip
Setting
$$
F(X)\coloneq-\frac1\beta\log\int_\Gamma\e^{-\beta H(x|X)}\gamma(\d x),
$$
we now  discuss a few consequences of Proposition~\ref{prop:genHamFR}.

\part{The Bochkov--Kuzovlev FR.} It follows from~\eqref{eq:HthetaInv} that the\index{FR!Bochkov--Kuzovlev}
Hamiltonian $H_0=H(\,\argdot\,|0)$ is time-reversal invariant. With
$S(\,\argdot\,|X)=\beta(H_0-F(0))$, the initial measure $\nu_X=\nu_{\widehat{X}}$ is the
associated equilibrium state at inverse temperature $\beta$
and~\eqref{eq:genHamFR} gives the detailed Bochkov--Kuzovlev fluctuation
relation~\eqref{equ:BKFR},
$$
\frac{\d\PP_{\widehat{X}}\circ\Theta}{\d\PP_X}(\bsx)
=\e^{-\beta(H_0(x_\tau)-H_0(x_0))}=\e^{-\beta W_\tau^\mathrm{ex}}.
$$

\part{The Jarzynski--Crooks FR.} With $S(\,\argdot\,|X)=\beta(H(\,\argdot\,|X)-F(X))$,
the initial\index{FR!Jarzynski--Crooks} measures
$$
\nu_X(\d x)=\e^{-\beta(H(x|X(0))-F(X(0)))}\gamma(\d x),
$$
and
$$
\nu_{\widehat{X}}\circ\theta(\d x)=\e^{-\beta(H(x|X(\tau))-F(X(\tau)))}\gamma(\d x),
$$
are, respectively, the equilibrium measure at inverse temperature $\beta$ for
the initial and final Hamiltonians $H(\,\argdot\,|X(0))$ and
$H(\,\argdot\,|X(\tau))$. In this case, and in view of~\eqref{eq:HthetaInv}, the
relation~\eqref{eq:genHamFR} reads
$$
\frac{\d\PP_{\widehat{X}}\circ\Theta}{\d\PP_X}(\bsx)
=\e^{\beta(F(X(\tau))-F(X(0)))}\e^{-\beta(H(x_\tau|X(\tau))-H(x_0|X(0)))}
=\e^{\beta\Delta F}\e^{-\beta W_\tau^\mathrm{in}},
$$
which is the detailed form of the Jarzynski--Crooks fluctuation relation, from
which~\eqref{equ:JarIntFR} immediately follows.

\part{FR in open systems.} We consider an open system $\cZ$ coupled to thermal\index{FR!in open systems}
reservoirs $\cR_1,\ldots,\cR_m$. The phase space of the joint system
$\cZ+\cR_1+\cdots+\cR_m$ is
$$
\Gamma=\Gamma_\cZ\times\Gamma_\cR,\qquad\Gamma_\cR=\bigtimes_{k=1}^m\Gamma_{\cR_k}.
$$
Accordingly, we denote phase space points by $x=(z,r)$, $r=(r_1,\ldots,r_m)$ and
Liouville's measure by $\d x=\d z\otimes\d r$, $\d r=\otimes_{k=1}^m\d r_k$. The
controls are $X=(Z,R)$ and the Hamiltonian of the coupled system is
$$
H(x|X)=H_\cZ(z|Z)+\sum_{k=1}^mH_k(r_k)+\sum_{k=1}^mR_kQ_k(z,r_k).
$$
Assuming $H_\cZ(\theta(z)|Z)=H_\cZ(z|\vartheta(Z))$, $H_k\circ\theta=H_k$ and
$Q_k\circ\theta=\theta_k Q_k$ with $\theta_k=\pm1$ for each $k$,
implies~\eqref{eq:HthetaInv} with $\vartheta(R)_k=\theta_k R_k$.

Given $\boldsymbol{\beta}=(\beta_1,\ldots,\beta_m)$, consider the initial
measure $\nu_X=\mu_Z\otimes\mu_{\boldsymbol{\beta}}$, where the product of Gibbs
states
$$
\mu_{\boldsymbol{\beta}}(\d r)=\bigotimes_{k=1}^m\e^{-\beta_k(H_k(r_k)-F_k)}\d r_k,
\qquad
F_k=-\frac{1}{\beta_k}\log\int_{\Gamma_{\cR_k}}\e^{-\beta_kH_k(r_k)}\d r_k,
$$
describes each reservoir $\cR_k$ in thermal equilibrium at inverse temperature
$\beta_k$ and
$$
\mu_Z(\d z)=\e^{-s(z|Z)}\d z,
$$
so that
$$
S(x|X)=s(z|Z)+\sum_{k=1}^n\beta_k\left(H_k(r_k)-F_k\right).
$$
Then the relation~\eqref{eq:genHamFR} reads
$$
\frac{\d\PP_{\widehat{X}}\circ\Theta}{\d\PP_X}(\bsx)
=\e^{-\left(s(\theta(z_\tau)|\vartheta(Z(\tau)))-s(z_0|Z(0))\right)}
\e^{-\sum_{k=1}^n\beta_k\Delta E_k},
$$
where
$$
\Delta E_k=H_k(r_{k,\tau})-H_k(r_{k,0})
$$
is the energy dissipated in the reservoir $\cR_k$. Note that, with
$\{\,\argdot\,,\,\argdot\,\}$ denoting the Poisson bracket,
$$
\beta_k\Delta E_k=\int_0^\tau R_k(t)J_k(t)\d t,
\qquad
J_k(t)=\beta_k\{Q_k,H_k\}(x_t),
$$
so that $J_k$ can be identified as the current conjugated to the affinity $R_k$.

Assuming the reservoirs to be large enough to remain in thermal equilibrium
during the whole process, one can also identify
$$
\Delta S=\sum_{k=1}^n\beta_k\Delta E_k=\sum_{k=1}^n\int_0^\tau R_k(t)J_k(t)\d t
$$
with the entropy dumped in the reservoirs, that is the total entropy generated
by the process.

In the special case
$$
s(z|Z)=\beta(H_\cZ(z|Z)-F(Z)),\qquad
F(Z)=-\frac{1}{\beta}\log\int_{\Gamma_{\cZ}}\e^{-\beta H_\cZ(z|Z)}\d z,
$$
one has
$$
\frac{\d\PP_{\widehat{X}}\circ\Theta}{\d\PP_X}(\bsx)
=\e^{\beta\Delta F}\e^{-\beta\Delta E-\Delta S}
$$
where
$$
\Delta F=F(Z(\tau))-F(Z(0)),\qquad
\Delta E = H_\cZ(z_\tau|Z(\tau))-H_\cZ(0|Z(0)),
$$
are, respectively, the free energy difference between the two {\sl equilibrium
states} $\mu_{Z(0)}$ and $\mu_{Z(\tau)}$ of the system $\cZ$, and the change in
the internal energy of this system during the process. This leads to the
fluctuation relation
$$
\EE_X[\e^{-\beta\Delta E-\Delta S}]=\e^{-\beta\Delta F}.
$$

\printbibliography[heading=bibintoc,title={References}]


\chapter{Large Deviations and Fluctuation Relations}
\label{chap:LD and RF}

\abstract*{After introducing the minimal framework needed to make sense of a
fluctuation theorem, we briefly review some basic facts about large deviations
estimates that play a crucial role in the formulation of  fluctuation theorems.
The reader is referred to~\cite{Dembo2000,Hollander2000} for detailed
introductions to this vast subject.
}

\abstract{After introducing the minimal framework needed to make sense of a
fluctuation theorem, we briefly review some basic facts about large deviations
estimates that play a crucial role in the formulation of  fluctuation theorems.
The reader is referred to~\cite{Dembo2000,Hollander2000} for detailed
introductions to this vast subject.
}

\vskip 1cm

Our starting point is a family of measurable
spaces~$(\Omega_\lambda,\fF_\lambda)$ indexed by a parameter~$\lambda$ varying
in a directed set~$\cL$. Each of these spaces is equipped with a pair
$(\PP_\lambda,\wP_\lambda)$ of probability measures. We say that the family
$(\PP_\lambda,\widehat{\PP}_\lambda)_{\lambda\in\cL}$ is {\sl in} \ndex{involution}
whenever, for each $\lambda\in\cL$,
$\widehat{\PP}_\lambda=\PP_\lambda\circ\Theta_\lambda$ where
$\Theta_\lambda:\Omega_\lambda\to\Omega_\lambda$  is a measure space
automorphism satisfying
$\Theta_\lambda\circ\Theta_\lambda=\Id_{\Omega_\lambda}$. Such an
involution implements a $\ZZ_2$-action on $\Omega_\lambda$. In many cases of
interest one takes\footnote{$\ZZ_+=\{0,1,2,\ldots
\}$ and $\RR_+=[0,\infty[.$} $\cL=\ZZ_+$ or $\cL=\RR_+$ and interprets the values
taken by~$\lambda$ as discrete or continuous time instances. In such a situation,
$(\Omega_\lambda,\fF_\lambda,\PP_\lambda)$ describes the space-time statistics
of the physical system under consideration over the finite time
interval~$[0,\lambda]$ and $\wP_\lambda$ is a reference measure related to the
same system running under a different protocol.\footnote{In the physics
literature, the process $(\Omega_\lambda,\fF_\lambda,\wP_\lambda)$ is sometimes
called {\em conjugate} to $(\Omega_\lambda,\fF_\lambda,\PP_\lambda)$.}\;A
possible candidate for this reference measure is
$\widehat{\PP}_\lambda=\PP_\lambda\circ\Theta_\lambda$ where $\Theta_\lambda$
implements the time-reversal symmetry of the system, in which case the resulting
pair of measures is in involution. However, as we shall illustrate in
Chapter~\ref{chap:Illustrations}, our general setup can also be applied to other physical
problems not necessarily related to dynamical processes. We shall denote by
$\EE_\lambda$ and $\wE_\lambda$ the expectation functionals associated to
$\PP_\lambda$ and $\wP_\lambda$.

\section{Large Deviations}
\label{sec:LDP}

Assume that, for each $\lambda$ in the directed set $\cL$, $\xi_\lambda$ is a
random variable on the measurable space $(\Omega_\lambda,\fF_\lambda)$ taking
its values in a Polish ({\sl i.e.,} separable complete metric) space $\fX$, and the
function $\cL\ni\lambda\mapsto r_\lambda>0$ is such that $\lim_{\lambda\in\cL}
r_\lambda=+\infty$. One says that the family $(\xi_\lambda)_{\lambda\in\cL}$
satisfies the {\bf Large Deviation Principle} (LDP) w.r.t.\;the family of
probability measures $(\PP_\lambda)_{\lambda\in\cL}$ with \ndex{scale}
$(r_\lambda)_{\lambda\in\cL}$ and \ndex{rate} $I$ whenever $I:\fX\to[0,+\infty]$ has
closed level sets,\footnote{$I$ has closed/compact level sets whenever
$\{x\in\fX\mid I(x)\le\alpha\}$ is closed/compact for all $\alpha\in\RR$. This implies
in particular that $I$ is lower semicontinuous.} and is such that for any Borel
set $X\subset\fX$,\index{LDP}
\begin{equation}
\begin{split}
-\inf_{x\in\mathring{X}} I(x)
&\le \liminf_{\lambda\in\cL}r_\lambda^{-1}
\log \PP_\lambda\{\xi_\lambda\in X\}\\[4pt]
&\le \limsup_{\lambda\in\cL}r_\lambda^{-1}
\log \PP_\lambda\{\xi_\lambda\in X\}
\le -\inf_{x\in\overline X}I(x),
\end{split}
\label{Eq:LDP}
\end{equation}
where $\mathring{X}/\overline X$ denotes the interior/closure of $X$. $I$ is
said to be a \ndex{good rate} whenever it has compact level sets. If the lower
bound in~\eqref{Eq:LDP} only holds for Borel sets $X$ that are subsets of
$\cX\subset\fX$, we shall say that the {\bf LDP holds locally} in $\cX$. We
remark that considering $X =\fX$ in \eqref{Eq:LDP} implies that $\inf_{x\in \fX}
I(x)=0$, and in particular that $I$ is not everywhere $+\infty$.

In most applications, the index $\lambda$ parametrizes the size of a sample and
$\xi_\lambda$ describes some property of samples of size $\lambda$.  Suppose
that the  law of large numbers holds as the size of the sample increases, {\sl i.e.,}
that there is $\bar{\xi}\in\fX$ such that
$\lim_{\lambda\in\cL}\xi_\lambda=\bar{\xi}$ holds in probability. The LDP then
describes, on the exponential scale, the rate at which this limit is achieved,
and the rate function $I$ has a single zero at $\bar{\xi}$.

The {\bf contraction principle}~\cite[Theorem~4.2.1]{Dembo2000}\footnote{
\cite{Dembo2000} does not address the case where $\cL$ is a general directed
set, but the proof is almost identical.} states that if the above LDP holds with\index{contraction principle}
a good rate $I$, and $\phi:\fX\to\fY$ is a continuous map taking its values in a
Polish space $\fY$, then the family $(\eta_\lambda)_{\lambda\in\cL}$ defined
by $\eta_\lambda=\phi(\xi_\lambda)$ satisfies the LDP
w.r.t.\;$(\PP_\lambda)_{\lambda\in\cL}$ with scale
$(r_\lambda)_{\lambda\in\cL}$ and good rate $J$, where
$$
J(y)=\inf\left\{I(x)\mid x\in\phi^{-1}(\{y\})\right\}.
$$
The same conclusion holds~\cite[Theorem~4.2.23]{Dembo2000} if $\phi$ is merely
measurable and such that, for a sequence of continuous maps $\phi_n:\fX\to\fY$,
one has
$$
\lim_{n\to\infty}\sup\left\{d(\phi(x),\phi_n(x))\mid x\in I^{-1}(]{-}\infty,\alpha])\right\}=0
$$
for all $\alpha\in\RR$.

{\bf Varadhan's Lemma}~\cite[Theorem~4.3.1]{Dembo2000} asserts that if the family
$(\xi_\lambda)_{\lambda\in\cL}$ satisfies the  LDP\index{Varadhan's Lemma}
w.r.t. $(\PP_\lambda)_{\lambda\in\cL}$ with scale
$(r_\lambda)_{\lambda\in\cL}$ and good rate $I$, and $\phi:\fX\to\RR$ is a continuous
function such that
$$
\limsup_{\lambda\in\cL}\frac1{r_\lambda}
\log\EE_\lambda\left[\e^{\alpha r_\lambda\phi(\xi_\lambda)}\right]<\infty
$$
for some $\alpha>1$, then
$$
\lim_{\lambda\in\cL}\frac1{r_\lambda}
\log\EE_\lambda\left[\e^{r_\lambda\phi(\xi_\lambda)}\right]
=\sup\left\{\phi(x)-I(x)\mid x\in\fX\right\}.
$$
Moreover~\cite[Theorem~III.17]{Hollander2000}, defining the tilted measure\footnote{The
notation $\PP_\lambda^\phi\ll\PP_\lambda$ means that $\PP_\lambda^\phi$ is absolutely
continuous with respect to $\PP_\lambda$.} $\PP_\lambda^\phi\ll\PP_\lambda$ with density
$$
\frac{\d\PP_\lambda^\phi}{\d\PP_\lambda}
=\frac{\e^{r_\lambda\phi(\xi_\lambda)}}%
{\EE_\lambda\left[\e^{r_\lambda\phi(\xi_\lambda)}\right]},
$$
on $(\Omega_\lambda,\fF_\lambda)$ and setting
\beq
J(x)=I(x)-\phi(x)-\inf\left\{I(y)-\phi(y)\mid y\in\fX\right\},
\label{Eq:Tilting}
\eeq
the family $(\xi_\lambda)_{\lambda\in\cL}$ satisfies the  LDP
w.r.t.\;$(\PP_\lambda^\phi)_{\lambda\in\cL}$ with scale
$(r_\lambda)_{\lambda\in\cL}$ and rate $J$.

Given the scale $(r_\lambda)_{\lambda\in\cL}$, two families
$(P_\lambda)_{\lambda\in\cL}$ and $(Q_\lambda)_{\lambda\in\cL}$ of
probability measures on $\fX$ are {\bf exponentially equivalent} whenever there\index{exponential equivalence}
exists a family of probability spaces $(\Omega_\lambda,\fF_\lambda,\PP_\lambda)$
and $\fX$-valued random variables  $(\xi_\lambda)_{\lambda\in\cL}$ and
$(\eta_\lambda)_{\lambda\in\cL}$ whose laws under $\PP_\lambda$ are given by
$P_\lambda$ and $Q_\lambda$ respectively, and  such that for all $\delta >0$,
the sets $\Gamma_{\lambda,\delta}=\{\omega\in\Omega_\lambda\mid
d(\xi_\lambda(\omega),\eta_\lambda(\omega))>\delta\}\in\fF_\lambda$ satisfy
\[
\limsup_{\lambda \in\cL}\frac1{r_\lambda}P_\lambda(\Gamma_{\lambda,\delta})=-\infty.
\]

Suppose that $\fX=\RR^d$ and that for all   $\alpha\in\RR^d$ the
limit\footnote{Here $\langle\,\argdot\,,\,\argdot\,\rangle$ denotes the Euclidean inner
product on $\RR^d$.}
\begin{equation}\label{eq:deflambdaalpha}
\Lambda(\alpha)\coloneq\lim_{\lambda\in\cL}\frac1{r_\lambda}
\log\EE_\lambda\left[\e^{r_\lambda\langle\alpha,\xi_\lambda\rangle}\right]\in[-\infty,\infty]
\end{equation}
exists and defines a lower semicontinuous function.
Then, $\Lambda$ is a convex function, its \ndex{essential domain}
$\Dom(\Lambda)=\{\alpha\in\RR^d\mid \Lambda(\alpha)<\infty\}$ is convex,
and if $d=1$ the map $\Dom(\Lambda)\ni\alpha\mapsto\Lambda(\alpha)$ is
continuous.\footnote{If $d>1$, continuity only holds in the relative interior
of $\Dom(\Lambda)$, and in particular in its interior,
see~\cite[Theorem~10.1]{Rockafellar1970}.} Suppose that $0$ belongs to
the interior of $\Dom(\Lambda)$ and that $\Lambda$ is differentiable on
this interior. If either $\Dom(\Lambda)=\RR$ or
\beq
\lim_{\substack{\alpha\to\partial\Dom(\Lambda)\\\alpha\in\Dom(\Lambda)}}|\Lambda'(\alpha)|=+\infty,
\label{Eq:Steep}
\eeq
then the {\bf Gärtner--Ellis theorem}~\cite[Theorem~2.3.6]{Dembo2000}  asserts\index{Gärtner--Ellis theorem}
that the family $(\xi_\lambda)_{\lambda\in\cL}$ satisfies the LDP
w.r.t.\;$(\PP_\lambda)_{\lambda\in\cL}$ with scale
$(r_\lambda)_{\lambda\in\cL}$ and good rate $I$ given by the \ndex{Legendre--Fenchel
transform} of $\Lambda$,
$$
I(x)=\sup_{\alpha\in\RR^d}(\langle\alpha,x\rangle-\Lambda(\alpha)).
$$
If $\Dom(\Lambda)\not=\RR$ and Condition~\eqref{Eq:Steep} fails, then the same
LDP holds locally on the set $\cX$ of all $x\in\RR^d$ for which  there exists
$\alpha$ in the interior of $\Dom(\Lambda)$ such that
\[
I(y)>I(x)+\langle\alpha,y-x\rangle
\]
holds for all $y\not=x$. Such a $x\in\RR^d$ is said to be an \ndex{exposed point} of
$I$, and the corresponding $\alpha$ is the associated \ndex{supporting hyperplane}.

We say that two families $(\xi_\lambda)_{\lambda\in\cL}$ and
$(\eta_\lambda)_{\lambda\in\cL}$ of random variables on
$(\Omega_\lambda,\fF_\lambda)$ taking their values in the same Polish space
$\fX$ are {\bf physically equivalent} whenever\index{physical equivalence}
$$
\lim_{\lambda\in\cL}\frac1{r_\lambda}\sup_{\omega\in\Omega_\lambda}
d(\xi_\lambda(\omega),\eta_\lambda(\omega))=0.
$$
This implies that the laws of $r_\lambda^{-1}\xi_\lambda$ and
$r_\lambda^{-1}\eta_\lambda$ under $\PP_\lambda$ are exponentially equivalent,
{\sl i.e.,} that
$$
\limsup_{\lambda\in\cL}\frac1{r_\lambda}
\log\PP_\lambda\left\{d(\xi_\lambda,\eta_\lambda)
>\varepsilon r_\lambda\right\}=-\infty
$$
holds for all $\varepsilon>0$. It follows~\cite[Theorem~4.2.13]{Dembo2000} that
if $(\xi_\lambda)_{\lambda\in\cL}$ satisfies the LDP
w.r.t.\;$(\PP_\lambda)_{\lambda\in\cL}$ with scale
$(r_\lambda)_{\lambda\in\cL}$ and good rate $I$, so does
$(\eta_\lambda)_{\lambda\in\cL}$. In addition, if
$(\xi_\lambda)_{\lambda\in\cL}$ and $(\eta_\lambda)_{\lambda\in\cL}$ are
physically equivalent, then replacing $\xi_\lambda$ with $\eta_\lambda$ in
\eqref{eq:deflambdaalpha} does not change the function $\Lambda$.

\section{Transient Fluctuation Relation}
\label{Ssec:TransientFR}

In this section we describe a minimal abstract framework allowing for the
formulation of transient fluctuation relations. Our point here is that such
relations are tautological consequences of a simple regularity property of a
pair of probability measures.

To a pair $(\PP_\lambda,\wP_\lambda)$ of equivalent probability measures
on~$(\Omega_\lambda,\fF_\lambda)$ we associate the \ndex{entropy production
observable}
\begin{equation}
\sigma_\lambda\coloneq\log\frac{\d\PP_\lambda}{\d\wP_\lambda}.
\label{Eq:sigma_lambda}
\end{equation}
This is a real-valued random variable, and we denote
by~$P_\lambda$ its law under~$\PP_\lambda$. Recall that the family
$(\PP_\lambda,\wP_\lambda)_{\lambda\in\cL}$ is {\sl in involution}
whenever, for each $\lambda\in\cL$,
$\wP_\lambda=\PP_\lambda\circ\Theta_\lambda$ where
$\Theta_\lambda:\Omega_\lambda\to\Omega_\lambda$  is a measure space
automorphism satisfying
$\Theta_\lambda\circ\Theta_\lambda=\Id_{\Omega_\lambda}$.

The very definition of~$\sigma_\lambda$ implies a number of simple, yet important
properties.

\part{Non-vanishing.}$\sigma_\lambda$ vanishes $\PP_\lambda$-a.e.\;iff\,
$\wP_\lambda=\PP_\lambda$. From our point of view this is a trivial and
uninteresting situation. In the following we assume that
$\wP_\lambda\not=\PP_\lambda$. In cases where $\PP_\lambda$ and
$\wP_\lambda$ are in involution, this means that the
$\ZZ_2$-symmetry $\Theta_\lambda$ is broken.

\part{Alternating.}The entropy production of the pair
$(\wP_\lambda,\PP_\lambda)$ is given by
$\hat{\sigma}_\lambda=-\sigma_\lambda$. If $\PP_\lambda$ and
$\wP_\lambda$ are in involution, then
$\hat{\sigma}_\lambda=\sigma_\lambda\circ\Theta_\lambda$, and hence
$\sigma_\lambda$ is odd under the $\ZZ_2$-action. In particular, the range
of $\sigma_\lambda$ is symmetric w.r.t.\;$0$.

\part{Mean entropy production.}Positive values of $\sigma_\lambda$ are preferred
under $\PP_\lambda$. To see this, recall that the \ndex{relative entropy}\footnote{also
called Kullback--Leibler divergence.} of two equivalent probability measures
$\PP$ and $\QQ$ is defined by
\beq
\Ent(\PP\mid\QQ)\coloneq
\int\log\left(\frac{\d\PP}{\d\QQ}\right)\d\PP.
\label{equ:EntDef}
\eeq
By Jensen's inequality, one has
$$
\Ent(\PP\mid\QQ)=-\int\log\left(\frac{\d\QQ}{\d\PP}\right)\d\PP
\ge-\log\int\frac{\d\QQ}{\d\PP}\d\PP=0,
$$
with equality iff $\PP=\QQ$.
It follows that the mean entropy production
\beq
\ep_\lambda\coloneq\EE_\lambda[\sigma_\lambda]=\Ent(\PP_\lambda\mid\wP_\lambda)
\label{equ:epDef}
\eeq
is non-negative and vanishes iff $\PP_\lambda=\wP_\lambda$.
Expressing this fact in terms of the law $P_\lambda$ of
$\sigma_\lambda$ under $\PP_\lambda$ yields the
inequality
\begin{equation} \label{Eq:Jarzynski_inequality}
\ep_\lambda=\int_\RR s\,P_\lambda(\d s)\ge0.
\end{equation}
Thus, if $\PP_\lambda$ and $\wP_\lambda$ are in involution, then $\ep_\lambda$
plays the role of an order parameter for the $\ZZ_2$-symmetry $\Theta_\lambda$:
Inequality~\eqref{Eq:Jarzynski_inequality} becomes strict whenever this symmetry
is broken.

\part{Jarzynski's identity.}A more quantitative statement
can be obtained by considering the \ndex{Rényi relative $\alpha$-entropy}
defined by
\beq
\Ent_\alpha(\PP\mid\QQ)
\coloneq\log\int\left(\frac{\d\QQ}{\d\PP}\right)^\alpha\d\PP,\qquad(\alpha\in\RR).
\label{equ:EntalphaDef}
\eeq
In terms of entropy production, we have
\beq
e_\lambda(\alpha)\coloneq\Ent_\alpha(\PP_\lambda\mid\wP_\lambda)
=\log\EE_\lambda\left[\e^{-\alpha\sigma_\lambda}\right]
=\log \int_\RR\e^{-\alpha s}P_\lambda(\d s).
\label{Eq:e_alpha_def}
\eeq
The function $\RR\ni\alpha\mapsto e_\lambda(\alpha)\in]{-}\infty,+\infty]$ is
the cumulant generating function of the random variable $-\sigma_\lambda$.
It is convex, lower semicontinuous, and vanishes for
$\alpha\in\{0,1\}$. Hence, $e_\lambda $ is non-positive on~$[0,1]$ and non-negative
outside~$[0,1]$. It admits an analytic continuation to the strip
$\{z\in\CC\mid 0<\Re z<1\}$. The Rényi entropy vanishes identically as a
function of $\alpha$ iff $\ep_\lambda=0$ and is strictly negative
for $\alpha\in]0,1[$ iff $\ep_\lambda>0$. Expressing the relation
$e_\lambda(1)=0$ in terms of~$P_\lambda$, we derive \ndex{Jarzynski
identity}~\cite{Jarzynski2000}
\begin{equation}
\int_\RR\e^{-s}P_\lambda(\d s)=1.
\label{Eq:Jarzynski_identity}
\end{equation}
Invoking Markov's inequality, we derive the following exponential
bound on the tail of $P_\lambda$
\begin{equation}\label{Eq:Jarzy_and_Markov}
P_\lambda(]{-}\infty,-s])=\PP_\lambda\{\sigma_\lambda\le-s\}\le\e^{-s}.
\end{equation}
To formulate the last properties, let us define
$$
\hat{e}_\lambda(\alpha)\coloneq\Ent_\alpha(\wP_\lambda|\PP_\lambda)
$$
and denote by $\widehat{P}_\lambda$ the law of $\sigma_\lambda$ under
$\wP_\lambda$.

\bigskip
\begin{proposition}[Transient Fluctuation Relation]\index{FR!transient/finite-time}
\label{Prop:Transient_FR}
\ben
\item The following relations hold:
\begin{gather}
\hat e_\lambda(\alpha)=e_\lambda(1-\alpha) \quad\mbox{for $\alpha\in\RR$},
\label{Eq:e_lambda_symmetry}\\[6pt]
\frac{\d P_\lambda}{\d \widehat P_\lambda}(s)=\e^s
\quad\mbox{for $s\in\RR$}.
\label{Eq:P_lambda_FR}
\end{gather}
\item If the family $(\PP_\lambda,\wP_\lambda)_{\lambda\in\cL}$ is in
involution, then
\begin{gather}
\hat e_\lambda=e_\lambda,
\label{Eq:hat_e_lambda}\\[6pt]
\widehat P_\lambda=P_\lambda\circ\vartheta,
\label{Eq:hat_P_lambda}
\end{gather}
where $\vartheta$ denotes the reflection~$s\mapsto-s$.
\een
\end{proposition}
Relations~\eqref{Eq:e_lambda_symmetry} and~\eqref{Eq:P_lambda_FR} are in fact
equivalent: the validity of one of them implies the other. We shall refer to
them  as the {\it transient FR\/}. It refines (and implies)
Inequality~\eqref{Eq:Jarzynski_inequality} and its basic appeal is
its universal form. In applications to non-equilibrium physics, the transient FR
with $e_\lambda $ not identically vanishing or, equivalently, $P_\lambda$
not concentrated at $s=0$, is a manifestation of time-reversal symmetry breaking
and emergence of the 2$^{\rm nd}$~Law of Thermodynamics.

\proof[of Proposition~\ref{Prop:Transient_FR}]\index{FR!transient/finite-time}
{\bf (1)} Relation~\eqref{Eq:e_lambda_symmetry} is a simple consequence of the
definition of Rényi's entropy,
\begin{align}
e_\lambda(1-\alpha)
&=\Ent_{1-\alpha}(\PP_\lambda\mid\wP_\lambda)\nonumber\\
&=\log\EE_\lambda[\e^{-(1-\alpha)\sigma_\lambda}]
=\log\wE_\lambda[\e^{\alpha\sigma_\lambda}]\label{Eq:sunnyW}\\
&=\log\wE_\lambda[\e^{-\alpha\hat\sigma_\lambda}]
=\Ent_\alpha(\wP_\lambda\mid\PP_\lambda)
=\hat e_\lambda(\alpha).\nonumber
\end{align}
To prove~\eqref{Eq:P_lambda_FR}, we exponentiate~\eqref{Eq:sunnyW} and
rewrite the result in terms of~$P_\lambda$ and $\widehat{P}_\lambda$ to get
$$
\int_\RR \e^{-(1-\alpha)s}P_\lambda(\d s)
=\int_\RR \e^{\alpha s}\widehat P_\lambda(\d s).
$$
Using now the analyticity of $e_\lambda(z)$ in the strip $0<\Re z<1$ and the
fact that the characteristic function uniquely defines the corresponding
probability measure, we deduce that~$\e^{-s}P_\lambda(\d s)$
and~$\widehat{P}_\lambda(\d s)$ coincide. This is equivalent
to~\eqref{Eq:P_lambda_FR}.

\part{(2)} Relation~\eqref{Eq:hat_e_lambda} is a direct consequence of the facts
that $\wP_\lambda=\PP_\lambda\circ\Theta_\lambda$ and
$\hat{\sigma}_\lambda=-\sigma_\lambda=\sigma_\lambda\circ\Theta_\lambda$, which give that
\begin{equation}\label{Eq:chainineqehatl}\begin{split}
\hat e_\lambda(\alpha)
&=\log\wE_\lambda[\e^{-\alpha\hat\sigma_\lambda}]
=\log\wE_\lambda[\e^{\alpha\sigma_\lambda}]\\
&=\log\EE_\lambda[\e^{\alpha\sigma_\lambda\circ\Theta_\lambda}]
=\log\EE_\lambda[\e^{-\alpha\sigma_\lambda}]=e_\lambda(\alpha).
\end{split}
\end{equation}
Exponentiating the third equality in~\eqref{Eq:chainineqehatl} yields
$$
\int\e^{\alpha s}\widehat{P}_\lambda(\d s)=\int\e^{-\alpha s}P_\lambda(\d s)
=\int\e^{\alpha\vartheta(s)}P_\lambda(\d s)
=\int\e^{\alpha s}P_\lambda\circ\vartheta(\d s)
$$
and the same characteristic function argument as in the proof of~(1)
yields~\eqref{Eq:hat_P_lambda}.

\qed

\section{Fluctuation Theorem}
\label{Ssect:FT}

Going back to the abstract setup of Section~\ref{Ssec:TransientFR}, we proceed
beyond the transient regime by selecting a scale function $\cL\ni\lambda\mapsto
r_\lambda$. In a specific context, the choice of $r_\lambda$ is naturally linked
to the structure of~$\cL$. For example, if $\cL=\ZZ_+$ or $\RR_+$, its elements
being interpreted as time instances, then $r_\lambda=\lambda$.

\begin{definition}\label{Def:FT}\index{FT}
\hspace{0em}
\ben
\item We say that the {\sl Fluctuation Theorem} holds for the family
$(\PP_\lambda,\wP_\lambda)_{\lambda\in\cL}$ with
scale $(r_\lambda)_{\lambda\in\cL}$ and rate $(I,\hat I)$ whenever the
following two conditions are satisfied:
\begin{enumerate}[label=(\roman*)]
\item $(r_\lambda^{-1}\sigma_\lambda)_{\lambda\in\cL}$ satisfies the LDP
w.r.t.\;$(\PP_\lambda)_{\lambda\in\cL}$ with scale
$(r_\lambda)_{\lambda\in\cL}$ and rate $I$.

\item $(r_\lambda^{-1}\sigma_\lambda)_{\lambda\in\cL}$ satisfies the  LDP
w.r.t.\;$(\widehat{\PP}_\lambda)_{\lambda\in\cL}$ with scale
$(r_\lambda)_{\lambda\in\cL}$ and rate $\hat I$.
\end{enumerate}

\item For each $\lambda\in\cL$, let $\widetilde{\sigma}_\lambda$ be a
real-valued random variable on $\Omega_\lambda$. We say that
$(\widetilde{\sigma}_\lambda)_{\lambda\in\cL}$ satisfies the {\sl Fluctuation
Theorem} w.r.t.\;the family $(\PP_\lambda,\wP_\lambda)_{\lambda\in\cL}$ with
scale $(r_\lambda)_{\lambda\in\cL}$ and rate $(I,\hat I)$ whenever the above
conditions~(i) and~(ii) are satisfied with $\sigma_{\lambda}$ replaced by
$\widetilde{\sigma}_\lambda$.

\item If the two LDPs mentioned in conditions~(i) and~(ii) only hold locally on
the interval $\cJ\ni0$, then we say that the corresponding Fluctuation Theorem
holds on $\cJ$. This Fluctuation Theorem is called global when it holds on
$\cJ=\RR$, local otherwise.
\een
\end{definition}

\remark[1] Recall that $P_\lambda/\widehat{P}_\lambda$ denotes
the law of $\sigma_\lambda$ under $\PP_\lambda/\wP_\lambda$. Thus, the
LDP estimates expressing the FT can be rewritten as
\begin{align}
-\inf_{s\in\mathring{S}} I(s)
\le \liminf_{\lambda\in\cL}r_\lambda^{-1}\log P_\lambda(r_\lambda S)
&\le \limsup_{\lambda\in\cL}r_\lambda^{-1}\log P_\lambda(r_\lambda S)
\le -\inf_{s\in \overline S}I(s),\label{Eq:FT_P}\\[10pt]
-\inf_{s\in\mathring{S}} \hat I(s)
\le \liminf_{\lambda\in\cL}r_\lambda^{-1}\log\widehat{P}_\lambda(r_\lambda S)
&\le \limsup_{\lambda\in\cL}r_\lambda^{-1}\log\widehat{P}_\lambda(r_\lambda S)
\le -\inf_{s\in \overline S}\hat I(s),\label{Eq:FT_hatP}
\end{align}
for any Borel set $S\subset\RR$.

\remark[2] Suppose that the FT holds on $\cJ$
for the family $(\PP_\lambda,\wP_\lambda)_{\lambda\in\cL}$ with
scale $(r_\lambda)_{\lambda\in\cL}$ and rate $(I,\hat I)$. Let
$(\sigma_\lambda)_{\lambda\in\cL}$ be the associated entropy production.
Then, any family $(\widetilde{\sigma}_\lambda)_{\lambda\in\cL}$ of real valued random
variables which is physically equivalent to $(\sigma_\lambda)_{\lambda\in\cL}$
satisfies the FT on $\cJ$ w.r.t.\;$(\PP_\lambda,\wP_\lambda)_{\lambda\in\cL}$ with scale
$(r_\lambda)_{\lambda\in\cL}$ and rate $(I,\hat I)$.

\section{Asymptotic Fluctuation Relation}
\label{Ssect:AsymptoticFR}

In Section~\ref{Ssec:TransientFR}, the transient FR was obtained as a trivial
consequence of the definition of entropy production. Here, we show how
the Fluctuation Theorem, together with the transient FR, yield an asymptotic
FR, a symmetry relation for the rate functions $(I,\hat I)$.\index{FR!asymptotic}

\begin{proposition}\label{Prop:RateSymmetry}
Suppose that the FT holds on~$\cJ$ for the family
$(\PP_\lambda,\wP_\lambda)_{\lambda\in\cL}$.
Then the corresponding rate functions satisfy the relation
\begin{equation}
\hat I(s)=I(s)+s
\label{Eq:RateSymmetry}
\end{equation}
for all $s\in\mathring{\cJ}$. Moreover, if the family
$(\PP_\lambda,\wP_\lambda)_{\lambda\in\cL}$ is in involution, then
\beq
\hat I(s)=I(-s)
\label{Eq:hatRate}
\eeq
for all $s\in\mathring{\cJ}\cap(-\mathring{\cJ})$.
\end{proposition}

\proof In view of the transient FR~\eqref{Eq:P_lambda_FR} we have
$$
\e^{r_\lambda\inf S}\widehat{P}_\lambda(r_\lambda S)\le
P_\lambda(r_\lambda S)
\le\e^{r_\lambda\sup S}\widehat{P}_\lambda(r_\lambda S)
$$
for any Borel set $S\subset\cJ$. Using~\eqref{Eq:FT_P} and~\eqref{Eq:FT_hatP}
we obtain
\begin{align*}
-\inf_{s\in\mathring{S}}I(s)
&\le\liminf_{\lambda\in\cL}r_\lambda^{-1}\log P_\lambda(r_\lambda S)\\
&\le \limsup_{\lambda\in\cL}r_\lambda^{-1}\log\bigl(\e^{r_\lambda \sup S}\widehat{P}_\lambda(r_\lambda S)\bigr)
\le\sup S-\inf_{s\in\overline{S}}\hat I(s),\\[10pt]
\inf S-\inf_{s\in\mathring{S}}\hat I(s)
&\le\liminf_{\lambda\in\cL}r_\lambda^{-1}\bigl(\e^{r_\lambda\inf S}\log \widehat{P}_\lambda(r_\lambda S)\bigr)\\
&\le\limsup_{\lambda\in\cL}r_\lambda^{-1}\log P_\lambda(r_\lambda S)
\le-\inf_{s\in\overline{S}}I(s).
\end{align*}
Taking $S=]b-\varepsilon,b+\varepsilon[$ with $b\in\mathring{\cJ}$
and $\varepsilon>0$ small enough, we derive
\beq
\begin{split}
\inf_{|s-b|<2\varepsilon}\hat I(s)\le\inf_{|s-b|\le\varepsilon}\hat I(s)
\le b+\varepsilon+\inf_{|s-b|<\varepsilon}I(s),\\[10pt]
b-\varepsilon+\inf_{|s-b|<2\varepsilon}I(s)\le b-\varepsilon+\inf_{|s-b|\le\varepsilon}I(s)
\le \inf_{|s-b|<\varepsilon}\hat I(s).
\end{split}
\label{Eq:loc_4.9}
\eeq
Since the functions~$I$ and~$\hat I$ are lower semicontinuous, we have
$$
I(b)=\lim_{\varepsilon\downarrow0}\inf_{|s-b|<\varepsilon}I(s),\qquad
\hat I(b)=\lim_{\varepsilon\downarrow0}\inf_{|s-b|<\varepsilon}\hat I(s).
$$
Passing to the  limit in~\eqref{Eq:loc_4.9} we get
$$
\hat I(b)\le b+I(b)\le\hat I(b)
$$
for any $b\in\cJ$, which proves~\eqref{Eq:RateSymmetry}.

If the family
$(\PP_\lambda,\wP_\lambda)_{\lambda\in\cL}$ is in involution,
then~\eqref{Eq:hat_P_lambda} gives
$\widehat{P}_\lambda(r_\lambda S)=P_\lambda(-r_\lambda S)$, and the uniqueness
of the rate function yields~\eqref{Eq:hatRate}.\hfill\qed

\medskip
\remark[3] If the family $(\PP_\lambda,\wP_\lambda)_{\lambda\in\cL}$ is in
involution, then \eqref{Eq:hat_P_lambda} implies the following assertions:
\begin{itemize}
\item If the LDP of condition~(i) (resp.~(ii)) in Definition~\ref{Def:FT} holds
with rate function $I$ (resp.~$\hat I$), then the LDP of condition~(ii)
(resp.~(i)) holds with rate $\hat I$ (resp.~$I$) obeying~\eqref{Eq:RateSymmetry}
and~\eqref{Eq:hatRate}.

\item If the LDP of condition~(i) (resp.~(ii)) holds locally on $\cJ$, then the
LDP of condition~(ii) (resp.~(i)) holds locally on $-\cJ$. In particular the FT
holds on $\cJ\cap(-\cJ)$.

\item In Part~(3) of Definition~\ref{Def:FT}, the set $\cJ$ can always be taken
to be symmetric around $0$, since when {\sl both} conditions hold on $\cJ$, they
also hold on $\cJ\cup(-\cJ)$.
\end{itemize}
However, if the two families $(\PP_\lambda,\wP_\lambda)_{\lambda\in\cL}$ are not
assumed to be in involution, then conditions~(i) and~(ii) in
Definition~\ref{Def:FT}  are not equivalent in general. Assuming that $I$ and
$\hat I$ obey~\eqref{Eq:RateSymmetry}, one can still show the following: If the
LDP of condition~(i) (resp.~(ii)) holds, then the leftmost bound in~\eqref{Eq:FT_hatP}
(resp.~\eqref{Eq:FT_P}) holds for every Borel set $S$, and
the rightmost bound holds for every {\em precompact} Borel set $S$. LDPs whose
upper bound is restricted to precompact sets are called the {\sl weak LDPs} in
the literature. While some of the results presented below could be formulated in
terms of weak LDPs, we will refrain from discussing such technicalities.

\section{Entropic Pressure and Rate}
\label{Ssec:Entropic_Pressure}

\begin{definition}
Assume that the FT holds for the family
$(\PP_\lambda,\wP_\lambda)_{\lambda\in\cL}$ with scale
$(r_\lambda)_{\lambda\in\cL}$ and rate $(I,\hat I)$. The associated
\ndex{entropic pressure} is the function defined by
\beq
\RR\ni\alpha\mapsto e(\alpha)\coloneq\sup_{s\in\RR}\left(\alpha s-I(-s)\right).
\label{Eq:ealphaDef}
\eeq
\end{definition}

We note that the entropic pressure is nothing but the \ndex{Legendre--Fenchel
transform} of the function\footnote{The choice of the $-$ sign here is often a
nuisance, but we shall keep it for historical reasons.} $s\mapsto I(-s)$. We
shall see that it is closely related to various notions of pressure and other
thermodynamic potentials. We use the name ``entropic pressure'' to clearly
distinguish it from these other notions and to emphasize its connection with
entropy production.

From the perspective of the theory of convex functions, it
follows from its definition that the entropic pressure is a {\sl closed, proper
convex function}\footnote{A convex function $F:\RR^d\to[-\infty,+\infty]$ is
proper if it never takes the value $-\infty$ and is not everywhere $+\infty$. A
proper, convex function is closed when its epigraph
$\{(x,y)\in\RR^d\times\RR\mid y\ge F(x)\}$ is closed, which is equivalent to $F$
being lower semicontinuous.}~\cite[Theorem~12.2]{Rockafellar1970}. Invoking the
LDP yields $e(0)=0$. Setting
\[
\hat e(\alpha)\coloneq\sup_{s\in\RR}\left(\alpha s-\hat I(s)\right),
\]
the asymptotic FR~\eqref{Eq:RateSymmetry} further implies
$\hat e(\alpha)=e(1-\alpha)$, and the same reasoning yields that $e(1)=\hat e(0)=0$.
It immediately follows that $e(\alpha)\le0$ for $\alpha\in[0,1]$ and
$e(\alpha)\ge0$ for $\alpha\in\RR\setminus[0,1]$. Note that $\hat e$ is another
closed proper convex function. Moreover, if the pair
$(\PP_\lambda,\wP_\lambda)_{\lambda\in\cL}$ is in involution,
then~\eqref{Eq:hatRate} further gives that $\hat e(\alpha)=e(\alpha)$.

Suppose that
$$
\bar e(\alpha)\coloneq\limsup_{\lambda\in\cL}\frac1{r_\lambda}e_\lambda(\alpha)<\infty
$$
for all $\alpha\in\RR$. Then Varadhan's lemma implies that
\beq
e(\alpha)=\bar e(\alpha)=\lim_{\lambda\in\cL}\frac1{r_\lambda}e_\lambda(\alpha)
\label{Eq:LegendreI}
\eeq
for all $\alpha\in\RR$. In this case, all the above properties of the entropic
pressure are inherited from  the corresponding properties of $e_\lambda$
described in Section~\ref{Ssec:TransientFR}. We also note that if
$(\tilde{\sigma})_{\lambda\in\cL}$ is physically equivalent to entropy
production, then
$$
e(\alpha)=\lim_{\lambda\in\cL}\frac1{r_\lambda}
\log\EE_\lambda[\e^{-\alpha\tilde{\sigma}_{\lambda}}]
$$
holds for all $\alpha\in\RR$.

Recalling that the rate function~$I:\RR\to[0,+\infty]$ is lower semicontinuous,
whenever it is convex, a well-known property of the  Legendre--Fenchel
transform~\cite[Theorem~12.2]{Rockafellar1970} leads to the dual identity
\begin{equation}
I(-s)=\sup_{\alpha\in\RR}\bigl(\alpha s -e(\alpha)\bigr)
\label{Eq:IasLegendreTransform}
\end{equation}
which holds for all $s\in\RR$. Denoting by\footnote{See the discussion
following~\eqref{eq:defdome} below for definitions and conventions.}
$\partial^\pm e(\alpha)$ the right/left derivative of $e$ at $\alpha\in\RR$ and
by $\partial e(\alpha)=[\partial^-e(\alpha),\partial^+e(\alpha)]\cap \RR$ its
subdifferential, another well-known property of the Legendre--Fenchel
transform~\cite[Theorem~23.5]{Rockafellar1970} yields
\beq
s\in\partial e(\alpha)\quad\Longrightarrow\quad
I(-s)=\alpha s-e(\alpha).
\label{Eq:LegendreForma}
\eeq
If $I$ is not convex on~$\RR$, then~\eqref{Eq:IasLegendreTransform}
and~\eqref{Eq:LegendreForma} hold with~$I$ replaced by its convex
hull defined by
\beq
\breve{I}(s)\coloneq\sup_{\varphi\in\mathrm{conv}(I)}\varphi(s),
\label{eq:nabucco}
\eeq
where $\mathrm{conv}(I)$ denotes the set of all convex functions
$\varphi:\RR\to]{-}\infty,+\infty]$ such that $\varphi(s)\le I(s)$ for all
$s\in\RR$. Although there are some intriguing examples where the rate function
$I$ is strictly concave on some interval,\footnote{As we shall see in
Section~\ref{Ex:Lattice_gas}, the mean field lattice gas provides such an
example.} in most cases of interest $I$ is convex.

The above discussion can be turned around. Suppose that the function
$$
\bar{e}(\alpha)=\limsup_{\lambda\in\cL}\frac1{r_\lambda}
\log\,\EE_\lambda[\e^{-\alpha\sigma_\lambda}]
$$
is finite on a neighborhood of $[0,1]$, then the LDP upper bounds hold
(see the proofs of~\cite[Theorems~2.3.6 and~4.5.3]{Dembo2000}). Moreover,
if the limit
$$
e(\alpha)=\lim_{\lambda\in\cL}\frac1{r_\lambda}e_\lambda(\alpha)
$$
exists and defines a differentiable function on some open interval
$\fJ_d\supset[0,1]$, then the Gärtner--Ellis theorem implies that the FT holds
locally on the interval $$\cJ=]a_-\vee(-a_+),(-a_-)\wedge a_+[$$ with
$$
a_-=\inf_{\alpha\in\fJ_d}e'(\alpha),\qquad
a_+=\sup_{\alpha\in\fJ_d}e'(\alpha),
$$
and the convex rate functions given by\footnote{The function $I$ defined in
\eqref{eq:defIhatIJd} coincides with the right-hand side of
\eqref{Eq:IasLegendreTransform} at least in the interval  $[-a_+, a_-]\cap
\RR$.}
\begin{equation}\label{eq:defIhatIJd}
I(-s)=\sup_{\alpha\in\fJ_d}\bigl(\alpha s-e(\alpha)\bigr),\qquad
\hat I(s)=s+I(s).
\end{equation}
Moreover, if $e$ is steep, {\sl i.e.,} $a_\pm=\pm\infty$, then $I$ and $\hat I$
are good rates and the global FT holds. This construction provides a practical
route to the proof of the Fluctuation Theorem which has been followed by most
mathematically rigorous studies since the seminal work~\cite{Lebowitz1999}. The
same approach was adopted in~\cite{Jaksic2011}.

\section{Higher-level Fluctuation Theorems}

The main deficiency of the Gärtner--Ellis route to the FT is its failure in
the presence of {\sl phase transitions, i.e.,} non-differentiability of\index{phase transition}
the entropic pressure.\footnote{We shall see in our examples that
non-differentiability of the entropic pressure is indeed the signature of phase
transitions in the usual thermodynamical sense.} The problem is particularly
acute in cases where singularities of $e$ occur in the interval $[0,1]$, since
then even the local FT is out of reach. In Chapter~\ref{chap:Illustrations}, we shall give
some examples where such a situation occurs and for which it is therefore
necessary to develop other routes to the FT.

There are several alternatives to the Gärtner--Ellis theorem in the
literature:
\begin{itemize}
\item The \ndex{Donsker--Varadhan theory}~\cite{Donsker1975a,Donsker1975b} for additive
functionals of Markov processes is of particular interest due to the wide usage of
the Markov approximation in non-equilibrium statistical mechanics.

\item The approach developed by Ruelle and Lanford
in~\cite{Ruelle1965,Lanford1973} is well adapted to studying shift spaces over
finite alphabets under very general decoupling conditions for which the
thermodynamic formalism does not apply. Such decoupling conditions arise
naturally in multifractal analysis, in Gibbs states with hard-core interactions,
and in the statistics of repeated quantum measurement
processes~\cite{Cuneo2019}.

\item Kifer's criterion~\cite{Kifer1990} can be particularly useful for infinite
dimensional systems, see, e.g.~\cite{Jaksic2015b,Jaksic2018}.
\end{itemize}
These alternative methods also provide a way to lift the FT and the FR from
level-1, as formulated above, to level-2 or level-3. In the abstract setting of
Sections~\ref{Ssec:TransientFR} and~\ref{Ssect:FT}, such a lifting can be
formulated as follows. Suppose that for each $\lambda\in\cL$ there exists a
random variable $\xi_\lambda$ on $(\Omega_\lambda,\fF_\lambda)$, taking values
in a Polish space $\fX$ and such that the family
$(\xi_\lambda)_{\lambda\in\cL}$ satisfies the  LDP w.r.t.\;both
$(\PP_\lambda)_{\lambda\in\cL}$ and $(\wP_\lambda)_{\lambda\in\cL}$ with
scale $(r_\lambda)_{\lambda\in\cL}$ and respective rates $\II$ and
$\hat{\II}$. Suppose also that there is a measurable map $\varsigma:\fX\to\RR$
such that $r_\lambda\varsigma(\xi_\lambda)$ is physically equivalent to
$\sigma_\lambda$. In such a case, we shall say that this family satisfies the FT
whenever the lifted FR
\beq
\hat{\II}(\xi)=\II(\xi)+\varsigma(\xi)
\label{Eq:FRlift}
\eeq
holds for $\xi\in\fX$ (see~\cite{Maes2008,Bodineau2008,Barato2015}).
Setting
$$
I(s)\coloneq\inf\{\II(\xi)\mid\xi\in\fX,\varsigma(\xi)=s\},\qquad
\hat I(s)\coloneq\inf\{\hat\II(\xi)\mid\xi\in\fX,\varsigma(\xi)=s\},
$$
and observing that~\eqref{Eq:FRlift} implies that $(I,\hat I)$ satisfies the
FR~\eqref{Eq:RateSymmetry}, it follows from the contraction principle that the
FT holds for $(\PP_\lambda,\wP_\lambda)_{\lambda\in\cL}$ with scale
$(r_\lambda)_{\lambda\in\cL}$ and rates $(I,\hat I)$ provided that the rates
$\II$ and $\hat{\II}$ are good and the function $\varsigma$ is continuous.

\section{On the Meaning of the Fluctuation Theorem}
\label{sec:FR_Meaning}

Steady state FRs reduce to the Einstein--Green--Kubo fluctuation--dissipation formula
and Onsager's reciprocity relations in the limit of weak deviations from thermal
equilibrium, as explained in~\cite{Gallavotti1996a,Lebowitz1999,Jaksic2011}. In this
section, we shall derive consequences of the FR which extend beyond the linear
response regime. We refer to~\cite{Merhav2010,Seifert2012} for
discussions of other consequences of the steady state FRs.

Returning to the level of generality of Sections~\ref{Ssec:TransientFR}
and~\ref{Ssect:FT}, the rate function~$I$ and the associated entropic pressure $e$
can be linked to the \ndex{hypothesis testing} {\sl error exponents\/} of the family
$(\PP_\lambda,\wP_\lambda)_{\lambda\in\cL}$. These exponents
quantify the rate of separation between measures~$\PP_\lambda$
and~$\wP_\lambda$ as~$\lambda$ moves along the index set $\cL$. If the
elements of $\cL$ are instances of time and
$\wP_\lambda=\PP_\lambda\circ\Theta_\lambda$ where~$\Theta_\lambda$ implements
time reversal, these exponents quantify the emergence of the arrow of time and
can be viewed as quantitative forms of the second law of thermodynamics.

\part{Binary Hypothesis Testing.} Given $\omega\in\Omega_\lambda$, consider the
null hypothesis
$$
H_0: \omega \text{ was sampled from }\PP_\lambda,
$$
and the alternative hypothesis
$$
H_A: \omega \text{  was sampled from }\wP_\lambda.
$$
To any $\Gamma\in\fF_\lambda$, we associate the test according to which
$$
\text{if }\omega\in\Gamma\text{, then we reject }H_A\text{, otherwise we reject }H_0.
$$
The probability of the type-I error---wrongly rejecting the null hypothesis---is
$\PP_\lambda(\Gamma^c)$, while the type-II error---wrongly rejecting the
alternative hypothesis---has probability $\wP_\lambda(\Gamma)$. Thus, assuming
for simplicity that $\PP_\lambda$ and $\wP_\lambda$ have equal {\sl a priori}
probabilities,
$$
D(\PP_\lambda,\wP_\lambda,\Gamma)\coloneq\PP_\lambda(\Gamma^c)+\wP_\lambda(\Gamma),
$$
is twice the total error probability of the test $\Gamma$,
and a test $\Gamma_\mathrm{opt}$ is optimal whenever
$$
D(\PP_\lambda,\wP_\lambda,\Gamma_\mathrm{opt})=D(\PP_\lambda,\wP_\lambda)\coloneq
\inf_{\Gamma\in\fF_\lambda}D(\PP_\lambda,\wP_\lambda,\Gamma).
$$
The following variant of the \ndex{Neyman--Pearson lemma} identifies an optimal test:
\begin{proposition}\label{prop:Neyman}
Assume that $\PP_\lambda$ and $\wP_\lambda$ are equivalent. Then
\begin{equation}\label{eq:optimalchernoff}
D(\PP_\lambda,\wP_\lambda)=D(\PP_\lambda,\wP_\lambda,\{\sigma_\lambda\ge0\}),
\end{equation}
where $\sigma_\lambda$ is given by~\eqref{Eq:sigma_lambda}.
\end{proposition}

\proof For later reference, we prove the slightly more general statement
\beq
\inf_{\Gamma\in\fF_\lambda}\left(\PP_\lambda(\Gamma^c)
+\e^{sr_\lambda}\wP_\lambda(\Gamma)\right)
=\PP_\lambda\{\sigma_\lambda<sr_\lambda\}
+\e^{sr_\lambda}\wP_\lambda\{\sigma_\lambda\ge sr_\lambda\}.
\label{eq:NPgen}
\eeq
Indeed, for any $\Gamma\in\fF_\lambda$ and $s\in\RR$,
one has\footnote{$1_A$ denotes the indicator function of a set $A$.}
\begin{align*}
\PP_\lambda(\Gamma^c)+\e^{sr_\lambda}\wP_\lambda(\Gamma)
&=\EE_\lambda\left[1_{\Gamma^c}+\e^{sr_\lambda-\sigma_\lambda}1_\Gamma\right]\\
&\ge\EE_\lambda\left[\min(1,\e^{sr_\lambda-\sigma_\lambda})\right]\\
&=\EE_\lambda\left[1_{\{\sigma_\lambda<sr_\lambda\}}+\e^{sr_\lambda-\sigma_\lambda}1_{\{\sigma_\lambda\ge sr_\lambda\}}\right]\\
&=\PP_\lambda\{\sigma_\lambda<sr_\lambda\}+\e^{sr_\lambda}\wP_\lambda\{\sigma_\lambda\ge sr_\lambda\},
\end{align*}
which implies the claimed identity.
\hfill\qed

\medskip
Error exponents quantify the exponential decay rate of type-II errors as
$\lambda$ increases through $\cL$, under various conditions on the probability
of type-I errors.

\part{Stein error exponents.} The {\sl lower} and {\sl upper Stein exponents} are defined by
\index{Stein error exponents}
\begin{align*}
\underline{\fs}&\coloneq\inf\biggl\{\liminf_{\lambda\in\cL}r_\lambda^{-1}\log\wP_\lambda(\Gamma_\lambda)\,\bigg|\,
\Gamma_\lambda\in\fF_\lambda,\lim_{\lambda\in\cL}\PP_\lambda(\Gamma_\lambda^c)=0\biggr\},\\[2pt]
\bar\fs&\coloneq\inf\biggl\{\limsup_{\lambda\in\cL}r_\lambda^{-1}\log\wP_\lambda(\Gamma_\lambda)\,\bigg|\,
\Gamma_\lambda\in\fF_\lambda,\lim_{\lambda\in\cL}\PP_\lambda(\Gamma_\lambda^c)=0\biggr\}.
\end{align*}
Given $\gamma\in]0,1[$, we also set
$$
\fp_\gamma(\PP_\lambda,\wP_\lambda)
\coloneq\inf\bigl\{\wP_\lambda(\Gamma)\mid
\Gamma\in\fF_\lambda,\PP_\lambda(\Gamma^c)\le\gamma\bigr\}.
$$
The following result establishes a link between the Stein exponents, the
asymptotics of~$\fp_\gamma(\PP_\lambda,\wP_\lambda)$, and the weak law of large
numbers.

\begin{proposition}
\label{prop:Stein}
\hspace{0em}
\ben
\item Suppose that the net $(r_\lambda^{-1}\sigma_\lambda)_{\lambda\in\cL}$
converges in probability to a deterministic limit
$\epr\in[-\infty,+\infty]$.\footnote{This holds if, for instance, the FT holds
with a rate~$I$ vanishing at a unique point $\epr$.} Then, $\epr\ge0$ and, for
any $\gamma\in{]}0,1{[}$, one has
\begin{equation}
\underline{\fs}=\bar\fs
=\lim_{\lambda\in\cL}r_\lambda^{-1}\log\fp_\gamma(\PP_\lambda,\wP_\lambda)=-\epr.
\label{eq:Angelich}
\end{equation}
\item If the limit~\eqref{Eq:LegendreI} holds for $\alpha$ in a neighborhood
of\/ $0$, and  the entropic pressure~$e$ is differentiable at $\alpha=0$, then
\begin{equation}
\lim_{\lambda\in\cL}r_\lambda^{-1}\Ent(\PP_\lambda\mid\wP_\lambda)=-e^\prime(0).
\label{eq:destroy}
\end{equation}
If, in addition, the net $(r_\lambda^{-1}\sigma_\lambda)_{\lambda\in\cL}$
converges in the mean to a deterministic limit, then the conclusions of Part~(1)
hold with $\epr=-e'(0)$.
\een
\end{proposition}

\proof {\bf (1)} For any $\delta>0$, Equ.~\eqref{Eq:Jarzy_and_Markov} gives that
$$
\PP_\lambda\left\{r_\lambda^{-1}\sigma_{\lambda}\le-\delta\right\}\le\e^{-\delta r_\lambda},
$$
so that $\epr\in[0,+\infty]$. To prove~\eqref{eq:Angelich}, let us fix $\gamma\in{]}0,1{[}$.
We first show that
\begin{equation}
\underline{\fs}_\gamma\coloneq
\liminf_{\lambda\in\cL}r_\lambda^{-1}\log\fp_\gamma(\PP_\lambda,\wP_\lambda)\ge-\epr,
\label{eq:rainy2}
\end{equation}
assuming, w.l.o.g., that  $\epr<\infty$. The assumed convergence in probability
gives that for any $s>\epr$ there is $\lambda_s\in\cL$ such that
$\PP_\lambda\{\sigma_{\lambda}<sr_\lambda\}\ge(1+\gamma)/2$ for all
$\lambda\succ\lambda_s$. For any such $\lambda$ and any $\Gamma\in\fF_\lambda$,
by~\eqref{eq:NPgen},
\begin{align*}
\PP_\lambda(\Gamma^c)+\e^{sr_\lambda}\wP_\lambda(\Gamma)
&\geq \PP_\lambda\{\sigma_\lambda<sr_\lambda\}
+\e^{sr_\lambda}\wP_\lambda\{\sigma_\lambda\ge sr_\lambda\}
\ge\frac{1+\gamma}2,
\end{align*}
and it follows that
$$
\wP_\lambda(\Gamma)\ge\frac{1-\gamma}2\e^{-sr_\lambda}
$$
whenever $\lambda\succ\lambda_s$ and $\PP_\lambda(\Gamma^c)\le\gamma$. Thus,
$r_\lambda^{-1}\log\fp_\gamma(\PP_\lambda,\wP_\lambda)\ge -s$ holds for large
$\lambda$, and we conclude that $\underline{\fs}_\gamma\ge-s$ for any $s>\epr$,
which yields~\eqref{eq:rainy2}. Next, we prove the upper bound
\begin{equation}
\bar{\fs}\le-\epr.
\label{eq:rainy3}
\end{equation}
It follows from our assumption that, for any $s<\epr$, the sets
$\Gamma_\lambda=\{\sigma_\lambda\ge sr_\lambda\}\in\fF_\lambda$ satisfy
$\lim_{\lambda\in\cL}\PP_\lambda(\Gamma_\lambda^c)
=\lim_{\lambda\in\cL}\PP_\lambda\{r_\lambda^{-1}\sigma_\lambda<s\}=0$.
From the estimate
$$
\wP_\lambda(\Gamma_\lambda)=\EE_\lambda\left[\e^{-\sigma_\lambda}1_{\sigma_\lambda\ge sr_\lambda}\right]\le\e^{-sr_\lambda},
$$
we deduce that $r_\lambda^{-1}\log\wP_\lambda(\Gamma_\lambda)\le-s$, and we
conclude that $\bar{\fs}\le-s$ for any $s<\epr$, which yields~\eqref{eq:rainy3}.
The obvious inequalities $\underline{\fs}_\gamma\le\underline{\fs}$, and
$$
\bar{\fs}_\gamma
\coloneq\limsup_{\lambda\in\cL}r_\lambda^{-1}\log\fp_\gamma(\PP_\lambda,\wP_\lambda)\le\bar{\fs},
$$
together with~\eqref{eq:rainy2}-\eqref{eq:rainy3} give~\eqref{eq:Angelich} and
the claimed equalities $\underline{\fs}=\bar{\fs}=-\epr$.

\part{(2)} We repeat the argument in the proof of~\cite[Theorem~II.6.3]{Ellis1985}.
One has
$$
\Ent(\PP_\lambda\mid\wP_\lambda)=\EE_\lambda[\sigma_{\lambda}]=-e_\lambda'(0).
$$
Since the functions
$\alpha\mapsto e_\lambda(\alpha)=\log\EE_\lambda[\e^{-\alpha\sigma_\lambda}]$
are convex, the differentiability of the limiting function
$$
\alpha\mapsto e(\alpha)=\lim_{\lambda\in\cL}r_\lambda^{-1}e_\lambda(\alpha)
$$
at $\alpha=0$ gives, by~\cite[Theorem~25.7]{Rockafellar1970},
$$
e'(0)=\lim_{\lambda\in\cL}r_\lambda^{-1}e_\lambda'(0),
$$
from which~\eqref{eq:destroy} follows. Finally, it follows from the convergence
in the mean
$$
\lim_{\lambda\in\cL}\EE_\lambda[|r_\lambda^{-1}\sigma_\lambda-\epr|]=0,
$$
that
$$
\lim_{\lambda\in\cL}r_\lambda^{-1}\sigma_\lambda=\epr=-e'(0)
$$
holds in probability, so that Part~(1) applies.
\hfill\qed

\medskip
Recall that the {\it total variation distance\/} of two probability
measures~$\PP$ and~$\QQ$ on a measurable space~$(\Omega,\fF)$ is
$$
\|\PP-\QQ\|_{\mathrm{var}}
\coloneq\sup_{\Gamma\in\fF}\left(\PP(\Gamma)-\QQ(\Gamma)\right).
$$
If $\QQ$ is absolutely continuous with respect to~$\PP$, and~$\Delta$ stands for
the corresponding density, then
$$
\PP(\Gamma)-\QQ(\Gamma)=1-(\PP(\Gamma^c)+\QQ(\Gamma))
=1-\int_\Omega\left(1_{\Gamma^c}+\Delta 1_\Gamma\right)\d\PP
\le1-\int_\Omega\min(1,\Delta)\d\PP.
$$
Observing that the last inequality is saturated for $\Gamma=\{\Delta\le1\}$, we get
$$
\|\PP-\QQ\|_{\mathrm{var}}=1-\int_\Omega\min(1,\Delta)\d\PP.
$$
Thus, assuming that $\PP_\lambda$ and $\wP_\lambda$ are equivalent, it follows
from Lemma~\ref{prop:Neyman} that
\begin{equation}
1-\|\PP_\lambda-\wP_\lambda\|_{\mathrm{var}}
=\EE_\lambda\left[\min(1,\e^{-\sigma_\lambda})\right]=D(\PP_\lambda,\wP_\lambda).
\label{eq:NPChernoff}
\end{equation}

\part{Chernoff error exponents.} The {\sl lower} and {\sl upper Chernoff exponents} are
\index{Chernoff error exponents}
\begin{equation*}
\begin{split}
\underline{c}&\coloneq\liminf_{\lambda\in\cL}r_\lambda^{-1}
\log\bigl(1-\|\PP_\lambda-\wP_\lambda\|_{\mathrm{var}}\bigr),\\[6pt]
\bar{c}&\coloneq\limsup_{\lambda\in\cL}r_\lambda^{-1}
\log\bigl(1-\|\PP_\lambda-\wP_\lambda\|_{\mathrm{var}}\bigr).
\end{split}
\end{equation*}
The following result provides bounds on these exponents in terms of the rate $I$
and the entropic pressure $e$.

\begin{proposition}
\hspace{0em}
\ben
\item Suppose that $\PP_\lambda$ and $\wP_\lambda$ are equivalent measures for
all $\lambda\in\cL$, and that the FT holds with the scale
$(r_\lambda)_{\lambda\in\cL}$ and rates~$(I,\hat I)$. Then\footnote{Under the
stated assumptions, a slight refinement of the proof below, using also the LDP
upper bound, shows that actually
$\underline{c}=\bar{c}=-\min\big(-\inf_{s\le0}I(s),-\inf_{s>0}\hat I(s)\big)$.}
\begin{equation}
-I(0)\le\underline{c}\le\bar{c}
\le\inf_{\alpha\in[0,1]}\bar{e}(\alpha).
\label{eq:genere}
\end{equation}
\item If, in addition, the rate $I$ is convex and~\eqref{Eq:LegendreI} holds for
$\alpha\in[0,1]$, then the upper and lower Chernoff exponents coincide, and
\footnote{Note that if the family $(\PP_\lambda,\wP_\lambda)_{\lambda\in\cL}$ is
in involution, then	$-I(0)=e(1/2)$.}
\begin{equation}
\underline{c}=\bar{c}=\lim_{\lambda\in\cL}r_\lambda^{-1}
\log\bigl(1-\|\PP_\lambda-\wP_\lambda\|_{\mathrm{var}}\bigr)
=-I(0) = \inf_{\alpha\in[0,1]}\bar{e}(\alpha).
\label{4.17}
\end{equation}
\een
\end{proposition}

\proof {\bf(1)} It follows from~\eqref{eq:NPChernoff} that, for $\alpha\in[0,1]$,
$$
1-\|\PP_\lambda-\wP_\lambda\|_{\mathrm{var}}
=\EE_\lambda\left[\min(1,\e^{-\sigma_\lambda})\right]
\le\EE_\lambda\left[\e^{-\alpha\sigma_\lambda}\right]=\e^{e_\lambda(\alpha)}.
$$
Thus, taking the limit over $\lambda\in\cL$,
$$
\overline{c}\le\limsup_{\lambda\in\cL}r_\lambda^{-1}e_\lambda(\alpha)
=\bar{e}(\alpha),
$$
which yields $\overline{c}\le\inf_{\alpha\in[0,1]}\bar{e}(\alpha)$. Since
$\underline{c}\leq \overline{c}$, it remains to prove that
$\underline{c}\ge-I(0)$ in order to establish \eqref{eq:genere}. Observe that
for all $\varepsilon>0$,
\begin{align*}
	1-\|\PP_\lambda-\wP_\lambda\|_{\mathrm{var}}
&=D(\PP_\lambda,\wP_\lambda)
=\PP_\lambda\{\sigma_\lambda <0\}+\wP_\lambda\{\sigma_\lambda\ge 0\}
\ge \e^{-r_\lambda \varepsilon}\PP_\lambda\{\sigma_\lambda\le r_\lambda \varepsilon\},
\end{align*}
where we used again Lemma~\ref{prop:Neyman} and the definition of
$\sigma_\lambda$. Hence, invoking the LDP lower bound, we derive
$\underline{c}\ge-\inf_{s<\varepsilon}I(s)-\varepsilon\ge-I(0)-\varepsilon$, and
since $\varepsilon>0$ was arbitrary, we obtain $\underline{c}\ge-I(0)$.

\part{(2)} Suppose now that~\eqref{Eq:LegendreI} holds and that $I$ is convex.
In view of~\eqref{Eq:IasLegendreTransform}, one has $-I(0)=\inf_{\alpha\in\RR}e(\alpha)$.
Since $e(\alpha)$ is non-positive in $[0,1]$ and non-negative outside $[0,1]$, we
obtain $-I(0)=\inf_{\alpha\in[0,1]}e(\alpha)$, which, together with~\eqref{eq:genere},
implies~\eqref{4.17}.

\hfill\qed

\medskip
Thus, if\,  $\overline c > 0$, then the measures~$\PP_\lambda$ and~$\wP_\lambda$
concentrate on the disjoint subsets $\{\sigma_\lambda>0\}$
and~$\{\sigma_\lambda<0\}$, respectively, and separate with an exponential rate given
by, at least, $r_\lambda \overline c$, in the sense that:
\begin{equation}
\limsup_{\lambda\in\cL}r_\lambda^{-1}\log\PP_\lambda\{\sigma_\lambda\leq0\}\leq \overline c, \qquad
\limsup_{\lambda\in\cL}r_\lambda^{-1}\log\wP_\lambda\{\sigma_\lambda\geq0\}\leq \overline c.
\label{eq:Exposep}
\end{equation}
The two bounds in~\eqref{eq:Exposep} follow from the definition of $\overline c$
and the fact that both $\PP_\lambda\{\sigma_\lambda\leq0\}$ and
$\wP_\lambda\{\sigma_\lambda\geq0\}$ are bounded above by the quantity
$D(\PP_\lambda,\wP_\lambda,\{\sigma_\lambda\ge0\})$
in~\eqref{eq:optimalchernoff}.

The last result in this section provides more detailed information on the rates
at which the measures $\PP_\lambda$ and $\wP_\lambda$ concentrate on
complementary subsets. This information is carried by a third kind of exponents:

\part{Hoeffding error exponents.} Given $u\in\RR$, define\index{Hoeffding error exponents}
\begin{align*}
\underline{\fh}(u)
&\coloneq\inf\Bigl\{\liminf_{\lambda\in\cL}
r_\lambda^{-1}\log\wP_\lambda(\Gamma_\lambda)\,\Big|\,
\Gamma_\lambda\in\fF_\lambda,
\limsup_{\lambda\in\cL}
r_\lambda^{-1}\log\PP_\lambda(\Gamma_\lambda^c)<-u\Bigr\},\\[4pt]
\overline{\fh}(u)
&\coloneq\inf\Bigl\{\limsup_{\lambda\in\cL}
r_\lambda^{-1}\log\wP_\lambda(\Gamma_\lambda)\,\Big|\,
\Gamma_\lambda\in\fF_\lambda,
\limsup_{\lambda \in\cL}
r_\lambda^{-1}\log\PP_\lambda(\Gamma_\lambda^c)< -u\Bigr\},\\[4pt]
\fh(u)
&\coloneq\inf\Bigl\{\lim_{\lambda\in\cL}
r_\lambda^{-1}\log\wP_\lambda(\Gamma_\lambda)\,\Big|\,
\Gamma_\lambda\in\fF_\lambda,
\limsup_{\lambda \in\cL}
r_\lambda^{-1}\log\PP_\lambda(\Gamma_\lambda^c)< -u\Bigr\},
\end{align*}
where the infimum in the last relation is taken over all
families~$(\Gamma_\lambda)_{\lambda\in\cL}$ for which the limit exists.

\medskip
It  immediately follows from these definitions that $\underline{\fh}$,
$\overline{\fh}$, and~$\fh$ are non-decreasing and non-positive functions of $u$
such that
\beq
\underline{\fh}(u)=\overline{\fh}(u)=\fh(u)=-\infty
\label{eq:hone}
\eeq
for $u<0$. Moreover, the inequalities
\beq
\underline{\fh}(u)\le\overline{\fh}(u)\le\fh(u)\le0
\label{eq:htwo}
\eeq
hold for all $u\geq 0$.

\begin{proposition}\label{prop:Hoeffding}
Suppose that, for all $\lambda\in\cL$, $\PP_\lambda$ and $\wP_\lambda$ are
equivalent measures such that the FT holds with the scale
$(r_\lambda)_{\lambda\in\cL}$ and convex rates~$(I,\hat I)$. If the associated
entropic pressure satisfies
$$
e(\alpha)=\lim_{\lambda\in\cL}e_\lambda(\alpha)
$$
for all $\alpha\in[0,1]$, then the Hoeffding error exponents coincide and are given by
\begin{equation}
\underline{\fh}(u)=\overline{\fh}(u)=\fh(u)=-f(u)
\label{eq:hhhufu}
\end{equation}
for all $u\in\RR$, where the function $f$ is related to the entropic pressure by
\beq
f(u)\coloneq\sup_{\alpha\in]0,1]}\frac{-(1-\alpha)u-e(\alpha)}{\alpha}.
\label{Eq:fDef}
\eeq
\end{proposition}

The proof being quite technical, we postpone it to the final appendix to this chapter.

\section*{Appendix: Proof of Proposition~\ref{prop:Hoeffding}}

We start with some basic properties of the rate function $I$, its convex hull
$\breve{I}$~\eqref{eq:nabucco}, the corresponding entropic pressure
$e$~\eqref{Eq:ealphaDef}, and the function $f$ defined in~\eqref{Eq:fDef}.
Although convexity is assumed in Proposition~\ref{prop:Hoeffding}, we do not
assume that $I$ is convex in the upcoming technical discussion.

We recall that, as the Fenchel--Legendre transform of the function $s\mapsto I(-s)$,
the entropic pressure $e$ is a closed proper convex function on $\RR$.
It is continuous and admits (possibly infinite) left/right derivatives
$\partial^\mp e(\alpha)$ at each point $\alpha$ of its \ndex{essential domain}
\begin{equation}
\Dom(e)\coloneq\{\alpha\in\RR\mid e(\alpha) \text{ is finite}\}\supset[0,1],
\label{eq:defdome}
\end{equation}
see~\cite[Sections~23-24]{Rockafellar1970}. By convention, when $\alpha\notin\Dom(e)$,
we set $\partial^-e(\alpha)=\partial^+e(\alpha)=-\infty$ if $\alpha\leq\inf\,\Dom(e)$
and $\partial^-e(\alpha)=\partial^+e(\alpha)=\infty$ if $\alpha\geq\sup\,\Dom(e)$.
These derivatives satisfy $\partial^-e(\alpha)\le\partial^+e(\alpha)$ with equality
outside a countable subset of $\Dom(e)$, $\partial^+e(\alpha)\le\partial^-e(\alpha')$
whenever $\alpha\le\alpha'$, and for all $\alpha\in\RR$,
\begin{equation}
\lim_{\gamma\downarrow\alpha}\partial^\pm e(\gamma)=\partial^+e(\alpha),\qquad
\lim_{\gamma\uparrow\alpha}\partial^\pm e(\gamma)=\partial^-e(\alpha).
\label{eq:derlim}
\end{equation}
The \ndex{subdifferential} of $e$ at $\alpha\in\Dom(e)$ is defined as the closed interval
$$
\partial e(\alpha)\coloneq[\partial^-e(\alpha),\partial^+e(\alpha)]\cap \RR.
$$

\bep\label{Prop:s}
Assume that the FT holds for the family $(\PP_\lambda,\wP_\lambda)_{\lambda\in\cL}$
with scale $(r_\lambda)_{\lambda\in\cL}$ and rate $(I,\hat I)$. Let
$$
e(\alpha)=\sup_{s\in\RR}(\alpha s-I(-s))
$$
denote the associated entropic pressure, and
\beq
\breve{I}(s)=\sup_{\alpha\in\RR}(\alpha s-e(-\alpha))
\label{eq:breveI}
\eeq
the convex hull of $I$. Let $s_*\in\RR_+$ and
$\underline s_0,\overline s_0,\underline s_1,\overline s_1\in[0,\infty]$
be defined by
\[
s_\ast\coloneq-\inf_{\alpha\in\RR}e(\alpha),\qquad
\partial e(0)\eqcolon[-\overline{s}_0,-\underline{s}_0]\cap\RR,\qquad
\partial e(1)\eqcolon[\underline{s}_1,\overline{s}_1]\cap\RR,
\]
where, following the above conventions, $\overline{s}_0=+\infty$ whenever
$e(\alpha)=+\infty$ for all $\alpha<0$ and $\overline{s}_1=+\infty$
whenever $e(\alpha)=+\infty$ for all $\alpha>1$.

Then, the following hold:
\ben
\item For any $\alpha\in\RR$,
\[
\hat e(\alpha)\coloneq\sup_{s\in\RR}\bigl(\alpha s-\hat I(s)\bigr)=e(1-\alpha),
\]
and in particular
\[
s_\ast=-\inf_{\alpha\in\RR}\hat e(\alpha),\qquad
\partial\hat e(0)=[-\overline{s}_1,-\underline{s}_1]\cap\RR,\qquad
\partial\hat e(1)=[\underline{s}_0,\overline{s}_0]\cap\RR.
\]
Moreover, the convex hull of $\hat I$ is given by
$$
J(s)=\sup_{\alpha\in\RR}(\alpha s-\hat{e}(\alpha))=\breve{I}(s)+s.
$$
\item Denote by $\mathrm{co}(A)$ the convex hull of the set $A\subset\RR$,
{\sl i.e.,} the smallest convex subset of\, $\RR$ containing $A$. If
$\underline{s}_0<\infty$, then
\beq
\mathrm{co}(I^{-1}(\{0\}))\subset\breve{I}^{-1}(\{0\})
=[\underline{s}_0,\overline{s}_0]\cap \RR.
\label{eq:carmen}
\eeq
If $\underline{s}_1<\infty$, then
\beq
\mathrm{co}(\hat{I}^{-1}(\{0\}))\subset J^{-1}(\{0\})
=[-\overline{s}_1,-\underline{s}_1]\cap \RR.
\label{eq:jose}
\eeq
Moreover, the inclusion in~\eqref{eq:carmen} and~\eqref{eq:jose} becomes an
equality whenever its right-hand side is compact.
\item  For all $s\in[-\overline s_0, \overline s_1]\cap \RR$,
\begin{equation}
\breve{I}(-s)=\max_{\alpha\in[0,1]}(\alpha s-e(\alpha)).
\label{eq:invertLTrestr}
\end{equation}
\item One has
\beq
0\le\breve{I}(0)=s_\ast\le\min(\underline{s}_0,\underline{s}_1),
\label{eq:salome}
\eeq
where the two inequalities become equalities iff $e$ vanishes identically on $[0,1]$.
\item $s_\ast=0\,\Longleftrightarrow\,\underline{s}_0=\underline{s}_1=0$.
\item $f$ is a non-increasing closed proper convex function satisfying
\[
f(u)=\left\{
\begin{array}{ll}
+\infty&\text{for }u<0;\\[4pt]
\underline{s}_0&\text{for }u=0;\\[4pt]
s_\ast&\text{for }u=s_\ast;\\[4pt]
0&\text{for }u\ge\underline{s}_1.
\end{array}
\right.
\]
If $\underline{s}_0=\infty$ (resp.\;$\underline{s}_1=\infty$), then the second
(resp.\;last) equality in this list must be replaced by the relation
$$
f(0)=\lim_{u\downarrow0}f(u)=+\infty,
\qquad(\text{resp. }\lim_{u\uparrow\infty}f(u)=0).
$$
\item The function $s\mapsto\breve{I}(s)$ is strictly decreasing on
$]{-}\underline{s}_1,\underline{s}_0{[}$ and maps this interval
homeomorphically onto ${]}0,\underline{s}_1{[}$. The reciprocal
map is given by the restriction to ${]}0,\underline{s}_1{[}$ of the
function $g$ defined by $g(u)=f(u)-u$. In particular,
$$
f\circ\breve{I}(s)=J(s)
$$
holds for all $s\in{]}{-}\underline{s}_1,\underline{s}_0{[}$. Moreover,
if $\underline{s}_0<\infty$ (resp.\;$\underline{s}_1<\infty$), then
$g(0)=\underline{s}_0$ and $\breve{I}(g(0))=0$
(resp.\;$g(\underline{s}_1)=-\underline{s}_1$, and
$\breve{I}(g(\underline{s}_1))=\underline{s}_1$).
\item If $(\PP_\lambda,\wP_\lambda)_{\lambda\in\cL}$ is in involution,
then $\hat{e}=e$, $\underline{s}_0=\underline{s}_1$,
$\overline{s}_0=\overline{s}_1$, $s_\ast=-e(1/2)$, and the function
$f:[0,\underline{s}_0{[}\to[0,\underline{s}_0{[}$ is an involution
such that $f\circ\breve{I}(s)=\breve{I}(-s)$.
\een
\eep

\proof {\bf(1)} Follows directly from the asymptotic FR~\eqref{Eq:RateSymmetry}.

\part{(2)} Since $e$ is convex and vanishes for $\alpha\in\{0,1\}$, we must have
$\underline{s}_0\ge0$. Moreover, $s<\underline{s}_0$, {\sl i.e.,}
$-s>\partial^+e(0)$ implies that there exists $\alpha\in]0,1[$ such that
$-\alpha s>e(\alpha)$. It follows from Relation~\eqref{eq:breveI} that
$\breve{I}(s)>0$ and hence $I(s)>0$. If $\overline{s}_0<\infty$, then a
completely similar argument applies to $s>\overline{s}_0$. For
$-s\in{]}{-}\overline{s}_0,-\underline{s}_0{[}\subset\partial e(0)$,
Relation~\eqref{Eq:LegendreForma} yields that $\breve{I}(s)=0$, so the second
equality in~\eqref{eq:carmen} holds by continuity.

Since $I(s)\ge\breve{I}(s)\ge0$, one has $I^{-1}(\{0\})\subset\breve{I}^{-1}(\{0\})$,
and the convexity of $\breve{I}$
further implies that $\breve{I}^{-1}(\{0\})=\breve{I}^{-1}({]}{-}\infty,0{]})$
is convex, hence $\mathrm{co}(I^{-1}(\{0\}))\subset\breve{I}^{-1}(\{0\})$. In
cases where $\breve{I}^{-1}(\{0\})$ is compact, one has
$\breve{I}^{-1}(\{0\})=[a,b]$, with $\breve{I}(a)=\breve{I}(b)=0$ and
$\breve{I}(s)>0$ for $s\not\in[a,b]$. To prove the claimed equality
$\mathrm{co}(I^{-1}(\{0\}))=\breve{I}^{-1}(\{0\})$, it is enough to show that
$I(a)=0=I(b)$. We consider only the case of $I(a)$, by symmetry the same
reasoning applies to $I(b)$. We argue by contradiction, assuming that $I(a)>0$.
By lower semicontinuity, for some  $A>0$ and $\delta>0$,  $|s-a|<\delta$ implies
$I(s)>A$. It follows that there exists a convex function $\bar{I}\leq I$ such
that $\bar{I}(a)>0$, which contradicts the definition of $\breve{I}$. The
figures~\ref{Fig:Boxes} below show how to modify $\breve{I}$ to get $\bar I$ in
the five possible cases:
\begin{enumerate}[label=(\arabic*)]
\item $a<b$, $0<\breve{I}(s)<\infty$ for $a-\delta<s<a$ and $\breve{I}(s)=0$ for $a\le s<a+\delta$;
\item $a<b$, $\breve{I}(s)=\infty$ for $s<a$ and $\breve{I}(s)=0$ for $a\le s<a+\delta$;
\item $a=b$, $0<\breve{I}(s)<\infty$ for $0<|s-a|<\delta$;
\item $a=b$, $0<\breve{I}(s)<\infty$ for $a-\delta<s<a$ and $\breve{I}(s)=\infty$ for $s>a$;
\item $a=b$, $\breve{I}(s)=\infty$ for $s\neq a$.
\end{enumerate}
\begin{figure}
\begin{center}
\includegraphics[width=1\textwidth]{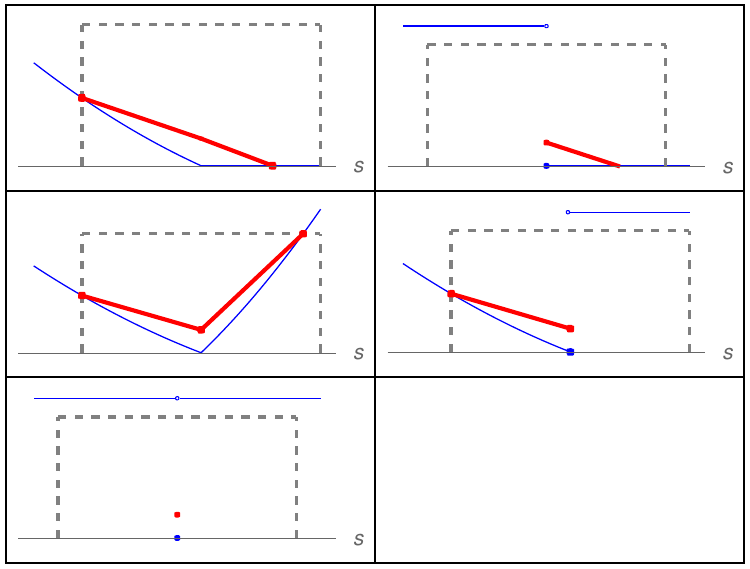}
\end{center}
\caption{The function $\breve{I}$ (thin lines) and the modification (thick
lines) leading to $\bar I$ for the 5 cases listed in the text. Since $I(a)>0$,
lower semicontinuity implies the existence of the dashed box
${]}a-\delta,a+\delta{[}\times[0,A]$ having no intersection with the graph of
$I$.}
\label{Fig:Boxes}
\end{figure}
Relations~\eqref{eq:jose} follow from the same arguments.

\part{(3)} Let $s\in[-\overline s_0, \overline s_1]\cap \RR = [\partial^-e(0),
\partial^+e(1)]\cap \RR$. Then, there exists some $\alpha_s \in [0,1]$ such that
$-s\in \partial e(\alpha_s)$, so~\eqref{Eq:LegendreForma}, applied here to
$\breve I$, implies that $I(-s)=\alpha_s s-e(\alpha_s)$. Since also
$\breve{I}(-s)=\sup_{\alpha\in \RR}(\alpha s-e(\alpha))$ by~\eqref{eq:breveI},
we obtain~\eqref{eq:invertLTrestr}.

\part{(4)}  By \eqref{eq:breveI} we have $\breve{I}(0)=s_\ast$ and
by~\eqref{eq:invertLTrestr} we further obtain
\beq
s_\ast=-\min_{\alpha\in[0,1]}e(\alpha).
\label{Eq:sstarForm}
\eeq
Since $e$ is convex and $e(0)=e(1)=0$, we conclude that $s_\ast=0$ iff $e$
vanishes identically on $[0,1]$. Invoking convexity, one further has
$e(\alpha)\ge\max(-\alpha\underline{s}_0,-(1-\alpha)\underline{s}_1)$ for all
$\alpha\in\RR$. Using Relation~\eqref{Eq:sstarForm} we derive
\[
s_\ast=\sup_{\alpha\in[0,1]}(-e(\alpha))\le
\sup_{\alpha\in[0,1]}\min(\alpha\underline{s}_0,(1-\alpha)\underline{s}_1)
=\min(\underline{s}_0,\underline{s}_1).
\]
Finally, we observe that the right-hand side of the last equality vanishes
iff $e$ vanishes identically on $[0,1]$.

\part{(5)} By~\eqref{eq:salome}, we only have to show that
$s_\ast=0\,\Longrightarrow\,\underline{s}_0=\underline{s}_1=0$. By Part~(4),
$s_\ast=0$ implies that $e$ vanishes on $[0,1]$, which further yields
$\underline{s}_0=-\partial^+e(0)=0$ and $\underline{s}_1=\partial^-e(1)=0$.

\part{(6)} From its definition~\eqref{Eq:fDef} as the pointwise supremum of a
family of non-increasing affine functions $(f_\alpha)_{\alpha\in{]}0,1]}$
vanishing at $\alpha=1$, we immediately infer that $f$ is a non-negative and
non-increasing, lower semicontinuous convex function such that $f(u)=+\infty$
for $u<0$. If $\underline{s}_0<\infty$, then from the fact that
$e(\alpha)\ge-\alpha\underline{s}_0$, we derive that $f(u)\le\underline{s}_0$
for $u\ge0$, with equality for $u=0$. If $\underline{s}_0=\infty$, then the
convexity of $e$ implies
$$
f(0)=\sup_{\alpha\in{]}0,1{]}}\frac{-e(\alpha)}{\alpha}
=-\inf_{\alpha\in{]}0,1{]}}\frac{e(\alpha)-e(0)}{\alpha-0}
=-\lim_{\alpha\downarrow0}\frac{e(\alpha)-e(0)}{\alpha-0}
=-\partial^+e(0)=\underline{s}_0=\infty,	
$$
and lower semicontinuity gives
$$
\liminf_{u\to0}f(u)\ge f(0)=\infty.
$$
If $\underline{s}_1<\infty$, then invoking again convexity we derive
\[
\underline{s}_1=\partial^-e(1)
=\lim_{\alpha\uparrow1}\frac{e(1)-e(\alpha)}{1-\alpha}
=-\lim_{\alpha\uparrow1}\frac{e(\alpha)}{1-\alpha},
=-\inf_{\alpha\in]0,1[}\frac{e(\alpha)}{1-\alpha},
\]
so that, for $\alpha\in]0,1[$,
\[
\frac{-(1-\alpha)u-e(\alpha)}{\alpha}
=\frac{1-\alpha}{\alpha}\left(-u-\frac{e(\alpha)}{1-\alpha}\right)
\le \frac{1-\alpha}{\alpha}\left(\underline{s}_1-u\right),
\]
which gives that $f(u)=0$ for $u\ge\underline{s}_1$.
If $\underline{s}_1=\infty$, then the function
$$
q(\alpha)\coloneq\frac{e(\alpha)-e(1)}{\alpha-1}
$$
is continuous, non-decreasing on the interval ${]}0,1{[}$, and $\lim_{\alpha\uparrow1}q(\alpha)=\underline{s}_1=\infty$.
Setting
$$
\alpha_\ast(u)\coloneq\sup\left\{\alpha\in{]}0,1{[}\mid q(\alpha)<u\right\},
$$
continuity gives $q(\alpha_\ast(u))=u$, so that
\beq
\lim_{u\to\infty}\alpha_\ast(u)=1
\label{eq:alphastar}
\eeq
and
$$
f(u)=\sup_{\alpha\in[\alpha_\ast(u),1]}\frac{(1-\alpha)(q(\alpha)-u)}{\alpha}
\le\frac1{\alpha_\ast(u)}\sup_{\alpha\in[\alpha_\ast(u),1]}\left((\alpha-1)u-e(\alpha)\right)
=\frac{\Delta(u)}{\alpha_\ast(u)},	
$$
where $\Delta(u)\coloneq(\bar{\alpha}(u)-1)u-e(\bar{\alpha}(u))$ is the largest
vertical distance between the graphs of $\alpha\mapsto e(\alpha)$ and
$\alpha\mapsto(\alpha-1)u$ over the interval $[\alpha_\ast(u),1]$, see
Figure~\ref{Fig:q}. Since $\alpha_\ast(u)\le\bar{\alpha}(u)$, for large enough
$u$ one has $\Delta(u)\le -e(\bar{\alpha}(u))\le -e(\alpha_\ast(u))$ and hence,
by~\eqref{eq:alphastar},
$$
\lim_{u\to\infty}f(u)\le\lim_{u\to\infty}\frac{-e(\alpha_\ast(u))}{\alpha_\ast(u)}=0.
$$
\begin{figure}
\begin{center}
\includegraphics[width=0.66\textwidth]{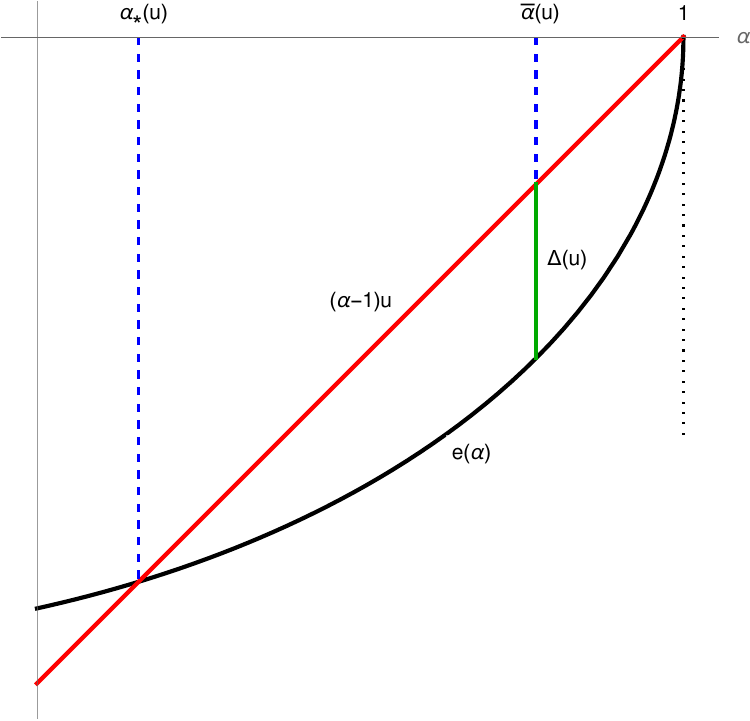}
\end{center}
\caption{\normalsize As the slope $u$ of the secant $\alpha\mapsto(\alpha-1)u$
diverges, its point of intersection $(\alpha_\ast(u),e(\alpha_\ast(u)))$ with
the graph of $e$ approaches $(1,0)$. Consequently, the maximal vertical distance
$\Delta(u)$ vanishes.}
\label{Fig:q}
\end{figure}

If $s_\ast=0$, then Part~(4) implies $f(s)=0$ for $s\ge0$, and we may assume that
$s_\ast>0$. By~\eqref{Eq:sstarForm} one has $e(\alpha)+s_\ast\ge0$, with equality
for some $\alpha\in]0,1[$. We conclude that
\[
f(s_\ast)=\sup_{\alpha\in]0,1]}\frac{-(1-\alpha) s_\ast-e(\alpha)}{\alpha} =  s_\ast -\inf_{\alpha\in]0,1]}\frac{e(\alpha)+s_\ast}{\alpha}=s_\ast.
\]

\part{(7)} As the sum of the non-increasing closed proper convex function $f$
and a strictly decreasing affine function, $g$ is a strictly decreasing closed
proper convex function. In particular, it is injective and, by Part~(6), it maps
${]}0,\underline{s}_1{[}$ homeomorphically onto
${]}{-}\underline{s}_1,\underline{s}_0{[}$ with
$$
\lim_{u\downarrow 0}g(u)=\underline{s}_0,\qquad\lim_{u\uparrow\underline{s}_1}g(u)
=-\underline{s}_1.
$$
For $s\in{]}{-}\underline{s}_1,\underline{s}_0{[}=-{]}\partial^+e(0),\partial^-e(1){[}$,
there exists $\gamma\in{]}0,1{[}$ such that $-s\in\partial e(\gamma)$ and
Relation~\eqref{Eq:LegendreForma} yields that
$$
r\coloneq\breve{I}(s)=-\gamma s-e(\gamma).
$$
By convexity, one further has $e(\alpha)\ge e(\gamma)+(\gamma-\alpha)s$ and it
follows that for $\alpha\in]0,1]$
$$
\frac{-r-e(\alpha)}{\alpha}=\frac{\gamma s+e(\gamma)-e(\alpha)}{\alpha}\le s,
$$
with equality for $\alpha=\gamma$. Thus, we can conclude that
$$
g\circ\breve{I}(s)=\sup_{\alpha\in]0,1]}\frac{-r-e(\alpha)}{\alpha}=s.
$$
If $\underline{s}_0<\infty$, then Part~(2) gives $\breve{I}(\underline{s}_0)=0$
and hence
$$
g(\breve{I}(\underline{s}_0))=g(0)=f(0)=\underline{s_0}.
$$
If $\underline{s}_1<\infty$, then the same argument gives
$\breve{I}(-\underline{s}_1)=\underline{s}_1$ so that
$$
g(\breve{I}(-\underline{s}_1))=g(\underline{s_1})
=f(\underline{s_1})-\underline{s_1}=-\underline{s_1}.
$$
It follows in particular that $\breve{I}$ is strictly decreasing on
${]}{-}\underline{s}_1,\underline{s}_0{[}$. Finally, we note that
\beq
f\circ\breve{I}(s)=\breve{I}(s)+g\circ\breve{I}(s)=\breve{I}(s)+s=J(s),
\label{eq:ludwig}
\eeq
and that the last assertions are direct consequences of Parts~(2) and~(6).

\part{(8)} In the involutive case, an immediate consequence
of~\eqref{Eq:hatRate} is that $\hat e=e$, and hence that
$e(1-\alpha)=\hat{e}(\alpha)=e(\alpha)$ for all $\alpha\in\RR$.
It follows that
$$
\underline{s}_1=\partial^-e(1)=-\partial^+e(0)=\underline{s}_0,\qquad
\overline{s}_1=\partial^+e(1)=-\partial^-e(0)=\overline{s}_0,
$$
and by symmetry $s_\ast=\inf_{\alpha\in[0,1]}e(\alpha)=e(1/2)$.
Moreover, $\breve{I}(-s)=\breve{I}(s)+s$ for any $s\in\RR$, and hence
Relation~\eqref{eq:ludwig} gives $\breve{I}(-s)=f(\breve{I}(s))$, so that
$$
f\circ f(\breve{I}(s))=f(\breve{I}(-s))=\breve{I}(s).
$$
It follows that $f\circ f=\Id$ on
$[0,\underline{s}_0]=\breve{I}([-\underline{s}_0,\underline{s}_0])$.
\hfill\qed

\bigskip
We now proceed with the proof of Proposition~\ref{prop:Hoeffding}, and for this
we assume that $I$ and $\hat I$ are convex, {\sl i.e.,} that $I=\breve I$ and
$\hat I = J$. We use the notation of Proposition~\ref{Prop:s},  and follow the
arguments of~\cite[Section~6.4]{Jaksic2012}. The rate functions $I$ and $\hat I$
are depicted in Figure~\ref{Fig:IhatI}.

\begin{figure}
\begin{center}
\includegraphics[width=0.8\textwidth]{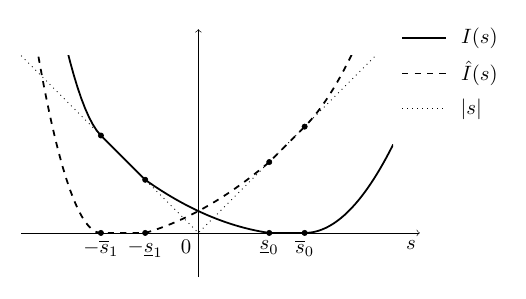}
\end{center}
\caption{The graph of $I$ and $\hat I$ in a situation where $I$ is convex (so
$I=\breve I$ and $\hat I = J$), and where $-\infty <-\overline s_1 <-\underline
s_1 <0<\underline s_0 < \overline s_0 < \infty$, in the notation of
Proposition~\ref{prop:Hoeffding}.}
\label{Fig:IhatI}
\end{figure}

First, we notice that~\eqref{eq:hhhufu} is immediate when $u<0$, in view
of~\eqref{eq:hone} and Proposition~\ref{Prop:s}~(6). We now
establish~\eqref{eq:hhhufu} in the case $u\geq 0$.

Consider a family $(\Gamma_\lambda)_{\lambda\in\cL}\subset\fF$. For all
$s\in\RR$ and $\varepsilon>0$, using first~\eqref{eq:NPgen}, then the
definition of $\sigma_\lambda$, and finally the LDP lower bound, we find
\begin{align*}
\liminf_{\lambda\in\cL}r_\lambda^{-1}
&\log\left(\PP_\lambda(\Gamma_\lambda^c)+\e^{sr_\lambda}\wP_\lambda(\Gamma_\lambda)\right)\\
\geq&\liminf_{\lambda\in\cL}r_\lambda^{-1}\log\left(\PP_\lambda\{\sigma_\lambda<sr_\lambda\} +\e^{sr_\lambda}\wP_\lambda\{\sigma_\lambda\ge sr_\lambda\}\right)\\
 \geq&\liminf_{\lambda\in\cL}r_\lambda^{-1}\log\left(\e^{-\varepsilon r_\lambda }\PP_\lambda\{ \sigma_\lambda\leq r_\lambda (s+\varepsilon)\}\right)
\ge-\inf_{s'< s+\varepsilon}I(s')-\varepsilon\geq -I(s)-\varepsilon.
\end{align*}
Since $\varepsilon>0$ is arbitrary, we conclude that
\begin{equation}
\liminf_{\lambda\in\cL}r_\lambda^{-1}\log\left(\PP_\lambda(\Gamma_\lambda^c) +\e^{sr_\lambda}\wP_\lambda(\Gamma_\lambda)\right)\ge -I(s).
\label{eq:lowerboundh}
\end{equation}
Then, since
$$
\liminf_{\lambda\in\cL}r_\lambda^{-1}\log\left(a_\lambda+b_\lambda\right)
\leq\max\left(\limsup_{\lambda\in\cL}r_\lambda^{-1}\log a_\lambda,
\liminf_{\lambda\in\cL}r_\lambda^{-1}\log b_\lambda\right)
$$
for any families $(a_\lambda)_{\lambda\in\cL}$ and $(b_\lambda)_{\lambda\in\cL}$ of
non-negative numbers, we conclude from~\eqref{eq:lowerboundh} that, if
$\limsup_{\lambda\in\cL}r_\lambda^{-1}\log\PP_\lambda(\Gamma_\lambda^c)<-I(s)$, then
$$
\liminf_{\lambda\in\cL}r_\lambda^{-1}\log\wP_\lambda(\Gamma_\lambda)\geq - I(s)-s=-\hat I(s).
$$
This shows that, for all $s\in\RR$,
\begin{equation}
\underline{\fh}(I(s))\geq -\hat I(s).
\label{eq:fhuis}
\end{equation}

If $\underline{s}_1<\infty$, then by Proposition~\ref{Prop:s}~(2)  we have
$\hat I(-\underline s_1)= 0$, and thus $I(-\underline s_1) = \underline s_1$.
Inequality~\eqref{eq:fhuis} then yields
$$
\underline{\fh}(\underline s_1)
=\underline{\fh}(I(-\underline s_1))
\geq -\hat I(-\underline s_1)= 0.
$$
By~\eqref{eq:htwo} and the fact that $\underline{\fh}$, $\overline{\fh}$
and~$\fh$ are non-decreasing, we conclude that
$$
\underline{\fh}(u)=\overline{\fh}(u)=\fh(u)=0
$$
for all $u\geq \underline s_1$.  By Proposition~\ref{Prop:s}~(6), we also have
 $f(u)=0$ for all $u\geq \underline s_1$ and thus~\eqref{eq:hhhufu} holds for
 all $u\geq\underline s_1$ if $\underline s_1<\infty$.

It thus remains to prove~\eqref{eq:hhhufu} for $u\in[0,\underline s_1[$. Since
the latter interval is empty when $\underline s_1=0$, we assume for the
remainder of this proof that $\underline s_1\in{]}0,\infty]$. By
Proposition~\ref{Prop:s}~(4)-(5), we have also $\underline s_0>0$. By
Part~(7) of the same proposition, the map
$$
S\coloneq{]}{-}\underline{s}_1,\underline{s}_0{[}\ni s
\mapsto I(s)\in{]}0,\underline{s}_1{[}
$$
is a decreasing homeomorphism. For $s\in\RR$ and $\alpha\in[0,1]$, we have
\begin{align*}
\PP_\lambda\{\sigma_\lambda<sr_\lambda\}+\e^{sr_\lambda}\wP_\lambda\{\sigma_\lambda\ge sr_\lambda\}
&=\EE_\lambda[1_{\{\sigma_\lambda<sr_\lambda\}}+\e^{sr_\lambda-\sigma_\lambda}1_{\{\sigma_\lambda\ge sr_\lambda\}}]\\
&=\e^{\alpha sr_\lambda}\EE_\lambda[\e^{-\alpha sr_\lambda}1_{\{\sigma_\lambda<sr_\lambda\}}
+\e^{(1-\alpha)sr_\lambda-\sigma_\lambda}1_{\{\sigma_\lambda\ge sr_\lambda\}}]\\
&\le\e^{\alpha sr_\lambda}\EE_\lambda[\e^{-\alpha\sigma_\lambda}1_{\{\sigma_\lambda<sr_\lambda\}}
+\e^{-\alpha\sigma_\lambda}1_{\{\sigma_\lambda\ge sr_\lambda\}}]\\
&=\e^{\alpha sr_\lambda}\EE_\lambda[\e^{-\alpha\sigma_\lambda}]=\e^{\alpha sr_\lambda+e_\lambda(\alpha)}.
\end{align*}
Invoking~\eqref{eq:invertLTrestr} and optimizing over $\alpha$, we further
obtain, for all $s\in S$, the upper bound
$$
\limsup_{\lambda\in\cL}r_\lambda^{-1}\log\left(\PP_\lambda\{\sigma_\lambda<sr_\lambda\}
+\e^{sr_\lambda}\wP_\lambda\{\sigma_\lambda\ge sr_\lambda\}\right)\le\inf_{\alpha\in[0,1]}(\alpha s+e(\alpha))=-I(s).
$$
In particular, for all $s\in S$,
\begin{align}
\limsup_{\lambda\in\cL}r_\lambda^{-1}\log\PP_\lambda\{\sigma_\lambda<sr_\lambda\}
&\le -I(s),\label{eq:upIs}\\
\limsup_{\lambda\in\cL}r_\lambda^{-1}\log\wP_\lambda\{\sigma_\lambda\ge sr_\lambda\}
&\le -I(s)-s=-\hat{I}(s).
\label{eq:uphatIs}
\end{align}
Let $s,s'\in S$ satisfy $s<s'$, so that $-I(s)<-I(s')$. By~\eqref{eq:upIs}, the family $(\Gamma_\lambda)_{\lambda\in\cL}$ given by
$\Gamma_\lambda =\{\sigma_\lambda \geq s r_\lambda\}$ satisfies
\begin{align*}
\limsup_{\lambda\in\cL}r_\lambda^{-1}\log\PP_\lambda(\Gamma_\lambda^c)&\le-I(s)<-I(s')
\end{align*}
and by applying \eqref{eq:fhuis} to $s'$ and using \eqref{eq:uphatIs}, we find
\begin{align*}
-\hat I(s')\leq \liminf_{\lambda\in\cL}r_\lambda^{-1}\log\wP_\lambda(\Gamma_\lambda)\leq \limsup_{\lambda\in\cL}r_\lambda^{-1}\log\wP_\lambda(\Gamma_\lambda)\leq -\hat I(s).
\end{align*}
Taking $s'\downarrow s$, the continuity of $\hat I$ gives
$$
\lim_{\lambda\in\cL}r_\lambda^{-1}\log\wP_\lambda(\Gamma_\lambda)=-\hat{I}(s),
$$
and hence $\fh(I(s)-\varepsilon)\le-\hat{I}(s)$ holds for $s\in S$ and $\varepsilon>0$.

If $I(s)$ is a point of continuity of $\fh$, then taking $\varepsilon\downarrow0$
gives, recalling~\eqref{eq:fhuis},
$$
-\hat{I}(s)\le\underline{\fh}(I(s))\le\overline{\fh}(I(s))\le \fh(I(s))\le-\hat{I}(s),
$$
and hence
\begin{equation}
\underline{\fh}(I(s))=\overline{\fh}(I(s))=\fh(I(s))=-\hat{I}(s).
\label{eq:pouti}
\end{equation}
Observing that the function
$$
u\mapsto
\fh_{(\Gamma_\lambda)_{\lambda\in\cL}}(u)
=\begin{cases}
\ds\liminf_{\lambda\in\cL}r_\lambda^{-1}\log\wP_\lambda(\Gamma_\lambda)
&\ds\text{if }\ \ \limsup_{\lambda\in\cL}r_\lambda^{-1}\log\PP(\Gamma_\lambda^c)<-u;\\[8pt]
+\infty&\text{otherwise,}
\end{cases}
$$
is upper semicontinuous, so is
$$
u\mapsto
\underline{\fh}(u)
=\inf_{(\Gamma_\lambda)_{\lambda\in\cL}}\fh_{(\Gamma_\lambda)_{\lambda\in\cL}}(u).
$$
Since $\underline{\fh}$ is non-decreasing, it is also right-continuous.
The same reasoning applies to the functions $\overline{\fh}$ and  $\fh$;
for the latter, one takes the infimum over all families
$(\Gamma_\lambda)_{\lambda\in\cL}$ such that
$\lim_{\lambda\in\cL}r_\lambda^{-1}\log\wP_\lambda(\Gamma_\lambda)$ exists.

Let $\cS\subset S$ be such that $I(\cS)$ is the set of discontinuity points of
$\fh$ in $I(S)$. One easily shows that $\cS$ is countable.
Relation~\eqref{eq:pouti} holds for any $s\in S\setminus\cS$, and by
right-continuity it holds for all $s\in S$. Hence, we have proved that
$$
\underline{\fh}(u)=\bar{\fh}(u)=\fh(u)=
-\hat{I}\circ I^{-1}(u)
$$
for all $u\in{]}0,\underline{s}_1{[}$. By  Proposition~\ref{Prop:s}~(7),
the right-hand side is equal to $-f(u)$ and hence we have
proved~\eqref{eq:hhhufu} for $u\in{]}0,\underline{s}_1{[}$.

In order to complete the proof, it only remains to show that~\eqref{eq:hhhufu}
holds for $u=0$, {\sl i.e.,} in view of  Proposition~\ref{Prop:s}~(6),
we need to show that
\begin{equation}
\underline{\fh}(0)=\overline{\fh}(0)=\fh(0)=-\underline s_0.
\label{eq:h0s0}
\end{equation}
If $\underline{s}_0<\infty$, then right-continuity gives
$$
\fh(0)=\lim_{u\downarrow0}\,\fh(u)=\lim_{s\uparrow\underline{s}_0}\,\fh(I(s))=
-\lim_{s\uparrow\underline{s}_0}\hat I(s)=-\hat{I}(\underline{s}_0)=-\underline{s}_0.
$$
In the opposite case, monotonicity gives
$$
\fh(0)\le\liminf_{u\downarrow0}\,\fh(u)=\liminf_{s\uparrow\infty}\,\fh(I(s))=
-\limsup_{s\uparrow\infty}\hat I(s)\le-\limsup_{s\uparrow\infty}s=-\infty,
$$
and hence $\fh(0)=-\infty=-\underline{s}_0$. The same arguments apply to
$\underline{\fh}$ and $\overline\fh$, and thus~\eqref{eq:h0s0} holds and the
proof of Proposition~\ref{prop:Hoeffding} is complete.\hfill\qed

\printbibliography[heading=bibintoc,title={References}]


\chapter{Illustrations}
\label{chap:Illustrations}

\abstract{This chapter provides illustrations of the various aspects of
Fluctuation Relations and Fluctuation Theorems. We revisit some paradigmatic
models of equilibrium and nonequilibrium statistical mechanics, deriving in each
case the Large Deviation Principle giving rise to a fluctuation theorem and
discussing its physical meaning and its particular features: convexity
properties, smoothness of the entropic pressure, phase transitions, etc.}

\abstract*{This chapter provides illustrations of the various aspects of
Fluctuation Relations and Fluctuation Theorems. We revisit some paradigmatic
models of equilibrium and nonequilibrium statistical mechanics, deriving in each
case the Large Deviation Principle giving rise to a fluctuation theorem and
discussing its physical meaning and its particular features: convexity
properties, smoothness of the entropic pressure, phase transitions, etc.}

\vskip 1cm

Let us introduce some standard notation to be repeatedly used in the following.
For $a,b\in\ZZ$,  $\llbracket a,b\rrbracket\coloneq[a,b]\cap\ZZ$, $\llbracket
a,b\llbracket\;\coloneq[a,b[\cap\ZZ$, etc. $\cP(\Omega)$ denotes the set of Borel
probability measures on a Polish space $\Omega$. Given a continuous map
$\varphi:\Omega\to\Omega$, we further denote by $\cP_\varphi(\Omega)$ the subset
of $\varphi$-invariant measures. $C(\Omega)$ and $C_b(\Omega)$ denote the space
of real-valued continuous functions on $\Omega$ and its subspace of bounded
elements. For $\mu\in\cP(\Omega)$ and $f\in C_b(\Omega)$ we write\label{ref:COmega}
$$
\langle f,\mu\rangle\coloneq\int_\Omega f\d\mu.
$$
Unless otherwise stated, $\cP(\Omega)$ is equipped with the (metrizable and separable)
topology of weak convergence. We recall that whenever $\Omega$ is compact, so is
$\cP(\Omega)$.

\section{Finite State Markov Chains}
\label{Ex:Markov_chain}

\subsection{Setup}

For a \ndex{Markov chain} on the state space~$\cA\coloneq\llbracket1,N\rrbracket$
we take $\cL=\NN$ as
the index set and, for each $t\in\NN$, denote by  $\Omega_t=\cA^t$  the $t$-step
path space of the process. Assuming that the transition matrix $P=[P_{xy}]_{x,y\in\cA}$
admits an invariant probability $p=[p_x]_{x\in\cA}$ ({\sl i.e.,} that $p=pP$), the path space
measure of the corresponding stationary Markov process $(p,P)$ is given by
$$
\Omega_t\ni \bsx=(x_1,\ldots,x_t)\mapsto
\PP_t(\bsx)\coloneq p_{x_1}P_{x_1x_2}\cdots P_{x_{t-1}x_t}.
$$
Given a second stationary Markov process $(\hat p,\widehat P)$ on the same state
space, the corresponding measure $\wP_t$ is equivalent to $\PP_t$ for all
$t\in\NN$ iff both conditions (i)~$p_x=0\Leftrightarrow\hat p_x=0$, and
(ii)~$P_{xy}=0\Leftrightarrow \widehat P_{xy}=0$ hold for all $x,y\in\cA$. Note
that if $P$ and $\widehat P$ are both irreducible,\footnote{$P$ is irreducible
iff, for any $x,y\in\cA$ there exists $N>0$ such that $(P^N)_{xy}>0$.} then the
Perron--Frobenius theorem implies that $p_x>0$ and $\hat p_x>0$ for all $x\in\cA$
so that Condition~(i) trivially holds. In the following, we shall assume that
$P$ and $\widehat P$ are irreducible and satisfy Condition~(ii).

\subsection{Transient Fluctuation Relation}

The entropy production random variable of the family
$(\PP_t,\wP_t)_{t\in\NN}$ is given by
\[
\sigma_t(\bsx)=\log\frac{\PP_t(\bsx)}{\wP_t(\bsx)}
=\log\frac{p_{x_1}}{\hat p_{x_1}}
+\sum_{s=1}^{t-1}\log\frac{P_{x_sx_{s+1}}}{\widehat P_{x_sx_{s+1}}}.
\]
Simple calculations yield that
\beq
\ep_t=\Ent(\PP_t|\wP_t)=\sum_{\bsx\in\Omega_t}\sigma_t(\bsx)\PP_t(\bsx)
=\Ent(p|\hat p)+(t-1)\sum_{x,y\in\cA}p_xP_{xy}\log\frac{P_{xy}}{\widehat P_{xy}},
\label{Eq:MarkovChainep}
\eeq
and
\begin{align}
e_t(\alpha)&=\Ent_\alpha(\PP_t|\wP_t)
=\log\sum_{\bsx\in\Omega_t}[\PP_t(\bsx)]^{1-\alpha}[\wP_t(\bsx)]^\alpha
=\log(q(\alpha)Q(\alpha)^{t-1}\boldsymbol{1}),\label{Eq:MarkovChainealpha}\\
\hat e_t(\alpha)&=\Ent_\alpha(\wP_t|\PP_t)
=\log\sum_{\bsx\in\Omega_t}[\PP_t(\bsx)]^{\alpha}[\wP_t(\bsx)]^{1-\alpha}
=\log(q(1-\alpha)Q(1-\alpha)^{t-1}\boldsymbol{1}),\nonumber
\end{align}
where the row vector $q(\alpha)$, the matrix $Q(\alpha)$, and the column
vector $\boldsymbol 1$ are given by
\[
q(\alpha)\coloneq[p_x^{1-\alpha}\hat p_x^\alpha]_{x\in\cA},\quad
Q(\alpha)\coloneq[P_{xy}^{1-\alpha}\widehat P_{xy}^\alpha]_{x,y\in\cA},\quad
\boldsymbol{1}\coloneq[1,1,\ldots,1]^T.
\]

As an important special case, let us consider the $\ZZ_2$-action on $\Omega_t$ induced
by the time-reversal map
$$
\Theta_t:(x_1,\ldots,x_t)\mapsto(x_t,\ldots,x_1),
$$
and set $\wP_t=\PP_t\circ\Theta_t$. Since $p_x>0$ for all $x\in\cA$, we can write
$$
\wP_t(\bsx)=p_{x_t}P_{x_tx_{t-1}}\cdots P_{x_1x_0}
=p_{x_1}\widehat P_{x_1x_2}\cdots\widehat P_{x_{t-1}x_t},
$$
with $\widehat P_{xy}=p_yP_{yx}p_x^{-1}$. It follows at once that $p\widehat P=p$.
Moreover, $\PP_t$ and $\widehat{\PP}_t$ are equivalent whenever
$P_{xy}=0\Leftrightarrow P_{yx}=0$, in which case $\widehat P$ is irreducible.
The family $(\PP_t,\widehat{\PP}_t)_{t\in\NN}$ is in involution and time-reversal
invariance $\PP_t=\wP_t$ holds iff $P=\widehat P$, which is a disguised form of the
\ndex{detailed balance} condition $p_xP_{xy}=p_yP_{yx}$. From~\eqref{Eq:MarkovChainep} one
 derives that the mean entropy production rate
$$
\lim_{t\to\infty}\frac{\ep_t}t
=\frac12\sum_{x,y\in\cA}p_x(P_{xy}-\widehat P_{xy})
\log\frac{P_{xy}}{\widehat P_{xy}}
$$
is strictly positive whenever the time-reversal symmetry is broken.

\subsection{Fluctuation Theorem}

We shall use~\eqref{Eq:MarkovChainealpha} to compute the entropic pressure of
the family $(\PP_t,\wP_t)_{t\in\NN}$. To this end, note that since the entries
of $Q(\alpha)$ are non-negative, one has
$$
\e^{-c(1+2|\alpha|)}\,\VERT Q(\alpha)^n\VERT\le q(\alpha)Q(\alpha)^n\boldsymbol{1}
\le\e^{c(1+2|\alpha|)}\,\VERT Q(\alpha)^n\VERT
$$
where the matrix norm is  defined as
$$
\VERT A\VERT\coloneq\sum_{x,y\in\cA}|A_{xy}|,
$$
and
$$
c\coloneq\max_{x\in\cA}\left(|\log p_x|,|\log\hat p_x|\right).
$$
It follows that
$$
\frac1t\left(\log\VERT Q(\alpha)^t\VERT-c(1+2|\alpha|)\right)
\le \frac1t e_t(\alpha)
\le\frac1t\left(\log\VERT Q(\alpha)^t\VERT+c(1+2|\alpha|)\right),
$$
and taking $t\to\infty$, Gelfand's formula yields
$$
e(\alpha)=\lim_{t\to\infty}\frac1t e_t(\alpha)=\log r(\alpha)
$$
where $r(\alpha)$ denotes the spectral radius of $Q(\alpha)$. One easily shows
that for all $\alpha\in\RR$ the matrix $Q(\alpha)$, which clearly has
non-negative entries, is irreducible. The Perron--Frobenius theorem then asserts
that $r(\alpha)$ is the dominant eigenvalue of $Q(\alpha)$. In particular,
$r(\alpha)$ is a strictly positive simple eigenvalue of $Q(\alpha)$. It follows
from analytic perturbation theory~\cite[Chapter~2]{Kato1966} that the function
$\RR\ni\alpha\mapsto r(\alpha)$ is real analytic. The same is true of
$\RR\ni\alpha\mapsto e(\alpha)$ and the Gärtner--Ellis theorem yields a global
FT with a rate function given by~\eqref{Eq:IasLegendreTransform}.\index{FT!for Markov chains}

\subsubsection{Level-3 Fluctuation Relation}

We conclude this example with a brief description of the Donsker--Varadhan process-level
LDP as an alternative approach to the FT for Markov chains. We refer
to~\cite{Deuschel1989,Dembo2000} for a more detailed exposition.\index{FR!level-3}

Let $\Omega\coloneq\cA^\NN$ be the set of infinite sequences $\bsx=(x_1,x_2,\cdots)$ of
elements of $\cA$. Equipped with the metric
$d(x,y)\coloneq2^{-\inf\{t\in\NN\,|\,x_t\not=y_t\}}$, $\Omega$ is a  compact metric
space. We denote the left shift $(x_1,x_2,\ldots)\mapsto(x_2,x_3,\ldots)$ on $\Omega$
by $\varphi$. The Borel $\sigma$-algebra $\cB$ of $\Omega$ is generated
by the cylinder sets $Z_I(\Gamma)\coloneq\{\bsx\in\Omega\,|\,(x_t)_{t\in I}\in\Gamma\}$,
where $I\subset\NN$ is finite and $\Gamma\subset\cA^I$. The set $\cP(\Omega)$, endowed
with the weak topology, is a compact metrizable space and $\cP_\varphi(\Omega)$ is closed
in $\cP(\Omega)$. We denote by $\PP$ and $\wP$ the elements of $\cP_\varphi(\Omega)$
induced by the sequences $(\PP_t)_{t\in\NN}$ and $(\wP_t)_{t\in\NN}$.

Analogous constructions apply to $\bar\Omega\coloneq\cA^\ZZ$,
and we shall also denote by $\varphi$ the left shift on this space. Any
$\QQ\in\cP_\varphi(\Omega)$ has a unique extension $\bar\QQ\in\cP_\varphi(\bar\Omega)$
which is characterized by
$$
\bar\QQ(Z_I(\Gamma))=\QQ(Z_{I+t}(\Gamma))
$$
for any finite $I\subset\ZZ$, any $\Gamma\subset\cA^I$, and any $t\in\NN$ such that
$I+t\subset\NN$. For $t\in\ZZ$, we denote by $\bar\cB_t$ the $\sigma$-algebra
generated by
$\{Z_I(\Gamma)\mid I\subset\rrbracket{-}\infty,t\rrbracket,I\text{ is finite},\Gamma\subset\cA^I\}$
and by $\bar\QQ_t$ the restriction of $\bar\QQ$ to $\bar\cB_t$. Finally, let
$\bar\QQ_0\otimes P$ be the probability measure on $\bar\cB_1$ defined by
$$
\int_{\bar\Omega} f(\ldots,x_0,x_1)\bar\QQ_0\otimes P(\d x)
\coloneq\sum_{y\in\cA}\int_{\bar\Omega}f(\ldots,x_0,y)P_{x_0y}\bar\QQ_0(\d x)
$$
for any bounded $\bar\cB_1$-measurable function $f:\bar\Omega\to\RR$.
If the matrix $P$ is irreducible, then the sequence of empirical
measures\footnote{$\delta_{\bsx}$ denotes the Dirac mass at $\bsx\in\Omega$.}
$$
\Omega\ni\bsx\mapsto\xi_t\coloneq\frac1t\sum_{s=0}^{t-1}\delta_{\varphi^s(\bsx)},
$$
viewed as $\cP(\Omega)$-valued random variables on $(\Omega,\cB,\PP)$,
satisfies the LDP
$$
-\inf_{\QQ\in\mathring{\Gamma}}\II(\QQ)
\le\liminf_{t\to\infty}\frac1t\log\PP\{\xi_t\in\Gamma\}
\le\limsup_{t\to\infty}\frac1t\log\PP\{\xi_t\in\Gamma\}
\le-\inf_{\QQ\in\bar{\Gamma}}\II(\QQ)
$$
for any Borel set $\Gamma\subset\cP(\Omega)$ with the good rate function
$$
\II(\QQ)\coloneq\left\{
\begin{array}{ll}
\Ent(\bar\QQ_1|\bar\QQ_0\otimes P)&\text{if }\QQ\in\cP_\varphi(\Omega);\\[4pt]
+\infty&\text{if }\QQ\in\cP(\Omega)\setminus\cP_\varphi(\Omega).
\end{array}
\right.
$$
If $\widehat P$ is irreducible, then the same LDP also holds w.r.t.\;the measure
$\wP$ with the corresponding rate function $\hat\II$. Moreover, the map
$$
\cP(\Omega)\ni\QQ\mapsto\varsigma(\QQ)\coloneq\langle\sigma,\QQ\rangle,\qquad
\sigma(\bsx)\coloneq\log\frac{P_{x_1x_2}}{\widehat P_{x_1x_2}},
$$
is continuous and one easily checks that
$$
t\varsigma(\xi_t)
=\sum_{s=1}^t\log\frac{P_{x_sx_{s+1}}}{\widehat P_{x_sx_{s+1}}}
=\sigma_t+o(t),\quad(t\to\infty).
$$
The contraction principle thus yields the FT with the rate functions
$$
I(s)\coloneq\inf\{\II(\QQ)\,|\,\QQ\in\cP(\Omega),\varsigma(\QQ)=s\},\qquad
\hat I(s)\coloneq\inf\{\hat\II(\QQ)\,|\,\QQ\in\cP(\Omega),\varsigma(\QQ)=s\}.
$$
Using the simple fact that, for any $\bsx\in\Omega$,
$$
\frac{\d\bar\QQ_0\otimes P}{\d\bar\QQ_0\otimes\widehat P}(\bsx)
=\frac{P_{x_0x_1}}{\widehat P_{x_0x_1}},
$$
one derives the process-level or level-3 FR
\begin{align*}
\hat{\II}(\QQ)=\Ent(\bar\QQ_1|\bar\QQ_0\otimes\widehat P)
&=\int\log\frac{\d\bar\QQ_1}{\d\bar\QQ_0\otimes\widehat P}\d\bar\QQ_1\\[4pt]
&=\int\left(\log\frac{\d\bar\QQ_1}{\d\bar\QQ_0\otimes P}
+\log\frac{\d\bar\QQ_0\otimes P}{\d\bar\QQ_0\otimes\widehat P}\right)\d\bar\QQ_1\\[4pt]
&=\II(\QQ)+\varsigma(\QQ)
\end{align*}
for $\QQ\in\cP_\varphi(\Omega)$. The same relation holds for
$\QQ\in\cP(\Omega)\setminus\cP_\varphi(\Omega)$, both sides of the resulting
identity being infinite.

\section{Mean-field Lattice Gas}
\label{Ex:Lattice_gas}

\subsection{Setup}

For $x=(x_1,\ldots,x_d)\in\ZZ^d$ we set $|x|\coloneq\max_i|x_i|$ and  denote by\index{mean-field lattice gas}
$|\lambda|$ the cardinality of a subset $\lambda\subset\ZZ^d$.
Let $\cL$ be the set of centered cubic boxes $\{x\in\ZZ^d\mid|x|\le n\}$
in $\ZZ^d$ ordered by inclusion and for $\lambda\in\cL$ let
$\Omega_\lambda=\{0,1\}^\lambda$ be the configuration space of a lattice gas
confined to the box $\lambda$. We shall consider the mean-field Hamiltonian
$$
\Omega_\lambda\ni\bsn=(n_x)_{x\in\lambda}\mapsto H_\lambda(\bsn)
\coloneq-\frac2{|\lambda|-1}\sum_{\substack{x,y\in\lambda\\ x\not=y}}n_x n_y
$$
with attractive interaction. The grand canonical
ensemble at inverse temperature $\beta$ and chemical potential $\mu$ is the
atomic measure on $\Omega_\lambda$ defined by
$$
\PP_\lambda(\{\bsn\})
\coloneq\e^{-\beta(H_\lambda(\bsn)-\mu N_\lambda(\bsn)+|\lambda|p_\lambda(\beta,\mu))},
$$
where $N_\lambda(\bsn)\coloneq\sum_{x\in\lambda}n_x$ is the total number of particles in the
box $\lambda$ and
\beq
p_\lambda(\beta,\mu)\coloneq\frac1{|\lambda|\beta}\log\sum_{\bsn\in\Omega_\lambda}
\e^{-\beta(H_\lambda(\bsn)-\mu N_\lambda(\bsn))}
\label{Eq:LatticeGasplambda}
\eeq
is the pressure of the gas. We denote expectation w.r.t.\;$\PP_\lambda$ by
$\langle\,\argdot\,\rangle_{\lambda,\beta,\mu}$.

The particle-hole symmetry
$$
\Theta_\lambda:(n_x)_{x\in\lambda}\mapsto(1-n_x)_{x\in\lambda}
$$
provides a $\ZZ_2$-action on $\Omega_\lambda$, and we set
$\wP_\lambda=\PP_\lambda\circ\Theta_\lambda$.

\subsection{Transient Fluctuation Relation}

A simple calculation gives that  the entropy production observable of the
involutive family $(\PP_\lambda,\widehat{\PP}_\lambda)_{\lambda\in\cL}$ is
\beq
\sigma_\lambda=2|\lambda|\beta(\mu+2)\left(\rho_\lambda-\frac12\right),
\label{Eq:MeanFieldsigmalambda}
\eeq
where $\rho_\lambda(\bsn)\coloneq|\lambda|^{-1}N_\lambda(\bsn)$ is the sample density.
For $\mu\neq -2$, Inequality~\eqref{Eq:Jarzynski_inequality} reduces to
$$
\begin{cases}
	\langle\rho_\lambda\rangle_{\lambda,\beta,\mu}\leq\frac12 & \text{if }\mu < -2;\\[6pt]
	\langle\rho_\lambda\rangle_{\lambda,\beta,\mu}\geq\frac12 & \text{if }\mu > -2,
\end{cases}
$$
which shows that $\ep_\lambda$ is an order parameter for
the particle-hole symmetry. Inequality~\eqref{Eq:Jarzy_and_Markov}
yields a quantitative version of this statement with the estimate
$$
\PP_\lambda\{\rho_\lambda-1/2<-a\}\le\e^{-2|\lambda|\beta(\mu+2)a},
$$
for $\mu>-2$ (and a similar estimate for $\mu<-2$).

The Rényi entropy of the lattice gas
\beq
e_\lambda(\alpha)=|\lambda|\beta\left(
p_\lambda(\beta,(1-2\alpha)(\mu+2)-2)-p_\lambda(\beta,\mu)
+\alpha(\mu+2)\right)
\label{Eq:MeanFieldealphalambda}
\eeq
vanishes identically for $\mu=-2$. The symmetry~\eqref{Eq:e_lambda_symmetry}
reduces, in this model, to the pressure identity
$$
p_\lambda(\beta,-2+\nu)-p_\lambda(\beta,-2-\nu)=\nu.
$$
Differentiation w.r.t.\,$\nu$ yields
$$
\langle\rho_\lambda\rangle_{\lambda,\beta,-2+\nu}
=1-\langle\rho_\lambda\rangle_{\lambda,\beta,-2-\nu},
$$
and in particular $\langle\rho_\lambda\rangle_{\lambda,\beta,-2}=1/2$.

\subsection{Fluctuation Theorem}

In terms of the sample density $\rho_\lambda$, the
pressure~\eqref{Eq:LatticeGasplambda} of the mean-field lattice gas reads
$$
p_\lambda(\beta,\mu)=\frac1{|\lambda|\beta}\log\sum_{\bsn\in\Omega_\lambda}
\e^{|\lambda|\beta\rho_\lambda(\bsn)(2\rho_\lambda(\bsn)+\mu)}
\e^{-\beta\rho_\lambda(\bsn)(1-\rho_\lambda(\bsn))\frac{2|\lambda|}{|\lambda|-1}},
$$
and since $\rho_\lambda(\bsn)\in[0,1]$ for all $\bsn\in\Omega_\lambda$, we have
$$
p(\beta,\mu)=\lim_{\lambda\in\cL}p_\lambda(\beta,\mu)
=\lim_{\lambda\in\cL}\frac1{|\lambda|\beta}\log\sum_{\bsn\in\Omega_\lambda}
\e^{|\lambda|\beta\rho_\lambda(\bsn)(2\rho_\lambda(\bsn)+\mu)}.
$$
Introducing the product $\BB^\lambda\coloneq\otimes_{x\in\lambda}\BB$ of the Bernoulli
distribution $\BB(0)=\BB(1)=1/2$, we can write
$$
\sum_{\bsn\in\Omega_\lambda}
\e^{|\lambda|\beta\rho_\lambda(\bsn)(2\rho_\lambda(\bsn)+\mu)}
=2^{|\lambda|}\BB^\lambda[\e^{|\lambda|f(\rho_\lambda)}]
$$
with $f(r)\coloneq\beta r(2r+\mu)$ and $\rho_\lambda=|\lambda|^{-1}\sum_{x\in\lambda}n_x$,
the $n_x$ being i.i.d.\;$\BB$-distributed random variables. Combining Cramér's
LDP~\cite[Theorem~2.2.3]{Dembo2000} with Vara\-dhan's lemma leads to the expression
\beq
p(\beta,\mu)=\frac1\beta\sup_{r\in[0,1]}F(\beta,\mu,r),
\label{Eq:MeanFieldPressureVar}
\eeq
where $F(\beta,\mu,r)\coloneq f(r)-r\log r-(1-r)\log(1-r)$. The right-hand side
of~\eqref{Eq:MeanFieldPressureVar} can be evaluated explicitly. This elementary
analysis (see~\cite[Section~30]{Dorlas1999} or~\cite[Section~IV.4]{Ellis2006})
yields that the set of critical points
$$
\fR(\beta,\mu)\coloneq\left\{r\in[0,1]\,\bigg|\,\partial_rF(\beta,\mu,r)=0\right\}
$$
is a singleton if either $\beta<\beta_{\rm c}=1$ or $\beta\ge\beta_{\rm c}$ and
\[
|1+\mu/2|>g(\beta)=\sqrt{\beta_{\rm c}^{-1}-\beta^{-1}}
-\beta^{-1}\mathrm{arth}(\sqrt{\beta_{\rm c}^{-1}-\beta^{-1}}).
\]
It contains $3$ elements if $\beta>\beta_{\rm c}$ and $|1+\mu/2|<g(\beta)$.
Moreover, one has
$$
p(\beta,\mu)=\frac1\beta F(\beta,\mu,\varrho(\beta,\mu)),
$$
where
$$
\varrho(\beta,\mu)\coloneq\left\{\bear{ll}
\min\,\fR(\beta,\mu)<\frac12,&\text{for }\mu<-2;\\[4pt]
\max\,\fR(\beta,\mu)>\frac12,&\text{for }\mu>-2.
\ear
\right.
$$
It follows that for $\beta>\beta_{\rm c}$,
$$
\lim_{\mu\uparrow-2}\varrho(\beta,\mu)<1/2
<\lim_{\mu\downarrow-2}\varrho(\beta,\mu).
$$
The function $(\beta,\mu)\mapsto p(\beta,\mu)$ is real analytic on the domain
$(\RR_+\times\RR)\setminus([\beta_{\rm c},\infty[\times\{-2\})$, and since
$$
\varrho(\beta,\mu)=\partial_\mu p(\beta,\mu),
$$
the pressure is non-differentiable as a function of $\mu$ at $\mu=-2$,
$\beta>\beta_{\rm c}$. The mean-field lattice gas thus exhibits a
spontaneous breaking of the particle-hole symmetry, see
Figure~\ref{Fig:MeanFieldPressure}. By~\eqref{Eq:MeanFieldealphalambda}
the entropic pressure is given by
\beq
e(\alpha)=\lim_{\lambda\in\cL}\frac1{|\lambda|}e_\lambda(\alpha)
=\beta\left(
p(\beta,(1-2\alpha)(\mu+2)-2)-p(\beta,\mu)
+\alpha(\mu+2)\right)
\label{Eq:MeanFieldealpha}
\eeq
which is a real analytic function of $\alpha\in\RR$ for $\beta<\beta_{\rm c}$,
but is non-differentiable at $\alpha=1/2$ for $\beta>\beta_{\rm c}$ and
$\mu\not=-2$. Thus, for $\beta<\beta_{\rm c}$ the Gärtner--Ellis theorem yields
a global FT with rate function given by~\eqref{Eq:IasLegendreTransform}. However,
due to the singularity appearing at $\alpha=1/2$, this approach fails for
$\beta>\beta_{\rm c}$.
\begin{figure}
\hskip-0.5cm\includegraphics[width=1.1\textwidth]{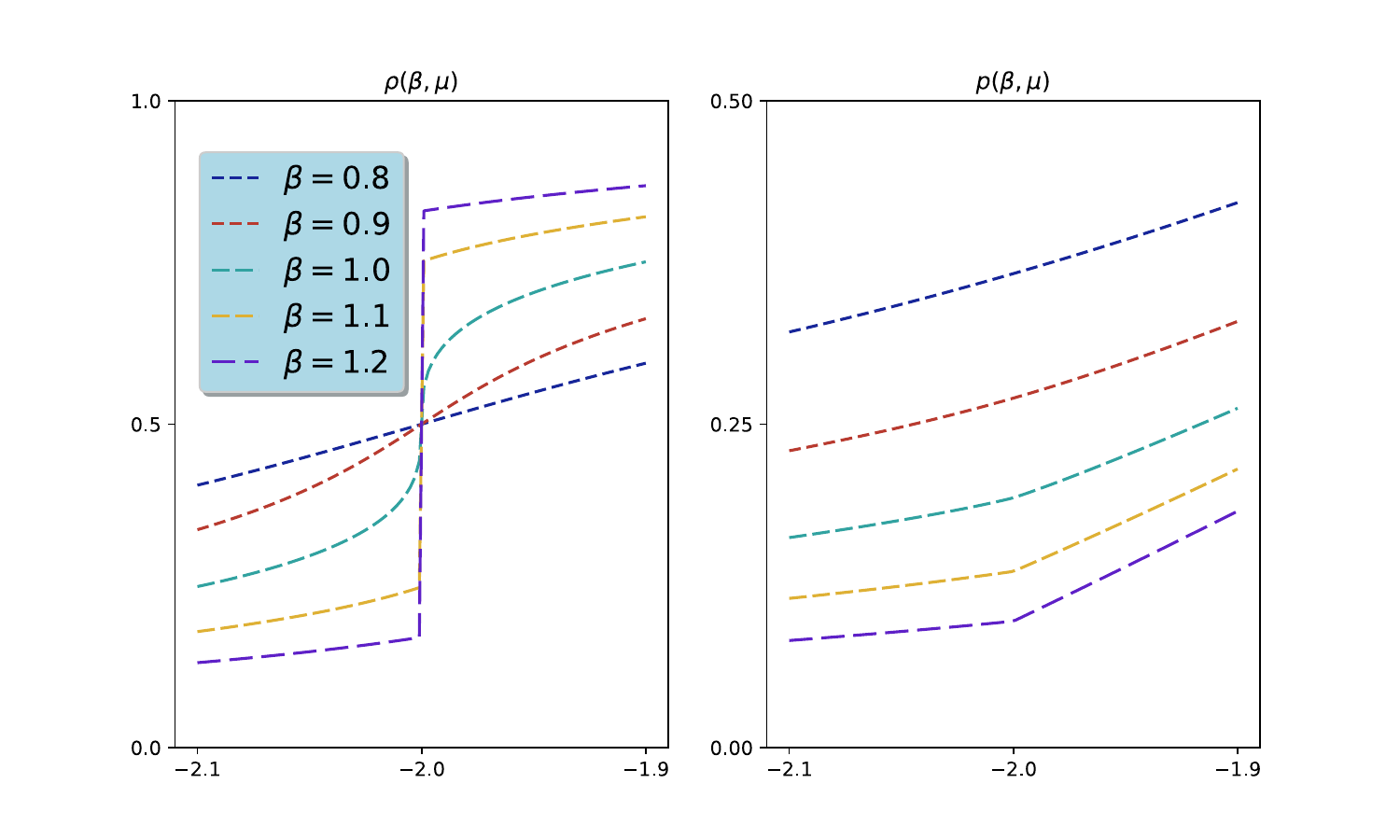}
\caption{The density $\varrho$ (left) and the pressure $p$ (right) of the
mean-field lattice gas as functions of $\mu$ near the critical value $\mu=-2$ for
various values of $\beta$ (shorter dashes mean lower $\beta$).}
\label{Fig:MeanFieldPressure}
\end{figure}

To overcome this, we shall now briefly sketch a derivation of the FT which does
not rely on the Gärtner--Ellis theorem and which works for any $\beta >0$. This
approach relies on standard techniques of large deviations theory, which easily
apply here due to the simple  structure of the mean-field lattice gas.

Let $\cA\coloneq\{0,1\}$. We shall identify without further notice $\cP(\cA)$ with the subset
$\{(1-r,r)\,|\,r\in[0,1]\}$ of $\RR^2$ and topologize accordingly.

For each $\lambda\in\cL$ define the level-2 empirical measure
$$
\Omega_\lambda\ni\bsn=(n_x)_{x\in\lambda}\mapsto\xi_\lambda
\coloneq\frac1{|\lambda|}\sum_{x\in\lambda}\delta_{n_x}\in\cP(\cA),
$$
where $\delta_{n}$ is the unit mass at $n\in\cA$. Considering
$\xi_\lambda$ as a $\cP(\cA)$-valued random variable on
$(\Omega_\lambda,\BB^\lambda)$, \ndex{Sanov's theorem}~\cite[Theorem~2.1.10]{Dembo2000}
yields the LDP
$$
-\inf_{\QQ\in\mathring{\Gamma}}\bI(\QQ)
\le\liminf_{\lambda\in\cL}\frac1{|\lambda|}
\log\BB^\lambda\{\xi_\lambda\in\Gamma\}
\le\limsup_{\lambda\in\cL}\frac1{|\lambda|}
\log\BB^\lambda\{\xi_\lambda\in\Gamma\}
\le-\inf_{\QQ\in\bar\Gamma}\bI(\QQ)
$$
for any Borel subset $\Gamma\subset\cP(\cA)$ with a rate function given by the
relative entropy
$$
\bI(\QQ)\coloneq\Ent(\QQ|\BB)=\log 2-S(\QQ),
$$
where
\beq
S(\QQ)\coloneq-\sum_{a\in\cA}\QQ(a)\log\QQ(a)
\label{equ:ShannonEnt}
\eeq
is the \ndex{Shannon entropy} of $\QQ$. Defining the map
$\cP(\cA)\ni\QQ\mapsto\bF(\QQ)\coloneq\langle G,\QQ\otimes\QQ\rangle\in\RR$, where
$$
G(a,b)\coloneq\beta a(2b+\mu)-\beta p(\beta,\mu),
$$
we observe that, for any $\lambda\in\cL$,
$$
2^{|\lambda|}\e^{-\frac12\frac{|\lambda|}{|\lambda|-1}}
\e^{|\lambda|\bF(\xi_\lambda)}\le
\frac{\d\PP_\lambda}{\d\BB^\lambda}
\le2^{|\lambda|}\e^{|\lambda|\bF(\xi_\lambda)}.
$$
It follows from Varadhan's lemma (more precisely from its tilted LDP
version~\eqref{Eq:Tilting}) that the LDP
$$
-\inf_{\QQ\in\mathring{\Gamma}}\II(\QQ)
\le\liminf_{\lambda\in\cL}\frac1{|\lambda|}
\log\PP_\lambda\{\xi_\lambda\in\Gamma\}
\le\limsup_{\lambda\in\cL}\frac1{|\lambda|}
\log\PP_\lambda\{\xi_\lambda\in\Gamma\}
\le-\inf_{\QQ\in\bar\Gamma}\II(\QQ)
$$
holds for $\Gamma\subset\cP(\cA)$ with the rate function
$$
\II(\QQ)\coloneq\bI(\QQ)-\bF(\QQ)-\inf_{\QQ'\in\cP(\cA)}(\bI(\QQ')-\bF(\QQ')).
$$
A simple calculation, using Relation~\eqref{Eq:MeanFieldPressureVar}, shows that
$$
\II(\QQ)=-\bF(\QQ)-S(\QQ).
$$
Finally, expressing the entropy production random
variable~\eqref{Eq:MeanFieldsigmalambda} as
$\sigma_\lambda=|\lambda|\varsigma(\xi_\lambda)$ with
\beq
\varsigma(\QQ)\coloneq\QQ\otimes\QQ(\sigma),\qquad\sigma\coloneq G-G\circ\Theta,
\label{Eq:MeanFieldvarsigma}
\eeq
(here $\Theta(a,b)\coloneq(1-a,1-b)$), the contraction principle shows that the global
FT holds with scale $r_\lambda=|\lambda|$ and rate function\index{FT!for the mean field lattice gas}
$$
I(s)=\inf\{\II(\QQ)\mid\QQ\in\cP(\cA),\varsigma(\QQ)=s\}.
$$
Observing that the constraint $\varsigma(\QQ)=s$ actually determines $\QQ$, with
$$
\QQ(1)=\frac12\left(1+\frac s{\beta(\mu+2)}\right)\eqcolon r,
$$
another straightforward calculation gives the following expression of the rate
function $I$ in terms of the variable $r$,
$$
I(s)=\left\{
\bear{ll}
\beta p(\beta,\mu)-F(\beta,\mu,r),&\text{for }r\in[0,1];\\[4pt]
+\infty,&\text{otherwise.}
\ear
\right.
$$

\remark[1] From the simple estimate
$$
|\e^{-\alpha\sigma_\lambda}|\le\e^{\beta|\lambda|\,|(\mu+2)\Re\alpha|}
$$
one easily deduces, invoking Vitali's convergence theorem, that
$$
e'(0)=\lim_{\lambda\in\cL}\frac1{|\lambda|}e'_\lambda(0)
=-\lim_{\lambda\in\cL}\frac1{|\lambda|}
\langle\sigma_\lambda\rangle_{\lambda,\beta,\mu}.
$$
Using~\eqref{Eq:MeanFieldealpha}, the left-hand side of this identity is
$$
e'(0)=2\beta(\mu+2)\left(\frac12-\partial_\mu p(\beta,\mu)\right),
$$
and recalling~\eqref{Eq:MeanFieldsigmalambda} we conclude that the
thermodynamic limit of the mean density is
$$
\lim_{\lambda\in\cL}\langle\rho_\lambda\rangle_{\lambda,\beta,\mu}=\varrho(\beta,\mu),
$$
and that of the mean entropy production rate is
$$
\epr\coloneq\lim_{\lambda\in\cL}\frac{\ep_\lambda}{|\lambda|}=2\beta(\mu+2)\left(\varrho(\beta,\mu)-\frac12\right).
$$

\remark[2] It follows from our explicit formula that $I'(s)=0$ iff
$r\in\fR(\beta,\mu)$. Since the latter set is not always a singleton, it may
happen that the rate function is not convex. Indeed, an elementary analysis shows
that $I$ is convex for $\beta<\beta_{\rm c}$, while for
$\beta>\beta_{\rm c}$ it is concave on the interval
$|s|<|\mu+2|\sqrt{\beta(\beta/\beta_{\rm c}-1)}$ and convex on its complement,
see Figure~\ref{Fig:MeanFieldRate}.

\remark[3]  The FR, $I(-s)=I(s)+s$, corresponds here to the elementary property
$$
F(\beta,\mu,1-r)=F(\beta,\mu,r)+\beta(\mu+2)(1-2r).
$$

\remark[4] Setting $\widehat{\QQ}(a)\coloneq\QQ(1-a)$ for $\QQ\in\cP(\cA)$, it follows from the
relation $\xi_\lambda\circ\Theta_\lambda=\hat{\xi}_\lambda$ that
$(\xi_\lambda)_{\lambda\in\cL}$ satisfies the LDP
w.r.t.\;$(\wP_\lambda)_{\lambda\in\cL}$ with the scale $r_\lambda=|\lambda|$ and
rate function $\hat{\II}$ given by $\hat{\II}(\QQ)\coloneq\II(\widehat{\QQ})$. Moreover,
rewriting Definition~\eqref{Eq:MeanFieldvarsigma} as
$$
\varsigma(\QQ)=\bF(\QQ)-\bF(\widehat{\QQ})
$$
immediately leads to the level-2 FR\index{FR!level-2}
$$
\hat{\II}(\QQ)=\II(\QQ)+\varsigma(\QQ),
$$
which is \eqref{Eq:FRlift}.

\begin{figure}
\hskip-1.2cm\includegraphics[width=1.2\textwidth]{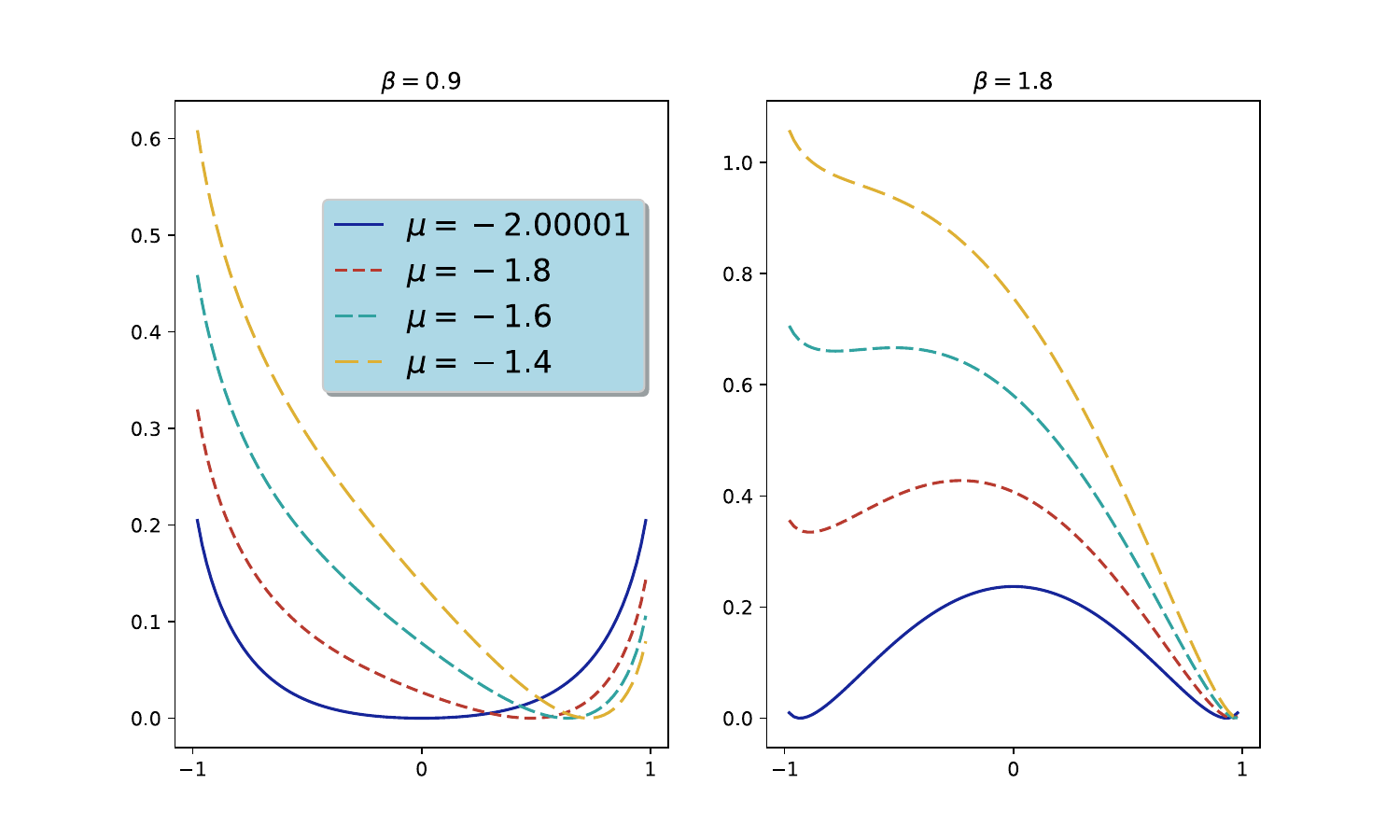}
\caption{The FT rate $I$ as a function of the rescaled variable
$\hat s=s/\beta(\mu+2)$ for $\beta=0.9<\beta_{\rm c}$ (left) and
$\beta=1.8>\beta_{\rm c}$ (right) and $\mu=-2,-1.8,-1.6,-1.4$.}
\label{Fig:MeanFieldRate}
\end{figure}

\begin{figure}
\begin{center}
\includegraphics[width=0.8\textwidth]{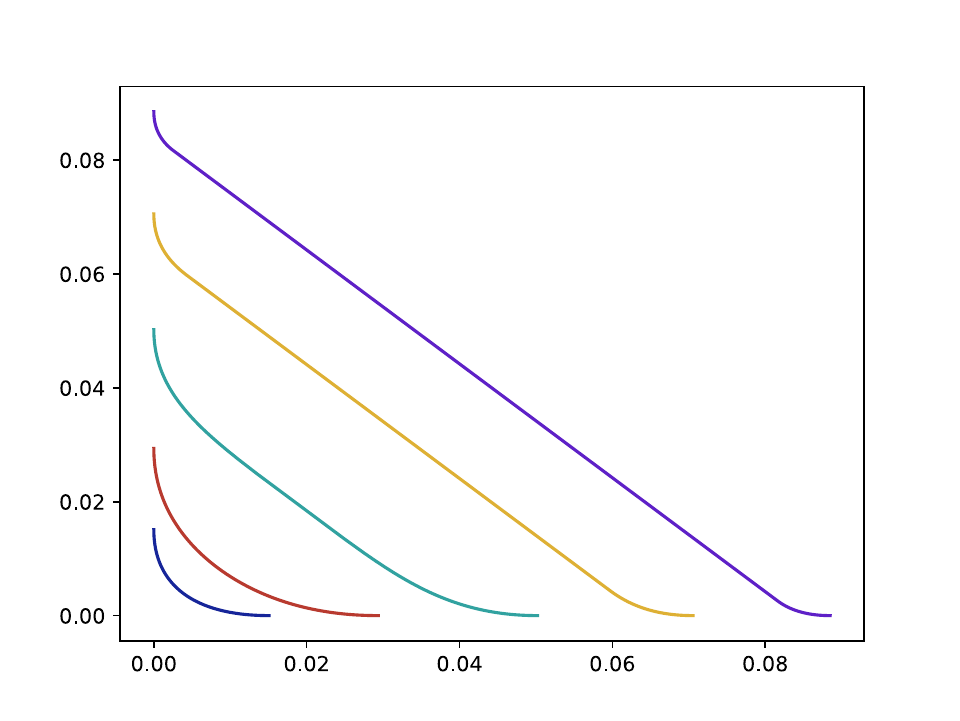}
\caption{The Hoeffding function $f$ (recall Equ.~\eqref{Eq:fDef}
and Proposition~\ref{prop:Hoeffding}) of the mean-field lattice gas for
$\beta=0.8~(\text{lowest curve}),0.9,1,1.1,1.2~(\text{highest curve})$
and $\mu=-1.9$.}
\end{center}
\label{Fig:MeanFieldf}
\end{figure}

\section{Ising Model}
\label{Ex:Ising}

\subsection{Setup}

With the same index set $\cL$ as in the previous example, let $\cA\coloneq\{-1,+1\}$ and
$\Omega_\lambda=\cA^\lambda$ be the configuration space of the \ndex{Ising
model} in the box $\lambda\in\cL$ and denote by
$$
\Omega_\lambda\ni\bss=(s_x)_{x\in\lambda}\mapsto H_\lambda(\bss)\coloneq
-\frac J2\sum_{\substack{x,y\in\lambda\\|x-y|=1}}s_xs_y
-h\sum_{x\in\lambda}s_x
$$
the ferromagnetic ($J>0$) Ising Hamiltonian with external magnetic
field~$h\in\RR$. The associated canonical ensemble at inverse
temperature~$\beta>0$ with free boundary condition is the atomic measure on
$\Omega_\lambda$ defined by
$$
\PP_\lambda(\{\bss\})\coloneq\e^{-\beta(H_\lambda(\bss)+|\lambda|p_\lambda(\beta,h))},
$$
where the real analytic function
\beq
p_\lambda(\beta,h)
\coloneq\frac1{|\lambda|\beta}\log\sum_{\bss\in\Omega_\lambda}\e^{-\beta H_\lambda(\bss)}
\label{Eq:Isingplambda}
\eeq
is the pressure.%
\footnote{The thermodynamic potential associated to the canonical ensemble
$\PP_\lambda$ is the Helmholtz free energy per spin
$f_\lambda(\beta,h)=-p_\lambda(\beta,h)$. We introduce the pressure
here in order to keep a more consistent terminology and notation.}
We shall denote by $\langle\,\argdot\,\rangle_{\lambda,\beta,h}$ the expectation
w.r.t.\;$\PP_\lambda$. A physically natural $\ZZ_2$-action on $\Omega_\lambda$
is provided by the spin-flip operation $\Theta_\lambda:\bss\mapsto-\bss$.
Defining then $\wP_\lambda=\PP_\lambda\circ\Theta_\lambda$, we obtain by
construction that the family
$(\PP_\lambda,\widehat{\PP}_\lambda)_{\lambda\in\cL}$ is in involution. We
note that $\wP_\lambda$ is obtained from $\PP_\lambda$ by replacing the external
field $h$ by $-h$.

\subsection{Transient Fluctuation Relation}

The entropy production random variable of the family
$(\PP_\lambda,\widehat{\PP}_\lambda)_{\lambda\in\cL}$ is given by
$$
\sigma_\lambda=\log\frac{\d\PP_\lambda}{\d\wP_\lambda}
=2|\lambda|\beta h m_\lambda,
$$
where
$$
m_\lambda(\bss)\coloneq\frac1{|\lambda|}\sum_{x\in\lambda}s_x
$$
is the specific magnetization. In this setting the general
Inequality~\eqref{Eq:Jarzynski_inequality} reads
$$
h\langle m_\lambda\rangle_{\lambda,\beta,h}\ge0,
$$
and implies  that for $h\not=0$ the mean magnetization is aligned
with the external field. Together with Markov's inequality, Jarzynski's
identity~\eqref{Eq:Jarzynski_identity}, which becomes
$$
\langle\e^{-2|\lambda|\beta h m_\lambda}\rangle_{\lambda,\beta,h}=1,
$$
allows us to deduce (see~\eqref{Eq:Jarzy_and_Markov}) the more precise statement
$$
\PP_\lambda\{hm_\lambda\le-E\}\le\e^{-2|\lambda|\beta E},
$$
showing that the probability of misaligned magnetization becomes exponentially
small in large boxes. The Rényi entropy~\eqref{Eq:e_alpha_def} is also easily
identified with
\beq
e_\lambda(\alpha)
=\log\langle\e^{-2|\lambda|\beta\alpha h m_\lambda}\rangle_{\lambda,\beta,h}
=|\lambda|\beta(p_\lambda(\beta,(1-2\alpha)h)-p_\lambda(\beta,h))
\label{Eq:Isingelambda}
\eeq
so that the symmetry~\eqref{Eq:e_lambda_symmetry} reduces to the pressure
identity\footnote{Note that this corresponds to Equ.~(3.9) in~\cite[Proposition~5.3.2]{Ruelle1969}.}
$$
p_\lambda(\beta,-h)=p_\lambda(\beta,h).
$$
Differentiation w.r.t\;$h$ yields the symmetry
$$
\langle m_\lambda\rangle_{\lambda,\beta,-h}
=-\langle m_\lambda\rangle_{\lambda,\beta,h},
$$
and in particular the absence of spontaneous magnetization
$\langle m_\lambda\rangle_{\lambda,\beta,0}=0$.

\subsection{Fluctuation Theorem}

The thermodynamic limit of the pressure~\eqref{Eq:Isingplambda}
$$
p(\beta,h)\coloneq\lim_{\lambda\in\cL}p_\lambda(\beta,h)
$$
exists for all $(\beta,h)\in\RR_+\times\RR$. It follows
from~\eqref{Eq:Isingelambda} that the entropic pressure is given by
$$
\RR\ni\alpha\mapsto e(\alpha)=\beta(p(\beta,(1-2\alpha)h)-p(\beta,h)).
$$
In dimension $d=1$, $p$ is given by (see~\cite[Chapter~14]{Huang1987})
$$
p(\beta,h)=J+\frac1\beta\log\left(
\cosh(\beta h)+\sqrt{\sinh^2(\beta h)+\e^{-4\beta J}}
\right).
$$
Thus, the entropic pressure
$$
e(\alpha)
=\log\left(\frac{\cosh((1-2\alpha)\beta h)
+\sqrt{\sinh^2((1-2\alpha)\beta h)
+\e^{-4\beta J}}}{\cosh(\beta h)+\sqrt{\sinh^2(\beta h)
+\e^{-4\beta J}}}\right)
$$
vanishes identically for $h=0$ and is a strictly convex real analytic function
for $h\not=0$. The global FT again follows from the Gärtner--Ellis theorem,
the rate function $I$ being given by the Legendre
transform~\eqref{Eq:IasLegendreTransform}.

In dimension $d>1$, no explicit formula is known for the pressure $p$.
However, it is known that a \ndex{phase transition} occurs at some critical inverse
temperature $0<\beta_{\rm c}<\infty$, see~\cite[Chapter~V]{Ellis2006}. For
$\beta<\beta_c$, $p(\beta,h)$ is a real analytic function of $h$ and so is the
entropic pressure: the global FT still follows from the Gärtner--Ellis theorem
in this case. For $\beta>\beta_c$, $p(\beta,h)$ is non-differentiable at $h=0$.
More precisely,
$$
(\partial_h^-p)(\beta,0)=-m(\beta)<0<m(\beta)=(\partial_h^+p)(\beta,0),
$$
where $m(\beta)$ is the spontaneous magnetization. Consequently, for $h\not=0$,
the entropic pressure is non-differentiable at $\alpha=1/2$, with
$$
(\partial^-e)(1/2)=-2\beta h m(\beta)<0<2\beta h m(\beta)=(\partial^+e)(1/2),
$$
and the Gärtner--Ellis route to the global FT fails. Due to the truly
interacting nature of the Ising model, a direct approach along the lines of our
analysis of the mean-field lattice gas seems more challenging. However, as we
shall now sketch, it is still within the reach of currently available
techniques. We refer to~\cite{Ruelle2004} for the needed facts on the
thermodynamics of lattice systems and to~\cite{Friedli2017} for the more
specific issues pertaining to the Ising model.

The configuration space of the infinitely extended $d$-dimensional Ising model
is $\Omega\coloneq\cA^{\ZZ^d}$ endowed with the product topology, a separable compact
metrizable space. It comes equipped with its Borel $\sigma$-algebra $\cB$. For
$\Lambda\subset\ZZ^d$, set $\Omega_\Lambda\coloneq\cA^\Lambda$ and let
$\cB_\Lambda$ be the associated Borel $\sigma$-algebra, which we identify with a
subalgebra of $\cB$. We denote the restriction of $\bss\in\Omega$ to $\Lambda$
by $\bss_\Lambda\in\Omega_\Lambda$ and the restriction of $\PP\in\cP(\Omega)$ to
$\cB_\Lambda$ by $\PP^{(\Lambda)}$. For $\bss\in\Omega_\Lambda$, the cylinder
subset of $\Omega$ based on $\bss$ is
$$
[\bss]\coloneq\{\bss'\in\Omega\,|\,\bss'_\Lambda=\bss\}.
$$
Let $\tau$ be the natural
action of $\ZZ^d$ on $\Omega$ and let $\cP_\tau(\Omega)$ be the set of
$\tau$-invariant elements of $\cP(\Omega)$. The pressure admits the variational
expression
\beq
\beta p(\beta,h)
=\sup_{\QQ\in\cP_\tau(\Omega)}\left(S(\QQ)+\beta\int G_h(s)\QQ(\d s)\right),
\label{Eq:IsingVariationalPressure}
\eeq
where
$$
\cP_\tau(\Omega)\ni\QQ\mapsto
S(\QQ)\coloneq\lim_{\lambda\in\cL}-\frac1{|\lambda|}\,
\sum_{\mathclap{\bss\in\Omega_\lambda}}\QQ([\bss])\log\QQ([\bss])
=\inf_{\lambda\in\cL}-\frac1{|\lambda|}\,
\sum_{\mathclap{\bss\in\Omega_\lambda}}\QQ([\bss])\log\QQ([\bss]),
$$
is the mean entropy, a non-negative, upper semicontinuous affine function, and
the potential
$$
\Omega\ni\bss\mapsto G_h(\bss)\coloneq s_0\left(h+\frac J2\sum_{|x|=1}s_x\right)
$$
is a continuous function.

A measure $\PP\in\cP(\Omega)$ is a \ndex{Gibbs state} if it satisfies the \ndex{DLR equation}
\beq
\PP([\bss_\lambda])=\int_{\Omega_{\lambda^c}}
\frac{\ds\e^{-\beta(H_\lambda(\bss_\lambda)+W_\lambda(\bss_\lambda,\bsr_{\lambda^c}))}}
{\ds\sum_{\bsr_\lambda\in\Omega_\lambda}
\e^{-\beta(H_\lambda(\bsr_\lambda)+W_\lambda(\bsr_\lambda,\bsr_{\lambda^c}))}}
\PP^{(\lambda^c)}(\d\bsr_{\lambda^c})
\label{Eq:DLR}
\eeq
for all $\bss_\lambda\in\Omega_\lambda$ and all $\lambda\in\cL$, where
$\lambda^c=\ZZ^d\setminus\lambda$ and
$$
W_\lambda(\bss,\bsr)\coloneq-J\sum_{(x,y)\in\partial\lambda}s_xr_y
$$
is the interaction energy across the boundary
$\partial\lambda\coloneq\{(x,y)\in\lambda\times\lambda^c\mid|x-y|=1\}$.
$\PP$ is an \ndex{equilibrium state} whenever it is $\tau$-invariant
and satisfies
$$
S(\PP)+\beta\langle G_h,\PP\rangle=\beta p(\beta,h).
$$
If $\PP\in\cP_\tau(\Omega)$, then $\PP$ is a Gibbs state iff it is an equilibrium state.
If $h\not=0$, then there is a unique Gibbs state $\PP_{\beta,h}$. Moreover,
$\PP_{\beta,h}\in\cP_\tau(\Omega)$ so that it is also the unique equilibrium state.

For each $\lambda\in\cL$ and $\bss\in\Omega$ define the empirical measure
$$
\xi_\lambda\coloneq\frac1{|\lambda|}\sum_{x\in\lambda}\delta_{\tau^x(\bss)}\in\cP(\Omega)
$$
where $\delta_\bss$ denotes the unit mass at $\bss$. It follows
from~\cite[Theorem~A]{Eizenberg1994} that the family
$(\xi_\lambda)_{\lambda\in\cL}$ satisfies a level-3 LDP
w.r.t.\;$\PP_{\beta,h}$ with scale $r_\lambda=|\lambda|$, {\sl i.e.,}
$$
-\inf_{\QQ\in\mathring{\Gamma}}\II(\QQ)
\le\liminf_{\lambda\in\cL}\frac1{|\lambda|}
\log\PP_{\beta,h}\{\xi_\lambda\in\Gamma\}
\le\limsup_{\lambda\in\cL}\frac1{|\lambda|}
\log\PP_{\beta,h}\{\xi_\lambda\in\Gamma\}
\le-\inf_{\QQ\in\bar\Gamma}\II(\QQ)	
$$
holds for all Borel sets $\Gamma\subset\cP(\Omega)$ with the good rate function
\beq
\II(\QQ)=\left\{\bear{ll}
\ds\beta p(\beta,h)-\beta\langle G_h,\QQ\rangle-S(\QQ),
&\text{if }\QQ\in\cP_\tau(\Omega);\\[12pt]
+\infty,&\text{if } \QQ\in\cP(\Omega)\setminus\cP_\tau(\Omega).
\ear
\right.
\label{Eq:IsingIIdef}
\eeq
Let $\Theta$ denote the spin flip on $\Omega$ and set
$\sigma(\bss)\coloneq G_h(\bss)-G_h\circ\Theta(\bss)=2hs_0$. One easily checks that
$$
|\lambda|\varsigma(\xi_\lambda)=\sigma_\lambda
$$
where $\varsigma$ is the continuous map
$$
\cP(\Omega)\ni\QQ\mapsto\varsigma(\QQ)\coloneq\beta\langle\sigma,\QQ\rangle.
$$
The contraction principle then implies that the LDP
\beq
\begin{split}
-\inf_{s\in\mathring{S}}I(s)
&\le\liminf_{\lambda\in\cL}\frac1{|\lambda|}
\log\PP_{\beta,h}\{|\lambda|^{-1}\sigma_\lambda\in S\}\\
&\le\limsup_{\lambda\in\cL}\frac1{|\lambda|}
\log\PP_{\beta,h}\{|\lambda|^{-1}\sigma_\lambda\in S\}
\le-\inf_{s\in\bar S}I(s)
\end{split}
\label{Eq:IsingGibbsLDP}
\eeq
holds for all Borel sets $S\subset\RR$ with the rate function
\beq
I(s)=\inf\left\{\II(\QQ)\mid\QQ\in\cP_\tau(\Omega),\varsigma(\QQ)=s\right\}.
\label{Eq:IsingRate}
\eeq
From the DLR equation~\eqref{Eq:DLR} we deduce that
$$
\frac{\d\PP_{\beta,h}^{(\lambda)}}{\d\PP_\lambda}(\bss_\lambda)
=\int_{\Omega_{\lambda^c}}
\frac{\ds\sum_{\bsr_\lambda\in\Omega_\lambda}
\e^{-\beta(H_\lambda(\bsr_\lambda)+W_\lambda(\bss_\lambda,\bsr_{\lambda^c}))}}
{\ds\sum_{\bsr_\lambda\in\Omega_\lambda}
\e^{-\beta(H_\lambda(\bsr_\lambda)+W_\lambda(\bsr_\lambda,\bsr_{\lambda^c}))}}
\,\PP_{\beta,h}^{(\lambda^c)}(\d\bsr_{\lambda^c}),
$$
and the estimate
$$
|W_\lambda(\bss_\lambda,\bsr_{\lambda^c})-W_\lambda(\bsr_\lambda,\bsr_{\lambda^c})|
\le J\sum_{(x,y)\in\partial\lambda}|(s_x-r_x)r_y|
\le 2J|\partial\lambda|
$$
implies that
$$
\e^{-2\beta J|\partial\lambda|}\le
\frac{\d\PP_{\beta,h}^{(\lambda)}}{\d\PP_\lambda}(\bss_\lambda)
\le\e^{2\beta J|\partial\lambda|}.
$$
The fact that $\lim_{\lambda\in\cL}|\lambda|^{-1}|\partial\lambda|=0$
immediately yields that~\eqref{Eq:IsingGibbsLDP} also holds with $\PP_{\beta,h}$
replaced by $\PP_\lambda$. Thus, the global FT holds. To identify the rate\index{FT!for the Ising model}
function $I$ we note that, due to the above mentioned properties of the mean
entropy, the function $\II$ is affine and lower semicontinuous on the compact
set $\cP_\tau(\Omega)$. Hence, the infimum in~\eqref{Eq:IsingRate} is achieved by
some $\QQ_s$. Thus, given $s_1,s_2\in\RR$, we can find
$\QQ_{s_1},\QQ_{s_2}\in\cP_\tau(\Omega)$ such that $I(s_i)=\II(\QQ_{s_i})$ and
$\varsigma(\QQ_{s_i})=s_i$. It follows that for any $\alpha\in[0,1]$ the measure
$\QQ_{[\alpha]}\coloneq\alpha\QQ_{s_1}+(1-\alpha)\QQ_{s_2}\in\cP_\tau(\Omega)$ is such
that $\II(\QQ_{[\alpha]})=\alpha I(s_1)+(1-\alpha)I(s_2)$ and
$\varsigma(\QQ_{[\alpha]})=\alpha s_1+(1-\alpha)s_2$. Hence,
\begin{align*}
I(\alpha s_1+(1-\alpha)s_2)&=\inf\left\{\II(\QQ)\mid
\QQ\in\cP_\tau(\Omega),\fs(\QQ)=\alpha s_1+(1-\alpha)s_2\right\}\\
&\le \II(\QQ_{[\alpha]})=\alpha I(s_1)+(1-\alpha)I(s_2),
\end{align*}
{\sl i.e.,} the rate function $I$ is convex. At this point we can either invoke
directly the relation between the rate function and the entropic pressure to conclude
that $I$ is given by the Legendre transform~\eqref{Eq:IasLegendreTransform}, or
compute explicitly the right-hand side of~\eqref{Eq:LegendreI},
using~\eqref{Eq:IsingRate} and~\eqref{Eq:IsingIIdef},
\begin{align*}
-\inf_{s\in\RR}(s\alpha+I(s))
&=-\inf_{s\in\RR}\,\inf\left\{s\alpha+\II(\QQ)\,|\,
\QQ\in\cP_\tau(\Omega),\varsigma(\QQ)=s\right\}\\
&=-\inf\left\{\alpha\varsigma(\QQ)+\II(\QQ)\,|\,\QQ\in\cP_\tau(\Omega)\right\}\\
&=\sup_{\QQ\in\cP_\tau(\Omega)}\left(
\beta\langle(1-\alpha)G_h+\alpha G_h\circ\Theta,\QQ\rangle+S(\QQ)\right)-\beta p(\beta,h).
\end{align*}
Since $(1-\alpha)G_h+\alpha G_h\circ\Theta=G_{(1-2\alpha)h}$, it follows
from the variational identity~\eqref{Eq:IsingVariationalPressure} that
$$
\sup_{s\in\RR}(\alpha s-I(-s))
=-\inf_{s\in\RR}(s\alpha+I(s))=\beta(p(\beta,(1-2\alpha)h)-p(\beta,h))=e(\alpha),
$$
and the convex function $I$ is again recovered
through~\eqref{Eq:IasLegendreTransform}.

\begin{figure}
\hskip-1.2cm\includegraphics[width=1.2\textwidth]{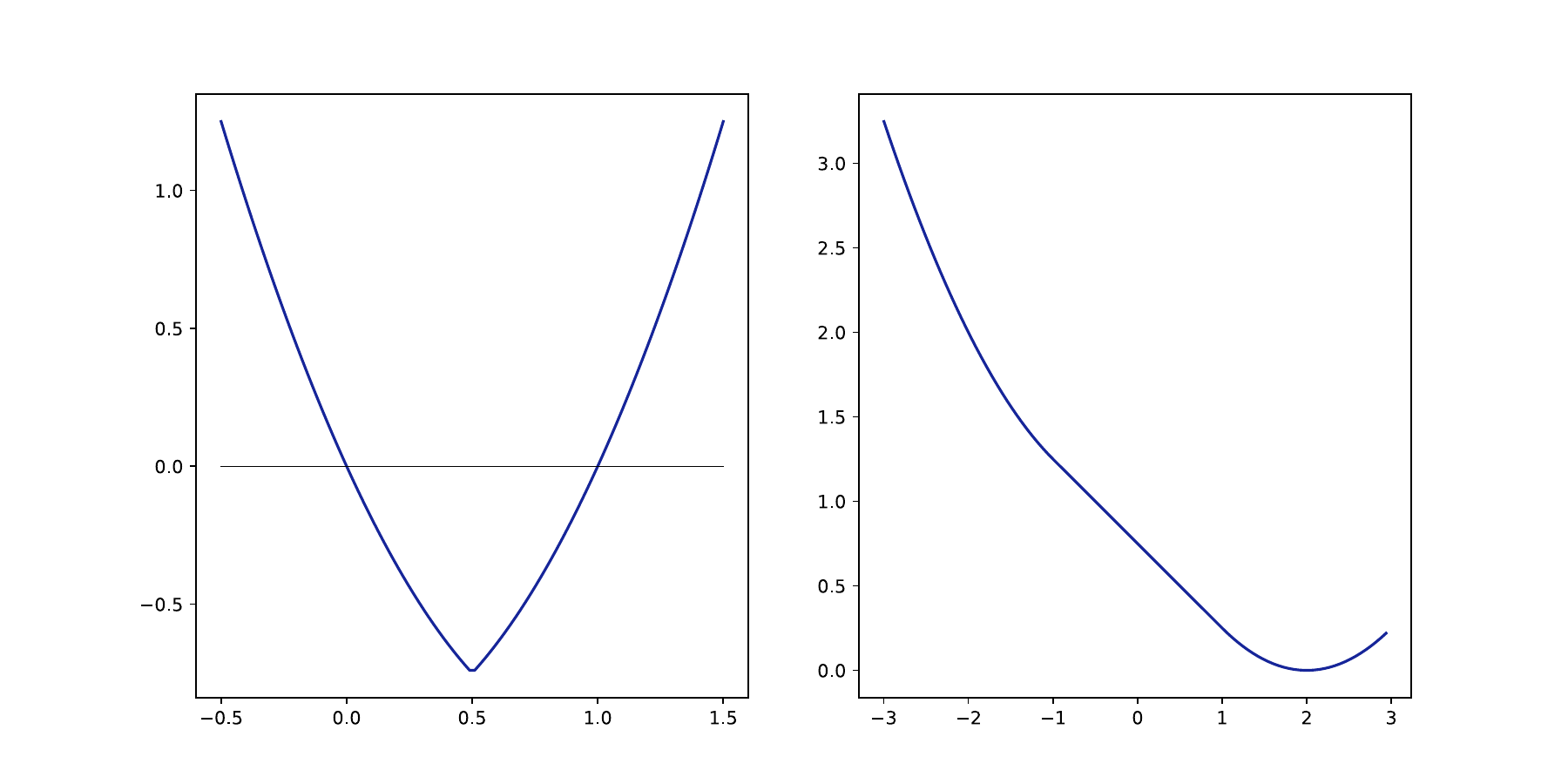}
\caption{A typical entropic pressure $e(\alpha)$ (left) and the corresponding
rate function $I(s)$ (right) for the Ising model in the phase transition regime
$\beta>\beta_{\rm c}$.}
\label{Fig:IsingRate}
\end{figure}

By Proposition~\ref{Prop:RateSymmetry}, we know that the FR, $I(-s)=I(s)+s$, holds
for all $s\in\RR$. The same conclusion immediately follows from
Equ.~\eqref{Eq:IasLegendreTransform} and the symmetry $e(1-\alpha)=e(\alpha)$ of
the entropic pressure which is equivalent to the symmetry $p(\beta,-h)=p(\beta,h)$ of
the pressure. Finally, the FR can also  be recovered from the level-3 fluctuation relation\index{FR!level-3}
$$
\II(\QQ\circ\Theta)=\II(\QQ)+\varsigma(\QQ),
$$
which holds for any $\QQ\in\cP(\Omega)$ as a direct consequence of
Definition~\eqref{Eq:IsingIIdef} and of the obvious fact that
$S(\QQ\circ\Theta)=S(\QQ)$.

The main lesson of this example is
that despite the fact that the entropic pressure is not everywhere differentiable,
the global FT holds with a rate function given by the Legendre transform of the
entropic pressure. Figure~\ref{Fig:IsingRate} shows the typical aspect of the
entropic pressure and rate function for the critical Ising Model. Note that,
due to the kink of the entropic pressure at $\alpha=1/2$, the rate function is
flat on the interval $[-\partial^+e(1/2),-\partial^-e(1/2)]$. The points in this
interval are said to be unexposed and the Gärtner--Ellis theorem fails to
provide the LDP lower bound at such points (see~\cite[Theorem~2.3.6]{Dembo2000}).

\section{Langevin Dynamics}
\label{Ex:Langevin}

Our previous examples are illustrations of the use of the contraction principle
in the derivation of the FR/FT, starting from the level-2 and/or level-3 LDP. The next example
shows that in some cases it may be relatively easy derive the LDP and prove a FR at
level-2 or level-3 using well-known techniques, but that the real difficulty appears
when trying to apply the contraction principle.\index{Langevin dynamics}

\subsection{Setup}

Let $\cG=(\cV,\cE)$ be a finite connected graph. For each vertex $i\in\cV$, let
$V_i\in C^\infty(\RR)$ be such that
$$
\inf_qV_i(q)\ge0,
\qquad
\lim_{|q|\to\infty}|V'_i(q)|=+\infty,
\qquad
\sup_q|V_i''(q)|<\infty.
$$
For each edge $e\in\cE$, let $W_{e}\in C_b^\infty(\RR)$%
\footnote{$C_b^\infty(\RR)$ denotes the space of smooth, bounded functions
whose derivatives of all orders are bounded.} be non identically vanishing
and such that
$$
W_e(-q)=W_e(q).
$$
Set
$$
U(q)\coloneq V(q)+W(q)\coloneq\sum_{i\in\cV}V_i(q_i)+\sum_{\{i,j\}\in\cE}W_{\{i,j\}}(q_i-q_j),
$$
and consider the Hamiltonian $H(q,p)\coloneq\frac12|p|^2+U(q)$ on the phase space
\[
\mathfrak{X}\coloneq\mathrm{T}^\ast\RR^\cV=\RR^\cV\oplus\RR^\cV
\]
equipped with its canonical
symplectic structure. We set $x\coloneq(q,p)$ and denote by $\d x=\d q\d p$ the
induced Liouville measure on $\fX$. Let $\Gamma$ and $T$ be diagonal matrices acting on
$\RR^\cV$ with strictly positive diagonal entries $\Gamma_i$ and $T_i$. The stochastic
differential equation (SDE)
\beq
\begin{split}
\d q(s)&=p(s)\d s;\\[2pt]
\d p(s)&=-\left(\nabla U(q(s))+T^{-1}\Gamma p(s)\right)\d s
+(2\Gamma)^{1/2}\d w(s),
\end{split}
\label{Eq:SDE}
\eeq
where $w$ denotes a standard Wiener process in $\RR^\cV$, provides a statistical
model for the dynamics of an open network $\cS$ whose internal dynamics is
governed by the Hamiltonian $H$ and where each node $i\in\cV$ is thermostated by
a reservoir $\cR_i$ in thermal equilibrium at temperature $T_i$. In the
physically more interesting case where $\Gamma$ is degenerate, the analysis of
the resulting dynamical system  becomes quite delicate and relies on the
detailed structure of the potential $U$. We refer the reader
to~\cite{Eckmann1999a,Eckmann1999b,Eckmann2000,Rey-Bellet2002a,Rey-Bellet2002b,Eckmann2003,Hairer2009} for
results pertaining to boundary driven chains of oscillators, and
to~\cite{Maes2003b,Eckmann2004,Bodineau2008,Cuneo2014,Jaksic2017,Cuneo2018} for more general networks.
Here, we shall only sketch the strategy that leads to the FT in the non-degenerate case
$\Gamma>0$. Our discussion heavily relies on the works~\cite{Wu2001,Bodineau2008} to
which we refer for more details.

The Markov generator associated to the SDE~\eqref{Eq:SDE} is given
by\footnote{Here $\{H,\,\argdot\,\}$ denotes a Poisson bracket.}
$$
L=\left(\nabla_p-T^{-1}p\right)\cdot\Gamma\nabla_p+\{H,\,\argdot\,\}.
$$
We denote by
$$
L^T=\nabla_p\cdot\Gamma\left(\nabla_p+T^{-1}p\right)-\{H,\,\argdot\,\}
$$
its formal adjoint, and by
\begin{align*}
\gamma(f,g)&\coloneq\frac12(Lfg-fLg-gLf)=\nabla_pf\cdot\Gamma\nabla_p g,\\[4pt]
\gamma(f)&\coloneq\gamma(f,f)=|\Gamma^{1/2}\nabla_pf|^2,
\end{align*}
the associated \ndex{carré du champ} operators. We recall the chain rule
\beq
L(F\circ f)=(F'\circ f)Lf+(F''\circ f)\gamma(f).
\label{eq:ChainRule}
\eeq
It follows that, for $\vartheta>T_{\rm max}\coloneq\max_i T_i$,
$$
L\,\e^{H/\vartheta}
=\vartheta^{-1}\left(
\tr(\Gamma)
-\Gamma^{1/2}p\cdot(T^{-1}-\vartheta^{-1})\Gamma^{1/2}p\right)
\,\e^{H/\vartheta}
\le\vartheta^{-1}\tr(\Gamma)\,\e^{H/\vartheta}.
$$
Invoking~\cite[Theorem~3.5]{Khasminskii2012} we can conclude that
SDE~\eqref{Eq:SDE} has global strong solutions which, together with the law
of the initial condition $x(0)$, assumed to be independent of the
driving Wiener process, define a Markov process with continuous path and hence,
for each $t\in\cL=\RR_+$, the law of $\bsx=(x(s))_{s\in[0,t]}$ is a probability
measure on $\Omega_t=C([0,t],\mathfrak{X})$. For $x\in\fX$, we shall denote by $\PP_{x,t}$
the law of the process started at $x(0)=x$ and by $\EE_{x,t}$ the corresponding
expectation functional. If $\mu\in\cP(\fX)$, then
$$
\PP_{\mu,t}\coloneq\int_{\fX}\PP_{x,t}\,\mu(\d x)
$$
is the law of the process started with $x(0)$ distributed according to $\mu$,
and we shall denote by $\mu_s$ the law of $x(s)$ under $\PP_{\mu,t}$.

Invoking Hörmander's commutator theorem~\cite[Theorem~22.2.1]{Hormander-III},
one easily establishes the micro-hypoellipticity of the operators $L$, $L^T$,
$\partial_t-L$ and $\partial_t-L^T$. It follows that for all $s>0$ the Markov
transition kernel of the process is strong Feller with a smooth density
$P^s(x,y)$ w.r.t.\;the Liouville measure on $\fX$, {\sl i.e.,} for any bounded
measurable function $f$ on $\fX$,
\[
\EE_{x,t}[f(x(t))]=\int_{\fX}P^t(x,y)f(y)\d y,
\]
where the function $(t,x,y)\mapsto P^t(x,y)$ is in
$C^\infty(]0,\infty[\times\fX\times\fX)$.

The system is controllable, {\sl i.e.,} for any $x,y\in\mathfrak{X}$ and any $t>0$
there exists a control $u\in C([0,t],\RR^\cV)$ such that the solution
$x_u=(q_u,p_u)$ of the system
\begin{align*}
\dot q(s)&=p(s),\\[2pt]
\dot p(s)&=-\left(\nabla U(q(s))+T^{-1}\Gamma p(s)\right)+u(s),
\end{align*}
with initial condition $x_u(0)=x$ satisfies $x_u(t)=y$. It follows from the
results of~\cite{Stroock1972}, see also~\cite{Eckmann1999a,Cuneo2018}, that for any
$(t,x,y)\in]0,\infty[\times\fX\times\fX$ one has $P^t(x,y)>0$. Moreover, the process
has a unique stationary measure $\mu_\st$ whose density
$\rho_\st=\frac{\d\mu_\st}{\d x}\in C^\infty(\fX)$ is everywhere positive and satisfies
$L^T\rho_\st=0$.

In the special case where the interaction potential $W$ vanishes
identically, the stationary measure is the {\sl local equilibrium state}
\beq
\nu(\d x)\coloneq Z^{-1}\e^{-S(x)}\d x, \qquad
S(x)\coloneq\frac12 p\cdot T^{-1}p+\sum_{i\in\cV}T_i^{-1}V_i(q_i),
\label{Eq:SleDef}
\eeq
which provides a convenient reference measure. Note that $\nu$ is invariant under
the time-reversal involution of $\fX$ given by $\theta:(q,p)\mapsto(q,-p)$.
A $\ZZ_2$-action on $\Omega_t$ is defined by
$$
\Theta_t:\bsx=(x(s))_{s\in[0,t]}\mapsto\hat\bsx=(\theta x(t-s))_{s\in[0,t]}.
$$
In the following, we shall consider the family $(\PP_t,\wP_t)_{t\in\RR_+}$ where
$$
\PP_t\coloneq\PP_{\mu_\st,t},\qquad \wP_t\coloneq\PP_t\circ\Theta_t.
$$
$\wP_t$ describes a stationary process with invariant measure
$\hat\mu_\st\coloneq\mu_\st\circ\theta$. For reasons which will become clear
below, we define a smooth function $\phi:\fX\to\RR$ by
\beq
\phi(x)\coloneq\log\frac{\d\hat\mu_\st}{\d\nu}(x).
\label{Eq:Langevinphidef}
\eeq
A simple calculation shows that the adjoint
$L^\ast$ of $L$ w.r.t.\;the inner product of the Hilbert space $L^2(\fX,\d\nu)$
satisfies the generalized \ndex{detailed balance} relation
\beq
\Theta L^\ast\Theta=L-\sigma
\label{Eq:LangevinGDB}
\eeq
where the operators $\Theta$ and $\sigma$ are defined, respectively, by
$f\mapsto f\circ\theta$ and multiplication with the function
$$
\sigma(x)\coloneq\{H,S\}(x)=-T^{-1}\!p\cdot\nabla W(q)=G(x)-G\circ\theta(x),\quad
G(x)=\frac14|T^{-1/2}\!(p-\nabla W(q))|^2.
$$
An immediate consequence of~\eqref{Eq:LangevinGDB} is that the
function~\eqref{Eq:Langevinphidef} satisfies
\beq
L\phi+\gamma(\phi)=\sigma.
\label{Eq:LangevinphiPDE}
\eeq

It follows from well-known
results on the time-reversal of diffusion processes (see~\cite{Anderson1982,Haussmann1986})
that $\wP_t$ is the law of the solutions
$\hat{\bsx}=(\hat q(s),\hat p(s))_{s\in[0,t]}$ of the SDE
\begin{align*}
\d \hat q(s)&=\hat p(s)\d s\\[2pt]
\d \hat p(s)&=-\left(\nabla U(\hat q(s))+T^{-1}\Gamma \hat p(s)
-2\Gamma(\nabla_p\phi)(\hat q(s),\hat p(s))\right)\d s
+(2\Gamma)^{1/2}\d\hat w(s)
\end{align*}
with initial condition $\hat x(0)$ distributed according to $\hat\mu_\st$, independently
of the standard Wiener process $\hat w$. In particular, the family
$(\wP_t)_{t\in\RR_+}$ describes a stationary Markov process with
generator $\hat L=L+2\gamma(\phi,\,\argdot\,)$ and transition kernel
\[\hat P^t(x,y)=\rho_\st(\theta y)P^t(\theta y,\theta x)\rho_\st(\theta x)^{-1}.\]
We shall denote by $\wP_{x,t}$ the probability induced on $\Omega_t$
by conditioning $\wP_t$ on the starting value $\hat x(0)=x$. Setting
$$
\eta(s)\coloneq\int_0^s (2\Gamma)^{1/2}\nabla_p\phi(x(r))\cdot\d w(r),
\qquad
[\eta](s)\coloneq\int_0^s|(2\Gamma)^{1/2}\nabla_p\phi(x(r))|^2\d r,
$$
Itô calculus and the generalized detailed balance, through
Relation~\eqref{Eq:LangevinphiPDE}, yield
$$
\eta(s)-\frac12[\eta](s)=\phi(x(s))-\phi(x(0))-\int_0^s\sigma(x(r))\d r.
$$
Following the lines of the proof of~\cite[Proposition~3.9]{Jaksic2017}, one shows
that the Doléans-Dade stochastic exponential
$$
[0,t]\ni s\mapsto Z(s)\coloneq\e^{\eta(s)-\frac12[\eta](s)}
$$
is a $\PP_{x,t}$-martingale. It follows from Girsanov's theorem that
$$
\frac{\d\PP_{x,t}}{\d\wP_{x,t}}(\bsx)=\exp\left[
\phi(x(0))-\phi(x(t))+\int_0^t\sigma(x(s))\d s\right].
$$

\subsection{Transient Fluctuation Relation}

Since
\[
\frac{\d\PP_t}{\d\wP_t}(\bsx)
=\frac{\d\PP_{\mu_\st,t}}{\d\wP_{\hat{\mu}_\st,t}}(\bsx)
=\frac{\d\PP_{\mu_\st,t}}{\d\wP_{\mu_\st,t}}(\bsx)\,
\frac{\d\mu_\st}{\d\hat\mu_\st}(x(0))
=\frac{\d\PP_{x(0),t}}{\d\wP_{x(0),t}}(\bsx)\,
\e^{\phi\circ\theta(x(0))-\phi(x(0))},
\]
the entropy production random variable of the pair $(\PP_t,\wP_t)$ is given
by
\beq
\sigma_t(\bsx)=\log\frac{\d\PP_t}{\d\wP_t}(\bsx)
=\phi\circ\theta(x(0))-\phi(x(t))+\int_0^t\sigma(x(s))\d s.
\label{Eq:LangevinEP}
\eeq
To connect this expression with thermodynamics, we note that by Itô's formula
the change in energy is given by
$$
\d H(x(t))=\sum_{i\in\cV}\left((2\Gamma_i)^{1/2} p_i(t)\d w_i(t)
+\Gamma_i(1-T_i^{-1}p_i(t)^2)\d t\right)
\eqcolon-\sum_{i\in\cV}\d\Phi_i(t).
$$
Interpreting $\d\Phi_i(t)$ as the instantaneous flux of energy dissipated in the
reservoir $\cR_i$, the total thermodynamic entropy dissipated in the environment
during the interval $[0,t]$ is given by
\begin{align*}
\fS_t(\bsx)&\coloneq\sum_{i\in\cV}\int_0^t\frac{\d\Phi_i(s)}{T_i}\\
&=-\int_0^t\left(
(2\Gamma)^{1/2} T^{-1} p(s)\cdot\d w(s)
-\frac12|(2\Gamma)^{1/2}T^{-1}p(s)|^2\d s
+\tr(\Gamma T^{-1})\d s\right).
\end{align*}
Applying Itô's formula to the function $S$ defined in~\eqref{Eq:SleDef} we
derive
$$
\d S(x(t))=(2\Gamma)^{1/2}T^{-1}p(t)\cdot\d w(t)
-\left(\frac12|(2\Gamma)^{1/2}T^{-1}p(t)|^2-\tr(\Gamma T^{-1})-\sigma(x(t))\right)\d t
$$
and hence
\beq
\fS_t(\bsx)=\log\frac{\d\nu}{\d x}(x(t))-\log\frac{\d\nu}{\d x}(x(0))
+\int_0^t\sigma(x(s))\d s.
\label{Eq:LangevinPEP}
\eeq
Comparison with~\eqref{Eq:LangevinEP} allows us to conclude that
$$
\sigma_t(\boldsymbol{x})=\log\rho_\st(x(0))
-\log\rho_\st\circ\theta(x(t))+\fS_t(\bsx).
$$
Thus, up to boundary terms, $\sigma_t$ coincides with the entropy dissipated
during the time interval $[0,t]$. To clarify the physical content of this
expression, we observe that it leads to the following form of the relative
entropy$$
\ep_t\coloneq\Ent(\PP_t|\wP_t)
=\EE_t[\sigma_t]=\Ent(\mu_\st|\hat\mu_\st)+\Delta S_\cR(t),
$$
where
$$
\Delta S_\cR(t)=\EE_{t}[\fS_t]=t\mu_\st(\sigma)
$$
is the expected increase in the entropy of the environment $\cR$.
Inequality~\eqref{Eq:Jarzynski_inequality} yields
$$
\mu_\st(\sigma)=\lim_{t\to\infty}\frac{\ep_t}{t}\ge0.
$$
Since the configurational entropy of the system $\cS$
\[
S_\cS(t)\coloneq-\int\rho_\st\log\rho_\st\d x
\]
is constant, we can rewrite the previous inequality as
$$
\Delta S_{\cS+\cR}(t)=\Delta S_\cS(t)+\Delta S_\cR(t)\ge0,
$$
which is a microscopic form of the $2^{\rm nd}$ law of thermodynamics.
Using Identity~\eqref{Eq:LangevinphiPDE},  we can further write
\[
\mu_\st(\sigma)=\mu_\st(L\phi+\gamma(\phi))=\mu_\st(\gamma(\phi)),
\]
from which we conclude that $\mu_\st(\sigma)=0$ iff
$\nabla_p\phi$ vanishes identically. Assume now that $\mu_\st(\sigma)=0$. Then
\beq
\rho_\st(q,p)=\e^{-\frac12 p\cdot T^{-1}p+f(q)}
\label{Eq:Langevinrhostspecial}
\eeq
for some $f\in C^\infty(\RR^\cV)$. Insertion into the equation $L^T\rho_\st=0$ yields
$$
p\cdot(\nabla f(q)+T^{-1}\nabla U(q))=0
$$
from which we deduce
$$
\nabla f+T^{-1}\nabla U=0,
$$
and hence
\beq
\left(\frac1{T_i}-\frac1{T_j}\right)
\frac{\partial^2 W}{\partial q_i\partial q_j}=0.
\label{Eq:LagevinThermalEquilibrium}
\eeq
Under our assumptions on the interaction potential $W$, this is only possible
whenever $T_i=T_j$ for all $\{i,j\}\in\cE$, {\sl i.e.,} $T_i=T_*$ for all $i\in\cV$.
This means that the environment $\cR$ is in thermal equilibrium at temperature
$T_*$. In this case, the process is stationary and in thermal equilibrium, with
\beq
\mu_\st(\d x)=Z^{-1}\e^{-H(x)/T_*}\d x,
\label{Eq:LangevinGibbsMeasure}
\eeq
the Gibbs measure at temperature $T_*$. Note that since $H\circ\theta=H$, the
measure~\eqref{Eq:LangevinGibbsMeasure} is time-reversal invariant and it follows
that the process is reversible, {\sl i.e.,} $\wP_t=\PP_t$. Reciprocally, if the process is reversible
then $\mu_\st\circ\theta=\mu_\st$ and the relation $L^T\rho_\st=0$ splits into the two
equations
$$
\nabla_p\cdot\Gamma(\nabla_p+T_*^{-1}p)\rho_\st=0,\qquad\{H,\rho_\st\}=0.
$$
The first one implies~\eqref{Eq:Langevinrhostspecial} and the second one leads again
to~\eqref{Eq:LagevinThermalEquilibrium}. Thus, the process is reversible iff the
environment $\cR$ is in thermal equilibrium and in the opposite case the mean
entropy production rate
$$
\epr\coloneq\lim_{t\to\infty}\frac{\ep_t}{t}
$$
is strictly positive and the Rényi entropy
$$
e_t(\alpha)=\Ent_\alpha(\PP_t|\wP_t)
$$
is non-vanishing for $\alpha\in]0,1[$.

\subsection{Fluctuation Theorems}

FTs are only known to hold for a few special classes of Langevin dynamical
systems. A local FT for boundary driven chains of anharmonic oscillators has
been obtained in~\cite{Rey-Bellet2002a}. The status of the global FT for general
harmonic networks has been discussed in detail in~\cite{Jaksic2017,Damak2020}.
Closely related examples are the local FT for Gaussian dynamical systems
obtained in~\cite{Jaksic2010b,Jaksic2016} and the local FT for strongly
dissipative stochastic partial differential equations proved
in~\cite{Jaksic2015a}. As we shall outline now, the main difficulty one
encounters in such derivations lies in the lack of compactness of the phase
space $\fX$ and in the unboundedness of the entropy production random variable.

Coming back to the setup of the Langevin equation~\eqref{Eq:SDE}, our current
assumptions imply that the Markov transition kernel $P^t$ and $\hat P^t$ are
strong Feller. Moreover, the Lyapunov functions described in
Lemma~\ref{Lem:Lyapunov} together with~\cite[Theorem~3.1]{Wu2001} give that the
Donsker--Varadhan theory~\cite{Donsker1983} applies to both the direct and the
time-reversed stationary processes. To formulate these results, we need to
introduce some additional objects and notations.

Equip the path spaces $\Omega\coloneq C(\RR_+,\fX)$ and $\bar\Omega\coloneq C(\RR,\fX)$ with
the topology of locally uniform convergence and the induced Borel structure. For
a subset $I\subset\RR$ and $\bsx=(x(s))_{s\in\RR}$, set $\bsx_I=(x(s))_{s\in
I}$ and let $\fF_I$ be the $\sigma$-subalgebra generated by $\bsx_I$. We denote
by $\tau$ the left shift on $\Omega$, {\sl i.e.,} the map
\[
\bsx=(x(t))_{t\in\RR_+}\mapsto\tau^s(\bsx)\coloneq(x(s+t))_{t\in\RR_+}.
\]
We shall use the same symbol to denote the left shift on $\bar\Omega$.

Any $\QQ\in\cP_\tau(\Omega)$ has a unique extension $\bar\QQ\in\cP_\tau(\bar\Omega)$,
the measure determined by $\bar\QQ(\Gamma)=\QQ\circ\tau^{-t}(\Gamma)$ for
$\Gamma\in\cF_{[-t,\infty[}$. The map $\QQ\mapsto\bar\QQ$ is bijective, its inverse
being the restriction
$\cP_\tau(\bar\Omega)\ni\QQ\mapsto\QQ|_{\fF_{[0,\infty[}}\in\cP_\tau(\Omega)$.
Let $\vartheta$ be the time reversal defined on $\bar\Omega$ by
$$
\bsx=(x(t))_{t\in\RR}\mapsto\vartheta(\bsx)\coloneq(\theta x(-t))_{t\in\RR},
$$
and denote by $\widehat{\QQ}\coloneq\QQ\circ\vartheta$ the involution induced on
$\cP_\tau(\bar\Omega)$. Setting $\widehat{\QQ}\coloneq\widehat{\bar\QQ}|_{\fF_{[0,\infty[}}$
defines an involution on $\cP_\tau(\Omega)$ such that
$\bar{\widehat{\QQ}}=\widehat{\bar\QQ}$. In particular, if $\PP\in\cP_\tau(\Omega)$ is
the measure induced by the family $(\PP_t)_{t\in\RR_+}$, then $\wP$ coincides
with the measure induced by the family $(\wP_t)_{t\in\RR_+}$.

For $\QQ\in\cP_\tau(\bar\Omega)$ and $t\in\RR$ we set
$\QQ_{t]}\coloneq\QQ|_{\fF_{]-\infty,t]}}$ and define
$$
\bar\QQ_{0]}\otimes\PP_t(\d\bsx)
\coloneq\bar\QQ(\d\bsx_{]-\infty,0]})\PP_{x(0),t}(\d\bsx_{[0,t]}).
$$

\part{Level-3 FT and FR.} The process-level empirical measures
$\boldsymbol{\xi}_t\in\cP(\Omega)$ defined by\index{FT!for Langevin dynamics}
$$
\Omega\ni\bsx\mapsto
\boldsymbol{\xi}_t\coloneq\frac1t\int_0^t\delta_{\tau^s(\bsx)}\d s
$$
satisfy the LDP w.r.t.\;the measures $\PP$ and $\wP$ with scale $r_t=t$, the
good rate functions being given by~\cite[Theorems~4.6 and~5.5]{Donsker1983}
$$
\II(\QQ)\coloneq\left\{
\begin{array}{ll}
\Ent(\bar\QQ_{1]}|\bar\QQ_{0]}\otimes\PP_1)
&\text{if }\QQ\in\cP_\tau(\Omega);\\[4pt]
+\infty&\text{otherwise},
\end{array}
\right.\quad
\hat{\II}(\QQ)\coloneq\left\{
\begin{array}{ll}
\Ent(\bar\QQ_{1]}|\bar\QQ_{0]}\otimes\wP_1)
&\text{if }\QQ\in\cP_\tau(\Omega);\\[4pt]
+\infty&\text{otherwise}.
\end{array}
\right.
$$
Since
$$
\log\frac{\d\bar\QQ_{0]}\otimes\PP_1}{\d\bar\QQ_{0]}\otimes\wP_1}
=\log\frac{\d\PP_{x(0),1}}{\d\wP_{x(0),1}}
=\phi(x(0))-\phi(x(1))+\int_0^1\sigma(x(s))\d s,
$$
it immediately follows that the level-3 FR\index{FR!level-3}
\beq
\hat{\II}(\QQ)=\II(\QQ)+\int\sigma(x(0))\QQ(\d\bsx)
\label{Eq:LangevinL3FR}
\eeq
holds for all $\QQ\in\cP_\tau(\Omega)$ such that both $\phi(x(0))$  and
$\sigma(x(0))$ belong to $L^1(\Omega,\QQ)$. Let us denote by $\phi_\pm$ the
positive/negative part of $\phi$. We shall see in the next paragraph that
$\sigma(x(0))$ and $\phi_+(x(0))$ belong to $ L^1(\Omega,\QQ)$ whenever
$\II(\QQ)<\infty$. Thus, \eqref{Eq:LangevinL3FR} holds  for all stationary
processes $\QQ$ such that $\II(\QQ)<\infty$ and\footnote{We are not aware of any
estimate of $\phi_-$ allowing us to deduce that
$\phi_-(x(0))\in L^1(\Omega,\QQ)$ from the fact that  $\II(\QQ)<\infty$ .}
$\phi_-(x(0))\in L^1(\Omega,\QQ)$. Moreover, invoking the uniqueness of the rate
function one easily derives the relation
\beq
\hat{\II}(\QQ)=\II(\widehat{\QQ})
\label{Eq:hatIIForm}
\eeq
for all $\QQ\in\cP_\tau(\Omega)$.

We shall now proceed to lower levels FR/FT via contractions. In order to keep the discussion
to a reasonable level of technicality, we will only outline the argument and leave two
technical lemmata to the~\hyperlink{AppA}{Appendix}.

\mypar{Level-2 FT and FR}
By contraction, the level-$2$ empirical measures $\xi_t\in\cP(\fX)$, defined by
$$
\Omega\ni\bsx\mapsto\xi_t\coloneq\frac1t\int_0^t\delta_{x(s)}\d s,
$$
satisfy a LDP w.r.t.\;the families $(\PP_t)_{t\in\RR_+}$ and $(\wP_t)_{t\in\RR_+}$
with scale $r_t=t$ and good rate functions
\begin{align*}
\bI(Q)=\inf\left\{\II(\QQ)\,\biggm|\,\QQ\in\cP_\tau(\Omega),
\int f(x(0))\QQ(\d\bsx)=\int f(x)Q(\d x)\text{ for all }f\in C_\mathrm{b}(\fX)
\right\},
\end{align*}
$\hat{\bI}$ being defined similarly. As immediate consequences of this
formula, we observe that Relation~\eqref{Eq:hatIIForm} yields that
\beq
\hat{\bI}(Q)=\bI(Q\circ\theta),
\label{Eq:LangevinLev2hat}
\eeq
and that $\II(\QQ)<\infty$ implies
$\QQ\in\cP_\tau(\Omega)$ with a marginal $Q\in\cP(\fX)$ satisfying
$\bI(Q)\le\bI(\QQ)<\infty$. The Donsker--Varadhan variational formula
(see~\cite{Donsker1976} or~\cite[Theorem 4.2.43]{Deuschel1989}),
$$
\bI(Q)=\sup_{\varphi}
\int\left(-L\varphi-\gamma(\varphi)\right)\d Q,
$$
and its analog for $\hat{\bI}$ hold, the supremum being taken over
all real-valued $C^2$-functions which are bounded below on $\fX$.
By~\cite[Proposition~2.8]{Jain1990}, $\bI(Q)<\infty$ implies that
$Q\ll\mu_\st$ and hence that $Q$ has a density w.r.t.\;the local
equilibrium measure $\nu$. Invoking~\cite[Corollary~2.2]{Wu2001},
it follows from the Lyapunov estimates of Lemma~\ref{Lem:Lyapunov} that
(recall Equ.~\eqref{Eq:SleDef})
$$
\int|\nabla S|^2\d Q<\infty
$$
whenever $\bI(Q)<\infty$ or $\hat\bI(Q)<\infty$. A direct consequence is that
$\sigma\in L^1(\fX,Q)$, and invoking the upper bound~\eqref{Eq:AppAphibound}
one can further deduce that $\phi_+\in L^1(\fX,Q)$.\footnote{This establishes
the claim made in the previous paragraph.} We shall prove in Lemma~\ref{Lem:L2FR}
that under the same condition $\gamma(\phi)\in L^1(\fX,Q)$ holds and that this
implies the level-2 FR\index{FR!level-2}
\beq
\hat{\bI}(Q)=\bI(Q)+\langle\sigma,Q\rangle,
\label{Eq:LangevinL2FR}
\eeq
for all $Q\in\cP(\fX)$ such that either $\bI(Q)<\infty$ or $\hat\bI(Q)<\infty$.

\remark For $Q\in\cP(\fX)$ such that $Q\ll\nu$ with
\beq
\log\frac{\d Q}{\d\nu}\in\cD\coloneq\left\{\varphi\in C^2(\fX)\,\biggm|\,
|\varphi|+|\nabla_p\varphi|\in L^2(\fX,\d Q)\right\},
\label{Eq:QQReg}
\eeq
one can decompose the level-2 rate function as
$$
\bI(Q)=\bK(Q)+I(Q|\nu)-\frac12\langle\sigma,Q\rangle,
$$
 (see~\cite[Proposition~2]{Bodineau2008}), with
$$
\bK(Q)=\sup_{\varphi\in\cD}\langle\{H,\varphi\}-\gamma(\varphi),Q\rangle,
$$
and the \ndex{relative Fisher information}
$$
I(Q|\nu)=4\int\gamma\left(\sqrt{\frac{\d Q}{\d\nu}}\right)\d\nu.
$$
One easily checks that both $\bK$ and $I(\,\argdot\,|\nu)$ are even under $\theta$,
so that the FR~\eqref{Eq:LangevinL2FR} follows for $Q$ satisfying~\eqref{Eq:QQReg}.

\part{Level-1 FT and FR.} Formally, the FT for our Langevin process follows
from~\eqref{Eq:LangevinEP} by contraction of the Donsker--Varadhan level-2 LDP,
the rate function being given by
\beq
I(s)=\inf\left\{\bI(Q)\,\big|\,Q\in\cP_\tau(\fX),\langle\sigma,Q\rangle=s \right\}.
\label{Eq:LangevinIDef}
\eeq
The corresponding FR~\eqref{Eq:RateSymmetry}-\eqref{Eq:hatRate} are directly
inherited from the level-2
relations~\eqref{Eq:LangevinL2FR}-\eqref{Eq:LangevinLev2hat} and the fact that
$\sigma\circ\theta=-\sigma$. The two obstructions to turning this formal
argument into a rigorous proof stem from the lack of compactness of the phase
space $\fX$:
\ben
\item Due to the unboundedness of $\phi$, the additive functional
$\widetilde{\sigma}_t$ defined by
$$
\widetilde{\sigma}_t(\bsx)\coloneq\int_0^t\sigma(x(s))\d s
$$
is {\sl not\/} physically equivalent to the entropy production $\sigma_t$ given
by~\eqref{Eq:LangevinEP}. We note that the same problem occurs when trying to
derive a fluctuation theorem for the dissipated thermodynamic entropy $\fS_t$,
the role of $\phi$ being there played by $S$.
\item The unboundedness of $\sigma$ makes the map
$$
\cP(\fX)\ni Q\mapsto\langle\sigma,Q\rangle
$$
discontinuous for the weak topology of $\cP(\fX)$.
\een

Thus, under the current assumptions, it is not known whether our Langevin
process satisfies the FT in the canonical form of Definition~\ref{Def:FT}. We
shall leave this important question for future investigations. On the positive
side, since $|\sigma(x)|\le C|\nabla S(x)|$ for some constant $C$,
invoking~\cite[Corollary~2.3]{Wu2001} and Lemma~\ref{Lem:Lyapunov} yields that
the functional $\widetilde{\sigma}_t$ satisfies the FT with rate
function~\eqref{Eq:LangevinIDef}.

From a more conceptual perspective, one may argue that the level-3 and level-2
FR~\eqref{Eq:LangevinL3FR} and~\eqref{Eq:LangevinL2FR}, which only depend on
time-reversal,  are more fundamental than the level-1
FR~\eqref{Eq:RateSymmetry}-\eqref{Eq:hatRate} which also depend on the choice of
an entropy production random variable. With this point of view, the natural
entropy production random variable is $\widetilde{\sigma}_t$, and it satisfies the
global FT. Rewriting~\eqref{Eq:LangevinPEP} as
$$
\widetilde{\sigma}_t(\bsx)=\fS_t(\bsx)+S(x(t))-S(x(0)),
$$
we express the total entropy production as the sum of the entropy dissipated into the
environment  and the increase in some local entropy. The latter vanishes on average since the
system is in a steady state, but this does not prevent fluctuations which contribute to
the fluctuations of $\widetilde{\sigma}_t$.

It is instructive to compare the situation just described with what happens in purely
harmonic networks.\footnote{Note that in such systems the potentials $V$ and $W$ are
quadratic functions.  Hence, $W$ is unbounded. Harmonic networks
therefore do not fall in the framework described in this subsection. }
There, one can show that the global FT holds for the
entropy production~$\sigma_t$ with a rate function satisfying the
FR~\eqref{Eq:RateSymmetry}-\eqref{Eq:hatRate}. The dissipated thermodynamic entropy
$\fS_t$ also satisfies the global LDP. However, the corresponding rate function only satisfies
the FR
$$
\Delta(s)\coloneq I(-s)-I(s)=s
$$
for $s$ in some finite interval around $0$. Outside of this interval, the {\sl symmetry function}
$\Delta$ becomes non-universal. We refer to~\cite{Jaksic2017} and references therein for a
discussion of this phenomenon which leads to what is sometimes referred to as\index{FT!extended}
the {\sl extended FT} in the physics literature. Invoking~\cite[Proposition~3.22]{Jaksic2017},
one deduces that a similar scenario holds for the additive functional $\widetilde{\sigma}_t$.
This is illustrated in Figure~\ref{Fig:LangevinSymmetry}. There, we consider a triangular
harmonic network with potential
$$
U(q)\coloneq\frac12(q_1^2+q_2^2+q_3^2)-\frac14(q_1q_2+q_2q_3+q_3q_1).
$$
The ratios of the reservoirs temperatures are given by
$$
T_1:T_2:T_3=2-\tau:3:1+\tau,\qquad\tau\in]0,1/2[,
$$
and $\Gamma$ is the unit matrix. On the interval $]1/2-\kappa_c(\tau),1/2+\kappa_c(\tau)[$
the entropic pressure is finite. It is infinite on the interior of its complement, and
its derivative diverges at its boundary. We have chosen $3$ values of $\tau$ which are the
largest solutions of  $\kappa_c(\tau)\in\{1,3/2,2\}$. Figure~\ref{Fig:LangevinSymmetry} displays
the departure of the symmetry function $\Delta(s)$ from the straight line $\Delta(s)=s$
for the rate function of the dissipated thermodynamic entropy $\fS_t$ (left pane)
and of the additive functional $\widetilde{\sigma}_t$ (right pane).

\begin{figure}
\centering
\includegraphics[width=5.6cm]{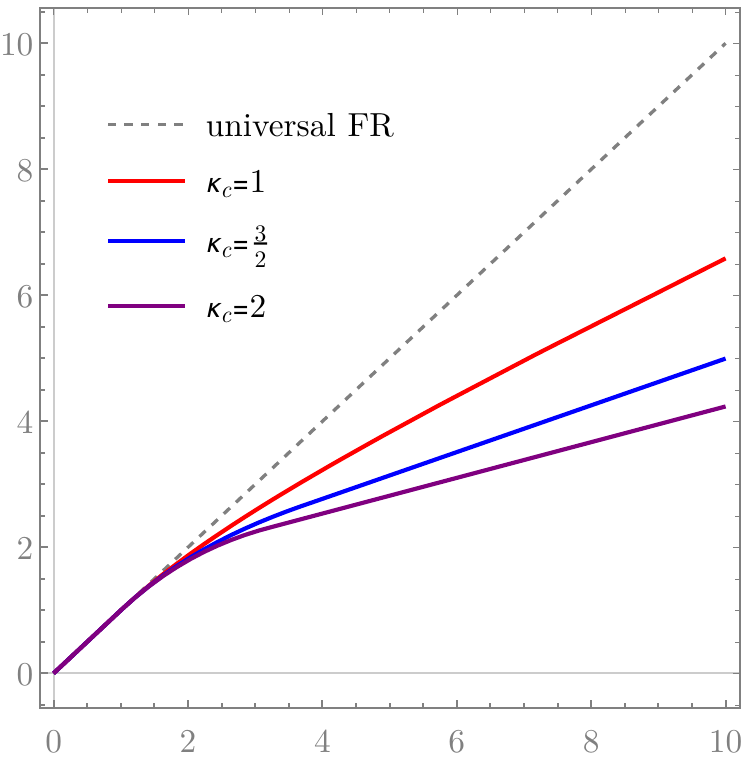}
\hskip 0.05cm
\includegraphics[width=5.6cm]{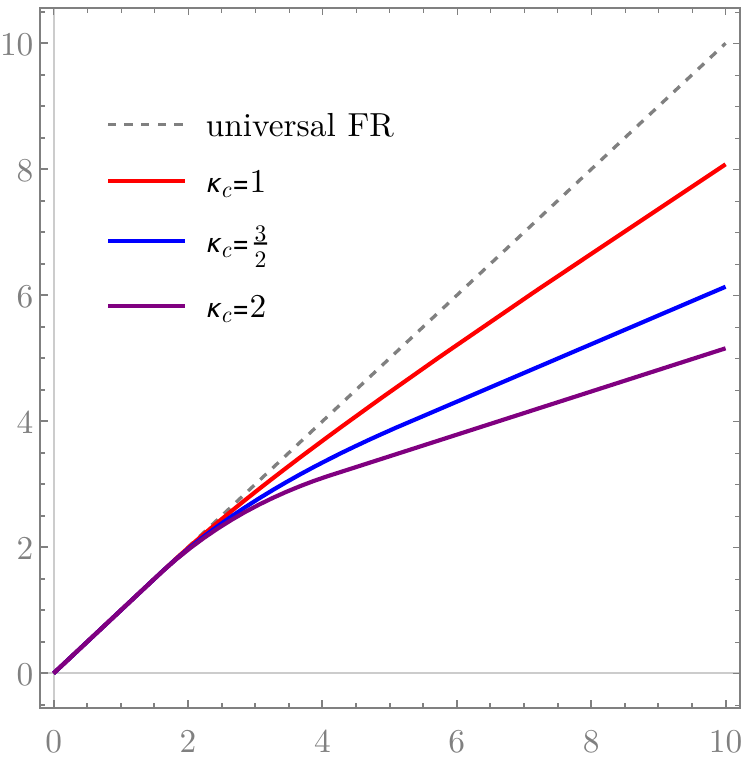}
\caption{The symmetry function $\Delta$ for the triangular network described in the text.
On the left: the case of the dissipated thermodynamic entropy $\fS_t$ in the steady state.
On the right: the case of the additive functional $\widetilde{\sigma}_t$.
The unit on both axes in both figures is the entropy production rate $\epr$.}
\label{Fig:LangevinSymmetry}
\end{figure}

\section{Reversible Smooth Dynamical Systems}
\label{Ex:RevDS}

\subsection{Setup}

Let $(\varphi^t)_{t\in\cL}$ be a group of homeomorphisms of  a Polish space $\Omega$,
with $\cL=\ZZ$ or $\RR$. In the second case, we also assume that $\varphi$ is
continuous as a map from $\cL\times\Omega$ to $\Omega$.  Assume that
$\mu\in\cP(\Omega)$ is such that $\mu_t\coloneq\mu\circ\varphi^{-t}$ is equivalent to
$\mu$ for all $t\in\cL$. For $\tau\in\cL_+\coloneq\cL\cap[0,\infty[$, set
$\Omega_\tau\coloneq C([0,\tau],\Omega)$ if $\cL=\RR$ and
$\Omega_\tau\coloneq\Omega^{\llbracket0,\tau\rrbracket}$ if $\cL=\ZZ$. Define the map
$\phi_\tau:\Omega\to\Omega_\tau$ by\index{smooth dynamical system}
$\phi_\tau(x)\coloneq(\varphi^t(x))_{t\in[0,\tau]\cap\cL}$ and the path space measure
$\PP_{\mu,\tau}$ on $\Omega_\tau$ by
$$
\PP_{\mu,\tau}\coloneq\mu\circ\phi_\tau^{-1}=\int \delta_{\phi_\tau(x)}\mu(\d x).
$$
Let $\theta$ be a time reversal, {\sl i.e.,} an involution of $\Omega$ such that
$\varphi^t\circ\theta=\theta\circ\varphi^{-t}$. Assume also that $\hat{\mu}\coloneq\mu\circ\theta$ is
equivalent to $\mu$. Define a $\ZZ_2$-action on the path space $\Omega_\tau$ by
$$
\bsx=(x(t))_{t\in[0,\tau]\cap\cL}\mapsto\Theta_\tau(\bsx)\coloneq(\theta x(\tau-t))_{t\in[0,\tau]\cap\cL}.
$$
Setting $\wP_{\mu,\tau}=\PP_{\mu,\tau}\circ\Theta_\tau$, we obtain a pair
$(\PP_{\mu,\tau},\widehat{\PP}_{\mu,\tau})_{\tau\in\cL_+}$
in involution. Its entropy production random variable is
$$
\sigma_\tau(\bsx)=\log\frac{\d\PP_{\mu,\tau}}{\d\wP_{\mu,\tau}}(\bsx)
=\log\frac{\d\mu}{\d\hat{\mu}_\tau}(x(0))
=\log\frac{\d\mu}{\d\hat{\mu}}(x(0))
-\log\frac{\d\mu_\tau}{\d\mu}(\theta x(0))
$$
where $\hat{\mu}_\tau=\mu_\tau\circ\theta = \hat{\mu}\circ \varphi^\tau$. Setting
$$
\kappa_\tau\coloneq\log\frac{\d\mu_\tau}{\d\mu},\qquad
\vartheta\coloneq\log\frac{\d\mu}{\d\hat{\mu}},
$$
one easily derives the cocycle relation
$$
\kappa_{\tau+t}=\kappa_\tau+\kappa_t\circ\varphi^{-\tau}.
$$
It follows that in the case $\cL=\ZZ$,
$$
\sigma_\tau(\bsx)=\vartheta(x(0))+\sum_{t=0}^{\tau-1}\widetilde{\sigma}(x(t)),
\qquad
\widetilde{\sigma}(x)\coloneq-\log\frac{\d\mu_1}{\d\mu}\circ\theta(x),
$$
and in the case $\cL=\RR$, under appropriate regularity assumptions, that
$$
\sigma_\tau(\bsx)=\vartheta(x(0))+\int_0^\tau\widetilde{\sigma}(x(t))\d t,
\qquad
\widetilde{\sigma}(x)\coloneq-\left.\frac{\d\ }{\d t}\frac{\d\mu_t}{\d\mu}\right|_{t=0}\circ\theta(x).
$$

Assuming that there exists a reference measure $\VV$ on $\Omega$ such that
$\VV\circ\theta=\theta$ and $\mu$ is equivalent to $\VV$, one can rewrite
the above formulas as
\beq
\sigma_\tau(\bsx)=
\log\rho(x(0))-\log\rho\circ\theta(x(\tau))
+\sum_{t=0}^{\tau-1}\sigma(x(t)),
\label{Eq:sigmadSD}
\eeq
with
$$
\sigma
\coloneq\log\frac{\d\VV\circ\varphi^{-1}}{\d\VV}\circ\varphi
\qquad
\rho\coloneq\frac{\d\mu}{\d\VV},
$$
in the case $\cL=\ZZ$, and
\beq
\sigma_\tau(\bsx)=
\log\rho(x(0))-\log\rho\circ\theta(x(\tau))
+\int_0^\tau\sigma(x(t))\d t
\label{Eq:sigmacSD}
\eeq
with
$$
\sigma\coloneq\left.\frac{\d\ }{\d t}\frac{\d\VV\circ\varphi^{-t}}{\d\VV}\right|_{t=0},
\qquad
\rho\coloneq\frac{\d\mu}{\d\VV},
$$
in the case $\cL=\RR$.

We note that the situation described in this example is similar to the one
discussed in~\cite{Jaksic2011}, with the only difference that we do not
require the measure $\mu$ to be $\theta$-invariant. The boundary terms
\beq
\log\rho(x(0))-\log\rho\circ\theta(x(\tau))
\label{Eq:BdryTerm}
\eeq
in~\eqref{Eq:sigmadSD} and~\eqref{Eq:sigmacSD} reflect the lack of time-reversal
invariance of $\mu$. Indeed, if $\mu\circ\theta=\mu$ then we can set $\VV=\mu$
and hence $\log\rho\equiv0$ in these formulas, and we recover the results
of~\cite{Jaksic2011} for the entropy production random variable.

\subsection{Phase Space Contraction Rate}
\label{Ex:DiffDS}

Continuing with the previous setting, let us further assume that $\Omega$ is equipped with the
structure of a smooth Riemannian manifold and that $\varphi$ is a diffeomorphism in the case
$\cL=\ZZ$, while the flow $t\mapsto\varphi^t$ is generated by a smooth vector
field $X$ in the case $\cL=\RR$. Denote by $\VV$ a Riemannian volume measure on
$\Omega$ and assume that $\VV\circ\theta=\VV$ and that $\mu$ and $\VV$ are
equivalent. Then, Relations~\eqref{Eq:sigmadSD} and~\eqref{Eq:sigmacSD}
hold with
\beq
\begin{split}
\rho=\frac{\d\mu}{\d\VV},\qquad
\sigma=\left\{
\begin{array}{ll}
-\log|\det(\varphi')|,&\text{if }\cL=\ZZ;\\[8pt]
-{\rm div}(X),&\text{if }\cL=\RR.
\end{array}
\right.
\end{split}
\label{equ:sigmadef}
\eeq
Thus, up to the boundary term~\eqref{Eq:BdryTerm}, the entropy production random variable
$\sigma_\tau$ is directly related to the \ndex{phase space contraction rate\/} $\sigma$.

\subsection{Transient Fluctuation Relations, Fluctuation Theorem and Steady States}
\label{ssect:Chaotic}

Physically interesting examples falling in the present context are Hamiltonian systems driven
out of equilibrium by external forces and coupled to Gaussian thermostats, see~\cite{Jaksic2011}
and references therein. It is precisely for such systems that {\sl Transient FRs\/} were first
discovered in the study of Second Law violating microstates~\cite{Evans1993,Evans1994}.
Inequality~\eqref{Eq:Jarzynski_inequality} can be written as\footnote{For simplicity we shall
only consider the case $\cL=\ZZ$, the other case being similar.}
$$
\int\left(\frac1\tau\sum_{t=0}^{\tau-1}\sigma\circ\varphi^t(x)\right)\mu(\d x)
\ge\frac1\tau\int\left(\log\rho\circ\theta\circ\varphi^\tau(x)-\log\rho(x)\right)\mu(\d x).
$$
Assuming that\footnote{Observe that $R=0$ whenever $\mu$ is a $\theta$-invariant
Riemannian volume.}
\beq
R=\sup_{x\in\Omega}|\log\rho(x)|<\infty,
\label{Eq:logrhobound}
\eeq
we derive the following asymptotic inequality for the time averaged phase space
contraction rate
$$
\liminf_{\tau\to\infty}\int\left(
\frac1\tau\sum_{t=0}^{\tau-1}\sigma\circ\varphi^t(x)\right)\mu(\d x)\geq 0.
$$
Here again Jarzynski's identity combined with Markov's inequality
(recall~\eqref{Eq:Jarzy_and_Markov}) yields a strengthening
$$
\mu\left(\left\{x\in\Omega\,\Bigg|\,
\frac1\tau\sum_{t=0}^{\tau-1}\sigma\circ\varphi^t(x)\le-s\right\}\right)
\le\e^{-s\tau+2R}.
$$

\remark The full phase space $\Omega=T^\ast\!M$ of a Hamiltonian
system  is unbounded and the natural reference measure there is the Liouville one.
This forces
$$
\inf_{x\in\Omega}\rho(x)=0,
$$
and incurs the failure of Assumption~\eqref{Eq:logrhobound}. The primary purpose of a thermostat is to
avoid this problem by constraining the motion to a compact submanifold
$\Sigma\subset\Omega$. Possible choices include surfaces of constant kinetic or
total energy. One expects that, in the thermodynamic limit, the statistical
properties of the motion do not depend too much on the choice of $\Sigma$, much
in the spirit of the equivalence of ensembles in equilibrium statistical
mechanics (see~\cite{Ruelle2000a} and references therein).

\medskip
The subtler FT and asymptotic FR were first
discussed in~\cite{Gallavotti1995b,Gallavotti1995c}, on the basis of the {\sl Chaotic
Hypothesis\/}. From our point of view, the two major outcomes of these works are
the formulation of the FT as a large deviation principle and its connection,
through the chaotic hypothesis, with the thermodynamic formalism of dynamical
systems (see also~\cite{Ruelle1999}).

Under the chaotic hypothesis of Gallavotti and Cohen, $(\Omega,\varphi,\mu)$ is a transitive
Anosov dynamical system: $\varphi$ is a diffeomorphism of class $C^2$ for which the compact
set $\Omega$ is hyperbolic, {\sl i.e.,} the tangent bundle of $\Omega$ has a $\varphi'$-invariant
decomposition
$$
T_x\Omega=E_x^{\,\rm u}\oplus E_x^{\,\rm s}
$$
into unstable and stable subbundles such that, for all $x\in\Omega$ and $t\in\cL_+$,
$$
\|(\varphi^{-t})'(x)|_{E_x^{\,\rm u}}\|\le C \lambda^t,\qquad
\|(\varphi^{t})'(x)|_{E_x^{\,\rm s}}\|\le C \lambda^t
$$
for some positive constants $C$ and $0<\lambda<1$. Transitivity further means that for any
nonempty open sets $U,V\subset\Omega$ and any $\tau\in\cL_+$ there is
$\cL_+\ni t>\tau$ such that $\varphi^{-t}(U)\cap V\not=\emptyset$.

For $\nu\in\cP_\varphi(\Omega)$, denote by $h_\varphi(\nu)$ the \ndex{Kolmogorov--Sinai entropy}
of $\nu$ w.r.t.\;$\varphi$. The \ndex{topological pressure} of $F\in C(\Omega)$ w.r.t.\;$\varphi$ is
defined by the variational expression
\beq
\fp_\varphi(F)\coloneq\sup_{\nu\in\cP_\varphi(\Omega)}\left(\langle F,\nu\rangle+h_\varphi(\nu)\right)
\label{Eq:CHvar}
\eeq
(see, e.g., \cite{Walters1982}).

Under the Chaotic Hypothesis, the Rényi entropy
$e_t(\alpha)=\Ent_\alpha(\PP_{\mu,\tau}|\wP_{\mu,\tau})$ has a large-time limit
$$
e(\alpha)=\lim_{t\to\infty}\frac1t e_t(\alpha)
$$
which defines a real-analytic function  on $\RR$. More precisely, setting
$$
v(x)\coloneq-\log \left|\det\left(\varphi'(x)|_{E_x^{\,\rm u}}\right)\right|,
$$
one has\footnote{See~\cite{Maes2003b} and~\cite[Proposition~11.2, Theorem~10.10]{Jaksic2011}.}
$$
e(\alpha)=\fp_\varphi(v-\alpha\sigma).
$$
By the Gärtner--Ellis theorem, the FT holds for the family
$(\PP_{\mu,\tau},\wP_{\mu,\tau})_{\tau\in\cL_+}$ with scale $r_t=t$ and
rate\index{FT!for chaotic dynamics}
\beq
I(s)=\sup_{\alpha\in\RR}\left(\alpha s-e(-\alpha)\right).
\label{Eq:CHrate}
\eeq
By Assumption~\eqref{Eq:logrhobound}, the additive functional
$$
\widetilde{\sigma}_\tau(\bsx)\coloneq\sum_{t=0}^{\tau-1}\sigma\circ\varphi^t(x(0))
$$
is physically equivalent to the entropy production $\sigma_\tau$.  It therefore also
satisfies the FT w.r.t.\;the family $(\PP_{\mu,\tau},\wP_{\mu,\tau})_{\tau\in\cL_+}$
with the same scale and rate.
Moreover, there is a unique SRB measure $\mu_+$, an ergodic element of
$\cP_\varphi(\Omega)$ characterized by the fact that
$$
\wlim_{\tau\to\infty}\frac1\tau\sum_{t=0}^{\tau-1}\delta_{\varphi^t(x)}=\mu_+
$$
for $\mu$-a.e.\;$x\in\Omega$. Thus, $\mu_+$ is a natural candidate for the
{\sl nonequilibrium steady state} (NESS)\index{NESS} of the system. It is
also the unique maximizer in~\eqref{Eq:CHvar} for the potential $v$,
{\sl i.e.,}
$$
\fp_\varphi(v)=\langle v,\mu_+\rangle+h_\varphi(\mu_+)=e(0)=0.
$$
In the language of topological dynamics, $\mu_+$ is the unique equilibrium state of
$(\Omega,\varphi)$ for the potential $v$. This state is typically quite different from $\mu$.
In fact, provided the entropy production rate is non-vanishing,
$$
\epr\coloneq\lim_{\tau\to\infty}\frac1\tau\sum_{t=0}^{\tau-1}\langle\sigma,\mu_t\rangle=\langle\sigma,\mu_+\rangle>0,
$$
$\mu$ and $\mu_+$ are mutually singular. Nevertheless, one can show that the
cumulant-generating function
$$
e_{+,\tau}(\alpha)=\log\langle\e^{-\alpha\widetilde{\sigma}_\tau},\mu_+\rangle
$$
satisfies
$$
\lim_{\tau\to\infty}\frac1\tau e_{+,\tau}(\alpha)=e(\alpha)
$$
for all $\alpha\in\RR$~\cite[Proposition~11.4]{Jaksic2011}. Thus, the additive functional
$\widetilde{\sigma}_\tau$ also
satisfies the FT w.r.t.\;the family $(\PP_{\mu_+,\tau},\wP_{\mu_+,\tau})_{\tau\in\cL_+}$
with the rate~\eqref{Eq:CHrate}.

\addcontentsline{toc}{section}{Appendix}
\hypertarget{AppA}{\section*{Appendix: Level-2 FR for Langevin Dynamics}}

Under the standing assumptions of Subsection~\ref{Ex:Langevin}, the following holds.

\bel
\label{Lem:Lyapunov}
Fix $\ell>2$ and $\delta\in]0,1/2[$. There exists $\varepsilon_0>0$ such that, for any
$\varepsilon\in]0,\varepsilon_0[$ there exists $c>0$ such that, for
any function $h\in C^2(\RR)$ that is bounded below and satisfies $0\le h'\le1$, the
functions
$$
\Psi(x)\coloneq\frac12 S(x)+\varepsilon p\cdot\Gamma^{-1}\nabla V(q)+h\circ\phi(x)
$$
and
$$
\widehat{\Psi}(x)\coloneq\frac12 S(x)+\varepsilon p\cdot\Gamma^{-1}\nabla V(q)-\frac1{1+\delta}\phi(x)
$$
are bounded below and satisfy
$$
G\coloneq-(L\Psi+\gamma(\Psi))\ge -\frac1c+c|\nabla S|^2
+(h'(1-\ell h')-h'')\circ\phi\,\gamma(\phi)
$$
and
$$
\widehat G\coloneq-(L\widehat{\Psi}+\gamma(\widehat{\Psi}))\ge-\frac1c+ c|\nabla S|^2.
$$
\eel

\proof We first show that $\Psi$ and $\widehat{\Psi}$ are bounded below.
Minimizing $\Psi$ over $p$, we get
\beq
\Psi\ge\frac12\left(\widetilde{V}
-\frac12|M\nabla\widetilde{V}|^2\right)+h\circ\phi,
\label{Eq:PSIlower}
\eeq
where
$$
\widetilde{V}\coloneq\sum_i\frac{V_i}{T_i},\qquad
M\coloneq2\varepsilon\Gamma^{-1}T^{3/2}.
$$
Setting $u\coloneq M^2\nabla\widetilde{V}(q)$ in the $2^{\rm nd}$-order Taylor formula
$$
\widetilde{V}(q-u)=\widetilde{V}(q)-u\cdot\nabla\widetilde{V}(q)
+\int_0^1u\cdot(D^2\widetilde{V})(q-su)u(1-s)\d s
$$
we get, with $C\coloneq\max_i\sup_q|T_i^{-1}V_i''(q)|$,
$$
\widetilde{V}(q)-|M\nabla\widetilde{V}(q)|^2
\ge\widetilde{V}(q-M^2\nabla\widetilde{V}(q))-\frac12C|M^2\nabla\widetilde{V}(q)|^2,
$$
which we rewrite as
$$
\widetilde{V}(q)-M\nabla\widetilde{V}(q)\cdot\left(
I-\frac12CM^2
\right)M\nabla\widetilde{V}(q)
\ge\widetilde{V}(q-M^2\nabla\widetilde{V}(q))\ge0.
$$
Now for small enough $\varepsilon$ one has $CM^2\le I$ and hence
$$
\widetilde{V}
-\frac12|M\nabla\widetilde{V}|^2\ge0.
$$
From~\eqref{Eq:PSIlower} we deduce that since $h$ is bounded below, so is $\Psi$.
To deal with $\widehat{\Psi}$ we note that, by~\cite[Proposition~3.5]{Eckmann1999b},
the function $\e^{S/2}\rho_\st=\e^{-S/2+\phi\circ\theta}$ belongs to the Schwartz class. It
follows that
\beq
\phi\le\widehat{C}+\frac12S
\label{Eq:AppAphibound}
\eeq
for some constant $\widehat C$ and hence
$$
\widehat{\Psi}\ge\frac{\delta}{1+\delta}\left(
\frac12S+\varepsilon(1+\delta^{-1})p\cdot\Gamma^{-1}\nabla V(q))\right)
-\frac{\widehat C}{1+\delta}.
$$
Repeating the previous argument gives that, for any given $\delta>0$,
$\widehat{\Psi}$ is bounded below provided $\varepsilon$ is small enough.

We now turn to the lower bounds for $G$. An explicit calculation gives
\begin{align*}
G(x)=&p\cdot\left(
\frac14T^{-1}\Gamma T^{-1}-\varepsilon\Gamma^{-1/2}(D^2V)(q)\Gamma^{-1/2}\right)p
+\varepsilon(1-\varepsilon)|\Gamma^{-1/2}(\nabla V)(q)|^2\\[4pt]
&+(h'(1-h')-h'')\circ\phi(x)\gamma(\phi)(x)\\[4pt]
&-\left(\Gamma^{1/2}T^{-1}p+2\varepsilon\Gamma^{-1/2}(\nabla V)(q)\right)
\cdot\left(\Gamma^{1/2}\nabla_p h\circ\phi\right)(x)
+\mathcal{O}(|(\nabla S)(x)|),
\end{align*}
where the remainder $\mathcal{O}(|\nabla S|)$ is uniform in $h$ for $0\le h'\le1$.
Invoking the Cauchy--Schwarz inequality on the first term of the last line in the
previous expression further yields, for arbitrary positive $\kappa$ and $\lambda$,
\begin{align*}
-\left(\Gamma^{1/2}T^{-1}p\right.&+\left.2\varepsilon\Gamma^{-1/2}(\nabla V)(q)\right)
\cdot\left(\Gamma^{1/2}\nabla_p h\circ\phi\right)(x)\\
\ge&-\frac1{2\kappa}|\Gamma^{1/2}T^{-1}p+2\varepsilon\Gamma^{-1/2}(\nabla V)(q)|^2
-\frac{\kappa}{2}(h'\circ\phi(x))^2\gamma(\phi)(x)\\
\ge&-\frac1{2\kappa}\left(1+\frac1\lambda\right)|\Gamma^{1/2}T^{-1}p|^2
-\frac{2\varepsilon^2}{\kappa}\left(1+\lambda\right)|\Gamma^{-1/2}(\nabla V)(q)|^2\\
&-\frac{\kappa}{2}(h'\circ\phi(x))^2\gamma(\phi)(x),
\end{align*}
so that
\begin{align*}
G(x)\ge&p\cdot\left(
\left(\frac14-\frac1{2\kappa}\left(1+\frac1\lambda\right)\right)
T^{-1}\Gamma T^{-1}-\varepsilon\Gamma^{-1/2}(D^2V)(q)\Gamma^{-1/2}
\right)p\\
&+\varepsilon\left(1-\varepsilon\left(1+\frac{2\varepsilon}{\kappa}\left( 1+\lambda\right)
\right)\right)|\Gamma^{-1/2}\nabla V|^2\\
&+\left(h'\left(1-\left(1+\frac{\kappa}{2}\right)h'\right)
-h''\right)\circ\phi(x)\gamma(\phi)(x)
+\mathcal{O}(|(\nabla S)(x)|).
\end{align*}
It follows that for $\kappa>2$, $\lambda>\kappa/(\kappa-2)$ and
small enough $\varepsilon$ there exists $c>0$ such that
$$
G(x)\ge c|(\nabla S)(x)|^2
+\left(h'\left(1-\left(1+\frac{\kappa}{2}\right)h'\right)
-h''\right)\circ\phi(x)\gamma(\phi)(x)
+\mathcal{O}(|(\nabla S)(x)|),
$$
which yields the claimed lower bound of $G$. The proof of the lower bound of
$\widehat{G}$ is similar, and we omit it.\hfill\qed

\bigskip
Setting
\[
\fF_Q[\varphi]\coloneq\langle L\varphi-\gamma(\varphi),Q\rangle,\qquad
\widehat{\fF}_Q[\varphi]\coloneq\langle\widehat{L}\varphi-\gamma(\varphi),Q\rangle,
\]
the Donsker--Varadhan level-2 rate functions for the direct and reversed
processes are given by\footnote{Since the functionals $\fF_Q$ and $\widehat{\fF}_Q$
are invariant under addition of a constant, the condition $\varphi\le0$ can be
replaced by $\varphi$ being bounded above. This will be used without further comment
in the sequel.}
\[
\bsI(Q)=\sup_{\varphi\in C^2(\fX),\varphi\le0}\fF_Q[\varphi],\qquad
\hat{\bsI}(Q)=\sup_{\varphi\in C^2(\fX),\varphi\le0}\widehat{\fF}_Q[\varphi].
\]
The following result gives the important properties of these rate functions.
\bel\label{Lem:L2FR}
 For any $Q\in\cP(\fX)$ the following hold:
\ben
\item If either $\bsI(Q)<\infty$ or $\hat\bsI(Q)<\infty$, then
$|\nabla S|\in L^2(\fX,\d Q)$.
\item $\hat{\bsI}(Q)=\bsI(Q\circ\theta)$.
\item If $\gamma(\phi)\in L^1(\fX,\d Q)$, then
\[
\bsI(Q\circ\theta)=\bsI(Q)+\langle\sigma,Q\rangle.
\]
\item If $\bsI(Q)<\infty$, then $\gamma(\phi)\in L^1(\fX,\d Q)$.
\een
\eel

\proof{\bf(1)} is a direct consequence of~\cite[Corollary~2.2]{Wu2001} and the previous
lemma.

\part{(2)} follows by contraction from the level-3 identity $\hat\II(\QQ)=\II(\widehat{\QQ})$.

\part{(3)} Recall that the generator of the time-reversed stationary process is given by
$\widehat{L}=L+2\gamma(\phi,\,\argdot\,)$. It follows that
\begin{align*}
\widehat{L}\varphi-\gamma(\varphi)&=L\varphi+2\gamma(\phi,\varphi)-\gamma(\varphi)
=L\varphi-\gamma(\varphi-\phi)+\gamma(\phi)\\
&=L(\varphi-\phi)-\gamma(\varphi-\phi)+L\phi+\gamma(\phi)
=L(\varphi-\phi)-\gamma(\varphi-\phi)+\sigma,
\end{align*}
and hence
$$
\widehat{\fF}_Q[\varphi]=\fF_Q[\varphi-\phi]+\langle\sigma,Q\rangle.
$$
Thus, by~(2), it is enough to show that
$$
\sup_{\varphi\in C_{\rm b}^2(\fX)}\fF_Q[\varphi-\phi]
=\sup_{\varphi\in C_{\rm b}^2(\fX)}\fF_Q[\varphi].
$$
Let $\psi\in C^\infty(\RR)$ be such that
$$
\psi(u)=\left\{\begin{array}{cl}
-1&\text{ for }u\le-3/2;\\
u&\text{ for }|u|\le1/2;\\
1&\text{ for }u\ge3/2,
\end{array}
\right.
$$
and set $\psi_r(u)=r\psi(u/r)$, $\chi_r(u)=u-\psi_r(u)$ for $r>0$. Then
$\phi_r=\psi_r\circ\phi\in C^2_{\rm b}(\fX)$,
and it suffices to show that $\fF_Q[\varphi-\phi_r]\to\fF_Q[\varphi-\phi]$ as
$r\to\infty$ for any $\varphi\in C^2_{\rm b}(\fX)$. From
$$
\fF_Q[\varphi-\phi_r]-\fF_Q[\varphi-\phi]
=\langle L(\chi_r\circ\phi)+\gamma(\chi_r\circ\phi)
+2\gamma(\phi_r-\varphi,\chi_r\circ\phi),Q\rangle,
$$
an explicit calculation, using the chain rule~\eqref{eq:ChainRule}, yields that
$$
\left|\fF_Q[\varphi-\phi_r]-\fF_Q[\varphi-\phi]\right|
\le4\int_{|\phi|>r/2}\left(|\sigma|+\gamma(\varphi)+\gamma(\phi)
\right)\d Q,
$$
and the result follows.

\part{(4)} For $R>0$, define $h_R\in C^2(\RR)$ by
$$
h_R(u)\coloneq\left\{\begin{array}{ll}
-3R/2&\text{ for }u<-2R;\\
-3R/2-4u-6u^2/R-3u^3/R^2-u^4/2R^3&\text{ for }-2R\le u\le-R;\\
u&\text{ for }u>-R.
\end{array}
\right.
$$
One easily shows that
\beq
-\frac{3R}{2}\le h_R(u),\qquad
0\le h_R'(u)\le 1,\qquad
0\le h_R''(u)\le\frac3{2R}.
\label{Eq:hRBounds}
\eeq
By Lemma~\ref{Lem:Lyapunov}, for sufficiently small $\varepsilon$ there is a positive
constant $c$ such that the $C^2$-function defined by
$\Psi(x)=S(x)/2+\varepsilon p\cdot\Gamma^{-1}\nabla V(q)+(h_R\circ\phi(x))/4$ is
bounded below and satisfies
$$
-(L\Psi+\gamma(\Psi))\ge-\frac1c+c|\nabla S|^2
+\frac14\left(h_R'\left(1-\frac34h_R'\right)-h_R''\right)\circ\phi\;\gamma(\phi)
$$
for all $R>0$. It follows that
$$
-\frac1c+c\int|\nabla S|^2\d Q
+\frac14\int\left(h_R'\left(1-\frac34h_R'\right)-h_R''\right)\circ\phi\;
\gamma(\phi)\d Q
\le\int(-L\Psi-\gamma(\Psi))\d Q\le\bI(Q).
$$
Thus, if $\bI(Q)<\infty$, then there is $\alpha<\infty$ such that
$$
\int\left(\frac14h_R'-h_R''\right)\circ\phi\,
\gamma(\phi)\d Q
\le
\int\left(h_R'\left(1-\frac34h_R'\right)-h_R''\right)\circ\phi\,
\gamma(\phi)\d Q\le\alpha
$$
for all $R>0$. Using the explicit form of $h_R$ and the
inequalities~\eqref{Eq:hRBounds}, one easily shows that
$$
\frac14h_R'(u)-h_R''(u)\ge f_R(u)\coloneq\left\{\begin{array}{ll}
0&\text{ for }u<-2R;\\
-\frac{3}{2R}&\text{ for }-2R\le u<-R;\\
\frac14&\text{ for }u\ge-R,
\end{array}
\right.
$$
from which we deduce that
$$
\langle f_R\circ\phi\,\gamma(\phi),Q\rangle\le\alpha
$$
for all $R>0$. Now we set
$$
f(u)\coloneq\sum_{n=0}^\infty 2^{-n}f_{2^nR}(u)
=\sum_{2^n\ge-u/R}2^{-n-2}-\sum_{-u/2R\le 2^n<-u/R}\frac3R 2^{-2n-1},
$$
so that
$$
\langle f\circ\phi\,\gamma(\phi),Q\rangle\le\frac\alpha2.
$$
Observing that $f(u)=\frac12$ for $u\ge-R$, while
$$
f(u)\ge\sum_{n\ge\log_2(|u|/R)}2^{-n-2}-\sum_{n\ge\log_2(|u|/R)-1}\frac3R 2^{-2n-1}
\ge \frac{R}{4|u|}-\frac{8R}{|u|^2}
=\frac{R}{4|u|}\left(1-\frac{32}{|u|}\right)
$$
is non-negative for $u<-R\le-32$, we conclude that for $R\ge32$ one has
$$
\frac12\int_{\phi\ge-R}\gamma(\phi)\d Q\le\langle f\circ\phi\,\gamma(\phi),Q\rangle\le\frac\alpha2.
$$
Letting $R\to\infty$ yields the result.\hfill\qed

\printbibliography[heading=bibintoc,title={References}]


\chapter{Level-2~LDP~and~Fluctuation~Theorem for~Chaotic~Maps}
\label{chap:Chaotic maps and FR}

\abstract{We restrict our attention to  chaotic dynamical systems on compact
metric spaces and derive FT/FR  for their periodic orbit ensembles. The
discussion is based on Ruelle--Walters thermodynamic formalism, which is reviewed
in some detail. This chapter is motivated by~\cite{Maes2003b}.}

\abstract*{We restrict our attention to  chaotic dynamical systems on compact
metric spaces and derive FT/FR  for their periodic orbit ensembles. The
discussion is based on Ruelle-Walters thermodynamic formalism, which is reviewed
in some detail. This chapter is motivated by~\cite{Maes2003b}.}

\vskip 1cm

\section{A Class of Continuous Dynamical Systems}
\label{sec:CDS}

Let $(M,d)$ be a compact metric space with Borel $\sigma$-algebra~$\fM$. We
recall that~$C(M)$ denotes the Banach space of
real-valued continuous functions on $M$ with the
sup-norm $\|\,\argdot\,\|$. The set $\cP(M)$ of Borel probability measures on~$M$
is equipped with the topology of weak convergence (denoted $\rightharpoonup$)
and the corresponding Borel $\sigma$-algebra. This makes
$\cP(M)$ a convex compact metrizable subspace of the dual of $C(M)$ endowed with
its weak-$\ast$ topology. Given $v\in C(M)$ and $\QQ\in\cP(M)$, we denote by $\bra
v,\QQ\ket$ the integral of~$v$ with respect to~$\QQ$.
The Dirac mass at $x\in M$ is denoted by
$\delta_x\in\cP(M)$.

In the following, we shall always assume:

\medskip\noindent
\hypertarget{hyp.C}{\textbf{(C)}}\label{C}
{\sl $\varphi:M\to M$ is a continuous map.}

\medskip\noindent
On occasions, we shall strengthen the above standing assumption to:

\medskip\noindent
\hypertarget{hyp.H}{\textbf{(H)}}\label{H}
{\sl $\varphi:M\to M$ is a homeomorphism,}

\medskip\noindent
and we shall always explicitly mention when Condition~\hH is assumed. We set
$\cT=\NN$ in case~\hC and $\cT=\ZZ$ in case~\hH, so that, in
both cases, the orbit of $x\in M$ is $(\varphi^t(x))_{t\in\cT}$. The set of
$\varphi$-invariant elements of $\cP(M)$ is denoted by~$\cP_\varphi(M)$. It is a
closed convex subset of $\cP(M)$ and a Choquet simplex whose extreme points are
the elements of $\cE_\varphi(M)$, the set of $\varphi$-ergodic measures. We
denote by $h_\mathrm{Top}(\varphi)$ the \ndex{topological entropy} of $\varphi$ and by
$h_\varphi(\PP)$ the \ndex{Kolmogorov--Sinai entropy} of $\PP\in\cP_\varphi(M)$
w.r.t.\;$\varphi$. The following assumptions will play a central role in our
discussion.

\medskip\noindent
\hypertarget{hyp.USCE}{\textbf{(USCE)}}\label{usce}
{\sl $h_\mathrm{Top}(\varphi)<\infty$, and the map
$\cP_\varphi(M)\ni\PP\mapsto h_\varphi(\PP)$ is upper semicontinuous.}

\medskip\noindent
\hypertarget{hyp.ED}{\textbf{(ED)}}\label{ed}
{\sl The set $\cE_\varphi(M)$ is entropy-dense in $\cP_\varphi(M)$, {\sl i.e.,} for any
$\PP\in\cP_\varphi(M)$ there exists a sequence $(\PP_n)_{n\in\NN}\subset\cE_\varphi(M)$
such that
\[
\PP_n\rightharpoonup\PP\quad\text{and}\quad h_\varphi(\PP_n)\to h_\varphi(\PP)
\]
as $n\to\infty$}

\medskip\noindent
The \ndex{topological pressure} of $v\in C(M)$ w.r.t.\,$\varphi$ is
\beq
\fp_\varphi(v)\coloneq\sup_{\PP\in\cP_\varphi(M)}\left(\bra v,\PP\ket+h_\varphi(\PP)\right).
\label{PressureDef}
\eeq
Moreover, $\PP\in\cP_\varphi(M)$ is called \ndex{equilibrium measure} {\sl for $v$} whenever it
achieves the supremum in this variational formula, {\sl i.e.,}
\beq
\bra v,\PP\ket+h_\varphi(\PP)=\fp_\varphi(v).
\label{EquiDef}
\eeq

\begin{proposition}\label{proppress}
Under Condition~\hC the following hold:
\ben
\item $\fp_\varphi(0)=h_\mathrm{Top}(\varphi)$.
\item Either $\fp_\varphi(v)$ is infinite for all $v\in C(M)$, or
$h_\mathrm{Top}(\varphi)$ is finite and so is $\fp_\varphi(v)$ for any
$v\in C(M)$. Moreover, in the latter case, the map
$C(M)\ni v\mapsto\fp_\varphi(v)\in\RR$ is convex and satisfies
\[
|\fp_\varphi(v)-\fp_\varphi(u)|\le\|v-u\|
\]
for any $u,v\in C(M)$.
\een

\medskip\noindent
If, in addition, Condition~\hUSCE holds, then one further has:

\ben\setcounter{enumi}{2}
\item The dual variational expression
$$
h_\varphi(\PP)=\inf_{v\in C(M)}\left(\fp_\varphi(v)-\bra v,\PP\ket\right)
$$
holds for all\/ $\PP\in\cP_\varphi(M)$. \item For any $v\in C(M)$, the set of
equilibrium measures for $v$ is a non-empty compact Choquet simplex. It is also
a face of $\cP_\varphi(M)$ whose extreme points are ergodic measures. \item For
any\/ $\PP\in\cE_\varphi(M)$ there exists $v\in C(M)$ such that $\PP$ is the
unique equilibrium measure for $v$.
\een
\end{proposition}
We refer the reader to~\cite{Walters1982} for more details and basic properties
of the pressure and equilibrium states. Parts~(1)--(3) of
Proposition~\ref{proppress} can be found there. For Parts~(4)--(5)
see~\cite[Theorem~2]{Ruelle2004} and~\cite[Theorem~1]{Phelps2002}.

\medskip\noindent
We finish this section with a brief discussion of our main assumptions~\hUSCE and~\hED.

Given $\varepsilon>0$ and an integer $t\ge1$ we denote by
\[
B_t(x,\varepsilon)\coloneq\{y\in M\mid d(\varphi^s(x),\varphi^s(y))<\varepsilon\text{ for }0\le s<t\}
\]
the Bowen ball of radius $\varepsilon$ centered at $x\in M$.
 A map $\varphi$ satisfying Condition~\hC (respectively Condition~\hH) is
called {\sl forward expansive} (respectively {\sl expansive}) whenever there exists $r>0$
such that,
\beq
\Gamma(x,r)\coloneq\bigcap_{t\in\cT}\overline{B_t(x,r)}=\{x\}
\label{equ:GammaDef}
\eeq
for any $x\in M$. The number $r$ is called the expansiveness constant of $\varphi$. It depends on
the metric $d$, but \ndex{expansiveness} itself only depends on the induced topology of $M$.
Expansiveness is sufficient to ensure Condition~\hUSCE, see~\cite[Corollary~7.11.1]{Walters1982}
and~\cite[Section~2.4]{Barreira2011}.

For $I=\llbracket a,b\rrbracket\subset\cT$, we denote orbit segments by
\[
\varphi^I(x)\coloneq\left(\varphi^t(x)\right)_{t\in I},
\]
and call \ndex{specification} a finite family of such segments
\[
\xi=\left(\varphi^{I_i}(x_i)\right)_{i\in\llbracket 1,n\rrbracket}.
\]
The integers $n$ and
$L(\xi)\coloneq\max\{|t-t'|:t,t'\in\cup_{i\in\llbracket 1,n\rrbracket}I_i\}$
are called the {\sl rank} and the {\sl length} of $\xi$ respectively.
The specification $\xi$ is $N${\sl-separated} whenever
$$d(I_i,I_j)=\min_{t_i\in I_i,t_j\in I_j}|t_i-t_j|\ge N$$
for all distinct $i,j\in\llbracket 1,n\rrbracket$. It is
$\varepsilon${\sl-shadowed} by $x\in M$ whenever
\[
\max_{i\in\llbracket 1,n\rrbracket}
\max_{t\in I_i}d(\varphi^t(x),\varphi^t(x_i))<\varepsilon.
\]

\medskip\noindent
\hypertarget{hyp.S}{\textbf{(S)}}\label{s}
{\sl $\varphi$ has the {\rm Specification Property}\footnote{Many
variants of the specification property appear in the literature;
see~\cite{Kwietniak2016} for a review.}if, for any $\delta>0$, there is
$N(\delta)\ge1$ such that any $N(\delta)$-separated specification
$\xi=\left(\varphi^{I_i}(x_i)\right)_{i\in\llbracket 1,n\rrbracket}$
is $\delta$-shadowed by some point $x\in M$.}

\medskip
Property~\hS implies Condition~\hED,
see~\cite[Theorem~B]{Eizenberg1994} whose proof, given for case~\hH, extends
without change to case~\hC. For a detailed discussion of Property~\hED, we
refer the reader to~\cite{Comman2017}.

\section{A Large Deviation Principle for Empirical Measures}
\label{sec:LDP for EM}

To any $x\in M$ we associate the sequence $(\mu_t^x)_{t\in\NN}$
of empirical measures, defined by
\beq
\mu_t^x\coloneq\frac1t\sum_{0\le s< t}\delta_{\varphi^s(x)},
\label{equ:mutDef}
\eeq
so that
\[
\bra v,\mu_t^x\ket=\frac1t S_tv(x)
\]
where
\[
S_tv\coloneq v+v\circ\varphi+\cdots+v\circ\varphi^{t-1}.
\]
Given $\PP\in\cP_\varphi(M)$, we shall consider $\mu_t^\argdot$ as a
$\cP(M)$-valued random variable. The main result in this section is the  LDP for
the sequence $(\mu_t^\argdot)_{t\in\NN}$ w.r.t\;appropriate sequences of
measures.

\begin{definition}
Given $v\in C(M)$, we say that $(\PP_t)_{t\in\NN}\subset\cP(M)$ is a
$(v,\varphi)$-sequence whenever
\[
\lim_{t\to\infty}\frac1t\log\bra\e^{S_tu},\PP_t\ket=\fp_\varphi(v+u)-\fp_\varphi(v)
\]
holds for all $u\in C(M)$.

\end{definition}

The following is a direct application of~\cite[Theorem 5.2]{Comman2009}.

\begin{theorem}\label{thm:LDP-2}
Assume~\hC, \hUSCE, \hED and let $(\PP_t)_{t\in\NN}$ be a $(v,\varphi)$-sequence
for some $v\in C(M)$. Then, w.r.t.\;this sequence, the empirical measures
$(\mu_t^\argdot)_{t\in\NN}$ satisfy the LDP
\[
-\inf_{\QQ\in\mathring{\Gamma}}\II(\QQ)
\le\liminf_{t\to\infty}\frac1t\log\PP_t\{\mu_t^\argdot\in\Gamma\}
\le\limsup_{t\to\infty}\frac1t\log\PP_t\{\mu_t^\argdot\in\Gamma\}
\le-\inf_{\QQ\in\bar\Gamma}\II(\QQ)
\]
for any Borel set $\Gamma\subset\cP(M)$.\footnote{Recall that  $\mathring{\Gamma}/\bar\Gamma$ denotes
the interior/closure of\/ $\Gamma$.} The rate  $\II:\cP(M)\to[0,\infty]$ is given by
$$
\II(\QQ)\coloneq\left\{\bear{ll}
\fp_\varphi(v)-\bra v,\QQ\ket-h_\varphi(\QQ)&\text{ if }\,\QQ\in\cP_\varphi(M);\\[6pt]
+\infty&\text{otherwise.}
\ear\right.
$$
\end{theorem}

\medskip
We note that since the entropy function $\QQ\mapsto h_\varphi(\QQ)$ is affine on
$\cP_\varphi(M)$,\footnote{See~\cite[Theorem~8.1]{Walters1982}.} so is the rate
$\II$. The inequality $\II(\QQ)\ge0$ follows from
Proposition~\ref{proppress}~(3). Moreover, comparing with~\eqref{EquiDef} we
conclude that $\II(\QQ)=0$ iff $\QQ$ is an equilibrium measure for $v$. Finally,
we observe that $\II$ is the restriction to $\cP(M)$ of the Fenchel--Legendre
transform of $u\mapsto\fp_\varphi(v+u)-\fp_\varphi(v)$.

By contraction, one immediately obtains the following

\begin{corollary}\label{cor:LDP-1}
Under the assumptions of Theorem~\ref{thm:LDP-2}, and for any $u\in C(M)$, one has the LDP
\[
     -\inf_{r\in\mathring{B}}I(r)
\le\liminf_{t\to\infty}\frac1t\log\PP_t\{t^{-1}S_tu\in B\}
\le\limsup_{t\to\infty}\frac1t\log\PP_t\{t^{-1}S_tu\in B\}
\le-\inf_{r\in\bar B}I(r)
\]
for any Borel set $B\subset\RR$. The good convex rate  $I:\RR\to[0,\infty]$ is given by
\[
I(r)=\inf\left\{\II(\QQ)\mid \QQ\in\cP_\varphi(M),\bra u,\QQ\ket=r\right\}.
\]
 \end{corollary}

We finish this section with some general remarks on $(v,\varphi)$-sequences.

\remark[1] If $(\PP_t)_{t\in\NN}$ is a $(v,\varphi)$-sequence and
$v'\in C(M)$, then one easily checks that the sequence
$(\PP^{(v')}_t)_{t\in\NN}$ defined by
\[
\frac{\d\PP^{(v')}_t}{\d\PP_t}=\frac{\e^{S_tv'}}{\bra\e^{S_tv'},\PP_t\ket}
\]
is a $(v+v',\varphi)$-sequence.

\remark[2] If $(\PP_t)_{t\in\NN}$ is a $(v,\varphi)$-sequence and
$(v_t)_{t\in\NN} \subset C(M)$ is such that
\[
\|v_t-S_tv\|=o(t)
\]
as $t\to\infty$, then it is again straightforward to check that the sequence
$(\widetilde{\PP}_t)_{t\in\NN}$ defined by
\[
\frac{\d\widetilde{\PP}_t}{\d\PP_t}=\frac{\e^{v_t-S_tv}}{\bra\e^{v_t-S_tv},\PP_t\ket}
\]
is also a $(v,\varphi)$-sequence. In view of the equivalence of asymptotically
additive sequences with additive ones~\cite{Cuneo2020}, this shows that
Theorem~\ref{thm:LDP-2} and Corollary~\ref{cor:LDP-1} also apply to asymptotically
additive sequences.

\section{Periodic Orbit Ensemble}

In order to exploit the  results of the preceding section, we need some concrete
scheme to construct $(v,\varphi)$-sequences. In this section we shall invoke the
following elementary result:

\begin{lemma}
If the sequence $(E_t)_{t\in\NN}$ of subsets of $M$ is such that,
for all $u\in C(M)$,
\beq
\lim_{t\to\infty}\frac1t\log\sum_{x\in E_t}\e^{S_tu(x)}=\fp_\varphi(u),
\label{pseqdef}
\eeq
then, for any $v\in C(M)$, setting
\beq
\PP_t\coloneq Z_t^{-1}\sum_{x\in E_t}\e^{S_tv(x)}\delta_x,
\qquad
Z_t=Z_t(v)\coloneq\sum_{x\in E_t}\e^{S_tv(x)},
\label{pseqmeasure}
\eeq
defines a $(v,\varphi)$-sequence.
\end{lemma}

Among the possible candidates for the subsets $E_t$, we will consider the sets
\[
\Fix(\varphi^t)\coloneq\{x\in M\mid\varphi^t(x)=x\}
\]
of $t$-periodic orbits, and call \ndex{periodic orbit ensemble} the corresponding
measures~\eqref{pseqmeasure}. The significance of periodic orbits in the modern
theory of dynamical systems goes back to seminal works of Bowen~\cite{Bowen1970}
and Manning~\cite{Manning1971}. In the context of the FT, periodic orbits played
an important role in the early numerical works~\cite{Evans1993}. Ruelle’s proof
of the Gallavotti–Cohen fluctuation theorem for Anosov
diffeomorphisms~\cite{Ruelle1999} was centered around periodic orbits. Further
insights were obtained in~\cite{Maes2003b} where, following the general scheme
of~\cite{Lebowitz1999,Maes1999}, the pairs $(\PP_t,\wP_t)$ and the entropy
production observable $\sigma$ were introduced, and the transient fluctuation
relation was discussed. The work~\cite{Maes2003b} primarily concerned Gibbsian
type FT for regular potentials $v$. Finally, let us remark that in the context
of lattice spin systems periodic orbits arise naturally from the consideration
of periodic boundary conditions. Other boundary conditions lead to
asymptotically additive perturbations to which the last remark of
Section~\ref{sec:LDP for EM} apply. Thus, the periodic orbit ensemble coincides
with the canonical ensemble of statistical mechanics.

To formulate the main result of this section, let us introduce the following
condition.

\medskip\noindent
\hypertarget{hyp.PAP}{\textbf{(PAP)}}\label{pap} {\sl $\varphi$ has the {\rm
Periodic Approximation Property} whenever $E_t=\Fix(\varphi^t)$ is finite for
all $t\ge1$ and satisfies Relation~\eqref{pseqdef} for all $u\in C(M)$.}

\medskip
Theorem~\ref{thm:LDP-2} has the immediate

\begin{corollary}\label{cor:per}
Suppose that Conditions~\hC, \hED, \hUSCE and \hPAP are satisfied. Then, for any
$v\in C(M)$, the conclusions of Theorem~\ref{thm:LDP-2} and Corollary~\ref{cor:LDP-1}
hold for the periodic orbit ensemble
\[
\PP_t\coloneq Z_t^{-1}\sum_{x\in\Fix(\varphi^t)}\e^{S_tv(x)}\delta_x,
\qquad
Z_t=Z_t[v]\coloneq\sum_{x\in\Fix(\varphi^t)}\e^{S_tv(x)}.
\]
\end{corollary}

As seen in Section~\ref{sec:CDS}, if $\varphi$ is forward expansive or expansive
and has the specification property~\hS, then Conditions~\hED and~\hUSCE hold. To
describe a condition ensuring~\hPAP we need to introduce a variant of the
specification property~\hS.

\medskip\noindent
\hypertarget{hyp.WPS}{\textbf{(WPS)}}\label{wps}
{\sl $\varphi$ has the {\rm Weak Periodic Specification Property} if, for any $\delta>0$,
there is a sequence of integers $(m_\delta(t))_{t\ge t_0}$ such that for all $t\ge t_0$,
\[
0\le m_\delta(t )<t,\qquad\lim_{\delta\downarrow0}\lim_{t\to\infty}\frac{m_\delta(t)}t=0,
\]
and for any $x\in M$, we have
\[
\Fix(\varphi^t)\cap B_{t-m_\delta(t)}(x,\delta)\not=\emptyset.
\]}
\begin{theorem}\label{thm:pap}
If $\varphi$ is forward expansive or expansive and satisfies Condition~\hWPS, then
Condition~\hPAP holds.
\end{theorem}

\proof
The argument is based on characterizations of the pressure that we recall now,
referring the reader to~\cite[Section~9.1]{Walters1982} for details.

Let $E$ be a finite subset of $M$ and $\varepsilon>0$. $E$ is called
$(\varepsilon,t)$-{\sl separated\/} if $y\notin B_t(x,\varepsilon)$ for any
distinct $x,y\in E$, and $(\varepsilon,t)$-{\sl spanning\/} if the family
$(B_t(x,\varepsilon))_{x\in E}$ covers~$M$.
For $v\in C(M)$ define
\beq
\begin{split}
\underline{P}_t(v,\varepsilon)&\coloneq\inf\biggl\{\sum_{x\in E}\e^{S_tv(x)}\,\bigg|\,
E\subset M\text{ is
 }(\varepsilon,t)\text{-spanning}\biggr\},\\
\overline{P}_t(v,\varepsilon)&\coloneq\sup\biggl\{\sum_{x\in E}\e^{S_tv(x)}\,\bigg|\,
E\subset M\text{ is
 }(\varepsilon,t)\text{-separated}\biggr\}.
\end{split}
\label{defPbars}
\eeq
For any $v\in C(M)$, the four limiting quantities
\beq
\begin{split}
\lim_{\varepsilon\downarrow0}\limsup_{t\to\infty}\frac1t\log\underline{P}_t(v,\varepsilon),
&\qquad\lim_{\varepsilon\downarrow0}\liminf_{t\to\infty}\frac1t\log\underline{P}_t(v,\varepsilon),\\
\lim_{\varepsilon\downarrow0}\limsup_{t\to\infty}\frac1t\log\overline{P}_t(v,\varepsilon),
&\qquad\lim_{\varepsilon\downarrow0}\liminf_{t\to\infty}\frac1t\log\overline{P}_t(v,\varepsilon),
\end{split}
\label{limitPbars}
\eeq
coincide with the topological pressure $\fp_\varphi(v)$.

Let $r$ denote the expansiveness constant of $\varphi$. For fixed $t\ge1$,
uniform continuity and periodicity imply that there is $\delta>0$ such that if
$x,y\in\Fix(\varphi^t)$  and $d(x,y)\le\delta$, then $d(\varphi^t
(x),\varphi^t(y))\le r$ for all $t\in\cT$. It follows that $d(x,y)>\delta$ for
any distinct $x,y\in\Fix(\varphi^t)$ and compactness implies that
$\Fix(\varphi^t)$ may contain only finitely many points. Let $v\in C(M )$, and
let $\varepsilon\in]0,r[$. Then, in view of periodicity, for any $t\ge1$ the
set $\Fix(\varphi^t)$ is $(\varepsilon, t)$-separated. It follows that
\[
\sum_{x\in\Fix(\varphi^t)}\e^{S_tv(x)}\le\overline{P}_t(v,\varepsilon)
\]
whence, by the above representation of the pressure, we conclude that
\beq
\limsup_{t\to\infty}\frac1t\log\sum_{x\in\Fix(\varphi^t)}\e^{S_tv(x)}\le\fp_\varphi(v).
\label{equ:limsupper}
\eeq
Since $\varphi$ satisfies Condition~\hWPS, the set $\Fix(\varphi^t)$ is
$(\varepsilon,t-m_\varepsilon(t))$-spanning. It follows that
\[
\sum_{x\in\Fix(\varphi^t)}\e^{S_tv(x)}\ge\underline{P}_{t-m_\varepsilon(t)}(v,\varepsilon)
\]
Combining this with the above representation of the pressure, and using the condition on
$m_\varepsilon(t)$, we see that
\[
\liminf_{t\to\infty}\frac1t\log\sum_{x\in\Fix(\varphi^t)}\e^{S_tv(x)}\ge\fp_\varphi(v),
\]
which, together with~\eqref{equ:limsupper}, ends the proof.\hfill\qed

\medskip
By Proposition~\ref{proppress}~(4), Condition~\hUSCE ensures that the set of
equilibrium states for $v\in C(M)$ is never empty. A non-constructive proof just
argues that an upper semicontinuous function achieves its maximum on a compact
set. The following result provides an explicit construction.
\begin{proposition}
If Conditions~\hUSCE and~\hPAP hold, then any weak limit point of the periodic
orbit ensemble~$(\PP_t)_{t\in\NN}$ is an equilibrium state for $v$.
\end{proposition}

\proof Let the sequence $t_n\to\infty$ be such that
$\PP_{t_n}\rightharpoonup\PP$. Since $\PP_t\in\cP_\varphi(M)$ for all $t$, so
does $\PP$. In view of Hölder’s inequality, the function
$v\mapsto\log Z_t[v]\in\RR$ is convex. Moreover, for any $u,v\in C(M)$, the
function $\alpha\mapsto f(\alpha)\coloneq\log Z_t[v+\alpha u]$ is differentiable, and
its derivative at zero is given by $f'(0)=\bra S_tu,\PP_t\ket$. By convexity,
we have $f(1)-f(0)\ge f'(0)$, which gives
$\log Z_t[v+u]-\log Z_t[v]\ge\bra S_tu,\PP_t\ket$.
Setting $t=t_n$ in the above inequality, dividing it by $t_n$ , and passing to
the limit yields $\fp_\varphi(v+u)-\fp_\varphi(v)\ge\bra u,\PP\ket$ which
can be rewritten as
\[
\fp_\varphi(v+u)-\bra v+u,\PP\ket\ge\fp_\varphi(v)-\bra v,\PP\ket.
\]
Taking the supremum over $u\in C(M)$ and invoking
Proposition~\ref{proppress}~(3) completes the proof.\hfill\qed

\section{Fluctuation Theorems for the Periodic Orbit Ensemble}
\label{sec:FT for Periodic Ensemble}

In this subsection we describe FRs and FTs for the periodic orbit ensemble. We
consider first the general setting of Definition~\ref{Def:FT}: let $v,\hat v\in C(M)$
and denote by $(\PP_t)_{t\in\NN}$ and $(\wP_t)_{t\in\NN}$ the corresponding periodic orbit
ensemble for the dynamical system $(M,\varphi)$. Without loss of generality, we
assume that $\fp_\varphi(v)=\fp_\varphi(\hat v)=0$. Under the assumptions of
Corollary~\ref{cor:per}, we deduce from Theorem~\ref{thm:LDP-2} that the empirical
measures $(\mu_t^\argdot)_{t\in\NN}$ satisfy the LDP w.r.t.\;both
$(\PP_t)_{t\in\NN}$ and $(\wP_t)_{t\in\NN}$, with respective rate functions
\[
\II(\QQ)=-\bra v,\QQ\ket-h_\varphi(\QQ),\qquad
\widehat{\II}(\QQ)=-\bra \hat v,\QQ\ket-h_\varphi(\QQ).
\]
Thus, we have the level-2 FT with the FR\index{FR!level-2}
\[
\widehat{\II}(\QQ)=\II(\QQ)+\ep(\QQ),
\]
where, obviously, $\ep(\QQ)=\bra v-\hat v,\QQ\ket$. Recalling that $\II(\QQ)=0$
iff $\QQ$ is an equilibrium measure for $v$, we deduce that $\ep(\QQ)\ge0$ for
any such measure, while $\ep(\widehat{\QQ})\le0$ whenever $\widehat{\QQ}$ is an
equilibrium measure for $\hat v$.

The entropy production associated to the pair $(\PP_t,\wP_t)$ is
$$
\sigma_t=\log\frac{\d\PP_t}{\d\wP_t}=S_t(v-\hat v)-\log\frac{Z_t[v]}{Z_t[\hat v]}.
$$
Condition~\hPAP yields
\[
\lim_{t\to\infty}\frac1t\log\frac{Z_t[v]}{Z_t[\hat v]}
=\fp_\varphi(v)-\fp_\varphi(\hat v)=0,
\]
which shows that $(\sigma_t)_{t\in\NN}$ is physically equivalent to
$(\widetilde{\sigma}_t)_{t\in\NN}$, where
\[
\widetilde{\sigma}_t\coloneq S_t(v-\hat v).
\]
Thus, $(\frac1t\sigma_t)_{t\in\NN}$ and $(\frac1t\widetilde{\sigma}_t)_{t\in\NN}$
satisfy the same LDP w.r.t.\;$(\PP_t)_{t\in\NN}$ and $(\wP_t)_{t\in\NN}$
with respective rate functions
\[
I(s)=\inf\{\II(\QQ)\mid\QQ\in\cP_\varphi(M),\ep(\QQ)=s\},
\qquad
\hat I(s)=\inf\{\widehat{\II}(\QQ)\mid\QQ\in\cP_\varphi(M),\ep(\QQ)=s\}.
\]
Moreover, the FR
$$
\hat I(s)=I(s)+s
$$
holds for all $s\in\RR$. The entropic pressures
\[
e(\alpha)=\lim_{t\to\infty}\frac1t
\log\bra\e^{-\alpha\sigma_t},\PP_t\ket,\qquad
\hat e(\alpha)=\lim_{t\to\infty}\frac1t
\log\bra\e^{\alpha\sigma_t},\wP_t\ket,
\]
exist for all $\alpha\in\RR$, are given by
\[
e(\alpha)=\fp_\varphi\left((1-\alpha)v+\alpha\hat v\right),\qquad
\hat e(\alpha)=\fp_\varphi\left((1-\alpha)\hat v+\alpha v\right),
\]
and satisfy
\[
\hat e(\alpha)=e(1-\alpha).
\]
They are related to the rate function by Fenchel--Legendre transformation
$$
I(s)=\sup_{\alpha\in\RR}\left(\alpha s-e(-\alpha)\right),\qquad
\hat I(s)=\sup_{\alpha\in\RR}\left(\alpha s-\hat e(\alpha)\right).
$$

The interpretation of these results rests on the general picture described in
Section~\ref{sec:FR_Meaning}: let $\PP$ and $\wP$ denote equilibrium states for
$v$ and $\hat v$ such that $\PP_{t_n}\rightharpoonup\PP$ and
$\wP_{t_n}\rightharpoonup\wP$ for some subsequence $t_n\to\infty$.\footnote{In
absence of phase transition, {\sl i.e.,} when both $v$ and $\hat v$ have a unique
equilibrium measure, one has $\PP_t\rightharpoonup\PP$ and
$\wP_t\rightharpoonup\wP$ as $t\to\infty$.} Then $\PP_{t_n}$ and $\wP_{t_n}$ are
equivalent measures for all $n$, but their limit points $\PP$ and $\wP$ may be
mutually singular, and they definitely are whenever $\ep(\PP)>0$. Then the
various exponents described in Section~\ref{sec:FR_Meaning} quantify the rate at
which the supports of $\PP_{t_n}$ and $\wP_{t_n}$ separate.

\medskip
Next, let us specialize the discussion to pairs of periodic orbit ensemble in involution. To
this end, we make the following additional assumption:
\newcommand{\hSym}{\hyperlink{hyp.Sym}{{\bf (Sym)}}\xspace}

\medskip\noindent
\hypertarget{hyp.Sym}{\textbf{(Sym)}}\label{reversal} {\sl There is a continuous involution $\theta:M\to M$
such that:
\begin{itemize}
\item $\varphi\circ\theta=\theta\circ\varphi$ in case~\hC;
\item $\varphi\circ\theta=\theta\circ\varphi^{-1}$ in case~\hH.
\end{itemize}
 }

\medskip
For $v\in C(M)$ and $\QQ\in\cP_\varphi(M)$, we then set $\hat v\coloneq v\circ\theta$
and $\widehat{\QQ}\coloneq\QQ\circ\theta^{-1}$, so that
$\bra\hat v,\QQ\ket=\bra v,\widehat{\QQ}\ket$. Let us remark that
$\widehat\QQ\in\cP_\varphi(M)$. Indeed, for any $v\in C(M)$ we have
\[
\langle v\circ\varphi,\widehat\QQ\rangle
=\langle v\circ\varphi\circ\theta,\QQ\rangle
=\langle v\circ\theta\circ\varphi^{\pm1},\QQ\rangle
=\langle v\circ\theta,\QQ\rangle
=\langle v,\widehat\QQ\rangle.
\]
In case~\hH, Assumption~\hSym is the standard time-reversal invariance condition
that appears in most of the  works on the FR and FT. To the best of our
knowledge, it was not previously observed that Assumption~\hSym also leads to FT
in case~\hC. Since in the latter case $\varphi$ is not required to be
invertible, this allows us to extend FR/FT to a large class of irreversible
dynamical systems.

Observe that Condition~\hSym implies
\begin{equation}
S_t(v\circ\theta)=\sum_{0\le s<t}v\circ\theta\circ\varphi^s
=\sum_{0\le s<t}v\circ\varphi^{\pm s}\circ\theta
=(S_tv)\circ\theta_t
\label{eq:SnVtheta}
\end{equation}
where
\[
\theta_t\coloneq\left\{
\bear{ll}
\theta&\text{in case }\hC;\\[4pt]
\theta\circ\varphi^{t-1}&\text{in case }\hH,
\ear
\right.
\]
is involutive. In the case~\hC this is immediate, and in the case~\hH it follows from
\[
\theta_t^{-1}=\varphi^{1-t}\circ\theta^{-1}=\varphi^{1-t}\circ\theta
=\theta\circ\varphi^{t-1}=\theta_t.
\]
In addition, we have
\begin{lemma}\label{lem:Sym}
Under Condition~\hSym, the following holds for any $\QQ\in\PP_\varphi(M)$ and any
$v\in C(M)$:
\begin{align}
h_\varphi(\widehat\QQ)&=h_\varphi(\QQ),\label{Eq:htheta}\\[2mm]
\fp_\varphi(\hat v)&=\fp_\varphi(v).\label{eq:Psymtheta}
\end{align}
Moreover, $\theta$ and $\theta_t$ act bijectively on $\Fix(\varphi^t)$.
\end{lemma}

\proof Since the Kolmogorov--Sinai entropy is a conjugacy invariant
(see~\cite[Theorem 4.11]{Walters1982}), we deduce
\[
h_{\varphi^{\pm1}}(\QQ)=h_\varphi(\QQ\circ \theta^{-1})=h_\varphi(\widehat\QQ)
\]
from the fact that $\theta^{-1}\circ\varphi\circ\theta=\varphi^{\pm1}$, which
proves~\eqref{Eq:htheta} in case~\hC. To deal with case~\hH, it suffices to
take into account~\cite[Theorem~4.13]{Walters1982} which states that
$h_{\varphi^{-1}}(\QQ)=h_\varphi(\QQ)$.

In order to prove~\eqref{eq:Psymtheta}, we observe that, by~\eqref{PressureDef}
and~\eqref{Eq:htheta},
\begin{align*}
\fp_\varphi(\hat v)
&=\sup_{\QQ\in\cP_\varphi(M)}\bigl(\langle v\circ\theta,\QQ\rangle+h_\varphi(\QQ)\bigr)
=\sup_{\QQ\in\cP_\varphi(M)}\bigl(\langle v,\widehat\QQ\rangle+h_\varphi(\widehat\QQ)\bigr)\\
&=\sup_{\QQ\in\cP_\varphi(M)}\bigl(\langle v,\QQ\rangle+h_\varphi(\QQ)\bigr)
=\fp_\varphi(v).
\end{align*}
To prove the last statement, it suffices to show that $\Fix(\varphi^t)$ is
invariant under $\theta$, a property that follows immediately from
Condition~\hSym.\hfill\qed

\medskip
In the following we shall again assume, without loss of generality, that $v\in C(M)$
is such that $\fp_\varphi(v)=0$. Relation~\eqref{eq:SnVtheta} and the last
statement of Lemma~\ref{lem:Sym} immediately imply that $Z_t[\hat v]=Z_t[v]=Z_t$
and\footnote{The notation $\wP_t=\PP_t\circ\theta_t^{-1}$ is not conflicting
with our previous convention $\widehat{\QQ}=\QQ\circ\theta^{-1}$ which we have
restricted to $\QQ\in\cP_\varphi(M)$. Indeed, for such $\QQ$ one has
$\QQ\circ\theta_t^{-1}=\QQ\circ\theta^{-1}$.}
\begin{equation*}
\wP_t=\PP_t\circ\theta_t^{-1}
=Z_t^{-1}\sum_{x\in\Fix(\varphi^t)}\e^{S_t\hat v(x)}\delta_x
\end{equation*}
for all $t$. In particular, the measures~$\PP_t$ and~$ \wP_t$ are equivalent, and
the entropy production
\begin{equation}
\label{4.31}
\sigma_t=\log\frac{\d\PP_t}{\d\wP_t}=S_t(v-\hat v)
\end{equation}
satisfies $\sigma_t\circ\theta_t=-\sigma_t$. For $\alpha\in\RR$, the entropic pressure is
\[
e(\alpha)=\lim_{t\to\infty}\frac1t
\log\bra\e^{-\alpha\sigma_t},\PP_t\ket
=\fp_\varphi\left((1-\alpha)v+\alpha\hat v\right).
\]
For $\QQ\in\cP_\varphi(M)$ one has $\ep(\QQ)=\langle v-\hat v,\QQ\rangle
=\langle v,\QQ-\widehat\QQ\rangle$, so that
$$
\ep(\widehat \QQ) = -\ep(\QQ).
$$

\begin{theorem}\label{thm:per FT}\index{FT!for chaotic maps}
In addition to the hypotheses of Corollary~\ref{cor:per}, suppose that
Condition~\hSym holds. Then the rate function~$\II$ of the LDP
for the empirical measures~\eqref{equ:mutDef} under the periodic orbit ensemble
$(\PP_t)_{t\in\NN}$ satisfies the FR
$$
\II(\widehat\QQ)=\II(\QQ)+\ep(\QQ)
$$
for any $\QQ\in\cP_\varphi(M)$. Furthermore, the sequence
$(\frac1t\sigma_t)_{t\in\NN}$ under the laws $(\PP_t)_{t\in\NN}$  satisfies the
LDP with the good convex rate function given by
$$
I(s)=\inf\{\II(\QQ):\QQ\in\cP_\varphi(M), \ep(\QQ)=s\},
$$
which is also the Legendre transform of the function $\alpha \mapsto e(-\alpha)$.
Moreover, the FR
\beq
I(-s)=I(s)+s
\label{FRagain}
\eeq
holds for all $s\in\RR$.
\end{theorem}

\section{Weak Gibbs Measures}
\label{sec:WeakGibbs}

We continue our study of FT/FR in chaotic dynamical systems, now focusing on
Gibbs and weak Gibbs measures.
\begin{definition}\label{def:weak Gibbs}
We say that~$\PP\in\cP(M)$ is a \ndex{weak Gibbs measure} for $v\in C(M)$ if
for any $t\ge1$ and any $\varepsilon>0$ there is $K_t(\varepsilon)\ge1$ such that
\beq
K_t(\varepsilon)^{-1}\,\e^{S_tv(x)-t\fp_\varphi(v)}\le \PP\bigl(B_t(x,\varepsilon)\bigr)
\le K_t(\varepsilon)\,\e^{S_tv(x)-t\fp_\varphi(v)},
\label{equ:weakGibbs}
\eeq
for every $x\in M$, and
\beq
\lim_{\varepsilon\downarrow0}\limsup_{t\to\infty}\frac1t\log K_t(\varepsilon)=0.
\label{equ:weakGibbsconst}
\eeq
If $K_t(\varepsilon)$ does not depend on $t$ and $\varepsilon$, then $\PP$ is a
\ndex{Gibbs measure} for $v$.
\end{definition}

We emphasize that the definition of (weak) Gibbs measure does not require
$\PP\in\cP_\varphi(M)$. Notice also that it follows from~\eqref{equ:weakGibbs} that the
support of~$\PP$ coincides with~$M$. The following lemma shows that if the
latter property is satisfied, then it suffices to require the validity
of~\eqref{equ:weakGibbs} for $\PP$-a.e.\;$x\in M$. This observation is technically useful when
transfer operators are used to construct weak Gibbs measures;
see~\cite[Section 2]{Kesseboehmer2001}, \cite[Appendix B]{Climenhaga2010}.

\begin{lemma}
Let $v\in C(M)$ and $\PP\in\cP(M)$ be such that $\supp(\PP)=M$. Assume that for all $t\ge1$ and
all $\varepsilon>0$, \eqref{equ:weakGibbs} holds for $\PP$-a.e.\;$x\in M$, with~$K_t(\varepsilon)$
satisfying~\eqref{equ:weakGibbsconst}. Then $\PP$ is weak Gibbs for $v$.
\end{lemma}

\proof Given $t\geq 1$ and $\varepsilon>0$, let~$A\subset M$ be a set of full $\PP$-measure
on which~\eqref{equ:weakGibbs} holds and let $x\in M$ be an arbitrary point. Since
$\overline{A}=M$, there is $x'\in A\cap B_t(x,\varepsilon/2)$, such that
$$
\PP(B_t(x, \varepsilon))\geq\PP(B_t(x', \varepsilon/2))
\geq K_t(\varepsilon/2)^{-1}\e^{S_tv(x')-t\fp_\varphi(v)}
\geq K_t'(\varepsilon)^{-1}\e^{S_tv(x)-t\fp_\varphi(v)-t\gamma_\varepsilon(v)},
$$
and
$$
\PP(B_t(x, \varepsilon))\leq\PP(B_t(x',2\varepsilon))
\leq K_t(2\varepsilon)\e^{S_tv(x')-t\fp_\varphi(v)}
\leq K_t'(\varepsilon)\e^{S_tv(x)-t\fp_\varphi(v)+t\gamma_\varepsilon(v)},
$$
where
\begin{equation*}
K_t'(\varepsilon)\coloneq\max\left(K_t(2\varepsilon), K_t(\varepsilon/2)\right)
\end{equation*}
and
\beq
\gamma_\varepsilon(v)\coloneq\max\left\{|v(x)-v(y)|\mid x,y\in M,d(x,y)\le\varepsilon\right\}.
\label{equ:gammaDef}
\eeq
The condition~\eqref{equ:weakGibbsconst} and the uniform continuity of $v$ give
\[
\lim_{\varepsilon\downarrow 0}\limsup_{t\rightarrow\infty}\frac{1}{t}\log K_t'(\varepsilon)=0,\qquad
\lim_{\varepsilon\downarrow 0}\gamma_\varepsilon(v)=0,
\]
from which the statement follows.\hfill\qed

\medskip
We next show that invariant weak Gibbs measures are equilibrium states
(the converse is not true in general).

\begin{lemma}\label{lem:WGestequil}
Let $v\in C(M)$ and assume that $\PP\in\cP_\varphi(M)$ is a weak Gibbs measure for $v$.
Then $\PP$ is an equilibrium measure for $v$.
\end{lemma}

\proof By~\eqref{equ:weakGibbs} we have, for all $x\in M$,
\beq
\lim_{\varepsilon\downarrow0}\limsup_{t\to\infty}\frac1t
\left(S_tv(x)-\log\PP(B_t(x,\varepsilon))\right)=\fp_\varphi(v).
\label{equ:wG1}
\eeq
On the other hand, the Brin--Katok local entropy formula~\cite{Brin1983}
implies that for $\PP$-a.e.~$x\in M$,
\beq
\lim_{\varepsilon\downarrow0}\limsup_{t\to\infty}\frac1t
\left(S_tv(x)-\log\PP(B_t(x,\varepsilon))\right)=h_\varphi(\PP,x)+\overline{v}(x),
\label{equ:wG2}
\eeq
where $h_\varphi(\PP,\,\argdot\,)$ and $\overline{v}$ are $\varphi$-invariant
functions such that $\langle h_\varphi(\PP,\,\argdot\,),\PP\rangle=h_\varphi(\PP)$ and
$\langle\overline{v},\PP\rangle=\langle v,\PP\rangle$.
Equating the right-hand sides of~\eqref{equ:wG1} and~\eqref{equ:wG2} and integrating
with respect to~$\PP$, we arrive at~\eqref{EquiDef}.
\hfill\qed

\medskip
\begin{lemma}\label{lem:WGPress}
Suppose that Condition~\hUSCE holds and that $\PP\in\cP(M)$ is a  weak Gibbs measure
for $v\in C(M)$. Then, setting  $\PP_t=\PP$ for all $t\in\NN$ yields a $(v,\varphi)$-sequence.
\end{lemma}

\proof The proof follows that of~\cite[Proposition~3.2]{Kifer1990} (see
also~\cite[Proposition~10.3]{Jaksic2011}).

For any $u\in C(M)$, $\varepsilon>0$, $t\ge1$, and any $(\varepsilon,t)$-spanning set $E$,
using~\eqref{equ:weakGibbs} we derive
\begin{align*}
\bra\e^{S_tu},\PP\ket
&\le\sum_{x\in E}\bra\boldsymbol{1}_{B_t(x,\varepsilon)}\e^{S_tu},\PP\ket
\leq\e^{t\gamma_\varepsilon(u)}\sum_{x\in E}\e^{S_tu(x)}\PP(B_t(x,\varepsilon))\\
&\leq K_t(\varepsilon)\e^{-t(\fp_\varphi(v)-\gamma_\varepsilon(u))}\sum_{x\in E}
\e^{S_t(v+u)(x)},
\end{align*}
where the modulus of continuity $\gamma_\varepsilon$ is as in~\eqref{equ:gammaDef}.
Taking the infimum over all $(\varepsilon,t)$-spanning set $E$ and recalling
the definition~\eqref{defPbars} yields
\[
\frac1t\log\bra\e^{S_tu},\PP\ket\le\gamma_\varepsilon(u)-\fp_\varphi(v)+\frac1t\log K_t(\varepsilon)+\frac1t\log\underline{P}_t(v+u,\varepsilon).
\]
By uniform continuity, $\gamma_\varepsilon(u)\to0$ as $\varepsilon\downarrow0$. Together
with~\eqref{equ:weakGibbsconst} and~\eqref{limitPbars}, this gives
\begin{equation*}
\limsup_{t\to\infty}\frac1t\log\bra\e^{S_tu},\PP\ket
\leq\fp_\varphi(v+u)-\fp_\varphi(v).
\end{equation*}
To prove the reverse inequality, we proceed similarly, observing that for any
$(\varepsilon,t)$-separated set $E$ we have
\[
\bra\e^{S_tu},\PP\ket
\geq K_t^{-1}(\varepsilon)\e^{-t(\fp_\varphi(v)+\gamma_\varepsilon(u))}
\sum_{x\in E}\e^{S_t(v+u)(x)}.
\]
Taking the supremum over all $(\varepsilon,t)$-separated sets and repeating the above argument
we obtain the desired result.
\hfill\qed

\medskip
The following result is again a direct consequence of Lemma~\ref{lem:WGPress},
Theorem~\ref{thm:LDP-2} and Corollary~\ref{cor:LDP-1}.
\begin{proposition}\label{thm:wG LDP-2}
Assume that Conditions~\hC, \hUSCE and~\hED hold and that\/ $\PP$ is a weak
Gibbs measure for $v\in C(M)$. Then the conclusions of Theorem~\ref{thm:LDP-2}
and Corollary~\ref{cor:LDP-1} hold with~$\PP_t$ replaced by\/~$\PP$.
\end{proposition}

Recall that $\sigma_t$ is defined by~\eqref{4.31}. By Lemma~\ref{lem:WGPress},
the limit
\[
e(\alpha)\coloneq\lim_{t\to\infty}\frac{1}{t}
\log\int_{M}\e^{-\alpha \sigma_t}\d\PP
\]
exists for all $\alpha\in\RR$ and is given by
\[
e(\alpha)=\fp_\varphi\left((1-\alpha)v+\alpha\,v\circ\theta\right)-\fp_\varphi(v) .
\]
The proof of the following result is exactly the same as that of
Theorem~\ref{thm:per FT}.
\begin{theorem}
If in  addition to the hypotheses of Theorem~\ref{thm:wG LDP-2},
Condition~\hSym holds, then all the conclusions of
Theorem~\ref{thm:per FT} hold with $\PP_t$ replaced by $\PP$.
\end{theorem}

\remark[1] Our assumptions do not exclude the situation where the sequence
$(\frac1t\sigma_t)_{t\in\NN}$ is exponentially equivalent to 0, {\sl i.e.,}
the situation where $I(0)=0$ and $I(s)=+\infty$ for all $s\neq 0$. In this case
the symmetry~\eqref{FRagain} is trivial, and $\ep(\PP)=0$ (note
that~\eqref{FRagain} always implies that $\ep(\PP)\geq0$). It is a non-trivial
question to determine whether a given system is truly out of equilibrium, {\sl
i.e.,} if $\ep(\PP)>0$. For $\PP\in\cP_\varphi(M)$, we have
$\ep(\PP)=\langle\sigma,\PP\rangle$, where $\sigma=v-v\circ\theta$. In the
context of the chaotic hypothesis of Gallavotti--Cohen,
Section~\ref{ssect:Chaotic}, when $\PP$ is the SRB measure, the inequality
$\ep(\PP)>0$ is called {\em dissipativity} and is often {\em assumed}
(see~\cite{Bonetto2000} for a model where dissipativity can be proved
and~\cite[Proposition~11.5]{Jaksic2011} for generic results in this direction).
See also~\cite{Ruelle1996,Bonetto2006} for discussions and~\cite[Theorem
5.2]{Maes2003b}, \cite[Corollary 10.16 (7)]{Jaksic2011}.

\remark[2] Weak Gibbs measures have been extensively studied in the recent
literature on multifractal formalism; see~\cite{Cuneo2019} for references and
additional information.

\printbibliography[heading=bibintoc,title={References}]


\chapter{The Tent Map}
\label{Chap:Tent}

\abstract{To illustrate the variety and depth of mathematical tools needed to
study  FR/FT in the realm of chaotic dynamics, in this chapter we present a
detailed study of the tent map. This example is rich enough to allow  for the
discussion of  interesting phenomena such as phase transitions and their
consequences, and provides a blueprint for related study of other interval
maps.}

\abstract*{To illustrate the variety and depth of mathematical tools needed to
study  FR/FT in the realm of chaotic dynamics, in this chapter we present a
detailed study of the tent map. This example is rich enough to allow  for the
discussion of  interesting phenomena such as phase transitions and their
consequences, and provides a blueprint for related study of other interval
maps.}

\vskip1cm

\section{Basic Properties}

The \ndex{tent map} is the piecewise affine unimodal map of the unit interval
$M\coloneq[0,1]$ defined by
\[
\varphi(x)\coloneq 1-|1-2x|.
\]
It has $0$ and $2/3$ as unstable fixed points and the orbit of the critical
point $c=1/2$ is $c\to1\to0\to0\cdots$. Since $2/7$ has period $3$, $\varphi$
has periodic orbits of any period by \v Sarkovskii's
theorem~\cite[Theorem~II.3.10]{Collet1980}. Moreover, $|\varphi'(x)|=2>1$ for
$x\not=c$ implies that all these periodic orbits are unstable. One easily shows
that, for $t\in\NN$ and $j\in\llbracket0,2^{t-1}\rrbracket$,
\[
|2^{t-1}x-j|\le\frac12\Rightarrow\varphi^t(x)=2|2^{t-1}x-j|,
\]
from which one deduces
\[
\Fix(\varphi^t)=\left\{\frac{2(j-1)}{2^t-1}\bigg| j\in\llbracket1,2^{t-1}\rrbracket\right\}
\bigcup
\left\{\frac{2j}{2^t+1}\bigg| j\in\llbracket1,2^{t-1}\rrbracket\right\},
\]
and in particular $|\Fix(\varphi^t)|=2^t$. The set of $\varphi$-periodic points
is dense in $M$, more precisely $\dist(x,\Fix(\varphi^t))\le2^{-t}$
holds for any $x\in M$. For later reference, let us introduce
$$
\cC\coloneq\bigcup_{n\in\ZZ_+}\varphi^{-n}(\{c\})
=\left\{j 2^{- k}\,\bigg|\, k\in\NN, j\in\llbracket0,2^k\rrbracket,j\text{ odd}\right\},
$$
the dense set of preimages of the critical point $c$, and $M_0\coloneq M\setminus\cC$.

\begin{figure}
\centering
\includegraphics[width=12cm]{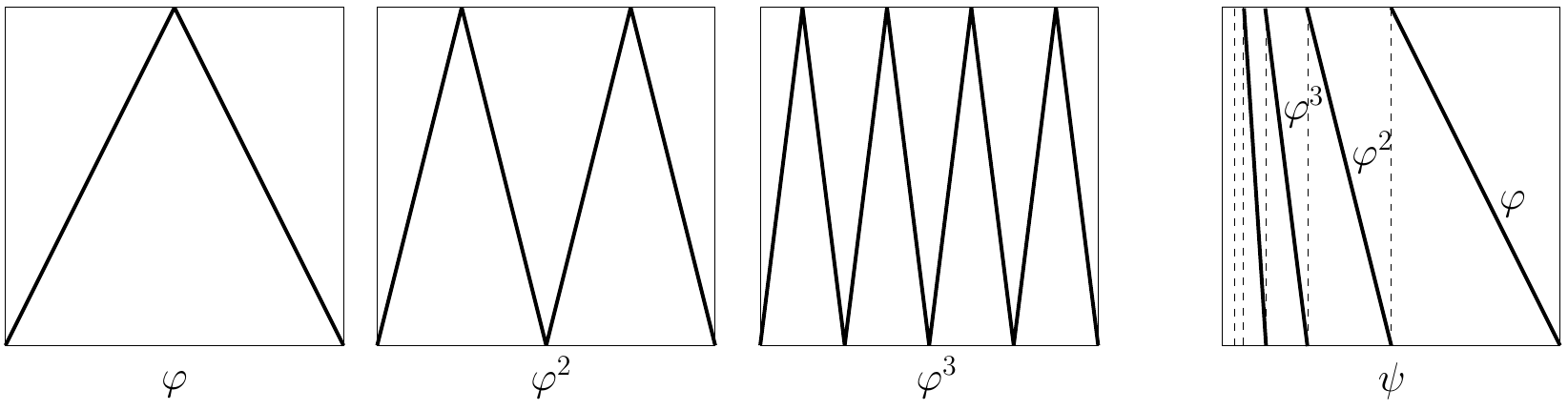}
\caption{The tent map $\varphi$, its first iterates, and the induced map $\psi$
which will be central in Section~\ref{sec:inducing}.}
\label{Fig:Tent}
\end{figure}

Since $\varphi^t$ has $2^t$ intervals of monotonicity, the Misiurewicz
formula~\cite[Theorem~1]{Misiurewicz1980} gives the topological entropy of the
tent map
\beq
h_{\rm Top}(\varphi)=\log2.
\label{eq:Tenthtop}
\eeq
The elementary calculation
\[
\int_0^1f\circ\varphi(x)\d x
=\int_0^{1/2}f(2x)\d x+\int_{1/2}^1f(2(1-x))\d x=\int_0^1f(x)\d x
\]
shows that Lebesgue measure $\lambda$ is $\varphi$-invariant. By Rokhlin's formula
\beq
h_\varphi(\lambda)
=\int_0^{1/2}\log\frac{\d\lambda\circ\varphi|_{[0,1/2]}}{\d\lambda}\d\lambda
+\int_{1/2}^1\log\frac{\d\lambda\circ\varphi|_{[1/2,1]}}{\d\lambda}\d\lambda
=\log2,
\label{TentRokhlin}
\eeq
so $\lambda$ has maximal entropy. One easily infers from the above mentioned
structure of periodic orbits that $\varphi$ is topologically transitive, {\sl
i.e.,} for any open subsets $U,V\subset M$ there exists $t\ge0$ such that
$U\cap\varphi^t(V)\not=\emptyset$. By~\cite[Theorem~4]{Hofbauer1981}, $\varphi$
is intrinsically ergodic: $\lambda$ is the unique $\varphi$-invariant measure of
maximal entropy. In particular, $\lambda$ is the only $\varphi$-ergodic measure
which is absolutely continuous w.r.t.\;the Lebesgue measure on $M$.

Since $\varphi$ has periodic points of period $3$, it follows
from~\cite[Theorem~2.20]{Ruette2017} that it is topologically mixing: for any
open subsets $U,V\subset M$ there exists $t\ge0$ such that
$U\cap\varphi^s(V)\not=\emptyset$ holds for all $s\ge t$.
By~\cite[Theorem~3.4]{Ruette2017}, it follows that $\varphi$ has the
specification property~\hS and hence satisfies Condition~\hED.

The symmetry $\varphi(1-x)=\varphi(x)$ implies that, for any $r>0$,
$1-x\in\Gamma(x,r)$ whenever $|x-c|\le r/2$ (recall the
definition~\eqref{equ:GammaDef}). Hence, the tent map fails to be forward
expansive. It is also easy to see that $\Fix(\varphi^t)$ is not
$(\varepsilon,t)$-separated for $\varepsilon>2/(4^t-1)$. Thus, we cannot apply
the results of Chapter~\ref{chap:Chaotic maps and FR} directly.

We shall see how to recover Properties~\hUSCE and~\hPAP in the next section. In
Sections~\ref{sec:inducing}--\ref{sec:GCP} we will focus on the basic properties
of the tent map's pressure $\fp_\varphi(v)$. This will lead us, in
Section~\ref{sec:PhaseTransitions}, to formulate a simple criterion for the occurrence of
phase transitions, {\sl i.e.,} singularities of the pressure
$\fp_{\varphi}(\kappa v)$, as a function of the coupling strength $\kappa$. We
will encounter an obstruction to the fluctuation theorem for the tent map
in Section~\ref{sec:square}, when discussing the last necessary ingredient
of Theorem~\ref{thm:per FT}, namely Assumption~\hSym. We will overcome
this problem by ``squaring'' the system. In the last two sections we will
discuss a specific example.

\section{Symbolic Dynamics}
\label{Ssec:TentSymbolic}

The tent map admits a simple and well-known coding associated to the Markov
partition of $M$ given by $J_0\coloneq[0,1/2]$ and $J_1\coloneq[1/2,1]$. Denote by
\beq
x=0.b_0b_1b_2\cdots=\sum_{k\in\ZZ_+}\frac{b_k}{2^{k+1}}
\label{Eq;binexp}
\eeq
the binary expansion of $x\in M$. We will call the sequence
$\bsb=(b_k)_{k\in\ZZ_+}$ the $\bsb$-code of $x$. Note  that
$$
1-x=\sum_{k\in\ZZ_+}\frac{1-b_k}{2^{k+1}}=0.\hat b_0\hat b_1\hat b_2\cdots,
$$
where $\hat b_k\coloneq 1-b_k$. Since $x\in J_{b_0}$, one has
\[
\varphi(x)=\begin{cases}
0.b_1b_2b_3\cdots&\text{ if }b_0=0;\\[4pt]
0.\hat b_1\hat b_2\hat b_3\cdots&\text{ if }b_0=1.
\end{cases}
\]
Thus, denoting by $j$ the surjection
$\Sigma=\{0,1\}^{\ZZ_+}\ni\bsb=(b_k)_{k\in\ZZ_+}\mapsto x\in M$
defined in~\eqref{Eq;binexp}, we can write
$$
\varphi\circ j=j\circ\zeta,
$$
where $\zeta:\Sigma\to\Sigma$ is given by $\zeta(\bsb)\coloneq(b_0\oplus
b_{k+1})_{k\in\ZZ_+}$, the symbol $\oplus$ denoting the ``exclusive or'' (xor)
boolean operator.\footnote{Note that $(\{0,1\},\oplus)$, as an Abelian group, is
isomorphic to $(\ZZ_2,+)$. In particular, $b_1\oplus b_2\oplus\cdots\oplus b_n$
is $1$ whenever the word $b_1\cdots b_n$ contains an even number of $1$'s, and
$0$ otherwise.}

The so-called Gray coding $\varkappa:\Sigma\to\Sigma$ defined by
$\bss=\varkappa(\bsb)$, with
$$
s_0\coloneq b_0,\qquad s_k\coloneq b_{k-1}\oplus b_k,
$$
maps $\bsb$-code into $\bss$-code in such a way that
$\varkappa\circ\zeta(\bsb)=\varkappa(\bsb')=\bss'$, where
$$
s_0'\coloneq b_0'=b_0\oplus b_1=s_1,
\quad
s_k'\coloneq b_{k-1}'\oplus b_k'=(b_0\oplus b_k)\oplus(b_0\oplus b_{k+1})=b_k\oplus b_{k+1}=s_{k+1}.
$$
Clearly,
$$
\varkappa\circ\zeta=\sigma\circ\varkappa,
$$
where $\sigma$ denotes the left shift on $\Sigma$ and one easily shows that
$\varkappa$ is a bijection whose inverse is given by
$\bsb=\varkappa^{-1}(\bss)$, with
$$
b_0\coloneq s_0,\qquad b_k\coloneq s_k\oplus b_{k-1}.
$$
Setting $\ell=j\circ\varkappa^{-1}$, we finally get
\beq
\varphi\circ\ell=\ell\circ\sigma.
\label{Eq:TentSemiConj}
\eeq

The product topology makes $\Sigma$ a compact metrizable space, and we shall
equip this space with the compatible metric
$$
\Sigma\times\Sigma\ni(\bsa,\bsb)\mapsto d(\bsa,\bsb)\coloneq 2^{-\inf\{k\in\ZZ_+\mid a_k\not=b_k\}}.
$$
We further denote by $\Sigma^\ast\coloneq\{w=b_1\cdots b_k\in\{0,1\}^k\mid k\in\NN\}$
the set of finite-length words,  and by $|w|$ the length of such a word. Given
$w_1\in\Sigma^\ast$ and $w_2\in\Sigma^\ast\cup\Sigma$, we write their
concatenation as $w_1w_2$ and the periodic repetition of $w_1$ as
$\overline{w_1}\in\Sigma$.

The following lemma summarizes the basic properties of the map $\ell$.
\begin{lemma}\label{TentEll}
\ben
\item $\ell:\Sigma\to M$ is onto. For any $x\in M$, $x=\ell(\bss)$
iff $\varphi^k(x)\in J_{s_k}$ holds for all $k\in\ZZ_+$.
\item For any $\bss,\bss'\in\Sigma$, $|\ell(\bss')-\ell(\bss)|\le d(\bss',\bss)$.
\item $\Sigma_0\coloneq\ell^{-1}(M_0)=\Sigma\setminus
\{w1\bar0\in\Sigma\mid w\in\Sigma^\ast,|w|\ge1\}$.
\item The map $\ell:\Sigma_0\to M_0$ is a continuous bijection with a measurable inverse.
\item $\Fix(\sigma^t)=\{\overline{w}\mid w\in\Sigma^\ast,|w|=t\}\subset\Sigma_0$ and
 $\Fix(\varphi^t)=\ell(\Fix(\sigma^t))$.
 \item For any $\PP\in\cP_\varphi(M)$ one has $\PP(\cC)=0$ and
  $h_\varphi(\PP)=h_\sigma(\PP\circ\ell)$.
\een
\end{lemma}

\proof {\bf (1)} Is a direct consequence of the above construction.

\part{(2)} The claim follows from
$$
|j(\bss')-j(\bss)|\le\sum_{k\ge-\log_2d(\bss',\bss)}2^{-k-1}=d(\bss',\bss),
$$
after noticing that $\varkappa$ and its inverse are isometries.

\part{(3)} One has $\ell^{-1}(M\setminus\cC)=\Sigma\setminus\ell^{-1}(\cC)$ and
$\bss\in\ell^{-1}(\cC)$ iff
\beq
\varkappa^{-1}(\bss)\in j^{-1}(\cC)=\{w0\bar1,w1\bar0\mid w\in\Sigma^\ast\},
\label{Eq:jinvc}
\eeq
and hence $\bss\in\{\varkappa(w0\bar1),\varkappa(w1\bar0)\mid w\in\Sigma^\ast\}$
from which the result immediately follows.

\part{(4)} Given~(2), the isometry of $\kappa$ and its inverse mentioned in its
proof, and the identity in~\eqref{Eq:jinvc}, it suffices to show that the map
$j:\Sigma\setminus\{w0\bar1,w1\bar0\mid w\in\Sigma^\ast\}\to M_0$ is injective
with a measurable inverse. To deal with the first property, suppose that
$\bss,\bss'\in\Sigma$ are such that $j(\bss')=j(\bss)$. W.l.o.g.\;we can assume
that for some $n\in\ZZ_+$ one has $d(\bss',\bss)=2^{-n}$, $s'_n=0$ and $s_n=1$. It
follows that
$$
\frac{s_n-s'_n}{2^{n+1}}=\frac1{2^{n+1}}=\sum_{k>n}\frac{s'_k-s_k}{2^{k+1}},
$$
from which one easily concludes that,  for any $k>n$, $s'_k-s_k=1$ and hence
$s'_k=1$ and $s_k=0$. This means that $\{\bss',\bss\}\subset j^{-1}(\cC)$. The
measurability of the inverse map follows at once from the fact that
$\ell(\bss)=x\in M_0$ iff $s_k=1_{]1/2,1]}\circ\varphi^k(x)$.

\part{(5)} The first identity is obvious. The inclusion following it is a direct
consequence of~(3). Since $\Fix(\varphi^t)\cap\cC=\emptyset$, the last assertion
follows from~(4).

\part{(6)} Let $\PP\in\cP_\varphi(M)$ and set $m\coloneq\PP(\{1/2\})$. For $n\in\ZZ_+$,
further set $\cC_n\coloneq\varphi^{-n}(\{1/2\})$. It follows that $\PP(\cC_n)=m$ for
all $n\in\ZZ_+$, and since $\cC$ is the disjoint union of all $\cC_n$
$$
1\ge\PP(\cC)=\sum_{n\in\ZZ_+}\PP(\cC_n)=\sum_{n\in\ZZ_+}m,
$$
which gives  that $m=0$ and hence $\PP(\cC)=0$. Hence,  for any
$\PP\in\cP_\varphi(M)$,  the map $\ell^{-1}:M_0\to\Sigma_0$ induces an
isomorphism between the measurable dynamical systems $(M,\varphi,\PP)$ and
$(\Sigma,\sigma,\PP\circ\ell)$, and the result follows. \hfill\qed

\medskip
It follows from~\eqref{Eq:TentSemiConj} and the previous lemma that  the tent
map $(M,\varphi)$ is a (topological) factor of the full shift $(\Sigma,\sigma)$.
Consequently, the inequalities~\cite[Section~4.4 and
Theorem~9.8~(4)]{Walters1982}
$$
h_\varphi(\PP)\le h_\sigma(\PP\circ\ell^{-1}),\qquad
\fp_\varphi(v)\le\fp_\sigma(v\circ\ell),
$$
hold for all $\PP\in\cP_\varphi(M)$ and all $v\in C(M)$. Since the shift
$\sigma$ is forward expansive and has the specification property~\hS, it
satisfies condition~\hPAP and
$$
\fp_\sigma(v\circ\ell)=\lim_{t\to\infty}\frac1t\log
\sum_{x \in\Fix(\sigma^t)}\e^{v\circ\ell(x)+v\circ\ell\circ\sigma(x)+\cdots
+v\circ\ell\circ\sigma^{t-1}(x)}.
$$
Invoking Lemma~\ref{TentEll}~(5), we can write
\beq
\fp_\varphi(v)\le\fp_{\sigma}(v\circ\ell)
=\lim_{t\to\infty}\frac1t\log\sum_{x \in\Fix(\varphi^t)}\e^{S^tv(x)}.
\label{Tentpressle}
\eeq
\begin{theorem}
\label{TenPAP}
The tent map $(M,\varphi)$ satisfies both the~\hUSCE and the~\hPAP properties.
\end{theorem}

\proof Since the full shift on $\Sigma$ is forward expansive with expansiveness
constant $r=1/2$,  the map $\QQ\mapsto h_\sigma(\QQ)$ is upper
semicontinuous~\cite[Theorem~8.2]{Walters1982}.  We claim that
$\PP_k\circ\ell\rightharpoonup\PP\circ\ell$ in $\cP_\sigma(\Sigma)$ whenever
$\PP_k\rightharpoonup\PP$ in $\cP_\varphi(M)$. The upper semicontinuity of the
map $\PP\mapsto h_\varphi(\PP)$ then follows from Lemma~\ref{TentEll}~(6). To
prove our claim, we have to show that
$\langle f\circ\ell^{-1},\PP_k\rangle\to\langle f\circ\ell^{-1},\PP\rangle$
for all $f\in C(\Sigma)$. By a density argument, it suffices to consider
functions of the form $f(s)=F(s_0,s_1,\ldots,s_{n-1})$ for $n\ge1$. For $x\in
M_0$, we then have $f\circ\ell^{-1}(x)=F\circ S_n(x)$ where
$$
S_n(x)\coloneq\left(1_{[1/2,1]}(x),1_{[1/2,1]}\circ\varphi(x),\ldots,
1_{[1/2,1]}\circ\varphi^{n-1}(x)\right).
$$
Observing that the function $S_n$ is constant on each
$I_m=]m2^{-n},(m+1)2^{-n}[$, we deduce that the set of discontinuities of
$F\circ S_n$ is $\{m2^{-n}\mid m\in\rrbracket0,2^n\llbracket\}\subset\cC$. Since
$\PP(\cC)=0$, applying~\cite[Theorem~5.2]{Billingsley1968}\footnote{Theorem~5.2 of the 1st edition of Billingsley's
book disappeared from the 2nd edition. It is, however, a simple consequence of~\cite[Theorem~2.7]{Billingsley1999}.}
leads to the desired relation
$$
\langle f\circ\ell^{-1},\PP_k\rangle=\int_M F(S_n(x))\PP_k(\d x)\to
\int_M F(S_n(x))\PP(\d x)=\langle f\circ\ell^{-1},\PP\rangle.
 $$

By Lemma~\ref{TentEll}~(4) $\cP_\varphi(M)$ is the image of $\cP_\sigma(\Sigma)$
by the map $\QQ\mapsto\PP=\QQ\circ\ell^{-1}$. Therefore,
$$
\fp_{\sigma}(v\circ\ell)=\sup_{\QQ\in\cP_\sigma(\Sigma)}
\left(\langle v,\QQ\circ\ell^{-1}\rangle+h_\sigma(\QQ)\right)
=\sup_{\PP\in\cP_\varphi(M)}
\left(\langle v,\PP\rangle+h_\sigma(\PP\circ\ell)\right),
$$
and Lemma~\ref{TentEll}~(6) yields that $\fp_{\varphi}(v)=\fp_{\sigma}(v\circ\ell)$.
The result follows from the identity~\eqref{Tentpressle}.\hfill\qed

\medskip
The tent map thus satisfies all the requirements of
Corollary~\ref{cor:per}.\footnote{Recall that~\hS implies~\hED.} Consequently,
the LDPs of Theorem~\ref{thm:LDP-2} and Corollary~\ref{cor:LDP-1} hold for the
corresponding periodic orbit ensemble. Before discussing a related fluctuation
theorem, we shall discuss some interesting features of the pressure map.

\section{Inducing}
\label{sec:inducing}

This section introduces the inducing scheme which will allow us, through the
representation of the tent map dynamics as a shift over a countable alphabet, to
achieve some control over the pressure $\fp_\varphi$. This will give us the opportunity
to discuss some interesting features of the LDP associated to the periodic orbit
ensemble of the tent map.

Setting $I_t\coloneq{}]2^{-t-1},2^{-t}]$ for $t\in\ZZ_+$ and $I_\infty\coloneq\{0\}$ we have
$$
I_{t+1}=[0,1/2]\cap\varphi^{-1}(I_t),
$$
and so, whenever $x\in I_t$ for some $t\ge1$, one has $\varphi^s(x)\in I_{t-s}$
for $s\in\llbracket0,t\rrbracket$, {\sl i.e.,} $t$ is the first return time of
$x$ to $I_0$.

For $x\in M$, let
$$
\tau(x)\coloneq\sum_{n\in\ZZ_+}1_{I_n}(x)(n+1),
$$
and define $\psi:[0,1[{}\to[0,1[{}$ by setting\footnote{Note that
$\tau(0)=0$, so that $\psi(0)=0$. Note also that $\varphi^{\tau(x)}(x)<1$ for
all $x\in M$. See Figure~\ref{Fig:Tent}.} $\psi(x)\coloneq\varphi^{\tau(x)}(x)$. With
$\MM\coloneq\ZZ_+\cup\{\infty\}$, denote by $\fs$ the left shift on the
subshift
$$
\fS\coloneq\left\{\bst=(t_k)_{k\in\ZZ_+}\in\MM^{\ZZ_+}\,\bigg|\, A_{t_kt_{k+1}}=1
\text{ for all }k\in\ZZ_+\right\}
$$
associated to the transition matrix
$$
A_{st}\coloneq\begin{cases}
1&\text{if }s=0\text{ and }t\in\ZZ_+;\\
1&\text{if }s\in\NN\text{ and }t\in\MM;\\
1&\text{if }s=t=\infty;\\
0&\text{otherwise}.
\end{cases}
$$
Define the $\bst$-code through the map $T:[0,1[{}\to\fS$,
$$
T(x)=(t_k(x))_{k\in\ZZ_+},\qquad
t_k(x)\coloneq\begin{cases}
s\in\ZZ_+&\text{ if }\psi^k(x)\in I_s;\\[6pt]
\infty&\text{ if }\psi^k(x)=0,
\end{cases}
$$
and the map $X:\fS\to[0,1[{}$ by
$$
X(\bst)\coloneq\sum_{k\in\ZZ_+}2^{-(t_0+\cdots+t_k)}(-2)^{-k},
$$
where it is understood that $2^{-\infty}=0$.

The following result states that $T=X^{-1}$ provides a bijective coding which conjugates the
induced map $\psi$ to the left shift $\fs$.
\bep\label{PropInducing}
\ben
\item $T\circ\psi=\fs\circ T$.
\item $\psi\circ X=X\circ\fs$.
\item For any $\bst\in\fS$, one has
$$
\{X(\bst)\}=\bigcap_{k \in\ZZ_+}\psi^{-k}(I_{t_k}).
$$
\item $X\circ T=\Id$, in particular, $X$ is surjective and $T$ is injective.
\item $T\circ X=\Id$, in particular, $T$ is surjective and $X$ is injective.
\item  For $t_0,t_1,\ldots,t_{s-1}\in\ZZ_+$ one has\footnote{Here, as above, $\overline{w}$
denotes the periodic repetition of the word $w$.}
\[
X (\overline{t_0t_1\cdots t_{s-1}})
=\frac{\ds\sum_{r=0}^{s-1}2^{-(t_0+\cdots+t_r)}(-2)^{-r}}{1-2^{-(t_0+\cdots+t_{s-1})}(-2)^{-s}},
\]
and for $u\in\ZZ_+$
\[
X (u\overline{t_0t_1\cdots t_{s-1}})
=2^{-u}\left(1-\frac{1}{2}X (\overline{t_0t_1\cdots t_{s-1}}) \right) .
\]
\een
\eep

\proof {\bf (1)} Is an immediate consequence of the definition of the map $T$.

\part{(2)} Using the crude estimate
$$
|X(\bst)|\le\sum_{k\in\ZZ_+}2^{-k}=2
$$
which holds for $\bst\in\fS$, and observing that the inequality is  strict
whenever $\bst\not=\bar0$ while $X(\bar0)=2/3$, we conclude that $|X(\bst)|<2$.
Invoking twice the relation
\beq
X(\bst)=2^{-t_0}\left(1-\frac{1}{2}X\circ\fs(\bst)\right),
\label{EqXshift}
\eeq
one shows that
\beq
X(\bst)\in2^{-t_0}\left(1-2^{-t_1}{}]0,1[{}\right),
\label{EqX1}
\eeq
from which one easily deduces that $X(\bst)\in I_{t_0}$, provided $t_1\not=0$.
In the case $t_1=0$, Relation~\eqref{EqX1} gives $X(\bst)\in{]}0,2^{-t_0}{[}$.
Since we must have $t_2<\infty$, repeated applications of~\eqref{EqXshift} yield
$X\circ\fs^2(\bst)>0$, $X\circ\fs(\bst)<1$, and $X(\bst)>2^{-t_0-1}$, so that
finally, $X(\bst)\in I_{t_0}$ and hence $\tau\circ X(\bst)=t_0+1$. In
particular, using the fact that
$\varphi^{t_0+1}|_{I_{t_0}}:x\mapsto2^{t_0+1}(2^{-t_0}-x)$, it follows
from~\eqref{EqXshift} that, for any $\bst\in\fS$,
$$
\psi\circ X(\bst)=\varphi^{t_0+1}(X(\bst))=2^{t_0+1}(2^{-t_0}-X(\bst))=X\circ\fs(\bst).
$$

\part{(3)} The last relation  gives $\psi^k\circ X(\bst)=X\circ\fs^k(\bst)\in
I_{t_k}$, and hence
\[
X(\bst)=\sum_{k\in\ZZ_+}2^{-(t_0+\cdots+t_k)}(-2)^{-k}\in\bigcap_{k\in\ZZ_+}\psi^{-k}(I_{t_k}).
\]
To prove the stated equality, we have to show that whenever
\beq
\{x,y\}\subset\bigcap_{k \in\ZZ_+}\psi^{-k}(I_{t_k}),
\label{EqXTlink}
\eeq
we must have $x=y$. If $t_0=\infty$, then $x=y=0$. In the case $t_r=\infty$ for
some $r>0$ with finite $t_1,\ldots,t_{r-1}$, it is easy to conclude that in
$\psi^k(x)=\psi^k(y)=2^{-t_k}$ for $k<r$, and in particular $x=y$. Let us
consider now the case where all $t_k$ are finite. From the fact that
$\{x,y\}\subset I_{t_0}$ we deduce
\beq
|x-y|<2^{-t_0-1},\qquad|\psi(x)-\psi(y)|=2^{t_0+1}|x-y|.
\label{Eqpsi1}
\eeq
In the same way, since $\{\psi(x),\psi(y)\}\subset I_{t_1}$, one has
\beq
|\psi(x)-\psi(y)|<2^{-t_1-1},\qquad|\psi^2(x)-\psi^2(y)|=2^{t_1+1}|\psi(x)-\psi(y)|.
\label{Eqpsi2}
\eeq
Combining the last identity of~\eqref{Eqpsi1} with the first inequality
of~\eqref{Eqpsi2} we conclude that $|x-y|<2^{-t_0-t_1-2}$. Iterating this
argument leads to
\beq
t_k(x)=t_k(y) \text{ for }0\le k<n\quad\Longrightarrow\quad
|x-y|<2^{-t_0-t_1-\cdots-t_{n-1}-n}\le2^{-n},
\label{eq:Nikola}
\eeq
and letting $n\to\infty$ yields again that $x=y$.

\part{(4)} Let $x\in[0,1[{}$. By definition
$\bst=(t_k)_{k\in\ZZ_+}=T(x)$ is such that
\beq
x\in\bigcap_{k \in\ZZ_+}\psi^{-k}(I_{t_k}).
\label{Eqxincap}
\eeq
By Part~(3), $y=X\circ T(x)=X(\bst)$ satisfies
$$
y\in\bigcap_{k \in\ZZ_+}\psi^{-k}(I_{t_k}).
$$
Thus, \eqref{EqXTlink} holds and by the proof of Part~(3) one has $y=x$.

\part{(5)} Let $\bst=(t_k)_{k\in\ZZ_+}\in\fS$. By Part~(3), $x=X(\bst)$
satisfies~\eqref{Eqxincap} and hence $T(x)=\bst$.

\part{(6)}  Write
\[
X(\overline{t_0t_1\cdots t_{s-1}})
=\sum_{n\in\ZZ_+}\sum_{k=ns}^{(n+1)s-1}2^{-(t_0+\cdots+t_k)} (-2)^{-k} ,
\]
with $t_j=t_r$ for $j\in\ZZ_+$, $r\in\llbracket0,s-1\rrbracket$ and $j\equiv r$
(mod $s$). Then, setting $k=ns+r$,
\begin{eqnarray*}
X (\overline{t_0t_1\cdots t_{s-1}})
&=&\sum_{n\in\ZZ_+}\sum_{r=0}^{s-1}
2^{-(t_0+\cdots+t_{s-1})n-(t_0+\cdots+t_r)}(-2)^{-ns-r}\\
&=&\sum_{n\in\ZZ_+}\left[2^{-(t_0+\cdots+t_{s-1})}(-2)^{-s}\right]^n
\sum_{r=0}^{s-1}2^{-(t_0+\cdots+t_r)}(-2)^{-r}\\
&=&\frac{\ds\sum_{k=0}^{s-1}2^{-(t_0+\cdots+t_k)}(-2)^{-k}}
{1-2^{-(t_0+\cdots+t_{s-1})}(-2)^{-s}}
\end{eqnarray*}
yields the first formula. The second one now follows from~\eqref{EqXshift}.
\qed

\medskip
\remark[1] The following construction is left as an
exercise. Given the Gray-code $\bss=\ell^{-1}(x)$ of $x\in M_0$, one gets its
$\bst$-code $\bst=T(x)$ by counting the number of $0$'s before the first $1$ and
the number of $0$'s separating the following $1$'s. For example the $\bss$-code
\[
\ell^{-1}(2^{-1/2})=11101111100001101000101010\cdots
\]
yields
\begin{center}
\includegraphics[scale=0.47]{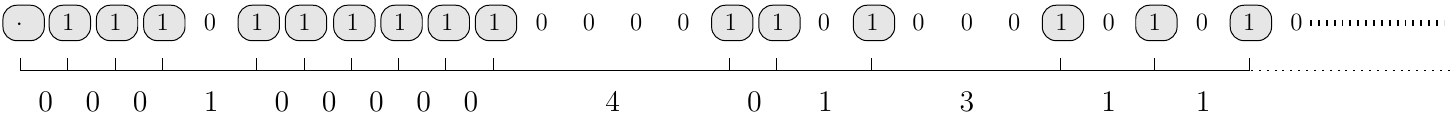}
\end{center}
and hence the corresponding $\bst$-code is $T(2^{-1/2})=00010000401311\cdots$.

\remark[2] From the properties of the intervals $I_t$
($t\in\ZZ_+$) described at the beginning of this section, one easily sees that
the tent map on the non-compact set $M_0$ can also be coded by the {\sl renewal
shift:} the left shift on the space
$$
\fR\coloneq
\left\{\bsu=(u_k)_{k\in\ZZ_+}\in\MM^{\ZZ_+}\,\bigg|\, u_{k+1}=u_k-1\text{ if }u_k>0\right\}.
$$
The $\bsu$-code is defined by
\[
\bsu=U(x)=(u_k(x))_{k\in\ZZ_+},\qquad u_k(x)\coloneq\begin{cases}
s\in\ZZ_+&\text{ if }\varphi^k(x)\in I_s;\\[6pt]
\infty&\text{ if }\varphi^k(x)=0,
\end{cases}
\]
and  is obtained from the $\bst$-code by substituting each $t_k$ by the word
\[t_k(t_k-1)\cdots 10,\]
for example
\[
U(2^{-1/2})=0001000004321001032101010\cdots.
\]

\section{Grand Canonical Picture}
\label{sec:GCP}

Let $v\in C(M)$. Considering the Lebesgue measure $\lambda$ as well as the Dirac
masses $\delta_0$ and $\delta_{2/3}$, Definition~\eqref{PressureDef} and
Formulas~\eqref{eq:Tenthtop}, \eqref{TentRokhlin} yield the bounds
\[
\max(\log 2+\langle v,\lambda\rangle ,v(0),v(2/3))
\le\fp_\varphi(v)\le\max_{x\in M}(\log 2+v(x)).
\]
By Theorem~\ref{TenPAP},
\[
\fp_\varphi(v)=\limsup_{t\to\infty}\frac1t\log
\sum_{x\in\Fix(\varphi^t)}\e^{S_tv(x)},
\]
and since
\[
\sum_{x\in\Fix(\varphi^t)}\e^{S_tv(x)}=\e^{tv(0)}
+\sum_{x\in\Fix(\varphi^t)\setminus\{0\}}\e^{S_tv(x)},
\]
we can write
\begin{equation}
\fp_{\varphi}(v)
=\max\left(v(0),\limsup_{t\to\infty}\frac1t\log
\sum_{x\in\Fix(\varphi^t)\setminus \{0\}}\e^{S_tv(x)}\right) .
\label{TentPressureForm}
\end{equation}
Focusing on the second argument in the above $\max$, let us
denote by $\Orb_t(\varphi)$ the set of $t$-periodic orbits of $\varphi$.
We have
\[
\sum_{x\in\Fix(\varphi^t)}\e^{S_tv(x)}
=\sum_{o\in\Orb_t(\varphi)}|o|\,\e^{\sum_{x\in o}v(x)},
\]
and
\[
\sum_{x\in\Fix(\varphi^t)\cap I_0}\e^{S_tv(x)}
=\sum_{o\in\Orb_t(\varphi)}|o\cap I_0|\,\e^{\sum_{x\in o}v(x)}.
\]
Noticing that $1\le|o\cap I_0|\le|o|\le t$ for any $o\in\Orb_t(\varphi)$
distinct from the orbit of $0$, one easily shows that
\[
\limsup_{t\to\infty}\frac1t\log
\sum_{x\in\Fix(\varphi^t)\setminus\{0\}}\e^{S_tv(x)}=
\limsup_{t\to\infty}\frac1t\log\sum_{x\in\Fix(\varphi^t)\cap I_0}\e^{S_tv(x)}.
\]
To proceed, we further partition the set $\Fix(\varphi^t)\cap I_0$ according to
the total number of visits of the segment $(\varphi^r(x))_{r\in\llbracket
0,t-1\rrbracket}$ to $I_0$, {\sl i.e.,}
\[
\Fix(\varphi^t)\cap I_0=\bigcup_{s=1}^{t}M_{s,t},
\]
where
\[
M_{s,t}\coloneq\left\{\varphi^{\tau(x)-1}(x)\,\bigg|\,x\in\Fix(\psi^s),
\sum_{k=0}^{s-1}\tau\circ\psi^k(x)=t\right\}.
\]
In terms of $\bst$ and $\bsu$-codes, one has
\begin{gather*}
T(M_{s,t})=\left\{\bst=0\overline{t_0\cdots t_{s-1}}\,\bigg|\,\sum_{k=0}^{s-1}t_k=t-s\right\},\\
U(M_{s,t})=\left\{\bsu=0\overline{t_0(t_0-1)\cdots10\cdots t_{s-1}(t_{s-1}-1)\cdots10}\,\bigg|\,\sum_{k=0}^{s-1}t_k=t-s\right\}.
\end{gather*}
Associating to $v$ the function
\beq
\fS\ni\bst=t_0t_1\cdots\mapsto
v^\# (\bst)\coloneq\sum_{r=0}^{t_0} v\circ\varphi^r\circ X(\bst),
\label{hatV}
\eeq
and setting
\beq
S^sv^\#\coloneq\sum_{r=0}^{s - 1}v^\#\circ\fs^r,
\label{TkDef}
\eeq
one easily checks that, for $\bst=0\overline{t_0\cdots t_{s-1}}\in T(M_{s,t})$,
\[
(S_tv)\circ X(\bst)=S^sv^\#(\overline{t_0\cdots t_{s-1}}),
\]
each side of this identity being the sum of the values taken by $v$ on a segment
of length $t$ of the orbit of $X(\bst)$. Thus, we can write
\[
\sum_{x\in\Fix(\varphi^t)\cap I_0}\e^{S_tv (x)}
=\sum_{s=1}^{t}\sum_{\bst\in T(M_{s,t})}\e^{(S_tv)\circ X(\bst)}
=\sum_{s=1}^{t}\sum_{\substack{t_0\ldots,t_{s-1}\in\ZZ_+\\ t_0+\dots+t_{s-1}=t-s}}
\e^{S^sv^\#(\overline{t_0\cdots t_{s-1}})}.
\]
Invoking the Cauchy--Hadamard formula allows us to express the pressure
$\fp_\varphi(v)$ in terms of the radius of convergence $\rho_{\mathrm{gc}}(v)$
of the (grand canonical) partition function
\[
\fP_{\varphi}(v, z)\coloneq\sum_{t=1}^{\infty} z^t
\sum_{x\in\Fix(\varphi^t)\cap I_0}\e^{S_tv(x)},
\]
namely
\beq
\fp_{\varphi}(v)=\max (v(0),-\log\rho_{\mathrm{gc}}(v)) .
\label{pphiForm}
\eeq

To control the function $\fP_{\varphi}$ we shall express it in terms of the
countable Markov shift $(\fS, \fs)$,
\begin{equation}
\begin{aligned}
\mathfrak{P}_{\varphi}(v,z)
&=\sum_{t=1}^{\infty}z^t\sum_{x\in\Fix(\varphi^t)\cap I_0}\e^{S_tv(x)}\\
&=\sum_{t=1}^{\infty}z^t\sum_{s=1}^t\sum_{\bst\in T(M_{t,s})}\e^{S^sv^\#(\bst)}\\
&=\sum_{s=1}^{\infty}z^s\sum_{t=s}^\infty z^{t-s}\sum_{\bst\in T(M_{t,s})}\e^{S^sv^\#(\bst)}\\
&=\sum_{s=1}^{\infty}z^s\,\Xi_s (v,z),
\end{aligned}
\label{GCseries}
\end{equation}
with
\beq
\Xi_s(v,z)\coloneq\sum_{\bst\in\Fix(\fs^s)}z^{t_0+\cdots+t_{s-1}}\e^{S^sv^\#(\bst)}
=\sum_{t_0,\ldots,t_{s-1}\in\ZZ_+}
z^{t_0+\cdots+t_{s-1}}\e^{S^sv^\#(\overline{t_0\cdots t_{s-1}})}.
\label{XIdef}
\eeq

Besides this standard combinatorial device, we need to impose some more stringent
regularity conditions on the potential $v$. The $k$-variation of a function
$f:\fS\to\RR$ is defined as
\beq
\mathrm{Var}_k (f)\coloneq\sup\left\{|f (\bst)-f (\bss)|\,
\big|\,\bst,\bss\in\fS,t_j =s_j\text{ for } 0\le j < k\right\} .
\label{eq:Varkdef}
\eeq
This function has summable variations whenever
\[
\mathrm{Var}(f)\coloneq\sum_{k\ge1}\mathrm{Var}_k (f)<\infty,
\]
and is $\theta$-Hölder, for $\theta\in{]}0,1{[}$, whenever
\[
\sup_{k \ge1}\,\theta^{-k}\mathrm{Var}_k (f)<\infty.
\]
Clearly, the latter condition implies the former.

\bep\label{PropGC}
Suppose that the potential $v\in C(M)$ is such that $v^\#$ has summable variations.
\ben
\item For $s\in\NN$, the radius of convergence of the series~\eqref{XIdef}
is given by $\e^{-v(0)}$. The limit
\beq
p(v,\mu)\coloneq\lim_{s\to\infty}\frac1s\log\Xi_s(v,\e^{-\mu})
\label{rhoDef}
\eeq
exists for $\mu>v(0)$ and defines a non-increasing function of $\mu$.
\item The function $\mu\mapsto p(v,\mu)-\mu$ is strictly decreasing, moreover
\beq
\lim_{\mu\to+\infty}p(v,\mu)-\mu
=\inf_{\mu>v(0)}p(v,\mu)-\mu
=-\infty,
\label{Xikzero}
\eeq
and
\beq
\Delta(v)\coloneq\lim_{\mu\downarrow v(0)}p(v,\mu)-\mu
=\sup_{\mu>v(0)}p(v,\mu)-\mu\in{]}{-}\infty,+\infty].
\label{DeltaDef}
\eeq
\item The following alternative holds:
\begin{itemize}
\item  either
$\Delta(v)\le0$, $\rho_\mathrm{gc}(v)\ge\e^{-v(0)}$, and $\fp_{\varphi}(v)=v(0)$;
\item or $\Delta(v)>0$,
$$
\rho_\mathrm{gc}(v)=\e^{-\mu^\ast}
\coloneqq\inf\{\e^{-\mu}\mid p(v,\mu)-\mu>0\}<\e^{-v(0)},
$$
and $\fp_\varphi(v)=\mu^\ast>v(0)$.
\end{itemize}
\item The function $\kappa\mapsto\Delta(\kappa v)$ is convex and takes the value
$+\infty$ at $\kappa=0$. If $v^\#$ is $\theta$-Hölder for some
$\theta\in{]}0,1{[}$, then
\[
\kappa_\mathrm{c}(v)\coloneqq\sup\{\kappa\ge0\mid\Delta(\kappa v)>0\}\in{]}0,\infty],
\]
and the function $\kappa\mapsto\fp_{\varphi}(\kappa v)$ is real-analytic on the interval
$]0,\kappa_\mathrm{c}(v)[$.
\een
\eep

\remark In the dynamical system literature on countable Markov shifts, $p(v,0)$ is
called \ndex{Gurevich pressure} of the potential $v^\#$ for the induced system
$$
p(v,0)\eqcolon P_G(v^\#).
$$
Thus, the function $\mu\mapsto p(v,\mu)-\mu$, which plays a central role in the
preceding result, is related to the Gurevich pressure by
$$
P_G((v-\mu)^\#)=p(v,\mu)-\mu.
$$
In particular
$$
\Delta(v)=P_G((v-v(0))^\#)
$$
is the so-called {\sl discriminant} of the potential $v$
(see~\cite{Sarig2001,Pesin2006}). In the cases where $\Delta(v)>0$, the
potential $v$ is said to be {\sl positive recurrent.}

\proof {\bf(1)} We first consider the case $s=1$. The radius of convergence of the
series
$$
\Xi_1(v,z)=\sum_{t=0}^\infty z^t \e^{v^\#(\bar t)},
$$
is given by the Cauchy--Hadamard formula
$$
\varrho=\liminf_{t\to\infty}\,\e^{-v^\#(\bar t)/t}.
$$
By~\eqref{hatV} and Proposition~\ref{PropInducing}~(6) we have
\beq
\frac1t v^\#(\bar t)=\frac1t\sum_{s=0}^tv\circ\varphi^s\circ X(\bar t)
=\frac1t\sum_{s=0}^tv\left(\frac{2^{-s}}{1+2^{-t-1}}\right).
\label{EqDropZone}
\eeq
Noticing that
$$
\max_{s\in\llbracket0,t\rrbracket}\left|\frac{2^{-s}}{1+2^{-t-1}}-2^{-s}\right|\le2^{-t},
$$
and invoking the uniform continuity of $v$ and Cesàro's limit theorem we deduce
from~\eqref{EqDropZone} that
\beq
\lim_{t\to\infty}\frac1t v^\#(\bar t)=v(0).
\label{vhashlim}
\eeq
This yields $\varrho=\e^{-v(0)}$, as claimed.

To deal with the general case $s\in\NN$, observe that
\[
\Xi_{s+s'}(v,\e^{-\mu})
=\sum_{\bst\in\ZZ_+^s}\e^{-\mu(t_0+\cdots+t_{s-1})}
\sum_{\bsr\in\ZZ_+^{s'}}\e^{-\mu(r_0+\cdots+r_{s'-1})}
\e^{S^sv^\#(\overline{\bst\bsr})}\e^{S^{s'}v^\#(\overline{\bsr\bst})}.
\]
Since
\[
S^sv^\#(\bar\bst)-\sum_{j=1}^s\mathrm{Var}_j (v^\#)
\le S^sv^\#(\overline{\bst\bsr})
\le S^sv^\#(\bar\bst)+\sum_{j = 1}^s\mathrm{Var}_j(v^\#),
\]
setting $a_s\coloneq\log\Xi_s(v,\e^{-\mu})$ and $c\coloneq2\mathrm{Var}(v^\#)$, we derive
\beq
a_s+a_{s'}-c\le a_{s+s'}\le a_s+a_{s'}+c,
\label{SubSuper}
\eeq
from which one easily infers that
$$
\Xi_1(v,\e^{-\mu})^s\e^{-c(s-1)}\le\Xi_s(v,\e^{-\mu})
\le\Xi_1(v,\e^{-\mu})^s\e^{c(s-1)},
$$
and hence that the radius of convergence of the series~\eqref{XIdef} is equal to
$\e^{-v(0)}$ for all $s\in\NN$. The existence of the limit~\eqref{rhoDef}
follows from the sub/super-additivity properties expressed in ~\eqref{SubSuper}.
Fekete's lemma~\cite[Part~I, Problem~99]{Polya1978} provides the explicit
formulas
$$
p(v,\mu)
=\inf_{s\in\NN}\frac1s\left(\log\Xi_s(v,\e^{-\mu})+c\right)
=\sup_{s\in\NN}\frac1s\left(\log\Xi_s(v,\e^{-\mu})-c\right).
$$
Thus
\beq
\left|\log\Xi_s(v,\e^{-\mu})-sp(v,\mu)\right|\le c,
\label{vege}
\eeq
and it follows from~\eqref{GCseries} that, for $|z|=\e^{-\mu}<\rho_{\mathrm{gc}}(v)$,
\beq
\frac{\e^{-c}}{1-\e^{p(v,\mu)-\mu}}
\le\left|\mathfrak{P}_{\varphi}(v,z)\right|\le\frac{\e^{c}}{1-\e^{p(v,\mu)-\mu}}.
\label{eq:Pbrack}
\eeq
As a direct consequence of~\eqref{XIdef}, the (positive) function
$\mu\mapsto\Xi_s(v,\e^{-\mu})^{1/s}$ is non-increasing.
By Definition~\eqref{rhoDef}, the same is true of $\mu\mapsto p(v,\mu)$,

\part{(2)} Given the latter property, the claimed strict monotony is obvious.
Invoking~\eqref{vege}, the simple fact that $\Xi_s(v,0)^{1/s}=\e^{v(2/3)}$
yields
$$
\lim_{\mu\to\infty}p(v,\mu)\le
\lim_{\mu\to\infty}\frac1s\log\Xi_s(v,\e^{-\mu})+\frac{c}s
=\frac1s\log\Xi_s(v,0)+\frac{c}s=v(2/3)+\frac{c}s,
$$
which proves~\eqref{Xikzero}. Starting again with~\eqref{vege}, and keeping only
the first term in the series $\Xi_1$, we derive
$$
p(v,\mu)-\mu\ge \log\Xi_1(v,\e^{-\mu})-\mu-c\ge v(2/3)-\mu-c,
$$
which, in the limit $\mu\downarrow v(0)$, proves~\eqref{DeltaDef}.

\part{(3)} Suppose that $\Delta(v)\le0$, fix $\mu>v(0)$ so that
$p(v,\mu)-\mu=-\delta<0$, and notice that it follows from~\eqref{eq:Pbrack} that
$$
\sup_{|z|<\e^{-\mu}}\left|\mathfrak{P}_{\varphi}(v,z)\right|
\le\frac{\e^{c}}{1-\e^{-\delta}}<\infty,
$$
from which we infer that $\rho_{\mathrm{gc}}(v)\ge\e^{-\mu}$.
Letting $\mu\downarrow v(0)$ gives $\rho_{\mathrm{gc}}(v)\ge\e^{-v(0)}$
and~\eqref{pphiForm} allows us conclude that $\fp_{\varphi}(v)=v(0)$.

In the case  $\Delta(v)>0$,  $\mu^\ast=\sup\{\mu\mid v(0)<\mu< p(v,\mu)\}>v(0)$
and for any $v(0)<\mu<\mu^\ast$ one has $p(v,\mu)-\mu>0$ and therefore
\[
\limsup_{s\to\infty}\left(\e^{-\mu s}\Xi_s(v,\e^{-\mu})\right)^{1/s}=\e^{p(v,\mu)-\mu}>1.
\]
By Cauchy's root test, the grand canonical partition function~\eqref{GCseries} diverges for
$|z|\ge\e^{-\mu}$. Letting $\mu\uparrow\mu^\ast$, we derive
$\rho_\mathrm{gc}(v)\le\e^{-\mu^\ast}$. Reciprocally, if $\mu>\mu^\ast$, then
\[
p(v,\mu)-\mu<0,
\]
and hence
\[
\limsup_{s\to\infty}\left(\e^{-\mu s}\Xi_s(v,\e^{-\mu})\right)^{1/s}<1.
\]
Invoking again the root test, we conclude that~\eqref{GCseries} converges for
$|z|\le\e^{-\mu}$. Letting $\mu\downarrow\mu^\ast$, we derive
$\rho_\mathrm{gc}(v)\ge\e^{-\mu^\ast}$.

\part{(4)} For reasons of space, we only sketch the proof. The reader should
consult~\cite[Theorem~5]{Sarig2001} and references therein for more details.

For $\kappa=0$, an elementary calculation gives $p(0,\mu)-\mu=-\log(\e^\mu-1)$,
from which we conclude
$$
\Delta(0)=\sup_{\mu>0}\log\frac1{\e^\mu-1}=+\infty.
$$
By~\eqref{rhoDef} and~\eqref{DeltaDef}, we have
\[
\Delta(\kappa v)=\sup_{\nu>0}\lim_{s\in\NN}
\left(\frac1s\log\Xi_s(\kappa v,\e^{-\nu-\kappa v(0)})-\nu-\kappa v(0)\right).
\]
Observing that
$\Xi_s(\kappa v,\e^{-\nu-\kappa v(0)})=\e^{s\kappa v(0)}\Xi_s(\kappa(v-v(0)),\e^{-\nu})$
yields
$$
\Delta(\kappa v)
=\sup_{\nu>0}\lim_{s\in\NN}\left(\frac1s\log\Xi_s(\kappa(v-v(0)),\e^{-\nu})-\nu\right).
$$
From~\eqref{XIdef} we infer that the function
$\kappa\mapsto\Xi_s(\kappa(v-v(0)),\e^{-\nu})$ is log-convex, so that the
previous identity immediately implies the convexity of
$\kappa\mapsto\Delta(\kappa v)$. We also note, for later reference, that this
function is non-increasing provided $\max_{x\in M}v(x)=v(0)$.

By~\eqref{vhashlim}, given $\varepsilon>0$ there is $N$ such that, for any
$n\ge N$ and $\bst\in\fS$ one has
\begin{align*}
\mu n-\kappa v^\#(n\bst)&\ge\mu n-\kappa v^\#(\bar n)-\kappa\mathrm{Var}_1(v^\#)\\
&\ge\left(\mu-\kappa v(0)\right)n-\left|\frac1n v^\#(\bar n)-v(0)\right|n-\kappa\mathrm{Var}_1(v^\#)\\
&\ge\left(\mu-\kappa v(0)-\varepsilon\right)n-\kappa\mathrm{Var}_1(v^\#).
\end{align*}
It follows that the Ruelle--Perron--Frobenius transfer operator
$L_{\kappa,\mu}$, defined on $C_b(\fS)$ by
\[
(L_{\kappa,\mu} f)(\bst)\coloneq\sum_{n\in\ZZ_+}\e^{-\mu n+\kappa v^\#(n\bst)}f(n\bst),
\]
satisfies
$$
\|L_{\kappa,\mu}1\|_\infty<\infty
$$
provided $\mu>\kappa v(0)$. It $v^\#$ is $\theta$-Hölder and  $\mu>\kappa v(0)$,
we deduce that  $L_{\kappa,\mu}$ acts as a bounded operator on the Banach space
\[
\cL_\theta\coloneq\left\{
f\in C_b(\fS)\,\bigg|\,\|f\|_\theta<\infty\right\},
\]
where 
\[
\|f\|_\theta=\|f\|_\infty
+\sup_{\bst\not=\bss}\theta^{-\inf\{k\in\ZZ_+\mid t_k\not=s_k\}}|f(\bst)-f(\bss)|.
\]
Since, for any $\bst\in\fS$ and $s\in\NN$,
\begin{align*}
  (L^s_{\kappa,\mu}1)(\bst)&=\sum_{t_0,\ldots,t_{s-1}\in\ZZ_+}
\e^{-\mu(t_0+\cdots+t_{s-1})+\kappa S^sv^\#(t_0\cdots t_{s-1}\bst)}\\
&=\Xi_s(\kappa v,\e^{-\mu})(1+o(s)),
\end{align*}
the Perron--Frobenius theorem yields that the logarithm of the spectral radius
of the  operator $L_{\kappa,\mu}$ is
\[
p(\kappa v,\mu)=\limsup_{n\to\infty}\frac1n\log(L_{\kappa,\mu}^n1)(\bst).
\]
The crucial point is that one can show that $\e^{p(\kappa,\mu)}$ is a simple
discrete eigenvalue of  $L_{\kappa,\mu}$, so that the claimed analyticity
follows from regular perturbation theory.

\qed

\section{Phase Transitions}
\label{sec:PhaseTransitions}

Whenever the critical coupling $\kappa_\mathrm{c}(v)$ is finite,
Proposition~\ref{PropGC} implies that the function
$\kappa\mapsto\fp_\varphi(\kappa v)$ is real-analytic on
${]}0,\kappa_\mathrm{c}(v){[}\cup{]}\kappa_\mathrm{c}(v),\infty{[}$, but fails
to be analytic at $\kappa=\kappa_\mathrm{c}(v)$. In thermodynamics, such a
phenomenon is related to phase transition and is accompanied by a multiplicity
of the equilibrium states. In this section we elaborate on  the occurrence of
such a phase transition for the tent map.

To this end, we shall consider lower and upper bounds on the partition function.
For $v\in C(M)$ and $t\in\ZZ_+$, let
\beq
v_-(t)\coloneq\inf_{x\in I_t}v(x),\qquad
v_+(t)\coloneq\sup_{x\in I_t}v(x).
\label{eq:wpmDef}
\eeq From the definition~\eqref{hatV} and Proposition~\ref{PropInducing} we
deduce, for $\bst\in\fS$,
\[
\underline{v}(t_0)\coloneq\sum_{s=0}^{t_0}v_-(s)\le v^\#(\bst)
\le\sum_{s=0}^{t_0}v_+(s)\eqcolon\overline{v}(t_0),
\]
and hence, by~\eqref{TkDef},
\[
\sum_{j=0}^{s-1} \underline{v}(t_j)
\le S^sv^\#(\bst)\le\sum_{j=0}^{s-1}\overline{v}(t_j).
\]
By Definition~\eqref{XIdef},
\[
\left(\sum_{t\in\ZZ_+}\e^{-\mu(t+1)+\kappa\underline{v}(t)}\right)^s
\le\e^{-\mu s}\Xi_s(\kappa v,\e^{-\mu})
\le\left(\sum_{t\in\ZZ_+}\e^{-\mu(t+1)+\kappa\overline{v}(t)}\right)^s,
\]
and we derive from~(\ref{rhoDef}--\ref{DeltaDef}) that
$$
\log\left(\sum_{t\in\ZZ_+}\e^{-\mu(t+1)+\kappa\underline{v}(t)}\right)
\le p(\kappa,\mu)-\mu
\le\log\left(\sum_{t\in\ZZ_+}\e^{-\mu(t+1)+\kappa\overline{v}(t)}\right).
$$
In particular, setting $w=v-v(0)$, we obtain
\begin{equation}
\log\left (\sum_{t\in\ZZ_+}\e^{\kappa\underline{w}(t)}\right)
\le\Delta(\kappa)
\le\log\left(\sum_{t\in\ZZ_+}\e^{\kappa\overline{w}(t)} \right) .
\label{DeltaBracket}
\end{equation}
The following are elementary consequences of these inequalities.
\bel\label{lem:phaseornot}
If\, $\kappa_\mathrm{c}(v)<\infty$, then
\[
\lim_{t\to\infty}\underline{v}(t)=-\infty.
\]
Moreover, if\/ $\overline{v}(t)<0$ for any $t\in\ZZ_+$ and
$$
\sum_{t\in\ZZ_+}\e^{\kappa\overline{v}(t)}<\infty
$$
for some $\kappa>0$, then $\kappa_\mathrm{c}(v)<\infty$.

\eel

\section{A Map of the Square}
\label{sec:square}

One easily checks that the map of the unit interval defined by
$$
x=0.b_0b_1b_2b_3\cdots\mapsto \theta(x)\coloneq0.\hat b_0b_1\hat b_2b_3\cdots,
$$
is involutive and satisfies $\varphi\circ\theta=\theta\circ\varphi$.
However, this map is discontinuous on the dense set $\cC$.
In fact, there is no non-trivial {\sl continuous} involution
$\theta$ of the unit interval satisfying Assumption~\hSym for the tent map.
Indeed, injectivity forces such a $\theta$ to be either strictly increasing or
strictly decreasing. The only increasing involution is the identity, and hence
is trivial.\footnote{A continuous $\theta:[0,1]\to[0,1]$ is involutive iff its
graph is symmetric w.r.t.\;the diagonal.}\;If $\theta$ is a decreasing
involution, then surjectivity demands $\theta(0)=1$ and since
$$
\varphi\circ\theta(0)=\varphi(1)=0\neq1=\theta(0)=\theta\circ\varphi(0),
$$
Assumption~\hSym fails.\footnote{One faces a similar issue in the study of Anosov systems, see~\cite[Section~6]{Maes2003b}}

Thus, even though the tent map satisfies all the requirements for
the results of Section~\ref{sec:FT for Periodic Ensemble} to
apply,\footnote{Recall the discussion at the end of
Section~\ref{Ssec:TentSymbolic}.} the only possible application turns out to
be trivial. To reach a non-trivial conclusion, we shall consider the map of the
unit square $M\times M$ defined by $\phi(x,y)=(\varphi(x),\varphi(y))$, for which all
the requirements of Corollary~\ref{cor:per} also hold. In this case,
Assumption~\hSym is verified by the involution
$$
\theta(x,y)\coloneq(y,x).
$$
Given $v_1,v_2\in C(M)$, the potential
$$
v(x,y)\coloneq v_1(x)+ v_2(y),
$$
and given the associated periodic orbit ensemble, the empirical measure
$$
\mu_t^{x,y}\coloneq\frac1t\sum_{0\le s< t}\delta_{\phi^s(x,y)},
$$
satisfies the LDP with rate
$$
\II(\QQ)=\fp_\phi(v)-\langle v,\QQ\rangle-h_\phi(\QQ).
$$
The conclusions of Theorem~\ref{thm:per FT}, and in particular the FR
$$
\II(\widehat{\QQ})=\II(\QQ)+\ep(\QQ),
$$
hold, with the entropy production
$$
\ep(\QQ)=\int(v_1-v_2)(x)\QQ(\d x,\d y)-\int(v_1-v_2)(y)\QQ(\d x,\d y).
$$
Since $\Fix(\phi^t)=\Fix(\varphi^t)\times\Fix(\varphi^t)$, one has
$$
\fp_\phi(v)=\lim_{t\to\infty}\frac1t\log\sum_{x,y\in\Fix(\varphi^t)}\e^{S_tv_1(x)+S_tv_2(y)}
=\fp_\varphi(v_1)+\fp_\varphi(v_2),
$$
from which we deduce that the entropic pressure is given by
$$
e(\alpha)=\fp_\varphi((1-\alpha)v_1+\alpha v_2)+
\fp_\varphi((1-\alpha)v_2+\alpha v_1)
-\fp_\varphi(v_1)-\fp_\varphi(v_2).
$$

\section{Bracketing the Critical Point and the Pressure}

We finish this chapter with an analysis of a specific potential
exhibiting phase transitions, {\sl i.e.,} singularities of the pressure function.
Numerical calculation of the critical point $\kappa_c$ and of the pressure is a
delicate matter. Rigorous estimates on these quantities are more tractable.
It is fairly easy to get lower and upper bounds on the partition function and
to derive from these estimates bounds on the critical point and the pressure.

The \ndex{Hurwitz zeta function} is defined by the convergent expansion
\begin{equation}
\zeta(p,z)\coloneq\sum_{n\in\ZZ_+}(n+z)^{-p},\qquad
\Re p>1, \Re z> 0,
\label{HurwitzDef}
\end{equation}
and extends to an analytic function on
$(p,z)\in\CC^2\setminus(\{1\}\times{]}{-}\infty,0{]})$. The special value
$\zeta(p,1)$ coincides with \ndex{Riemann zeta function} $\zeta(p)$ on
$\CC\setminus\{1\}$. The functions
\begin{equation}
D_-(\kappa,\mu)\coloneq\sum_{n\in\ZZ_+}\e^{-\mu(n+1)-\kappa(\zeta(p)-\zeta (p, n + 3)-1)}-1, \label{DminusDef}
\end{equation}
\begin{equation}
D_+(\kappa,\mu)\coloneq\sum_{n\in\ZZ_+}\e^{- \mu(n+1)-\kappa(\zeta (p)-\zeta (p,n+2))}-1, \label{DplusDef}
\end{equation}
play the central role in

\begin{lemma}
Let the potential $v\in C (M)$ be given by
\begin{equation}
v(x)\coloneq-(1-\log_2x)^{-p},
\label{Vdef}
\end{equation}
with constant $p \in] 0, 1 [$.
\ben
\item The induced function $v^\#$ defined in~\eqref{hatV} is $\theta$-Hölder for
$\theta\ge1/2$, and hence has summable variations, more precisely
\[
\mathrm{Var}_k(v^\#)\le\frac{2p}{\log2}\zeta(p+1,2)2^{-k},\qquad
\mathrm{Var}(v^\#)\le\frac{3}{2}\frac{p}{\log2}\zeta(p+1,2).
\]
\item The function $(\kappa,\mu)\mapsto D_{\pm}(\kappa,\mu)$ is
analytic in $\cO\coloneq\{(\kappa,\mu)\in \CC^2\mid\Re(\mu)>0\}$ and continuous on
$\overline{\cO}\cap\{(\kappa,\mu)\in \CC^2\mid\Re(\kappa)>0\}$. The function
$\kappa\mapsto D_{\pm}(\kappa,0)$ is strictly decreasing on ${]}0,\infty{[}$,
diverges to $+\infty$ as $\kappa\downarrow 0$ and tends to $-1$ as $\kappa\rightarrow+\infty$.
In particular, it has a unique zero $\kappa_{\pm}(p)\in{]}0,\infty{[}$.
For fixed $\kappa<\kappa_{\pm}(p)$, the function
$\mu\mapsto D_{\pm}(\kappa,\mu)$ is strictly decreasing on ${]}0,\infty{[}$, tends to
$D_{\pm}(\kappa,0)>0$ as $\mu\downarrow0$ and to $-1$ as $\mu\rightarrow+\infty$.
In particular, it has a unique zero $\mu_{\pm}(p,\kappa)\in{]}0,\infty{[}$.
\item The following estimates hold for the critical coupling strength
$\kappa_c(p)$ and the pressure $\fp_{\varphi}(\kappa V)$,
\[
\kappa_-(p)\le\kappa_c(p)\le\kappa_+(p),
\]
\[
\min(\mu_-(p,\kappa),\mu_+(p,\kappa))
\le\fp_{\varphi}(\kappa v)\le\max(\mu_-(p,\kappa),\mu_+(p,\kappa))
\]
for $\kappa\le\kappa_c(p)$.
\een
\end{lemma}

\proof {\bf (1)} By~\eqref{hatV} and~\eqref{eq:Varkdef} one has
\[
\mathrm{Var}_k(v^\#)\le\sup\left\{\sum_{q=0}^{t_0(x)}
|v\circ\varphi^q(x)-v\circ\varphi^q(y)|
\,\bigg|\, x,y\in M, t_j(x)=t_j(y)\text{ for } 0\le j< k\right\}.
\]
Taking into account the convexity of $v$, we can write
\[
\sup_{x,y\in I_t}\frac{|v(x)-v(y)|}{|x-y|}\le|v'(2^{-t-1})|
=\frac{p}{\log 2}\frac{2^{t+1}}{(t+2)^{p+1}},
\]
so that, for $t_0(x)=t_0(y)$ and $0\le q\le t_0(x)$,
$$
|v\circ\varphi^q(x)-v\circ\varphi^q(y)|
\le\frac{p}{\log 2}\frac{2^{t_0(x)-q+1}}{(t_0(x)-q+2)^{p+1}}
|\varphi^q(x)-\varphi^q(y)|.
$$
For $x,y\in M$ such that $t_j(x)=t_j(y)$ for $0\le j<k$,
taking~\eqref{eq:Nikola} into account, we derive
\begin{align*}
|v\circ\varphi^q(x)-v\circ\varphi^q(y)|
&\le\frac{p}{\log 2}\frac{2^{t_0(x)-q+1}}{(t_0(x)-q+2)^{p+1}}2^{-(t_0(x)-q+t_1(x)+\cdots+t_{k-1}(x)+k)}\\
&\le\frac{2p}{\log 2}\frac{1}{(t_0(x)-q+2)^{p}}2^{-k},
\end{align*}
and hence
\[
\sum_{q=0}^{t_0(x)}|v\circ\varphi^q(x)-v\circ\varphi^q(y)|
\le\frac{2p}{\log2}2^{-k}\sum_{q=2}^{t_0(x)+2}\frac{1}{q^{p+1}}\le\frac{2p}{\log2}\zeta(p+1,2)2^{-k},
\]
{\sl i.e.,}
\[
\mathrm{Var}_k(v^\#)\le\frac{2p}{\log2}\zeta(p+1,2)2^{-k}.
\]
Summing over $k\ge1$ yields
\[
\mathrm{Var}(v^\#)\le\frac{2p}{\log2}\zeta(p+1,2).
\]

\part{(2)} We consider $D_+$, the analysis of $D_-$ is similar.
The relation
\begin{equation}
\zeta(p)-\zeta(p,m+1)=\sum_{k=1}^\infty k^{-p}-\sum_{k=0}^\infty(k+m+1)^{-p}
=\sum_{k=1}^m k^{-p},
\label{ZetaForm}
\end{equation}
which, in view of~\eqref{HurwitzDef}, holds for $\Re p>1$ and $m\ge1$,
extends by analyticity to any $p\in\CC\setminus\{1\}$ and $m\ge0$. Thus
\[
\phi_n\coloneq\zeta(p)-\zeta(p,n+2)=\sum_{k=1}^{n+1}k^{-p}
\ge\int_1^{n+2}x^{-p}\d x=\frac{(n+2)^{1-p}-1}{1-p},
\]
and the series in Definition~\eqref{DplusDef} admits the majorant
\[
\sum_{n=0}^\infty\e^{-(n+1)\Re\mu-\phi_n\Re\kappa},
\]
which is convergent whenever $\Re\mu>0$ or $\Re\mu=0$ and $\Re\kappa>0$. This
shows that $D_+$ is analytic in $\cO$ and continuous in
$\overline{\cO}\cap\{(\kappa,\mu)\in \CC^2\mid\Re(\kappa)>0\}$. The remaining
statements are simple consequences of the fact that the above majorant coincides
with the series in~\eqref{DplusDef} for real values of $\kappa$ and $\mu$.

\part{(3)} For the potential~\eqref{Vdef}, Definition~\eqref{eq:wpmDef} gives
\[
v_-(t)=-(t+1)^{-p}, \qquad
v_+(t)=-(t+2)^{-p},
\]
so that
\[
\underline{v}(t)=-\sum_{n=1}^{t+1}n^{-p}
\]
and
\[
\overline{v}(t)=1-\sum_{n=1}^{t+2}n^{-p}.
\]
Invoking Relation~\eqref{ZetaForm}, we can write
\begin{align*}
\underline{v}(t)&=-(\zeta (p)-\zeta(p,t+2)),\\[4pt]
\overline{v}(t)&=-(\zeta(p)-1-\zeta(p,t+3)),
\end{align*}
and, in view of~\eqref{DeltaBracket}, \eqref{DminusDef} and~\eqref{DplusDef},
$$
\log(1+D_-(\kappa,0))\le\Delta(\kappa v)\le\log(1+D_+(\kappa,0)),
$$
and more generally
$$
\log(1+D_-(\kappa,\mu))\le p(\kappa v,\mu)-\mu\le\log(1+D_+(\kappa,\mu)).
$$
According to the discussion in the previous section, the critical point
$\kappa_c (p)$ is determined by the condition $\Delta (\kappa )=0$. Since
the function $\kappa \mapsto \Xi_k (\kappa v, r)$ is obviously non-increasing,
recalling that
$$
\Delta(\kappa v)=\lim_{r\uparrow1}\liminf_{k\to\infty}\;\Xi_k(\kappa,r)^{-k},
$$
we deduce that the function $\kappa\mapsto\Delta(\kappa v)$ is
non-decreasing and it follows that
$$
\kappa_-(p) \le\kappa_c(p)\le\kappa_+(p).
$$
The same analysis applies to the pressure $\fp_{\varphi} (\kappa v)$,
which is the zero of the function $\mu\mapsto\rho(\kappa v,\mu)\e^{\mu}-1$
and is bracketed by the zeros of $D_{\pm}(\kappa,\mu)$.\hfill\qed

\section{Numerical Calculation of the Pressure}

Direct application of~\eqref{TentPressureForm} for numerical evaluation
converges slowly and, due to the exponential growth $|M_t|=2^t$, becomes
impractical. Selecting the subsequence of prime $n$'s largely improves the
situation. For prime $t$, the set $\mathrm{Fix}(\varphi^t)$ decomposes into the
two fixed points $0$, $2/3$ and periodic points of {\sl primitive period} $t$.
Thus
\[
\sum_{x\in\mathrm{Fix}(\varphi^t)}\e^{\kappa S_tv(x)}
=\e^{t\kappa v(0)}+\e^{t\kappa v(2/3)}+\sum_{x\in\mathrm{Per}_t}\e^{\kappa S_tv(x)},
\]
where $\mathrm{Per}_t$ denotes the set of points with primitive period $t$.
Denoting by $\Orb'_t$ the set of periodic orbits of primitive period $t$, we get
\[
\sum_{x\in\mathrm{Per}_t}\e^{\kappa S_tv(x)}=
\sum_{o\in\Orb'_t}\e^{\kappa v(o)},
\]
where
\[
v(o)\coloneq\sum_{x\in o}v(x).
\]
Thus
\[
\fp_{\varphi} (\kappa v)=\lim_{\substack{t\to\infty\\ t\text{ prime}}}
\frac1t\log\left(\e^{t\kappa v(0)}+\e^{t\kappa v(2/3)}
+t\sum_{o\in\Orb'_t}\e^{\kappa v(o)}\right),
\]
and hence
\[
\fp_{\varphi}(\kappa v)=\max\left(\kappa v(0),\kappa v(2/3),
\lim_{\substack{t\to\infty\\ t\text{ prime}}}
\frac1t\log\left(t\sum_{o\in\Orb'_t}\e^{\kappa v(o)}\right)\right).
\]
It turns out that the sequence
\[
\text{prime } t \mapsto p_t (\kappa)
\coloneq\frac1t\log\left(t\sum_{o\in\Orb'_t}\e^{\kappa v(o)}\right)
\]
is rapidly converging; The plots of $p_{11} (\kappa)$ and $p_{13} (\kappa)$
are indistinguishable for $\kappa\le1/2$, and slight departures near
$\kappa_c (p)$ appear for $\kappa>1/2$, see Figure~\ref{Fig:Convergence}.
Consistency with our rigorous lower/upper bounds is clearly displayed in
Figure~\ref{Fig:Consistency}; observe that the curvature is  small, but
increases with $p$.
The critical point $\kappa_c(p)$ is easily computed from this approximation and
Figure~\ref{Fig:kappacritical} shows the result of this calculation.

Going back to the square map of Section~\ref{sec:square}, consider the potential
$$
v(x,y)=\kappa_c(p_1)v_1(x)+\kappa_c(p_2)v_2(y),
$$
where $v_1$ and $v_2$ are of the form~\eqref{Vdef} with respective exponents
$p_1,p_2\in{]}0,1{[}$.

Figure~\ref{Fig:EntropicPressure} shows the entropic pressure $e(\alpha)$
and the associate rate $I(s)$ for $(p_1,p_2)=(0.9,0.1)$,
illustrating the occurrence of phase transitions. The singularities of the entropic
pressure materialize as jumps in its first derivative, located on the vertical lines.
To each of them corresponds a segment of unexposed points of the rate.

\printbibliography[heading=bibintoc,title={References}]

\newpage

\begin{figure}
\includegraphics[width=1.1\textheight,angle=90]{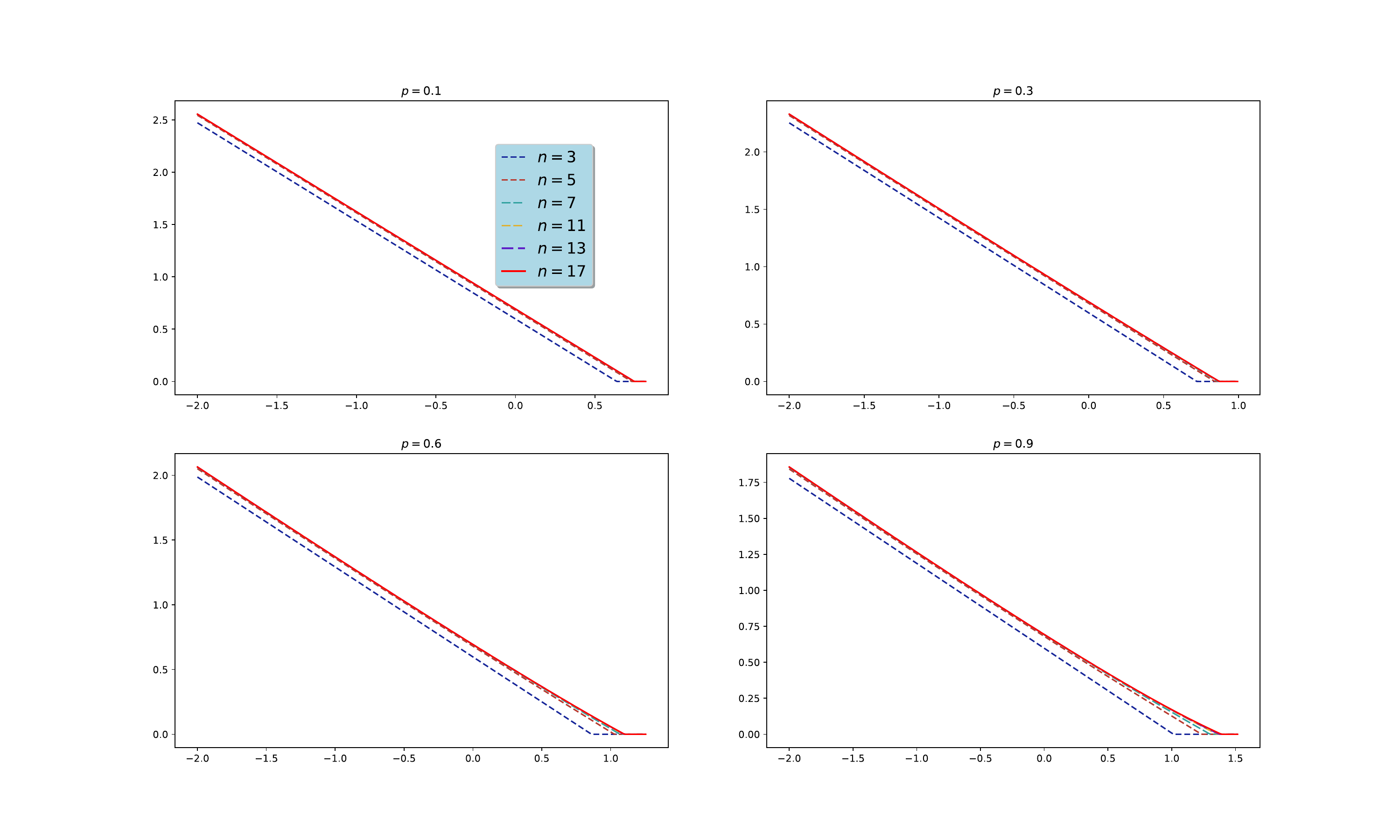}
\caption{The sequence $\{p_n(\kappa)\,|\,n=3,5,7,11,13,17\}$ for some values of
the exponent~$p$.}
\label{Fig:Convergence}
\end{figure}
\begin{figure}
\includegraphics[width=1.1\textheight,angle=90]{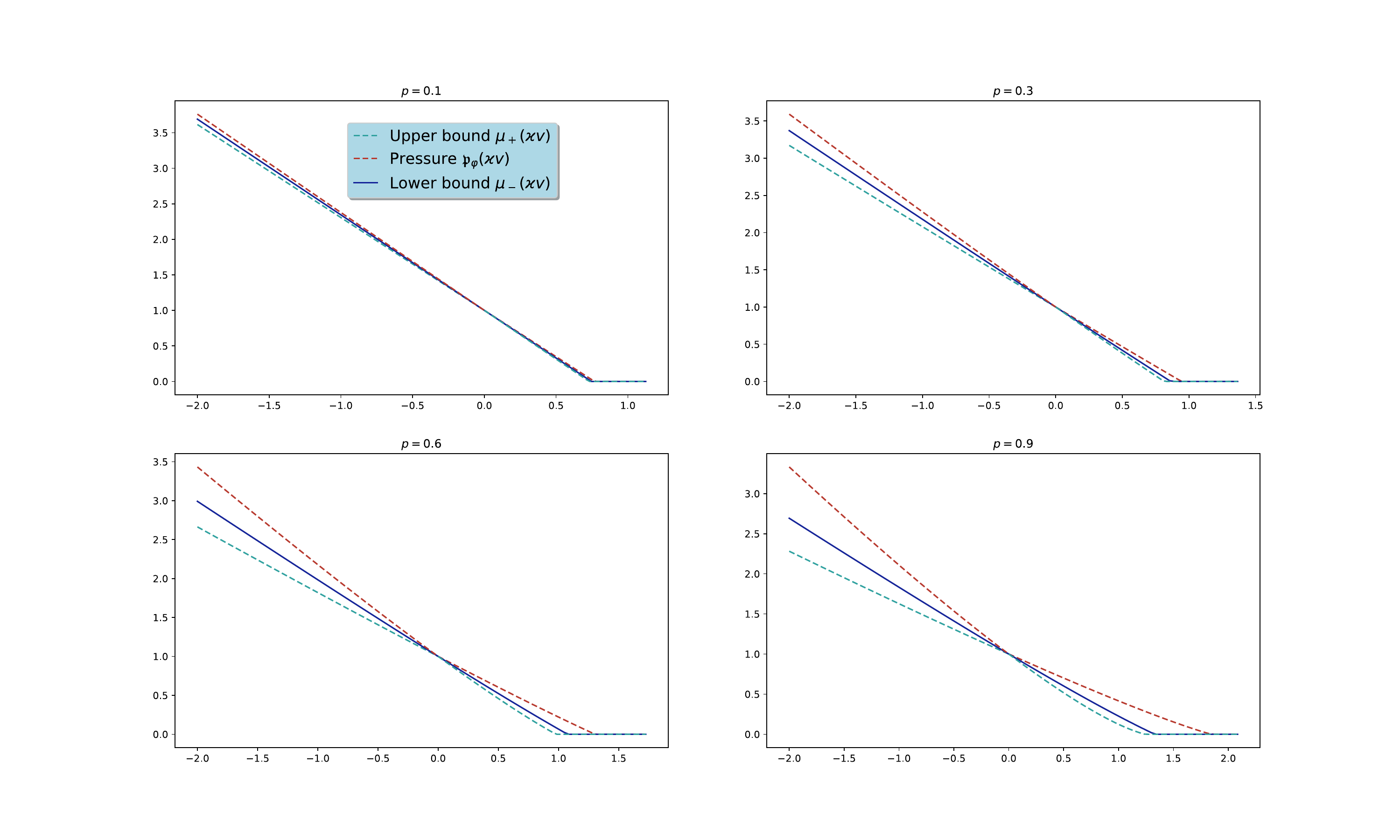}
\caption{The approximation of the pressure $\fp_\varphi(\kappa v)$ from
$p_{23}(\kappa)$ (solid) and the lower/upper bounds $\mu_\pm(\kappa)$ (dashed)
for some values of the exponent $p$.}
\label{Fig:Consistency}
\end{figure}

\begin{figure}
\begin{center}
\includegraphics[width=0.8\textwidth]{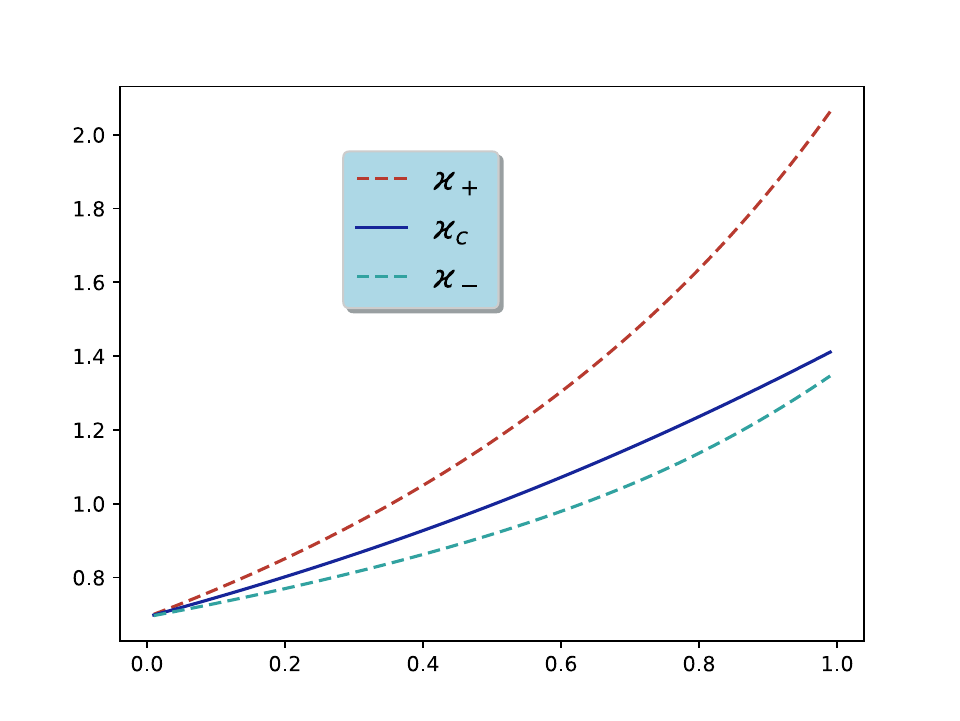}
\end{center}
\caption{The critical point $\kappa_c(p)$ computed from $p_{23}(\kappa)$ (solid) and the lower/upper bounds $\kappa_\pm(p)$ (dashed).}
\label{Fig:kappacritical}
\end{figure}

\begin{figure}
\begin{center}
\includegraphics[trim=5.0cm 1cm 0.0cm 1cm,width=1.3\textwidth]{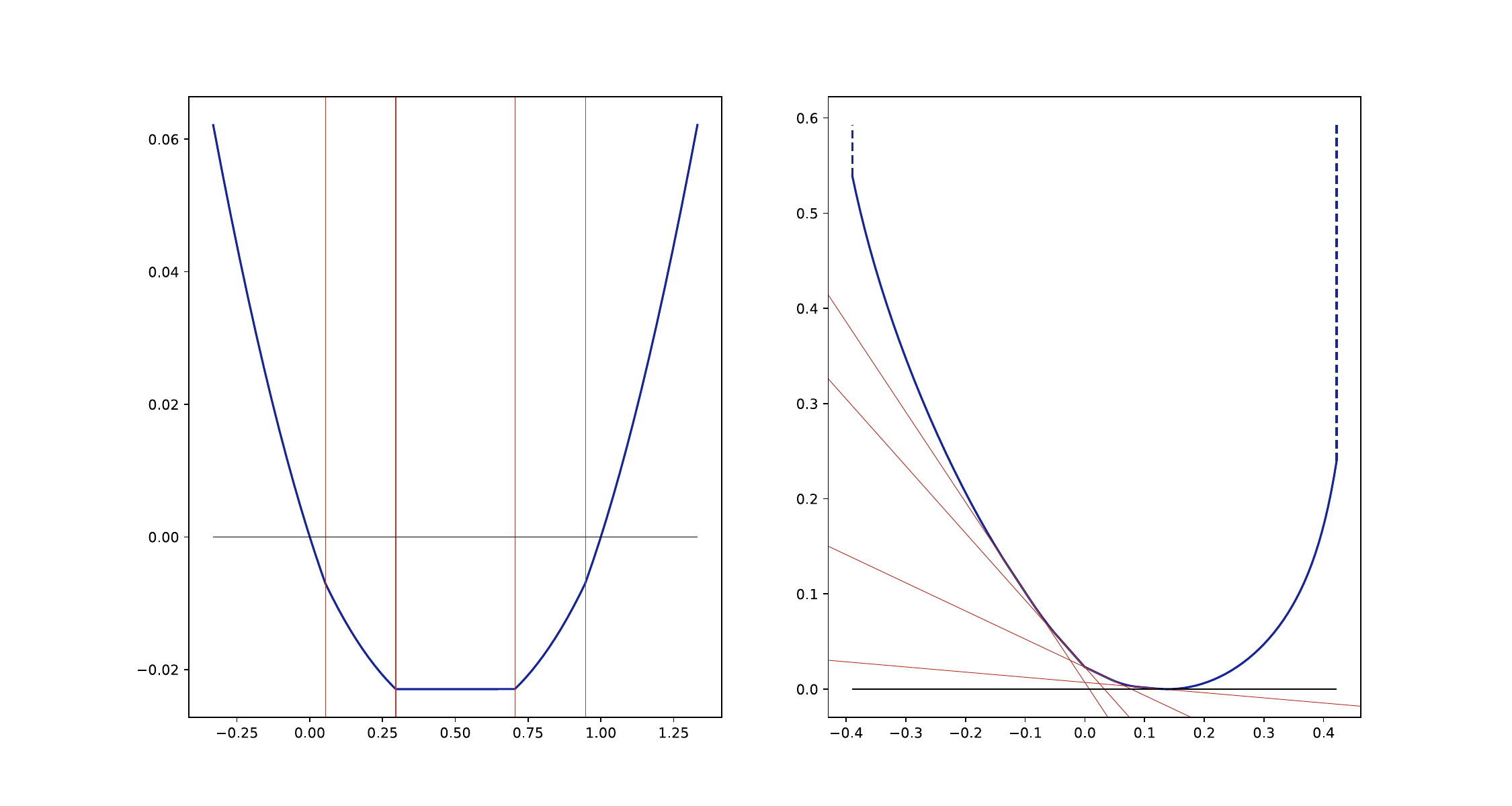}
\hspace{3cm}
\end{center}
\caption{The entropic pressure and the rate for $(p_1,p_2)=(0.9,0.1)$ (see main text for details).}
\label{Fig:EntropicPressure}
\end{figure}


\chapter*{Epilogue}

\addcontentsline{toc}{chapter}{Epilogue}

The panorama of FRs and FTs presented in this  short monograph arose through the
attempts of mathematical physicists to unravel the general structure behind the
respective  numerical and theoretical physics discoveries of the 1990s, and to
test the proposed structure on various  models of physical interest. The
emerging mathematical problems are in the field of  Large Deviation Theory. They
are  typically difficult and poorly understood. This should not come as a
surprise, since the rigorous proofs of FRs and FTs amount to the justification
of the Second Law in a given setting.

Boltzmann's ergodic hypothesis provides a historical parallel to our topic.
Roughly speaking, it states that the time and ensemble averages are the same,
and emerged in Boltzmann's attempts to understand the foundations of statistical
mechanics. It has led to the birth of ergodic theory as a mathematical
discipline. FRs and FTs deal with fluctuations that accompany  time-reversal
symmetry breaking  in non-equilibrium statistical mechanics.  Unlike in the case
of Boltzmann's ergodic hypothesis, the mathematical discipline to which they
relate, the theory of large deviations, was well-developed in early 1990s at the
time of the first discoveries, and has provided the natural structural setting
for their study. It remains to be seen whether the novel mathematical problems
brought by FRs and FTs will lead to  major new mathematical developments.

The follow-up monograph~\cite{Cuneo2025x} is primarily addressed to a  pure
mathematics audience. Just like ergodicity, FRs and FTs can be formulated
mathematically in great generality, and studied without any reference to
underlying physics. However, the existing Large Deviation Theory technology is
far from sufficient to address FRs and FTs in most models of physical interest.
The monograph~\cite{Cuneo2025x}  is focused  on the state-of-the-art
mathematical  techniques related to this study.

The two monographs are complementary to each other and can be read
independently. Neither is encyclopedic nor even a comprehensive review of the
subject, and perhaps at this point such a treatise cannot be written
anymore.\footnote{Among important missing topics, we mention the derivation  of  the
linear response theory (Green--Kubo formulas, Onsager reciprocity relation,
Fluctuation--Dissipation theorem) from FRs and FTs, see~\cite{Jaksic2011} and
references therein, and a discussion of  extensions of FRs and FTs to the
quantum domain,
see~\cite{Benoist2024e,Benoist2024f,Benoist2025a,Benoist2025b,Jaksic2010b} and
references therein.} We do hope, though, that they may provide a blueprint for
future studies of FRs and FTs by both mathematical physicists and
mathematicians.

\printbibliography[heading=bibintoc,title={References}]

\backmatter

\Extrachap{Index of Notations}

\subsection*{Abbreviations}

\runinhead{FR} fluctuation relation
\runinhead{FT} fluctuation theorem
\runinhead{LDP} large deviation principle
\runinhead{NESS} nonequilibrium steady state
\runinhead{SDE} stochastic differential equation

\subsection*{Assumptions}

\runinhead{(C)}  $\varphi$ is continuous, page~\pageref{C}
\runinhead{(H)} $\varphi$ is a homeomorphism, page~\pageref{H}
\runinhead{(USCE)}  upper semicontinuity of entropy, page~\pageref{usce}
\runinhead{(S)} specification property, page~\pageref{s}
\runinhead{(WPS)}  weak periodic specification property, page~\pageref{wps}
\runinhead{(ED)} the set $\cE_\varphi(M)$ is entropy-dense in $\cP_\varphi(M)$, page~\pageref{ed}
\runinhead{(PAP)}  periodic approximation of pressure, page~\pageref{pap}
\runinhead{(Sym)} reversal hypothesis, page~\pageref{reversal}

\subsection*{Sets}

\subruninhead{$\NN$} the set of natural numbers
\subruninhead{$\ZZ$} the set of integers
\subruninhead{$\llbracket a,b\rrbracket$} the set of integers in $[a,b]$
\subruninhead{$\ZZ_+$} the set of non-negative integers
\subruninhead{$\RR$} the set of real numbers
\subruninhead{$\RR_+$} the set of non-negative real numbers
\subruninhead{$\CC$} the set of complex numbers
\subruninhead{$C(X)$} space of continuous functions $f:X\to\RR$
\subruninhead{$C_b(X)$} subspace of bounded $f\in C(X)$, with the supremum norm~$\|\cdot\|$
\subruninhead{$\cP(X)$} set of Borel probability measures on $X$
\subruninhead{$\cP_\varphi(X)$} subset of $\varphi$-invariant elements of $\cP(X)$ (for $\varphi:X\to X$)
\subruninhead{$\cE_\varphi(X)$} subset of ergodic elements of $\cP_\varphi(X)$
\subruninhead{$(\Omega_\lambda,\fF_\lambda)_{\lambda\in\cL}$} a directed family of measurable spaces
\subruninhead{$\Fix(\varphi)$} set of fixed points of the map $\varphi$

\subsection*{Maps}

\subruninhead{$\PP,\wP,\QQ,\widehat\QQ$} probability measures
\subruninhead{$(\PP_\lambda,\wP_\lambda)$} a pair of probability measures on $(\Omega_\lambda,\fF_\lambda)$,
\subruninhead{$S(\PP)$} Shannon entropy of $\PP$, Equ.~\eqref{equ:ShannonEnt}
\subruninhead{$\Ent(\PP\,|\,\QQ)$} relative entropy, Equ.~\eqref{equ:EntDef}
\subruninhead{$\Ent_\alpha(\PP\,|\,\QQ)$} Rényi relative entropy, Equ.~\eqref{equ:EntalphaDef}
\subruninhead{$\sigma_\lambda$} entropy production observable, Equ.~\eqref{Eq:sigma_lambda}
\subruninhead{$\ep_\lambda$} entropy production, Equ.~\eqref{equ:epDef}
\subruninhead{$e_\lambda(\alpha)$} cumulant generating function of $-\sigma_\lambda$, Equ.~\eqref{Eq:e_alpha_def}
\subruninhead{$e(\alpha)$} entropic pressure, Equ.~\eqref{Eq:ealphaDef}
\subruninhead{$\sigma$} phase space contraction rate, Equ.~\eqref{equ:sigmadef}
\subruninhead{$\mu_t^x$} empirical measures, Equ.\eqref{equ:mutDef}

\subruninhead{$h_\mathrm{Top}(\varphi)$}  topological entropy of a map~$\varphi$
\subruninhead{$h_\varphi(\QQ)$} Kolmogorov--Sinai entropy of $\QQ$ w.r.t.\;$\varphi$

\subruninhead{$\fp_\varphi(v)$}  topological pressure of continuous potential $v$ w.r.t.\;$\varphi$, Equ.~\eqref{PressureDef}

\subsection*{Misc Symbols}

\subruninhead{$\partial^\pm$} left/right derivative, page~\pageref{eq:derlim}
\subruninhead{$\mathring{X}/\overline X$} the interior/closure of a set $X$
\subruninhead{$\QQ\ll\PP$} absolute continuity of $\QQ$ w.r.t.\;$\PP$
\subruninhead{$\frac{\d\QQ}{\d\PP}$} Radon--Nikodym derivative
\subruninhead{$\langle\,\argdot\,,\,\argdot\,\rangle$} Euclidean inner product on $\RR^d$
\subruninhead{$\Dom(f)$} essential domain of a convex function $f$
\subruninhead{$d(\,\argdot\,,\,\argdot\,)$} metric on a Polish space
\subruninhead{$\PP\otimes\QQ$} the product measure
\subruninhead{$\langle f,\PP\rangle$} the integral of $f$ w.r.t.\;$\PP$
\subruninhead{$|X|$} cardinality of the set $X$
\subruninhead{$1_X$} indicator function of the set $X$
\subruninhead{$\{\,\argdot\,,\,\argdot\,\}$} Poisson bracket

\printindex
\end{document}